\pdfoutput=1
\documentclass[acmsmall,nonacm,screen,appendix]{acmart}

\usepackage[capitalize]{cleveref}

\crefname{line}{Line}{Lines}
\crefname{rule}{rule}{rules}
\Crefname{rule}{Rule}{Rules}
\crefrangelabelformat{line}{#3#1#4--#5#2#6}

\usepackage{thmtools}
\usepackage{thm-restate}

\usepackage{booktabs}
\usepackage{microtype}
\usepackage{wrapfig}
\usepackage{natbib}
\usepackage[inline]{enumitem}
\usepackage{mathpartir}
\let\ifdraft\iffalse
\usepackage{style/notation}
\usepackage{style/appendix}

\definecolor{bluebell}{rgb}{0.25,0.5,0.75}

\usepackage{tcolorbox}
\tcbuselibrary{skins,breakable}
\ifappendix
  \newtcolorbox{result}{
    blanker,
    extras={interior engine=spartan},
    grow to left by=3.1mm,left*=0mm,
    grow to right by=2mm,right*=0mm,
    top=1mm,bottom=1mm,
    beforeafter skip balanced=0.1\baselineskip plus 2pt,
    borderline west={1.2mm}{0pt}{bluebell},
    breakable,
    colback=bluebell!10,
  }
\else
  \newtcolorbox{result}{
    blanker,
    extras={interior engine=spartan},
    grow to left by=2pt,left*=0mm,
    grow to right by=2pt,right*=0mm,
    top=1mm,bottom=1mm,
    beforeafter skip balanced=0.1\baselineskip plus 2pt,
    breakable,
    colback=bluebell!10,
  }
\fi

\additionalmaterial[The extended version~\cite{fullversion}]{\cite{fullversion}}

\ifappendix
  \setcopyright{none}
  \AtBeginDocument{
    \settopmatter{printccs=true,printacmref=true}
  }
  \makeatletter
  \def\@mkbibcitation{\bgroup \par\medskip\small\noindent{\bfseries This article is an extended version of:}\par\nobreak
      \noindent
      Jialu Bao, Emanuele D’Osualdo, and Azadeh Farzan. 2025.
      \textsc{Bluebell}: An Alliance of Relational Lifting and Independence for Probabilistic Reasoning.
      \textit{Proc. ACM Program. Lang.} 9, POPL, Article 58 (January 2025), 31~pages.
      \url{https://doi.org/10.1145/3704894}
    \par\egroup }
  \makeatother
\else
  \setcopyright{rightsretained}
  \acmJournal{PACMPL}
  \acmYear{2025}
  \acmVolume{9}
  \acmNumber{POPL}
  \acmArticle{58}
  \acmMonth{1}
  \acmDOI{10.1145/3704894}

  \received{2024-07-11}
  \received[accepted]{2024-11-07}
\fi

\newenvironment{mathfig}[1][]{\begin{figure}[tp]\adjustfigure[#1]}{\end{figure}}

\ifdraft

\else
\fi
\newcommand\relabel[1]{\label{app:#1}}
\ifdraft
  \renewcommand\relabel{\label}\fi

\begin{document}

\title{\thelogic: An Alliance of Relational Lifting and Independence for Probabilistic Reasoning}
\ifappendix
  \subtitle{(Extended Version)}
\fi

\author{Jialu Bao}
\orcid{0000-0002-2353-350X}
\email{jb965@cornell.edu}
\affiliation{\institution{Cornell University}
 \city{Ithaca}
 \country{NY, US}
}
\author{Emanuele D'Osualdo}
\email{emanuele.dosualdo@uni-konstanz.de}
\orcid{0000-0002-9179-5827}
\affiliation{\institution{MPI-SWS}
 \city{Saarland Informatics Campus}
 \country{Germany}
}
\affiliation{\institution{University of Konstanz}
\country{Germany}
}
\author{Azadeh Farzan}
\orcid{0000-0001-9005-2653}
\email{azadeh@cs.toronto.edu}
\affiliation{\institution{University of Toronto}
 \city{Toronto}
 \country{Canada}
}

\begin{abstract}
We present \thelogic, a program logic for reasoning about probabilistic programs where unary and relational styles of reasoning come together to create new reasoning tools.
Unary-style reasoning is very expressive and is powered by foundational mechanisms to reason about probabilistic behavior like \emph{independence} and \emph{conditioning}.
The relational style of reasoning, on the other hand, naturally shines when the properties of interest \emph{compare} the behavior of similar programs
(e.g. when proving differential privacy)
managing to avoid having to characterize the output distributions of the individual programs.
So far, the two styles of reasoning
have largely remained separate in the many program logics designed for the deductive verification of probabilistic programs.
In \thelogic, we unify these styles of reasoning through the introduction of a new modality called ``\supercond'' that can encode and illuminate the rich interaction between \emph{conditional independence} and \emph{relational liftings};
the two powerhouses from the two styles of reasoning.
\end{abstract}




\maketitle

\section{Introduction}\label{sec:intro}
Probabilistic programs are pervasive, appearing as
  machine learned subsystems,
  implementations of randomized algorithms,
  cryptographic protocols, and
  differentially private components,
among many more.
Ensuring reliability of such programs requires formal frameworks
in which correctness requirements can be formalized and verified
for such programs.
Similar to the history of classical program verification,
a lot of progress in this has come in the form of program logics
for probabilistic programs.
In the program logic literature,
there are two main styles of reasoning for probabilistic programs:
\emph{unary} and \emph{relational},
depending on the nature of the property of interest.
For instance, for differential privacy or cryptographic protocols correctness,
the property of interest is naturally expressible relationally.
In contrast, for example, specifying the expected cost of a
randomized algorithm is naturally done in the unary style.

Unary goals are triples $ \{P\}\ t\ \{Q\}$ where
$t$ is a probabilistic program,
$P$ and $Q$ are the pre- and post-conditions,
\ie \emph{predicates over distributions of stores}.
Such triples assert that
running~$t$ on an input store drawn from a distribution satisfying~$P$
results in a distribution over output stores which satisfies~$Q$.
Unary reasoning for probabilistic programs has made great strides,
producing logics for reasoning about
  expectations~\cite{kozen1983PDL,Morgan:1996,kaminski2016weakest,kaminski2019thesis,aguirre2021pre,Bartocci2022moment}
  and probabilistic independence~\cite{barthe2019probabilistic}.
  DIBI~\cite{bao2021bunched} and Lilac~\cite{lilac},
  which are the most recent, made a strong case for adding power to reason
about conditioning and independence in one logic.
Intuitively, conditioning on some random variable~\p{x}
allows to focus on the distribution of other variables
assuming~$\p{x}$ is some deterministic outcome~$v$;
two variables are (conditionally) independent if
knowledge of one does not give any knowledge of the other (under conditioning).
Lilac argued for (conditional) independence as the fundamental source of
modularity in the probabilistic setting.

Relational goals, in contrast, specify a desired relation between the output
distributions of \emph{two} programs~$t_1$ and~$t_2$,
for example, that~$t_1$ and~$t_2$ produce the same output distribution.
In principle, proving such goals can be approached in a unary style:
if the output distributions can be characterized
individually for each program,
then they can be compared after the fact.
More often than not, however,
precisely characterizing the output distribution of a program
can be extremely challenging.
Relational program logics like pRHL~\cite{barthe2009formal}
and its successors~\cite{barthe2009formal,barthe2015coupling,hsu2017probabilistic,gregersen2023asynchronous,AguirreBGGS19},
allow for a different and often more advantageous strategy.
The idea is to consider the two programs side by side,
and analyse their code as if executed in lockstep.
If the effect of a step on one side is ``matched'' by the corresponding
step on the other side,
then the overall outputs would also ``match''.
This way, the proof only shows that whatever is computed on one side,
will be matched by a computation on the other, without having to characterise
what the actual output is.

This idea of ``matching'' probabilistic steps is formalised in these logics
via the notion of \emph{couplings}~\cite{barthe2009formal,barthe2015coupling}.
The two programs can be conceptually considered to execute in
two ``parallel universes'', where they are oblivious to each other's randomness.
It is therefore sound to pretend their executions draw samples from
a common source of randomness (called a \emph{coupling})
in any way that eases the argument,
as long as the marginal distribution of the correlated runs in each universe
coincides with the original one.
For example, if both programs flip a fair coin,
one can force the outcomes of the coin flips to be the same
(or the opposite of each other,
depending on which serves the particular line of argument better).
Relating the samples in a specific way helps with
relating the distributions step by step, to support a relational goal.
Couplings, when applicable, permit relational logics to elegantly sidestep
the need to characterize the output distributions precisely.
As such, relational logics hit an ergonomic sweet spot in reasoning style
by restricting the form of the proofs that can be carried out.

Consider, for example, the code in
\cref{fig:between-code}.
The \code{BelowMax($x, S$)} procedure takes~$N$ samples
from a non-empty set~$S \subs \Int$,
according to an (arbitrary) distribution $\prob_S \of \Dist(S)$;
if any of the samples is larger than the given input~$x$
it declares~$x$ to be below the maximum of~$S$.
The \code{AboveMin($x, S$)} approximates in the same way
whether~$x$ is above the minimum of~$S$.
These are Monte Carlo style algorithms with a \emph{false bias}:
if the answer is false, they always correctly produce it,
and if the answer is true, then they correctly classify it
with a probability that depends on~$N$
(i.e., the number of samples).
It is a well-known fact that Monte Carlo style algorithms can be composed.
For example, \p{BETW\_SEQ} runs
\code{BelowMax($x, S$)} and \code{AboveMin($x, S$)}
to produce a \emph{false-biased} Monte Carlo algorithm
for approximately deciding whether~$x$ lies \emph{within} the extrema of~$S$.
Now, imagine a programmer proposed \p{BETW},
as a way of getting more mileage out of the number of samples drawn;
both procedures take~$2N$ samples,
but \p{BETW} performs more computation for each sample.
Such optimisations are not really concerned about
what the precise output distribution of each code is,
but rather that a \emph{true} answer is produced
with higher probability by \p{BETW};
in other words, its \emph{stochastic dominance} over \p{BETW\_SEQ}.

A unary program logic has only one way of reasoning
about this type of stochastic-dominance:
it has to analyze each code in isolation,
characterize its output distribution,
and finally assert/prove that one dominates the other.
In contrast, there is a natural \emph{relational strategy}
for proving this goal:
we can match the~$N$ samples of \p{BelowMax}
with~$N$ of the samples of \p{BETW}, and the~$N$ samples of \p{AboveMin}
with the remaining samples of \p{BETW} in lockstep,
and for each of these aligned steps,
\p{BETW} has more chances of turning \p{l} and \p{r} to~1
(and they can only increase).

\begin{figure*}
  \adjustfigure[\small]\lstset{
belowskip=-5pt,
    gobble=2,
  }\setlength\tabcolsep{0pt}\begin{tabular*}{\textwidth}{
      @{\extracolsep{\fill}}
      *{4}{p{\textwidth/4}}@{}
    }
\begin{sourcecode*}
  def BelowMax($x$,$S$):
    repeat $N$:
      q :~ $\prob_S$
      r := r || q >= $x$
  \end{sourcecode*}
&
  \begin{sourcecode*}
  def AboveMin($x$,$S$):
    repeat $N$:
      p :~ $\prob_S$
      l := l || p <= $x$
  \end{sourcecode*}
&
  \begin{sourcecode*}
  def BETW_SEQ($x$, $S$):
    BelowMax($x$,$S$);
    AboveMin($x$,$S$);
    d := r && l
  \end{sourcecode*}
  &
  \begin{sourcecode*}
  def BETW($x$,$S$):
    repeat $2 N$:
      s :~ $\prob_S$
      l := l || s <= $x$
      r := r || s >= $x$
    d := r && l
  \end{sourcecode*}
  \end{tabular*}
  \caption{A stochastic dominance example: composing Monte Carlo algorithms two different ways.
    \ifappendix All variables are initially~0,
      $N\in\Nat$ is some fixed constant,
      $S$ is a set of integers, and
      $\prob_S$ is some given distribution of elements of~$S$.
      Both \p{BETW\_SEQ($x$,$S$)} and \p{BETW($x$,$S$)}
      approximately decide whether~$x$ lies within the extrema of~$S$.\else $N\in\Nat$ is some fixed constant, and all variables are initially~0.\fi }
  \label{fig:between-code}
\end{figure*}

Unary logics can express information about distributions
with arbitrary levels of precision;
yet none can encode the simple natural proof idea outlined above.
This suggests an opportunity:
Bring native relational reasoning support to an expressive unary logic,
like Lilac.
Such a logic can be based on assertions over distributions, offering the precision and expressiveness of unary logics while natively supporting relational reasoning.
As a result, it would be able to encode the argument outlined above
at the appropriate level of abstraction.
To explore this idea,
let us outline the central principle that we would need
to import from relational reasoning:
\emph{relational lifting}.

Relational logics use variants of judgments of the form
$ \{R_1\} \m[\I1: t_1, \I2: t_2] \{R_2\}$, where
$t_1$ and $t_2$ are the two programs we are comparing and
$R_1$ and $R_2$ are the relational pre- and post-conditions.
$R_1$ and $R_2$ differ from unary assertions in two ways:
first they are used to relate two distributions
instead of constraining a single one.
Second, they are predicates over \emph{pairs of stores},
and not of distributions directly.
Let us call predicates of this type ``deterministic relations''.
If~$R$ was a deterministic predicate over a single store,
requiring it to hold with probability~1 would naturally lift it
to a predicate $\sure{R}$ over distributions of stores.
When~$R$ is a deterministic relation between pairs of stores,
its \emph{relational lifting}~$\cpl{R}$ relates two distributions over stores
$\prob_1,\prob_2 \of \Dist(\Store)$,
if (1) there is a distribution over \emph{pairs} of stores
$\prob \of \Dist(\Store\times\Store)$
such that its marginal distributions on the first and second store
coincide with $\prob_1$ and $\prob_2$ respectively,
(\ie $\prob$ is a \emph{coupling} of $\prob_1$ and $\prob_2$)
and (2) $\prob$ samples pairs of stores
satisfying the relation~$R$ with probability~1.
Such relational liftings can
encode a variety of useful relations between distributions.
For instance, let $R = (\Ip{x}{1} = \Ip{x}{2})$
relate stores~$s_{\I{1}}$ and~$s_{\I{2}}$
if they both assign the same value to \p{x};
then the lifting $\cpl{R}$ holds for two distributions
$\prob_{\I{1}},\prob_{\I{2}} \of \Dist(\Store)$
if and only if they induce the same distributions in \p{x}.
Similarly, the lifting $\cpl{\Ip{x}{1} \leq \Ip{x}{2}}$
encodes stochastic dominance of the distribution of \p{x}
in $\prob_{\I{2}}$ over the one in $\prob_{\I{1}}$.

Relational proofs built out of relational lifting
then work by using deterministic relations as assertions,
and showing that a suitably coupled lockstep execution
of the two programs satisfies each assertion with probability~1.
To bring relational reasoning to unary logics,
we want to preserve the fact that assertions are over distributions,
and yet support relational lifting as the key abstraction
to do relational reasoning.
This new logic can equally be viewed as a relational logic
with assertions over pairs of distributions
(rather than over pairs of stores).
With such a view,
seeing relational lifting as one among many constructs
to build assertions seems like a very natural, yet completely unexplored, idea.

What is entirely non-obvious is whether relational lifting
works well as an abstraction together with the other key ``unary'' constructs,
such as independence and conditioning,
that are the source of expressive power of unary logics.
\ifappendix The properties of couplings already provide a source of examples
for one way in which unary and relational facts might interact.
For example, from a well-known property of couplings,
\else
For example, from the properties of couplings,
\fi we know that establishing $\cpl{\Ip{x}{1} = \Ip{x}{2}}$ implies that
$\Ip{x}{1}$ and $\Ip{x}{2}$ are identically distributed;
this can be expressed as an entailment:
\begin{equation}
  \cpl{\Ip{x}{1} = \Ip{x}{2}}
  \lequiv
  \E \prob.
    \distAs{\Ip{x}{1}}{\prob}
    \land
    \distAs{\Ip{x}{2}}{\prob}
  \label{eq:rl-id-conv}
\end{equation}
The equivalence says that establishing a coupling that can
(almost surely) equate the values of $\Ip{x}{1}$ and $\Ip{x}{2}$,
amounts to establishing that the two variables are identically distributed.
The equivalence can be seen as a way to interface ``unary'' facts
and relational liftings.

Probability theory is full of lemmas of this sort and it is clearly undesirable to admit any lemma that is needed for one proof or another as an axiom in the program logic.
Can we have a logic in which they are derivable without having to abandon its nice abstractions?
Can the two styles be interoperable at the level of the logic?
In this paper, we provide an affirmative answer to this question by proposing a new program logic called \thelogic.

We propose that relational lifting does in fact have non-trivial and useful
interactions with independence and conditioning.
Remarkably, \thelogic's development is unlocked by
a more fundamental observation:
once an appropriate notion of conditioning is defined in \thelogic,
relational lifting and its laws can be derived from this foundational
conditioning construct.

The key idea is a new characterization of relational lifting as a form of
conditioning:
whilst relational lifting is usually seen as a way to induce a relation over distributions from a deterministic relation,
\thelogic\ sees it as a way to go from
a tuple of distributions to a relation between the values of some conditioned variables.
More precisely:
\begin{itemize}
  \item
    We introduce a new \emph{\supercond} modality in \thelogic\
    which can be seen, in hindsight,
    as a natural way to condition when dealing with tuples of distributions.
  \item
    We show that \supercond\ can represent uniformly
    both, conditioning \emph{à la} Lilac,
    and relational lifting as derived notions in \thelogic.
  \item
    We prove a rich set of general rules for \supercond,
    from which we can derive both known and novel proof principles
    for conditioning and for relational liftings in \thelogic.
\end{itemize}

Interestingly, our \supercond\ modality can replicate the same reasoning
style of Lilac's modality, while having a different semantics
(and validating an overlapping but different set of rules as a result).
This deviation in the semantics is a stepping stone for obtaining an
adequate generalization to the \pre n-ary case (unifying unary and binary as special cases).
We expand on these ideas in \cref{sec:overview}, using a running example.
More importantly, our \supercond\  enables \thelogic\ to
\begin{itemize}
\item Accommodate unary and relational reasoning
  in a fundamentally interoperable way: For instance, we showcase the interaction between lifting and conditioning in the derivation of our running example in \cref{sec:overview}.
\item Illuminate known reasoning principles: For instance, we discuss how \thelogic\ emulates pRHL-style reasoning
  in~\cref{sec:ex:prhl-style}.
\item Propose new tools to build program proofs: For instance, we discuss out-of-order coupling of samples through \ref{rule:seq-swap} in \cref{sec:overview:obox}.\item Enable the exploration of the theory of high-level constructs
   like relational lifting (via the laws of independence and \supercond): For instance, novel broadly useful rules \ref{rule:rl-merge} and \ref{rule:rl-convex}, discussed in \cref{sec:overview} can be derived within \thelogic.
\end{itemize}
All proofs, omitted details, and additional examples can be found in~\ifappendix the Appendix.\else\cite{fullversion}.\fi  \ifappendix\pagebreak\fi
\section{A Tour of \thelogic}
\label{sec:overview}

In this section we will highlight the main key ideas behind
\thelogic, using a running example.

\subsection{The Alliance}
\label{sec:overview:intro}

\begin{wrapfigure}[7]{R}{20ex}\begin{sourcecode*}[linewidth=.9\linewidth,xleftmargin=.1\linewidth,
    belowskip=-.5em,aboveskip=-1em,
    gobble=2
  ]
  def encrypt():
    k :~ Ber(1/2)
    m :~ Ber($p$)
    c := k xor m
  \end{sourcecode*}\caption{One time pad.}
  \label{fig:onetime-code}
\end{wrapfigure}
We work with a first-order imperative probabilistic programming language
consisting of programs~$t\in\Term$ that mutate a variable store~$\store\in\Store$
(\ie a finite map from variable names~$\Var$ to values~$\Val$).
We only consider discrete distributions
(but with possibly infinite support).
In \cref{fig:onetime-code} we show a simple example adapted from~\cite{barthe2019probabilistic}:
the \p{encrypt} procedure uses a fair coin flip to generate an encryption key
\p{k}, generates a plaintext message in boolean variable \p{m}
(using a coin flip with some bias~$p$)
and produces the ciphertext \p{c} by XORing the key and the message.
A desired property of the program is that the ciphertext should be
indistinguishable from an unbiased coin flip; as a binary triple:
\begin{equation}
  \{\True\}
  \m[
    \I1: \code{encrypt()},
    \I2: \code{c:~Ber(1/2)}
  ]
  \{
    \cpl{ \Ip{c}{1}=\Ip{c}{2} }
  \}
  \label{ex:xor:goal}
\end{equation}
where we use the $\at{i}$ notation to indicate the index of the program store that an expression references.
In \cref{sec:ex:one-time-pad}, we discuss a unary-style proof of this goal in \thelogic. Here, we focus on a relational argument, as a running example. The natural (relational) argument goes as follows.
When computing the final XOR,
  if $\p{m}=0$ then \code{c=k},
  and if $\p{m}=1$ then \code{c=!k}.
Since both $\Ip{k}{1}$ and $\Ip{c}{2}$ are distributed as \emph{unbiased} coins,
they can be coupled either so that they get the same value,
or so that they get opposite values (the marginals are the same).
One or the other coupling must be established
\emph{conditionally} on~$\Ip{m}{1}$, to formalize this argument.
Doing so in pRHL faces the problem that the logic is too rigid to permit one to
condition on $\Ip{m}{1}$ before $\Ip{k}{1}$ is sampled; rather it forces one to establish a coupling of $\Ip{k}{1}$ and $\Ip{c}{2}$ right when the two samplings happen.
This rigidity is a well-known limitation of relational logics,
which we overcome by ``immersing'' relational lifting
in a logic with assertions on distributions.
Recent work~\cite{gregersen2023asynchronous}
proposed workarounds based on ghost code for pre-sampling
(see~\cref{sec:relwork}).
We present a different solution based on framing, to the generic problem of out-of-order coupling, in~\cref{sec:overview:obox}.

Unconstrained by the pRHL assumption that every assertion has to be represented as a relational lifting, we observe three crucial components in the proof idea:
\begin{enumerate}
\item
  \emph{Probabilistic independence}
  between the sampling of $\Ip{k}{1}$ and $\Ip{m}{1}$,
  which makes conditioning on $\Ip{m}{1}$ preserve the
  distribution of $\Ip{k}{1}$;
\item
  \emph{Conditioning} to perform case analysis
  on the possible values of $\Ip{m}{1}$;
\item
  \emph{Relational lifting}
  to represent the existence of couplings imposing the desired
  correlation between $\Ip{k}{1}$ and $\Ip{c}{2}$.
\end{enumerate}
Unary logics like
  Probabilistic Separation Logics~(PSL)
  \cite{barthe2019probabilistic} and
  Lilac
explored how probabilistic independence
can be represented as \emph{separating conjunction},
obtaining remarkably expressive and elegant reasoning principles.
In \thelogic, we import the notion of independence from Lilac:
\thelogic's assertions are interpreted over
tuples of probability spaces~$\m{\psp}$,
and $ Q_1 * Q_2 $  holds on $\m{\psp}$ if
$\m{\psp}(i)$ can be seen as the \emph{independent product}
of $ \m{\psp}_1(i) $ and $\m{\psp}_2(i)$,
for each~$i$,
such that the tuples $\m{\psp}_1$ and $\m{\psp}_2$ satisfy $Q_1$ and $Q_2$
respectively.
This means that
  $\distAs{\Ip{x}{1}}{\prob} * \distAs{\Ip{y}{1}}{\prob}$
states that $\Ip{x}{1}$ and $\Ip{y}{1}$ are independent and identically distributed,
as opposed to
  $\distAs{\Ip{x}{1}}{\prob} \land \distAs{\Ip{y}{1}}{\prob}$
which merely declares the two variables as identically distributed
(but possibly correlated).
For a unary predicate over stores~$R$
we write $\sure{R\at{i}}$ to mean that
the predicate~$R$ holds with probability~1
in the distribution at index~$i$.

With these tools it is easy to get through the first two assignments
of \p{encrypt} and the one on component $\I2$ and get to a state
satisfying the assertion
\begin{equation}
  P =
  \distAs{\Ip{k}{1}}{\Ber{\onehalf}}
  *
  \distAs{\Ip{m}{1}}{\Ber{p}}
  *
  \distAs{\Ip{c}{2}}{\Ber{\onehalf}}
  \label{ex:xor:start}
\end{equation}

The next ingredient we need is conditioning.
We introduce a new modality $\CMod{\prob}$ for conditioning,
in the spirit of Lilac.
The modality takes the form
$
  \CC\prob v.K(v)
$
where~$\prob$ is a distribution, $v$ is a logical variable bound
by the modality (ranging over the support of~$\prob$),
and $K(v)$ is a family of assertions indexed by~$v$.
Before exploring the general meaning of the modality
(which we do in \cref{sec:overview:supercond}),
let us illustrate how we would represent conditioning on $\Ip{m}{1}$
in our running example.
We know $\Ip{m}{1}$ is distributed as~$\Ber{p}$;
conditioning on $\Ip{m}{1}$ in \thelogic\ would give us
an assertion of the form
$\CC{\Ber{p}} v. K(v)$
where~$v$ ranges over~$\set{0,1}$,
and the assertions
$K(0) = \sure{\Ip{m}{1}=0} * P_0 $
and
$K(1) = \sure{\Ip{m}{1}=1} * P_1 $
(for some $P_0,P_1$)
represent the properties of the distribution conditioned on
$ \Ip{m}{1}=0 $ and $ \Ip{m}{1}=1 $, respectively.
Intuitively, the assertion states that the current distribution
satisfies $K(v)$ when conditioned on $\Ip{m}{1}=v$.
Semantically, a distribution $\prob_0$ satisfies the assertion
$\CC{\Ber{p}} v. K(v)$
if there exists distributions $ \krnl_0, \krnl_1$ such that
$\krnl_0$ satisfies $ K(0)$,
$\krnl_1$ satisfies $ K(1)$, and
$\prob_0$ is the convex combination
$
  \prob_0 = p \cdot \krnl_1 + (1-p) \cdot \krnl_0.
$
When $K(v)$ constrains, as in our case, the value of a variable
(here $\Ip{m}{1}$) to be~$v$,
the only $\krnl_0$ and $\krnl_1$ satisfying the above constraints
are the distribution $\prob_0$ conditioned on the variable being 0 and 1 respectively.

Combining independence and conditioning with the third ingredient,
relational lifting~$\cpl{R}$ (we discuss more about how to define it
in~\cref{sec:overview:supercond}), we can now express with an assertion the desired
conditional coupling we foreshadowed in the beginning:
\begin{equation}
  Q =
  \CC{\Ber{p}} v.
    \left(
      \sure{\Ip{m}{1}=v}
      *
      \begin{cases}
        \cpl{ \Ip{k}{1} = \Ip{c}{2} }     \CASE v=0 \\
        \cpl{ \Ip{k}{1} = \neg\Ip{c}{2} } \CASE v=1
      \end{cases}
    \right)
  \label{ex:xor:ccouple}
\end{equation}
The idea is that we first condition on $\Ip{m}{1}$
so that we can see it as the deterministic value~$v$,
and then we couple $\Ip{k}{1}$ and $\Ip{c}{2}$ differently
depending on~$v$.

To carry out the proof idea formally, we are left with two subgoals.
The first is to formally prove the entailment
$
  P \proves Q.
$
Then, it is possible to prove that after
the final assignment to \p{c} at index \I{1},
the program is in a state that satisfies
$Q * \sure{\Ip{c}{1} = \Ip{k}{1} \xor \Ip{m}{1}}$.
To finish the proof we would need to prove that
$
  Q * \sure{\Ip{c}{1} = \Ip{k}{1} \xor \Ip{m}{1}}
  \proves
  \cpl{ \Ip{c}{1} = \Ip{c}{2} }.
$
These missing steps need laws
governing the interaction among independence, conditioning and relational lifting in this \pre n-ary setting.

\begin{result}
A crucial observation of \thelogic{}
is that, by choosing an appropriate definition for the
\supercond\ modality~$\CMod{\prob}$,
relational lifting can be encoded as a form of conditioning.
Consequently, the laws governing relational lifting can be derived
from the more primitive laws for \supercond.
Moreover, the interactions between relational lifting and independence can be derived through the primitive laws for the interactions between \supercond\ and independence.
\end{result}

\subsection{\SuperCond\ and Relational Lifting}
\label{sec:overview:supercond}

Now we introduce the \emph{joint} conditioning modality and its general \pre n-ary version.
Given
$\prob \of \Dist(A)$ and
a function $\krnl \from A \to \Dist(B)$ (called a \emph{Markov kernel}),
define the distribution $\bind(\prob, \krnl) \of \Dist(B)$ as
$
  \bind(\prob, \krnl) =
    \fun b.\Sum*_{a\in A} \prob(a) \cdot \krnl(a)(b)
$ and $
  \return(v) = \dirac{v}.
$
The $\bind$ operation represents a convex combination with coefficients in
$\prob$, while $\dirac{v}$ is the Dirac distribution, which assigns probability~1
to the outcome~$v$.
These operations form a monad with the distribution functor $\Dist(\hole)$,
a special case of the Giry monad~\cite{giry1982categorical}.
Given a distribution $\prob \of \Dist(A)$,
and a predicate~$K(a)$ over pairs of distributions
parametrized by values~$a\in A$,
we define
$
  \CMod{\prob} a\st K(a)
$
to hold on some $(\prob_1,\prob_2)$ if
\begin{align*}
  \exists \krnl_1,\krnl_2 \st
  \forall i \in \set{1,2} \st
    \prob_i = \bind(\prob, \krnl_i)
  \land
  \forall a \in \psupp(\prob) \st
    K(a) \text{ holds on }
    (\krnl_1(a), \krnl_2(a))
\end{align*}
Namely, we decompose the pair $(\prob_1,\prob_2)$ component wise
into convex compositions of $\prob$ and some kernels~$\krnl_1,\krnl_2$,
one per component.
Then for every $a$ with non-zero probability in~$\prob$, we require the predicate~$K(a)$ to hold for the pair of distributions
$ (\krnl_1(a), \krnl_2(a)) $.
The definition naturally extends to any number of indices.

One powerful application of the \supercond modality is to encode relational
liftings $\cpl{R}$.
Imagine we want to express $\cpl{ \Ip{k}{1} = \Ip{c}{2} }$.
It suffices to assert that there exists some distribution
$\prob \of \Dist(\Val\times\Val)$ over pairs of values such that
$
  \CC\prob (v_1,v_2).\bigl(
    \sure{\Ip{k}{1} = v_1} \land
    \sure{\Ip{c}{2} = v_2} \land
    \pure{v_1=v_2}
  \bigr),
$
where $\pure{\phi}$ denote the embedding of a pure fact~$\phi$ (\ie a meta-level statement) into the logic. Let us digest the formula step-by-step.
$\CC\prob (v_1,v_2). \bigl(\sure{\Ip{k}{1} = v_1} \land \sure{\Ip{c}{2} = v_2}\bigr)$
conditions the distribution at index \I{1} on $\Ip{k}{1} = v_1$ and
conditions the distribution at index \I{2} on $\Ip{c}{2} = v_2$;
such simultaneous conditioning is possible only if $\prob$ projected
to its first index, $\prob \circ \inv{\proj_1}$, is the marginal distribution
of $\Ip{k}{1}$ and $\prob$ projected to its second index, $\prob \circ
\inv{\proj_2}$, is the marginal distribution of $\Ip{c}{2}$.
Thus, $\prob$ is a joint distribution -- a.k.a. coupling -- of the marginal distributions of $\Ip{k}{1}$ and $\Ip{c}{2}$.
The full assertion $\CC\prob (v_1,v_2).(
    \sure{\Ip{k}{1} = v_1} \land
    \sure{\Ip{c}{2} = v_2} \land
    \pure{v_1=v_2})$
ensures that the coupling $\prob$ has non-zero probabilities only on pairs
$(v_1, v_2)$ where $v_1 = v_2$, and this is exactly the requirement of the relational lifting $\cpl{\Ip{k}{1} = \Ip{c}{2}}$.

The idea generalizes to arbitrary relation lifting $\cpl{R}$,
which are encoded using assertions of the form
$
  \E\prob.\CC{\prob} (\vec{v}_1,\vec{v}_2).\bigl(
    \sure{\vec{\p{x}}\at{\I1}=\vec{v}_1}
    \land
    \sure{\vec{\p{x}}\at{\I2}=\vec{v}_2}
    \land
    \pure{R(\vec{v}_1,\vec{v}_2)}
  \bigr).
$
The encoding hinges on the crucial decision in the design of the
\supercond modality of using the same distribution~$\prob$ to
decompose the distributions at all indices.
Then, the assertion inside the conditioning can force $\prob$ to be
a joint distribution of (marginal) distributions of program states at
different indices; and it can further force $\prob$ to have non-zero
probability only on pairs of program states that satisfy the relation~$R$.

\begin{result}
The remarkable fact is that our formulation of
relational lifting directly explains:
\begin{enumerate}
\item How the relational lifting can be \emph{established}:
  that is, by providing some joint distribution~$\prob$ for
  $\Ip{k}{1}$ and $\Ip{c}{2}$ ensuring~$R$ (the relation being lifted)
  holds for their joint outcomes;
  and
\item
  How the relational lifting can be \emph{used} in entailments:
  that is, it guarantees that if one conditions on the store,
  $R$~holds between the (now deterministic) variables.
\end{enumerate}
\end{result}

To make these definitions and connections come to fruition we need to
study which laws are supported by the \supercond\ modality
and whether they are expressive enough to reason about distributions pairs
without having to drop down to the level of semantics.

\subsection{The Laws of \SuperCond}

We survey the key laws for \supercond\ in this section, and explore a vital consequence of defining both conditional independence and relational lifting based on the \supercond modality: the laws of both can be derived from a set of expressive laws about \supercond\ alone. To keep the exposition concrete,
we focus on a small subset of laws that are enough to prove the example
of \cref{sec:overview:intro}.
Let us focus first on proving:
\begin{equation}
  \distAs{\Ip{k}{1}}{\Ber{\onehalf}}
  *
  \distAs{\Ip{m}{1}}{\Ber{p}}
  *
  \distAs{\Ip{c}{2}}{\Ber{\onehalf}}
  \proves
  \CC{\Ber{p}} v.
    \left(
      \sure{\Ip{m}{1}=v}
      *
      \begin{cases}
        \cpl{ \Ip{k}{1} = \Ip{c}{2} }     \CASE v=0 \\
        \cpl{ \Ip{k}{1} = \neg\Ip{c}{2} } \CASE v=1
      \end{cases}
    \right)
  \label{ex:xor:entail1}
\end{equation}

\noindent
We need the following primitive laws of \supercond:
\begin{proofrules}
  \infer*[lab=c-unit-r]{}{
  \distAs{\aexpr\at{i}}{\mu}
  \lequiv
  \CC\prob v.\sure{\aexpr\at{i}=v}
}   \label{rule:c-unit-r}

  \infer*[lab=c-frame]{}{
  P * \CC\prob v.K(v)
  \proves
  \CC\prob v.(P * K(v))
}   \label{rule:c-frame}

  \infer*[lab=c-cons]{
  \forall v\st
  K_1(v) \proves K_2(v)
}{
  \CC\prob v.K_1(v)
  \proves
  \CC\prob v.K_2(v)
}   \label{rule:c-cons}
\end{proofrules}

\Cref{rule:c-unit-r} can convert back and forth from
ownership of an expression~$E$ at~$i$ distributed as~$\prob$,
and the conditioning on~$\prob$ that makes~$E$ look deterministic.
\Cref{rule:c-frame} allows to bring inside conditioning
any resource that is independent from it.
\Cref{rule:c-cons} simply allows to apply entailments inside \supercond.
We can use these laws to perform conditioning on~$\Ip{m}{1}$:
\begin{eqexplain}
&
  \distAs{\Ip{k}{1}}{\Ber{\onehalf}}
  *
  \distAs{\Ip{m}{1}}{\Ber{p}}
  *
  \distAs{\Ip{c}{2}}{\Ber{\onehalf}}
\whichproves
  \distAs{\Ip{k}{1}}{\Ber{\onehalf}}
  *
  (\CC{\Ber{p}} v.\sure{\Ip{m}{1}=v})
  *
  \distAs{\Ip{c}{2}}{\Ber{\onehalf}}
\byrule{c-unit-r}
\whichproves
  \CC{\Ber{p}} v.
  \left(
    \sure{\Ip{m}{1}=v}
    *
    \distAs{\Ip{k}{1}}{\Ber{\onehalf}}
    *
    \distAs{\Ip{c}{2}}{\Ber{\onehalf}}
  \right)
\byrule{c-frame}
\end{eqexplain}
Here we use \ref{rule:c-unit-r} to convert ownership of~$\Ip{m}{1}$
into its conditioned form.
Then we can bring the other independent variables inside the conditioning
with \ref{rule:c-frame}.
This derivation follows in spirit the way in which Lilac
introduces conditioning, thus inheriting its ergonomic elegance.
Our rules however differ from Lilac's in both form and substance.
First, Lilac's rule for introducing conditioning (called \textsc{C-Indep}),
requires converting ownership of a variable into conditioning, and bringing some independent resources inside conditioning, as a single monolithic step.
In \thelogic\ we accomplish this pattern as a combination of our
\ref{rule:c-unit-r} and \ref{rule:c-frame},
which are independently useful.
Specifically, \ref{rule:c-unit-r} is bidirectional,
which makes it useful to recover unconditional facts from conditional ones.
Furthermore, we recognize that \ref{rule:c-unit-r} is nothing but
a reflection of the right unit law of the monadic structure of distributions
(which we elaborate on in \cref{sec:logic}).
This connection prompted us to provide rules that reflect the remaining
monadic laws (left unit~\ref{rule:c-unit-l} and
associativity~\ref{rule:c-fuse}). It is noteworthy that these rules do not follow from Lilac's proofs:
our modality has a different semantics, and our rules seamlessly apply to
assertions of any arity.

To establish the conditional relational liftings of the entailment in \eqref{ex:xor:entail1},
\thelogic\ provides a way to introduce relational liftings from ownership of the distributions
of some variables:
\begin{proofrule}
  \infer*[lab=coupling]{
  \prob \circ \inv{\proj_1} = \prob_1
  \\
  \prob \circ \inv{\proj_2} = \prob_2
  \\
  \prob(R) = 1
}{
  \distAs{\p{x}_1\at{\I1}}{\prob_1} *
  \distAs{\p{x}_2\at{\I2}}{\prob_2}
  \proves
  \cpl{R(\p{x}_1\at{\I1}, \p{x}_2\at{\I2})}
}
   \label{rule:coupling}
\end{proofrule}
The side conditions of the rule ask the prover to provide a coupling
$\prob \of \Dist(\Val\times\Val)$
of $\prob_1 \of \Dist(\Val)$ and $\prob_2 \of \Dist(\Val)$,
which assigns probability~1 to a (binary) relation~$R$.
If $\p{x}_1\at{\I1}$ and $\p{x}_2\at{\I2}$ are distributed as $\prob_1$ and $\prob_2$, respectively, then the relational lifting of $R$ holds between them
(as witnessed by the existence of~$\prob$).
Note that for the rule to apply, the two variables need to live in distinct indices.
\begin{result}
Interestingly, \ref{rule:coupling} can be derived from the
encoding of relational lifting and the laws of \supercond.
\end{result}
Remarkably, although the rule mirrors the
step of coupling two samplings in a pRHL proof,
it does not apply to the code doing the sampling itself,
but to the assertions representing the effects of those samplings.
This allows us to delay the forming of coupling to until
all necessary information is available (here, the outcome of $\Ip{m}{1}$).
We can use \ref{rule:coupling} to prove both entailments:
{\begin{equation}
  \distAs{\Ip{k}{1}}{\Ber{\onehalf}} *
  \distAs{\Ip{c}{2}}{\Ber{\onehalf}}
  \proves
  \cpl{ \Ip{k}{1} = \Ip{c}{2} }
\text{\; and \; }
  \distAs{\Ip{k}{1}}{\Ber{\onehalf}} *
  \distAs{\Ip{c}{2}}{\Ber{\onehalf}}
  \proves
  \cpl{ \Ip{k}{1} = \neg\Ip{c}{2} }
  \label{ex:xor-two-cpl}
\end{equation}}
In the first case we use the coupling which flips a single coin and returns
the same outcome for both components, in the second we flip a single coin
but return opposite outcomes.
Thus we can now prove:
\[
  \CC{\Ber{p}} v.
  \left(
    \sure{\Ip{m}{1}=v} *
    \begin{pmatrix}
    \distAs{\Ip{k}{1}}{\Ber{\onehalf}}
    \\ {}*
    \distAs{\Ip{c}{2}}{\Ber{\onehalf}}
    \end{pmatrix}
  \right)
  \proves
  \CC{\Ber{p}} v.
    \left(
      \sure{\Ip{m}{1}=v}
      *
      \begin{cases}
        \cpl{ \Ip{k}{1} = \Ip{c}{2} }     \CASE v=0 \\
        \cpl{ \Ip{k}{1} = \neg\Ip{c}{2} } \CASE v=1
      \end{cases}
    \right)
\]
by using \ref{rule:c-cons},
and using the two couplings of~\eqref{ex:xor-two-cpl} in the $v=0$ and $v=1$ respectively.
Finally, the assignment to \p{c} in \p{encrypt} generates the fact
$\sure{\Ip{c}{1} = \Ip{k}{1} \xor \Ip{m}{1}}$.
By routine propagation of this fact we can establish
$
  \CC{\Ber{p}} v. \cpl{ \Ip{c}{1} = \Ip{c}{2} }.
$
To get an unconditional lifting,
we need a principle explaining the interaction between lifting and conditioning.
\thelogic\ can derive the general rule:
\begin{proofrule}
  \infer*[lab=rl-convex]{}{
  \CC\prob \wtv.\cpl{R} \proves \cpl{R}
}
   \label{rule:rl-convex}
\end{proofrule}
which states that relational liftings are \emph{convex},
\ie closed under convex combinations.
\begin{result}
\ref{rule:rl-convex} is an instance of many rules on the interaction between relational lifting
and the other connectives (conditioning in this case)
that can be derived in \thelogic\ by exploiting the encoding of liftings as \supercond.
\end{result}

Let us see how this is done for \ref{rule:rl-convex} based on two other rules of \supercond:
\begin{proofrules}\small \infer*[lab=c-skolem]{
  \prob \of \Dist(\Full{A})
}{
  \CC\prob v. \E x \of X. Q(v, x)
  \proves
  \E f \of A \to X. \CC\prob v. Q(v, f(v))
}   \label{rule:c-skolem}

\infer*[lab=c-fuse]{}{
  \CC{\prob} v.
  \CC{\krnl(v)} w.
    K(v,w)
  \lequiv
  \CC{\prob \fuse \krnl} (v,w). K(v,w)
}   \label{rule:c-fuse}
\end{proofrules}
\Cref{rule:c-skolem} is a primitive rule which
follows from Skolemization of the implicit universal
quantification used on~$v$ by the modality.
\Cref{rule:c-fuse}
can be seen as a way to merge two nested conditioning or split one conditioning into two.
The rule uses $
  \prob \fuse \krnl \is
    \fun(v,w).{ \prob(v) \cdot \krnl(v)(w)},
$
a variant of $\bind$ that does not forget the intermediate~$v$.
\Cref{rule:c-fuse} is an immediate consequence of two primitive rules
\ifappendix(\ref{rule:c-assoc} and \ref{rule:c-unassoc}) \fi
that reflect the associativity of the~$\bind$ operation.

To prove \ref{rule:rl-convex}
we start by unfolding the definition of relational lifting
(we write $K(v)$ for the part of the encoding inside the conditioning):
\begin{eqexplain}
  \CC\prob v. \cpl{R}
  \lequiv {} &
  \CC\prob v.
    \E \hat{\prob}_0.
      \CC{\hat{\prob}_0} w. K(w)
\whichproves
  \E \krnl.
    \CC\prob v.
      \CC{\krnl(v)} w. K(w)
\byrule{c-skolem}
\whichproves
  \E \krnl.
    \CC {\prob \fuse \krnl} (v,w). K(w)
\byrule{c-fuse}
\whichproves
  \E \hat\prob_1.
    \CC{\hat\prob_1} (v,w). K(w)
\proves
  \cpl{R}
\bydef
\end{eqexplain}
The application of \ref{rule:c-skolem} commutes the existential
quantification of the joint distribution~$\hat{\prob}_0$ and the outer modality.
By \ref{rule:c-fuse} we are able to merge the two modalities and obtain again
something of the same form as the encoding of relational liftings.

\subsection{Outside the Box of Relational Lifting}
\label{sec:overview:obox}

One of the well-known limitations of pRHL is that it requires
a very strict structural alignment between the order of samplings
to be coupled in the two programs.
A common pattern that pRHL rules cannot handle is
showing that reversing the order of execution of two blocks of code
does not affect the output distribution,
\eg running
  \code{x:=Ber(1/2);y:=Ber(2/3)} versus
  \code{y:=Ber(2/3);x:=Ber(1/2)}.
In \thelogic, we can establish this pattern using a \emph{derived} general rule:
\begin{proofrule}
  \infer*[lab=seq-swap]{
    \{P_1\} {\m[\I1: t_1, \I2: t_1']} \{\cpl{R_1}\}
    \\
    \{P_2\} {\m[\I1: t_2, \I2: t_2']} \{\cpl{R_2}\}
  }{
    \{P_1 * P_2\}
    {\m[
     \I1: (t_1\p; t_2),
     \I2: (t_2'\p; t_1')
    ]}
    \{\cpl*{
      R_1
      \land
      R_2
    }\}
  }
  \label{rule:seq-swap}
\end{proofrule}
The rule assumes that the lifting of $R_1$ (resp. $R_2$) can be established
by analyzing $t_1$ and $t_1'$ ($t_2$ and~$t_2'$)
side by side from precondition $P_1$ ($P_2)$.
The standard sequential rule of pRHL would force an alignment
between the wrong pairs ($t_1$ with~$t_2'$, and $t_2$ with~$t_1'$).
Crucial to the soundness of the rule is the assumption
(expressed by the precondition~$P_1*P_2$ in the conclusion)
that $P_1$ and~$P_2$ are probabilistically independent.\footnote{
  In the full model of \thelogic,
  to ensure safe mutation, assertions also include ``write/read permissions''
  on variables (in the ``variables as resource''-style~\cite{BornatCY06}).
  In \ref{rule:seq-swap} the separation between $P_1$ and $P_2$ ensures,
  in addition to probabilistic independence, that if $t_1$ has write permissions
  on a variable~$\p{x}$, $t_2$ does not have read permissions on it and viceversa (and analogously for $t_1'$ and $t_2'$).
  The full model incurs in some permissions bookkeeping,
  which we omit in this section for readability;
  \cref{ex:perm-triples} shows how to fill in the omitted details.}
In contrast, because pRHL lacks the construct of independence,
it simply cannot express such a rule.

\begin{result}
\thelogic's treatment of relational lifting enables the study of the interaction
between lifting and independence,
unlocking a novel solution for forfeiting strict structural similarities between components required by relational logics.
\end{result}
Two ingredients of \thelogic\ cooperate to prove
\ref{rule:seq-swap}:
the adoption of a \emph{weakest precondition}~(WP) formulation of triples
(and associated rules)
and a novel property of relational lifting. Let us start with WP.
In \thelogic, a triple $\{P\}\ \m{t}\ \{Q\}$ is actually encoded as
the entailment
$ P \proves \WP {\m{t}} {Q} $.
Here, $P$ and $Q$ are both assertions on \pre n-nary tuples of distributions;
and throughout, we use the bold $\m{t}$ to denote an \pre n-nary
tuple of program terms.
The formula $ \WP {\m{t}} {Q}$  is a natural generalization of WP assertion to \pre n-nary programs:
$\WP {\m{t}} {Q}$
holds on
a \pre n-nary tuple of distributions~$\m{\prob}$,
if the tuple of output distributions
obtained by running each program in~$\m{t}$ on the corresponding component
of~$\m{\prob}$,
satisfies~$Q$.
\thelogic\ provides a number of rules for manipulating WP;
here is a selection needed for deriving \ref{rule:seq-swap}:
\begin{proofrules}
  \infer*[lab=wp-cons]{
  Q \proves Q'
}{
  \wpc{\m{t}}{Q}
  \proves
  \wpc{\m{t}}{Q'}
}   \label{rule:wp-cons}

  \infer*[lab=wp-frame]{}{
  P \sepand \wpc{\hpt}{Q}
  \proves
  \wpc{\hpt}{\liftA{P} \sepand Q}
}   \label{rule:wp-frame}

  \infer*[lab=wp-seq]{}{
  \WP {\m[i: t]}[\big]{
    \WP {\m*[i: \smash{t'}]} {Q}
  }
  \proves
  \WP {\m[i: (t\code{;}\ t')]} {Q}
}   \label{rule:wp-seq}

  \infer*[lab=wp-nest]{}{
  \wpc{\m{t}_1}{
    \wpc{\m{t}_2}{Q}
  }
  \lequiv
  \wpc{(\m{t}_1 \m. \m{t}_2)}{Q}
}   \label{rule:wp-nest}
\end{proofrules}
\Cref{rule:wp-cons,rule:wp-frame} are the usual consequence and framing rules
of Separation Logic, in WP form.
By adopting Lilac's measure-theoretic notion of independence as the interpretation for separating conjunction, we obtain a clean frame rule.
Among the WP rules for program constructs,
\cref{rule:wp-seq} takes care of sequential composition.
Notably, we only need to state it for unary WPs,
in contrast to other logics where supporting relational proofs
requires building the lockstep strategy into the rules.
We use the more flexible approach from the Logic for Hypertriple Composition (LHC)~\cite{d2022proving},
where a handful of arity-changing rules allow seamless integration
of unary and relational judgments.
One such rule is the \ref{rule:wp-nest} rule, which
establishes the equivalence of a WP on $ \m{t}_1 \m. \m{t}_2 $,
where $(\m.)$ is union of indexed tuples with disjoint indexes,
and two nested WPs involving $\m{t}_1$, and $\m{t}_2$ individually.
This for instance allows us to lift the unary \ref{rule:wp-seq}
to a binary lockstep rule:
\begin{derivation}
  \infer*[Right=\ref{rule:wp-nest}]{
  \infer*[Right={\ref{rule:wp-seq},\ref{rule:wp-cons}}]{
  \infer*[Right=\ref{rule:wp-nest}]{
  \infer*[Right=\ref{rule:wp-cons}]{
    P \proves
    \WP {\m[\I1: t_1]} {
      \WP {\m[\I2: t_2]}{Q'}
    }
    \\
    Q' \proves \WP {\m[\I1: t_1']} {
      \WP {\m[\I2: t_2']} {Q}
    }
  }{
    P \proves
    \WP {\m[\I1: t_1]} {
      \WP {\m[\I2: t_2]}{
        \WP {\m[\I1: t_1']} {
          \WP {\m[\I2: t_2']} {Q}
        }
      }
    }
  }}{
    P \proves
    \WP {\m[\I1: t_1]} {
      \WP {\m[\I1: t_1']} {
        \WP {\m[\I2: t_2]}{
          \WP {\m[\I2: t_2']} {Q}
        }
      }
    }
  }}{
    P \proves
    \WP {\m[\I1: (t_1\p; t_1')]} {
      \WP {\m[\I2: (t_2\p;t_2')]} {Q}
    }
  }}{
    P \proves \WP {\m[\I1: (t_1\p; t_1'), \I2: (t_2\p;t_2')]} {Q}
  }
\end{derivation}

The crucial idea behind \ref{rule:seq-swap} is that the two programs
$t_1$ and $t_2$ we want to swap rely on \emph{independent} resources,
and thus their effects are independent from each other.
In Separation Logic this kind of reasoning is driven by framing:
which is done through framing in Separation Logic:
while executing~$t_1$, frame the resources
needed for~$t_2$, which remain intact in the state left by~$t_1$.
Multiple applications of~\ref{rule:wp-frame} and other basic rules
get us to the post-condition $\cpl{R_1} \ast \cpl{R_2}$, but
we want to combine them into one relational lifting.
This is accommodated~by:
\begin{proofrule}
  \infer*[lab=rl-merge]{}{
  \cpl{R_1} * \cpl{R_2}
  \proves
  \cpl{R_1 \land R_2}
}
   \label{rule:rl-merge}
\end{proofrule}
We do not show the derivation here for space constraints,
but essentially it consists in unfolding the encoding of lifting,
and using \ref{rule:c-frame} and \ref{rule:c-fuse}
to merge the two \supercond\ modalities.

Using these rules we can construct the following derivation:
\begin{derivation}[\small]
  \infer*[Right=\ref{rule:rl-merge}]{
  \infer*[Right=\ref{rule:wp-nest}]{
  \infer*[Right={\ref{rule:wp-seq},\ref{rule:wp-nest}}]{
  \infer*[Right=\ref{rule:wp-frame}]{
  \infer*[Right=\ref{rule:wp-frame}]{
  \infer*[Right={\ref{rule:wp-frame},\ref{rule:wp-nest}}]{
  \infer*{
    P_1 \proves \WP {\m[\I1: t_1, \I2: t_1']}{\cpl{R_1}}
    \\
    P_2 \proves \WP {\m[\I1: t_2, \I2: t_2']}{\cpl{R_2}}
  }{
    P_1 * P_2
    \proves
    \WP{\m[\I1: t_1,\I2: t_1']}{
      \cpl{R_1}
    }
    *
    \WP{\m[\I1: t_2, \I2: t_2' ]}{
      \cpl{R_2}
    }
  }}{
    P_1 * P_2
    \proves
    \WP{\m[\I1: t_1]}[\big]{
\WP{\m[
       \I1: t_2,
       \I2: t_2'
      ]}*{
        \cpl{R_2}
      } * {}
      \WP{\m[\I2: t_1']}{
        \cpl{R_1}
      }
}
  }}{
    P_1 * P_2
    \proves
    \WP{\m[\I1: t_1]}[\Big]{
    \WP{\m[
     \I1: t_2,
     \I2: t_2'
    ]}[\big]{
\cpl{R_2} * {}
      \WP{\m[\I2: t_1']}{
        \cpl{R_1}
      }
} }
  }}{
    P_1 * P_2
    \proves
    \WP{\m[\I1: t_1]}[\Big]{
    \WP{\m[
     \I1: t_2,
     \I2: t_2'
    ]}[\big]{
    \WP{\m[\I2: t_1']}{
\cpl{R_1} * {}
\cpl{R_2}
} } }
  }}{P_1 * P_2
    \proves
    \WP{\m[\I1: (t_1\p; t_2)]}[\big]{
      \WP{\m[\I2: (t_2'\p; t_1')]}{
        \cpl{R_1} *
        \cpl{R_2}
      }
    }
  }}{P_1 * P_2
    \proves
    \WP{\m[
     \I1: (t_1\p; t_2),
     \I2: (t_2'\p; t_1')
    ]}[\big]{
\cpl{R_1} *
\cpl{R_2}
}
  }}{P_1 * P_2
    \proves
    \WP{\m[
     \I1: (t_1\p; t_2),
     \I2: (t_2'\p; t_1')
    ]}*{\cpl*{
R_1 \land
      R_2
}}
  }
\end{derivation}
We explain the proof strategy from bottom to top.
We first apply \ref{rule:rl-merge} to the postcondition
(thanks to \ref{rule:wp-cons}).
This step reduces the goal to proving the two relational liftings
can be established independently from each other.
Then we apply \ref{rule:wp-nest} and \ref{rule:wp-seq}
to separate the two indices, break the sequential compositions and recombine
the two inner WPs.
We then proceed by three applications of the \ref{rule:wp-frame} rule:
the first brings $\cpl{R_2}$ out of the innermost WP;
the second brings the WP on $\m[\I1:t_1']$ outside the middle WP;
the last brings the WP on $\m[\I1:t_2,\I2:t_2']$ outside the topmost WP.
An application of \cref{rule:wp-nest} merges the resulting nested WPs
on $t_1$ and $t_1'$.
We thus effectively reduced the problem to showing that the two WPs
can be established independently, which was our original goal.

The \ref{rule:rl-merge} rule not only provides an elegant way of overcoming
the long-standing alignments issue with constructing relational lifting,
but also shows how fundamental the role of probabilistic independence is
for compositional reasoning:
the same rule with standard conjunction is unsound!
Intuitively, if we just had $ \cpl{R_1} \land \cpl{R_2} $,
we would know there exist two couplings
$\prob_1$ and $\prob_2$,
justifying $\cpl{R_1}$ and $\cpl{R_2}$ respectively,
but the desired consequence $\cpl{R_1 \land R_2}$
requires the construction of a single coupling that justifies both relations
at the same time.
We can see this is not always possible by looking back at
\eqref{ex:xor-two-cpl}:
for two fair coins we can establish
$
  \cpl{ \Ip{k}{1} = \Ip{c}{2} }
  \land
  \cpl{ \Ip{k}{1} = \neg\Ip{c}{2} }
$,
but
$
  \cpl{
    \Ip{k}{1} = \Ip{c}{2}
    \land
    \Ip{k}{1} = \neg\Ip{c}{2}
  }
$ is equivalent to false.

 \section{Preliminaries: Programs and Probability Spaces}
\label{sec:prelims}

To formally define the model of \thelogic\ and validate its rules,
we introduce a number of preliminary notions.
Our starting point is the measure-theoretic approach of
\cite{lilac} in defining probabilistic separation.
We recall the main definitions here.
The main additional assumption we make throughout
is that the set of outcomes of distributions is countable.

\begin{definition}[Probability spaces]
\label{def:prob-sp}
  Given a set of possible \emph{outcomes} $\Outcomes$,
  a \salgebra\ $ \salg \in \SigAlg(\Outcomes) $ is
  a set of subsets of $\Outcomes$
    that is closed under countable unions and
    complements, and such that
    $\Outcomes \in \salg$.
  The \emph{full} \salgebra\ over~$\Outcomes$ is
  $ \Full{\Outcomes} = \powerset(\Outcomes) $,
		the powerset of $\Outcomes$.
  For $F \subs \powerset(\Outcomes)$, we write $\sigcl{F} \in \SigAlg(\Outcomes)$
    for the smallest \salgebra\ containing $F$.
Given~$\salg \in \SigAlg(\Outcomes)$,
  a \emph{probability distribution}~$\prob \in \Dist(\salg)$,
    is a countably additive function~$ \prob \from \salg \to [0,1] $
    with $\prob(\Outcomes)=1$.
  The support of a distribution $\prob \in \Dist(\Full{\Outcomes})$
  is the set of outcomes with non-zero probability
  $ \psupp(\prob) \is \set{ a \in \Outcomes | \prob(a) > 0 } $,
  where $\prob(a)$ abbreviates $\prob(\set{a})$.

  A \emph{probability space} $ \psp \in \ProbSp(\Outcomes) $ is
  a pair $ \psp = (\salg, \prob) $ of
    a \salgebra\ $\salg \in \SigAlg(\Outcomes)$ and
    a probability distribution~$\prob \in \Dist(\salg)$.
The \emph{trivial} probability space $\Triv{\Outcomes} \in \ProbSp(\Outcomes)$
  is the trivial \salgebra\ $ \set{\Outcomes,\emptyset} $
  equipped with the trivial probability distribution.
Given $\salg_1 \subs \salg_2$ and $\prob \in \Dist(\salg_2)$,
  the distribution $ \restr{\prob}{\salg_1} \in \Dist(\salg_1) $
  is the restriction of $\prob$ to $\salg_1$.
The \emph{extension pre-order}~$(\extTo)$ over probability spaces is defined as
  $
    (\salg_1, \prob_1) \extTo (\salg_2, \prob_2) \is
      \salg_1 \subs \salg_2 \land \prob_1 = \restr{\prob_2}{\salg_1}.
  $

  A function $f \from \Outcomes_1 \to \Outcomes_2$ is
  \em{measurable} on
  $ \salg_1\in\SigAlg(\Outcomes_1)$ and $\salg_2\in\SigAlg(\Outcomes_2) $
  if\/
  $ \forall \event \in \salg_2 \st {\inv{f}(\event) \in \salg_1} $.
When $\salg_2 = \Full{\Outcomes_2}$
	we simply say~$f$ is measurable~in~$\salg_1$.
\end{definition}

\begin{definition}[Product and union spaces]
\label{def:prod-union-sp}
  Given $ \salg_1 \in \SigAlg(\Outcomes_1),\salg_2 \in \SigAlg(\Outcomes_2) $,
  their product is the \salgebra{}
  $ \salg_1 \pprod \salg_2 \in \SigAlg(\Outcomes_1 \times \Outcomes_2) $
  defined as
  $ \salg_1 \pprod \salg_2 \is \sigcl{\set{\event_1 \times \event_2 | \event_1 \in \salg_1, \event_2 \in \salg_2}} $,
  and their union is the \salgebra{}
  $ \salg_1 \punion \salg_2 \is \sigcl{\salg_1 \union \salg_2} $.
The product of two probability distributions
  $ \prob_1 \in \Dist(\salg_1) $ and
  $ \prob_2 \in \Dist(\salg_2) $ is
  the distribution
  $ (\prob_1 \pprod \prob_2) \in \Dist(\salg_1 \pprod \salg_2) $
  defined by
  $ (\prob_1 \pprod \prob_2)(\event_1 \times \event_2) = \prob_1(\event_1)\prob_2(\event_2) $
  for all $\event_1 \in \salg_1$, $\event_2 \in \salg_2$.
\end{definition}

\begin{definition}[Independent product~\cite{lilac}]
\label{def:indep-comb}
  Given $ (\salg_1, \prob_1),(\salg_2, \prob_2) \in \ProbSp(\Outcomes) $,
  their \emph{independent product} is
  the probability space
  $(\salg_1 \punion \salg_2, \prob) \in \ProbSp(\Outcomes)$
  where for all~$ \event_1 \in \salg_1, \event_2 \in \salg_2 $,
  $
    \prob(\event_1 \inters \event_2) = \prob_1(\event_1) \cdot \prob_2(\event_2)
  $. It is unique, if it exists~\cite[Lemma 2.3]{lilac}. Let $ \psp_1 \iprod \psp_2 $ be the unique independent product
  of $\psp_1$ and $\psp_2$ when it exists, and be undefined otherwise.
\end{definition}

\paragraph{Indexed tuples}
To handle unary and higher-arity relational assertions in a uniform way,
we consider finite sets of indices $I \subs \Nat$,
and \pre I-indexed tuples of values of type~$X$,
represented as (finite) functions $\Hyp[I]{X}$.
We use boldface to range over such functions.
The syntax $ \m{x} = \m[i_0: x_0,\dots,i_n: x_n] $ denotes the function
$ \m{x} \in \Hyp[\set{i_0,\dots,i_n}]{X} $ with $\m{x}(i_k) = x_k$.
We often use comprehension-style notation \eg $\m{x} = \m[i: x_i | i\in I]$.
For $\m{x} \in \Hyp[I]{A}$ we let $\supp{\m{x}} \is I$.
Given some $ \m{x} \in \Hyp[I]{A} $ and some $J \subs I$,
the operation $ \m{x} \setminus J \is \m[i: \m{x}(i) | i \in I \setminus J] $
removes the components with indices in~$J$ from $\m{x}$.

\paragraph{Programs}
We consider a simple first-order imperative language.
We fix a finite set of \emph{program variables}~$\p{x} \in \Var$
and countable set of \emph{values}~$\val \in \Val \is \Int$
and define the program \emph{stores} to be
$ \store \in \Store \is \Var \to \Val $
(note that $\Store$ is countable).

\begin{mathfig}
  \begin{grammar}
\Expr \ni \expr \is
        \val | \code{x} | \prim(\vec{\expr})
    \qquad
    \vec{\expr} \grammIs \expr_1,\dots,\expr_n
    \qquad
    \prim \grammIs + | - | < | \dots
    \qquad
    \dist \grammIs \p{Ber} | \p{Unif} | \dots
    \\
    \Term \ni \term \is
        \Assn{x}{\expr}
      | \Sample{x}{\dist}{\vec{\expr}}
      | \Skip
| \Seq{\term_1}{\term_2}
      | \Cond{\expr}{\term_1}{\term_2}
      | \Loop{\expr}{\term}
  \end{grammar}
  \caption{Syntax of program terms.}
  \label{fig:term-syntax}
\end{mathfig}

Program \emph{terms}~$ \term \in \Term $ are formed according to the
grammar in \cref{fig:term-syntax}.
For simplicity,
booleans are encoded by using~$0 \in \Val$ as false and any other value as true.
We will use the events
  $\false \is \set{0}$ and
  $\true \is \set{n \in \Val | n\ne 0}$.
Programs use standard deterministic primitives~$\prim$,
which are interpreted as functions
$ \sem{\prim} \from \Val^n \to \Val $, where $n$ is the arity of $\prim$.
Expressions~$\expr$ are effect-free deterministic numeric expressions,
and denote, as is standard, a function
$ \sem{\expr} \from \Store \to \Val $,
\ie a random variable of $\Full{\Store}$.
We write $\pvar(\expr)$ for the set of program variables that occur
in~$\expr$.
Programs can refer to some collection of known
\emph{discrete} distributions~$\dist$,
each allowing a certain number of parameters.
Sampling assignments $ \code{x:~$\dist$($\vec{v}$)} $
sample from the distribution $\Sem{\dist}(\vec{v}) \from \Dist(\Full{\Val})$.
The distribution $ \Sem{\p{Ber}}(p) = \Ber{p} \of\Dist(\Full{\set{0,1}}) $
is the Bernoulli distribution assigning probability~$p$ to outcome~1.

Similar to Lilac, we consider a simple iteration construct
$ \code{repeat}\; e\; t $ which evaluates $e$ to a value $n \in \Val$
and, if $n>0$, executes $t$ in sequence $n$ times.
This means we only consider a subset of terminating programs.
The semantics of programs is
entirely standard and is defined in \appendixref{sec:appendix:definition}.
It associates each term~$t$ to a function
$
  \sem{t} \from \Dist(\Full{\Store}) \to \Dist(\Full{\Store})
$
from distributions of input stores to
distributions of output stores.

In the relational reasoning setting, one would consider multiple
programs at the same time and relate their semantics.
Following LHC~\cite{d2022proving},
we define \emph{hyper-terms} as $ \m{t} \in \Hyp[J]{\Term} $
for some finite set of indices~$J$.
Let~$I$ be such that~$\supp{\m{t}} \subs I$; the semantics
$
  \sem{\m{t}}_I \from
      \Hyp[I]{\Dist(\Full{\Store})} \to \Hyp[I]{\Dist(\Full{\Store})}
$
takes a \pre I-indexed family of distributions as input and outputs
another \pre I-indexed family of distributions:
\[
  \sem{\m{t}}_I(\m{\prob}) \is
    \fun i.
    \ITE{i \in \supp{\m{t}}}{
      \sem{\m{t}(i)}(\m{\prob}(i))
    }{
      \m{\prob}(i)
    }
\]
Note that the store distributions at indices in $ I \setminus \supp{t} $
are preserved as is.
We omit~$I$ when it can be inferred from context.
To refer to program variables in a specific component we use
elements of $I\times \Var$, writing $\ip{x}{i}$ for $(i,\p{x})$.
 \section{The \thelogic\ Logic}
\label{sec:logic}

We are now ready to define \thelogic's semantic model,
and formally prove its laws.

\subsection{A Model of (Probabilistic) Resources}

As a model for our assertions we use a modern presentation of
partial commutative monoids, adapted from \cite{KrebbersJ0TKTCD18},
called ``ordered unital resource algebras'' (henceforth~RA).
Instead of a partial binary operation, RAs are equipped with a \emph{total}
binary operation and a predicate $\raValid$ indicating which elements of the
carrier are considered \emph{valid} resources.
Partiality of the operation manifests as mapping some combinations of arguments
to invalid elements.

\begin{definition}[Ordered Unital Resource Algebra]
\label{def:ra}
  An \emph{ordered unital resource algebra}~(RA) is a tuple
  $
    (M, \raLeq, \raValid, \raOp, \raUnit)
  $
  where
  $ \raLeq \from M \times M \to \Prop $ is the reflexive and transitive
  \emph{resource order},
  $ \raValid \from M \to \Prop $ is the \emph{validity predicate},
  $ (\raOp) \from M \to M \to M $ is the \emph{resource composition},
  a commutative and associative binary operation on $M$,
  and
  $ \raUnit \in M $ is the \emph{unit} of $M$,
  satisfying, for all $a,b,c\in M$:
  \begin{mathpar}
  \raValid(\raUnit)

  \raUnit \raOp a = a

  \infer{
    \raValid(a \raOp b)
  }{
    \raValid(a)
  }

  \infer{
    a \raLeq b
    \\
    \raValid(b)
  }{
    \raValid(a)
  }

  \infer{
    a \raLeq b
  }{
    a \raOp c \raLeq b \raOp c
  }
  \end{mathpar}
\end{definition}

Any RA can serve as a model of basic connectives of separation logics;
in particular, $P*Q$ will hold on $a \in M$ if and only if
there are $a_1,a_2 \in M$ such that $ a = a_1 \raOp a_2 $ and
$P$ holds on $a_1$ and $Q$ holds on $a_2$.

\thelogic's assertions will be interpreted over a specific RA,
which we construct as the combination of other basic RAs.
The main component is the Probability Spaces RA,
which uses independent product as the RA operation.

\begin{definition}[Probability Spaces RA]
  The \emph{probability spaces RA}
  $ \PSpRA_\Outcomes $
  is the resource algebra
  ${(\ProbSp(\Outcomes) \dunion \set{\invalid}, \raLeq, \raValid, \raOp, \Triv{\Outcomes})}$
  where
  $\raLeq$ is the extension order with $\invalid$ added as the top element,
  that is,
    $ \psp_1 \raLeq \psp_2 \is \psp_1 \extTo \psp_2 $ and
    $ \forall a \in \PSpRA_\Outcomes\st
        a \raLeq \invalid$;
  $\raValid(a) \is a \neq \invalid$;
  composition is independent product:\[
    a \raOp b \is
      \begin{cases}
        \psp_1 \iprod \psp_2
          \CASE a=\psp_1, b=\psp_2, \text{ and }
                \psp_1 \iprod \psp_2 \text{ is defined}
        \\
        \invalid \OTHERWISE
      \end{cases}
  \]
\end{definition}

The fact that $\PSpRA_\Outcomes$ satisfies the axioms of RAs is
established in~\appendixref{sec:appendix:model} and builds on the analogous
construction in Lilac.
In comparison with the coarser model of PSL,
independent product represents a more sophisticated way of separating
probability spaces.
In PSL, separation of distributions requires the distributions to
involve disjoint sets of variables, ruling out
assertions like
$ \distAs{\p{x}}{\prob} * \sure{\p{x}=\p{y}} $
or
$ \own{\p{x}} * \own{\code{x xor y}} $,
which are satisfiable in Lilac and \thelogic.

\begin{example}
\label{ex:indip-prod}
  Assume there are only two variables $\p{x}$ and $\p{y}$.
  Let $X_v = \set{ \store | \store(\p{x}) = v}$
  and  $\psp_1 = (\salg_1,\prob_1) $ with
  $ \salg_1 = \sigcl{\set{ X_v | v \in \Val}} $
  and let $ \prob_1 $ give \p{x} the distribution of a fair coin,
  \ie $ \prob_1 $ is
  the extension to $\salg_1$ of
  $ \prob_1(X_0) = \prob_1(X_1) = \onehalf $.
  Intuitively, the assertion $\distAs{\p{x}}{\Ber{\onehalf}}$ holds
  on $\psp_1$
  (we will define the assertion's interpretation in \cref{sec:prob-assert}).
  Similarly, $\sure{\p{x}=\p{y}}$ holds on $\psp_2 = (\salg_2, \prob_2)$
  where
  $
    \salg_2 = \set{\emptyset, \Store, \set{E}, \Store\setminus E}
  $ with $ E = \set{\store | \store(\p{x})=\store(\p{y})} $
  and $\prob_2(E) = 1$.
  Note that $\salg_2$ is very coarse:
  it does not contain events that can pin the value of \p{x} precisely;
  thanks to this, $\prob_2$ does not need to specify what is the distribution of \p{x}, but only that \p{y} will coincide on \p{x} with probability~1.
  It is easy to see that the independent product of $\psp_1$ and $\psp_2$
  exists and is $\psp_3 = (\salg_1 \punion \salg_2, \prob_3)$
  where $\prob_3$ is determined by
  $ \prob_3(X_0 \inters E) = \prob_3(X_1 \inters E) = \onehalf $,
  \ie makes \p{x} \p{y} the outcomes of the same fair coin.
  This means $\psp_3$ is a model of $\distAs{\p{x}}{\Ber{\onehalf}} * \sure{\p{x}=\p{y}}$.
\end{example}

When state is immutable, like in Lilac (which uses a functional language),
the $\PSpRA_\Outcomes$ RA is adequate to support a logic for probabilistic independence.
There is however an obstacle in adopting independent product in a language
with mutable state.

\begin{example}
  Consider a simple assignment \code{x:=0}.
  In the spirit of separation logic's local reasoning,
  we would want to prove a
  small-footprint triple for the assignment, \ie one where the precondition
  only involves ownership of the variable \p{x}.
  We could try with $ \distAs{x}{\prob} $ (for arbitrary $\prob$)
  but we would run into problems proving the frame property:
  as we remarked, an assertion like $\sure{\p{x}=\p{y}}$ is
  a valid frame of $ \distAs{x}{\prob} $;
  yet if $y \ne 0$, such frame would not hold after the assignment.
\end{example}

We solve this problem by combining $\PSpRA$ with an RA of permissions over variables.
The idea is that in addition to information about the distribution,
assertions can indicate which ``write permissions'' we own on variables.
An assertion that owns write permissions on \p{x} would be incompatible
with any frame predicating on \p{x}.
Then a triple for assignment just needs to require write permission
to the assigned variable.
We model permissions using a standard fractional permission RA.

\begin{definition}The \emph{permissions RA}
  is defined as
  $(\Perm, \raLeq, \raValid, \raOp, \raUnit)$
  where
  $ \Perm \is \Var \to \PosRat $,
  $ {a \raLeq b} \is {\forall \p{x} \in \Var \st a(\p{x}) \leq b(\p{x})} $,
  $ \raValid(a) \is (\forall \p{x} \in \Var \st a(\p{x}) \leq 1) $,
  $ a_1 \raOp a_2 \is \fun \p{x}. a_1(\p{x}) + a_2(\p{x})$ and
  $ \raUnit = \fun \wtv.0 $.
\end{definition}

We now want to combine permissions with probability spaces.
The goal is to allow probability spaces to contain only information
about variables of which we have some non-zero permission.
This gives rise to the following definition.

\begin{definition}[Compatibility]
  Given a probability space $\psp \in \ProbSp(\Store)$ and
  a permission map $\permap \in \Perm$,
we say that $\psp$ is \emph{compatible} with $\permap$,
  written $\psp\compat\permap$,
  if there exists
  $\psp' \in \ProbSp((\Var \setminus S) \to \Val)$
  such that
  $\psp = \psp' \pprod \Triv{S \to \Val}$,
  where
  $S = \set{x \in \Var | \permap(x) = 0}.$
  Note that we are exploiting the isomorphism
  $
    \Store \iso
    ((\Var \setminus S) \to \Val)
      \times
      (S \to \Val).
  $
  We extend the notion to $ \PSpRA_{\Store} $
  by declaring $ \invalid \compat \permap \is \True$.
\end{definition}

We can now construct an RA which combines probability spaces and permissions.

\begin{definition}Let
  $
    \PSpPmRA \is
      \set{
        (\maybePsp, \permap)
          | \maybePsp \in \PSpRA_{\Store},
            \permap \in \Perm,
            \maybePsp \compat \permap
      }.
  $
  We associate with $\PSpPmRA$ the
  \emph{Probability Spaces with Permissions}~RA
  $
    (\PSpPmRA, \raLeq, \raValid, \raOp, \raUnit)
  $
  where
\begin{align*}
    \raValid((\maybePsp, \permap)) &\is
      \maybePsp \neq \invalid
      \land \forall\p{x}.\permap(\p{x}) \leq 1
    &
    (\maybePsp, \permap) \raOp (\maybePsp', \permap') &\is
      (\maybePsp \raOp \maybePsp', \permap \raOp \permap')
    \\
    (\maybePsp, \permap) \raLeq (\maybePsp', \permap') &\is
      \maybePsp \raLeq \maybePsp'
      \text{ and }
      \permap \raLeq \permap'
    &
    \raUnit &\is ( \Triv{\Store}, \fun \p{x}. 0)
  \end{align*}
\end{definition}

\begin{example}
\label{ex:perm-ra}
  Using $\PSpPmRA$,
  we can refine an assertion
  $\distAs{\p{x}}{\prob}$ into
  $(\distAs{\p{x}}{\prob})\withperm{x:q}$,
  which holds on resources $(\psp, \permap)$ where $\psp$ distributes \p{x}
  as $\prob$ and $\permap(\p{x})\geq q$.
  What this achieves is to be able to differentiate between an assertion
  $(\distAs{\p{x}}{\prob})\withperm{x:\onehalf}$ which allows frames to predicate on $\p{x}$ (\eg $\sure{\p{x}=\p{y}}$) and
  an assertion $(\distAs{\p{x}}{\prob})\withperm{x:1}$ which does not allow the
  frame and consequently allows mutation of~\p{x}.
\end{example}

While this allows for a clean semantic treatment of mutation and independence,
it does incur some bookkeeping of permissions in practice,
which we omitted in the examples of \cref{sec:overview}.
The necessary permissions are however easy to infer
from the variables used in the assertions,
as we will illustrate later in \cref{ex:perm-triples}.

Finally, to build \thelogic's model we simply construct an RA
of \pre I-indexed tuples of probability spaces with permissions.

\begin{definition}[\thelogic\ RA]
Given a set of indices~$I$ and a RA~$M$,
  the \emph{product RA} $ \Hyp[I]{M} $ is the pointwise lifting
  of the components of~$M$.
\thelogic's model is $\Model_I \is \Hyp{\PSpPmRA}$.
\end{definition}

\subsection{Probabilistic Hyper-Assertions}
\label{sec:prob-assert}

Now we turn to the assertions in our logic.
We take a semantic approach to assertions:
we do not insist on a specific syntax and instead characterize
what constitutes an assertion by its type.
In Separation Logic, assertions are defined relative to some RA~$M$,
as the upward closed functions $ M \to \Prop $.
An assertion $ P \from M \to \Prop $ is upward closed if
$ \forall a, a' \in M\st a \raLeq[M] a' \implies P(a) \implies P(a'). $
We write $ M \ucto \Prop $ for the type of upward closed assertions on $M$.
We define hyper-assertions to be assertions over $\Model_I$,
\ie $P \in \HAssrt_I \is \Model_I \ucto \Prop $.
Entailment is defined as
$
  (P \proves Q) \is
    \forall a \in M\st
      \raValid(a) \implies (P(a) \implies Q(a)).
$
Logical equivalence is defined as entailment in both directions:
$ P \lequiv Q \is (P \proves Q) \land (Q \proves P) $.
We inherit the basic connectives
(conjunction, disjunction, separation, quantification)
from SL, which are well-defined on arbitrary RAs, including~$\Model_I$.
In particular:
\begin{align*}
  P * Q &\is \fun a.
    \exists b_1,b_2 \st
      (b_1 \raOp b_2) \raLeq a \land
      P(b_1) \land
      Q(b_2)
  &
  \pure{\varphi} &\is \fun \wtv. \varphi
  &
  \Own{b} &\is \fun a. b \raLeq a
\end{align*}
Pure assertions $\pure{\varphi}$ lift meta-level propositions~$\varphi$
to assertions (by ignoring the resource).
$\Own{b}$ holds on resources that are greater or equal than~$b$ in the RA order;
this means~$b$ represents a lower bound on the available resources.

We now turn to assertions
that are specific to probabilistic reasoning in \thelogic{},
\ie the ones that can only be interpreted in $\Model_I$.
We use the following two abbreviations:
\begin{align*}
  \Own{\m{\salg}, \m{\prob}, \m{\permap}} &\is
    \Own{((\m{\salg}, \m{\prob}), \m{\permap})}
&
  \Own{\m{\salg}, \m{\prob}} &\is
    \E \m{\permap}. \Own{\m{\salg}, \m{\prob}, \m{\permap}}
\end{align*}

To start, we define \emph{\pre A-typed assertion expressions}~$ \aexpr $
which are of type
$ \aexpr \from \Store \to A $.
Note that the type of the semantics of a program expression $\sem{\expr} \from \Store \to \Val$ is a \pre\Val-typed assertion expression; because of this we seamlessly use program expressions in assertions, implicitly coercing them to their semantics.
Since in general we deal with hyper-stores $\m{\store} \in \Hyp{\Store}$,
we use the notation $\aexpr\at{i}$ to denote the application of $\aexpr$ to the
store $\m{\store}(i)$.
Notationally, it may be confusing to read composite expressions like
$ (\p{x}-\p{z})\at{i} $, so we write them for clarity with each program variable annotated with the index, as in $\ip{x}{i} - \ip{z}{i}$.

\paragraph{The meaning of owning~$ \distAs{\Ip{x}{1}}{\prob} $}
A function $ f \from A \to B $ is \emph{measurable} in a \salgebra{}
$ \salg \of \SigAlg(A) $ if $ \inv{f}(b) = \set{a \in A | f(a) = b} \in \salg $
for all $b\in B$.
An expression $\aexpr$ always defines a measurable function
(\ie a \emph{random variable})
in $\Full{\Store}$,
but might not be measurable in some sub-algebras of $\Full{\Store}$.
Lilac proposed to use measurability as the notion of ownership:
an expression~$\aexpr$ is owned in any resources that contains enough
information to determine its distribution, \ie that makes $\aexpr$ measurable.
While this makes sense conceptually,
we discovered it made another important connective of Lilac,
almost sure equality, slightly flawed
(in that it would not support the necessary laws).\footnote{In fact, a later revision~\cite{lilac2} corrected the issue,
  although with a different solution from ours.
  See~\cref{sec:relwork}.}
We propose a slight weakening of the notion of measurability which solves
those issues while still retaining the intent behind the meaning of ownership in relation to independence and conditioning.
We call this weaker notion ``almost measurability''.

\begin{definition}[Almost-measurability]
  \label{def:almost-meas}
  Given a probability space $ (\salg,\prob) \in \ProbSp(\Outcomes) $
  and a set $\event \subs \Outcomes$,
  we say $ \event $ is \emph{almost measurable} in $(\salg, \prob)$,
  written $ \almostM{\event}{(\salg, \prob)} $,
  if \[
    \exists \event_1,\event_2 \in \salg \st
    \event_1 \subs \event \subs \event_2
    \land
    \prob(\event_1)=\prob(\event_2)
  \]
  We say a function~$ \aexpr \from \Outcomes \to A $,
  is \emph{almost measurable} in $(\salg, \prob)$,
  written $ \almostM{\aexpr}{(\salg, \prob)} $,
  if $
  {\almostM{\inv{\aexpr}(a)}{(\salg, \prob)}}
  $
  for all $a \in A$.
When
  $ \event_1 \subs \event \subs \event_2$
  and
  $ \prob(\event_1)=\prob(\event_2)=p $,
  we can unambiguously assign probability~$p$ to $\event$,
  as any extension of~$\prob$ to $\Full{\Outcomes}$ must
  assign~$p$ to $\event$;
  then we write~$\prob(\event)$ for $p$.
\end{definition}

While almost-measurability does not imply measurability,
it constrains the current probability space to contain enough information
to uniquely determine the distribution of~$\aexpr$ in any
extension where~$\aexpr$ becomes measurable.
For example let $X=\set{\store | \store(\p{x}) = 42}$ and
$ \salg = \sigma(\set{\event}) = \set{\Store,\emptyset,X,\Store\setminus X}$.
If $ \prob(X) = 1 $, then $ \almostM{\p{x}}{(\salg,\prob)} $
holds but $\p{x}$ is not measurable in~$\salg$, as $\salg$ lacks events
for $\p{x}=v$ for all~$v$ except~$42$.
Nevertheless, any extension $(\salg',\prob') \extOf (\salg,\prob)$
where $\p{x}$ is measurable,
would need to assign $\prob'(X) = 1$ and
$\prob(\p{x}=v) = 0$ for every $v \ne 42$.

We arrive at the definition of the assertion
$\distAs{\aexpr\at{i}}{\prob}$ which requires
$\aexpr\at{i}$ to be almost-measurable,
determining its distribution as~$\prob$ in any extension
of the local probability space.
Formally, given $ \prob \of \Dist(\Full{A}) $ and $\aexpr \from \Store \to A$,
we define:
\begin{align*}
  \distAs{\aexpr\at{i}}{\prob} & \is
  \E \m{\salg},\m{\prob}.
  \Own{\m{\salg},\m{\prob}} *
  \pure{
    \almostM{\aexpr}{(\m{\salg}(i),\m{\prob}(i))}
    \land
    \prob = \m{\prob}(i) \circ \inv{\aexpr}
  }
\end{align*}
The assertion states that we own just enough information about the probability
space at index~$i$, so that its distribution is uniquely determined as~$\prob$ in any extension of the space.

Using the $\distAs{\aexpr\at{i}}{\prob}$ assertion
we can define a number of useful derived assertions:
\begin{align*}
  \expectOf{\aexpr\at{i}} = r & \is
    \E \prob.
      \distAs{\aexpr\at{i}}{\prob} *
      \pure{
r = \Sum*_{a\in\psupp(\prob)} \prob(a) \cdot a
      }
  &
  \sure{\aexpr\at{i}} &\is
\distAs{(\aexpr \in \true)\at{i}}{\dirac{\True}}
  \\
  \probOf{\aexpr\at{i}} = r & \is
    \E \prob.
    \distAs{\aexpr\at{i}}{\prob} *
    \pure{
      \prob(\true) = r
    }
  &
  \own{\aexpr\at{i}} &\is
    \E \prob. \distAs{\aexpr\at{i}}{\prob}
\end{align*}
Assertions about
expectations ($\expectOf{\aexpr\at{i}}$) and
probabilities ($\probOf{\aexpr\at{i}}$),
simply assert ownership of some distribution with the desired (pure) property.
The ``almost surely'' assertion
$\sure{\aexpr\at{i}}$ takes a boolean-valued expression $\aexpr$ and
asserts that it holds (at~$i$) with probability 1.
As remarked in \cref{ex:indip-prod}, an assertion like
$\sure{\Ip{x}{1}=\Ip{y}{1}}$ owns the expression~$(\Ip{x}{1}=\Ip{y}{1})$ but not necessarily~$\Ip{x}{1}$ itself:
the only events needed to make the expression almost measurable are
$ \Ip{x}{1}=\Ip{y}{1} $ and $\Ip{x}{1}\ne\Ip{y}{1}$,
which would be not enough to
make~$\Ip{x}{1}$ itself almost measurable.
This means that an assertion like
$ \own{\Ip{x}{1}} * \sure{\Ip{x}{1}=\Ip{y}{1}} $ is satisfiable.

\paragraph{Permissions}
The previous example highlights the difficulty with supporting mutable state:
owning $ \distAs{\Ip{x}{1}}{\prob} $ is not enough to allow safe mutation,
because the frame can record information like~$\sure{\Ip{x}{1}=\Ip{y}{1}}$,
which could be invalidated by an assignment to $\p{x}$.
Our solution is analogous to the ``variables as resource''
technique in Separation Logic~\cite{BornatCY06},
and uses the permission component of \thelogic's~RA.
To manipulate permissions we define the assertions:
\begin{align*}
  \perm{\ip{x}{i}:q} &\is
    \E\m{\psp},\m{\permap}.
      \Own{\m{\psp},\m{\permap}}
      * \pure{\m{\permap}(i)(\p{x}) = q}
  &
  P\withp{\m{\permap}} &\is
    P \land \E\m{\psp}.\Own{\m{\psp}, \m{\permap}}
\end{align*}
Now owning $\perm{\Ip{x}{1}:1}$ forbids any frame to retain information
about $\Ip{x}{1}$: any resource compatible with $\perm{\Ip{x}{1}:1}$
would have a \salgebra\ which is trivial on $\Ip{x}{1}$.
In practice, preconditions are always of the form
$ P\withp{\m{\permap}} $ where $\m{\permap}$ contains full permissions
for every variable the relevant program mutates,
and non-zero permissions for the other variables referenced in the assertions or program.
When framing, one would distribute evenly the permissions to each separated
conjunct, according to the variables mentioned in the assertions.
We illustrate this pattern concretely in \cref{ex:perm-triples}.

\begin{wrapfigure}[4]{R}{23ex}\centering
$
  \let\LabTirName\RuleNameLbl \infer*[lab=and-to-star]{
  \idx(P) \inters \idx(Q) = \emptyset
}{
  P \land Q \proves P \sepand Q
}   \label{rule:and-to-star}
$
\end{wrapfigure}
\paragraph{Relevant indices}
Sometimes it is useful to determine which indices are relevant for an assertion.
Semantically, we can determine if the indices $J \subs I$ are irrelevant to~$P$
by
$
\irrel_J(P) \is
  \forall a \in \Model_I \st
    \bigl(
      \exists \pr{a} \st
        \raValid(\pr{a})
        \land
        a = \pr{a} \setminus J
        \land P(\pr{a})
    \bigr)
    \implies P(a).
$
The set $\idx(P)$ is the smallest subset of $I$ so that
$ \irrel_{I\setminus \idx(P)}(P) $ holds.
\Cref{rule:and-to-star} states that separation between resources
that live in different indexes is the same as normal conjunction:
distributions at different indexes are neither independent nor correlated;
they simply live in ``parallel universes'' and can be related as needed.

\subsection{\SuperCond}

As we discussed in \cref{sec:overview},
the centerpiece of \thelogic{} is the \supercond\ modality,
which we can now define fully formally.

\begin{definition}[\Supercond\ modality]
  \label{def:c-mod}
  Let $ \prob \in \Dist(\Full{A}) $
  and $ K \from A \to \HAssrt_I $,
  then we define the assertion
  $ \CMod{\prob} K \of \HAssrt_I $
  as follows
  (where $ \m{\krnl}(I)(v) \is \m[i: \m{\krnl}(i)(v) | i \in I] $):
  \begin{align*}
    \CMod{\prob} K &\is
    \fun a.
    \begin{array}[t]{@{}r@{\,}l@{}}
      \E \m{\sigmaF}, \m{\mu}, \m{\permap}, \m{\krnl}.
      & (\m{\sigmaF}, \m{\mu}, \m{\permap}) \raLeq a
      \land
      \forall i\in I\st
      \m{\mu}(i) = \bind(\prob, \m{\krnl}(i))
   \\ & \land \;
   \forall v \in \psupp(\prob).
   K(v)(\m{\sigmaF}, \m{\krnl}(I)(v), \m{\permap})
    \end{array}
  \end{align*}

\end{definition}
The definition follows the principle we explained in
\cref{sec:overview:supercond}:
$ \CMod{\prob} K $ holds on resources where we own some
tuple of probability spaces which can all be seen
as the convex combinations of the same~$\prob$ and some kernel.
Then the conditional assertion~$K(v)$ is required to hold on the
tuple of kernels evaluated at~$v$.
Note that the definition is upward-closed by construction.

\begin{mathfig}[\small]
  \begin{proofrules}
    \infer*[lab=c-true]{}{
  \proves \CC\prob \wtv.\True
}     \label{rule:c-true}

    \infer*[lab=c-unit-l]{}{
  \CC{\dirac{v_0}} v.K(v)
  \lequiv
  K(v_0)
}     \label{rule:c-unit-l}

    \infer*[lab=c-transf]{
f \from \psupp(\prob') \to \psupp(\prob)
  \;\text{ bijective}
  \\\\
  \forall b \in \psupp(\prob') \st
    \prob'(b) = \prob(f(b))
}{
  \CC\prob a.K(a)
  \proves
  \CC{\prob'} b.K(f(b))
}     \label{rule:c-transf}

    \infer*[lab=c-and]{
  \idx(K_1) \inters \idx(K_2) = \emptyset
}{
  \CC{\prob} v. K_1(v)
    \land
  \CC{\prob} v. K_2(v)
  \proves
  \CC\prob v.
    (K_1(v) \land K_2(v))
}     \label{rule:c-and}

    \infer*[lab=sure-str-convex]{}{
  \CC\prob v.(K(v) * \sure{\aexpr\at{i}})
  \proves
  \sure{\aexpr\at{i}} * \CC\prob v.K(v)
}     \label{rule:sure-str-convex}

    \infer*[lab=c-pure]{}{
  \pure{\prob(\event)=1} * \CC\prob v.K(v)
  \lequiv
  \CC\prob v.(\pure{v \in \event} * K(v))
}     \label{rule:c-pure}
  \end{proofrules}
  \caption{Primitive Conditioning Laws.}
  \label{fig:cond-laws}
\end{mathfig}

We discussed a number of \supercond\ laws in \cref{sec:overview}.
\Cref{fig:cond-laws} shows some important primitive laws
that were left out.
\Cref{rule:c-true} allows to introduce a trivial modality;
together with \ref{rule:c-frame} this allows for the introduction
of the modality around any assertion.
\Cref{rule:c-unit-l} is a reflection of the left unit rule of the underlying
monad: conditioning on the Dirac distribution can be eliminated.
\Cref{rule:c-transf} allows for the transformation of the
convex combination using~$\prob$
into using~$\prob'$ by applying a bijection between their support
in a way that does not affect the weights of each outcome.
\Cref{rule:c-and} allows to merge two modalities using the same~$\prob$,
provided the inner conditioned assertions do not overlap
in their relevant indices.

\Cref{rule:sure-str-convex} internalizes a stronger version of convexity of
$ \sure{\aexpr\at{i}} $ assertions.
When $K(v) = \True$ we obtain convexity
$
  \CC\prob v.\sure{\aexpr\at{i}}
  \proves
  \sure{\aexpr\at{i}}.
$
Additionally the rule asserts that the unconditional~$\sure{\aexpr\at{i}}$
keeps being independent of the conditional $K$.

Finally, \cref{rule:c-pure} allows to translate facts that hold with probability~1 in~$\prob$ to predicates that hold on every~$v$ bound by conditioning on~$\prob$.

We can now give the general encoding of relational lifting
in terms of \supercond.

\begin{definition}[Relational Lifting]
\label{def:rel-lift}
  Let $X \subs I \times \Var$;
  given a relation~$R$ between variables in $X$,
  \ie $R \subs \Val^{X}$,
  we define
  (letting
  $
    \sure{\ip{x}{i} = \m{v}(\ip{x}{i})}_{\ip{x}{i}\in X} \is
      \LAnd_{\ip{x}{i}\in X}
        \sure{\ip{x}{i} = \m{v}(\ip{x}{i})}
  $):
  \begin{align*}
    \cpl{R} &\is
      \E \prob \of \Dist(\Val^{X}).
        \pure{\prob(R) = 1} *
        \CC\prob \m{v}.
          \sure{\ip{x}{i} = \m{v}(\ip{x}{i})}_{\ip{x}{i}\in X}
  \end{align*}
\end{definition}

\begin{example}
Let us expand \cref{def:rel-lift} on
$\cpl{ \Ip{k}{1} = \Ip{c}{2} }$.
The assertion concerns the variables
$X = \set{\Ip{k}{1}, \Ip{c}{2}}$;
to instantiate the definition we see the assertion as
the lifting $ \cpl{R_{=}} $
of the relation $ R_{=} \subs \Val^X $
defined as
$
  R_{=} =
    \set{ \m{v} \in \Val^{X}
        | \m{v}(\Ip{k}{1}) = \m{v}(\Ip{c}{2}) }
  ,
$
giving rise to the assertion
\[
  \E \prob.
    \pure{\prob(R_{=}) = 1} *
    \CC\prob \m{v}.
      \sure{\Ip{k}{1} = \m{v}(\Ip{k}{1})}
      \land
      \sure{\Ip{c}{2} = \m{v}(\Ip{c}{2})}
\]
Here, $\Val^X$ can be alternatively presented as
$ \Val^2 $, giving
$ R_{=} \equiv \set{(v,v) | v \in \Val} $.
With this reformulation, the encoding of \cref{def:rel-lift} becomes
\[
  \E \prob.
    \pure{\prob(R_{=}) = 1} *
    \CC\prob (v_1,v_2).
      \sure{\Ip{k}{1} = v_1}
      \land
      \sure{\Ip{c}{2} = v_2}
\]
Thanks to \ref{rule:c-pure}, the assertion can be rewritten as
$
  \E \prob.
    \CC\prob (v_1,v_2).
      \pure{R_{=}(v_1,v_2)} *
      \sure{\Ip{k}{1} = v_1}
      \land
      \sure{\Ip{c}{2} = v_2}
$
which can be simplified to
$
  \E \prob.
  \CC\prob (v_1,v_2).\bigl(
    \sure{\Ip{k}{1} = v_1} \land
    \sure{\Ip{c}{2} = v_2} \land
    \pure{v_1=v_2}
  \bigr)
$
(which is how we presented the encoding in \cref{sec:overview:supercond}).
Since $R_{=}$ is so simple,
by \ref{rule:c-transf} we can simplify the assertion even further and obtain
$
  \E \prob.
  \CC\prob v.\bigl(
    \sure{\Ip{k}{1} = v} \land
    \sure{\Ip{c}{2} = v}
  \bigr).
$
\end{example}

In \cref{rule:rl-merge},
the two relations might refer to different indexed variables,
\ie $R_1\in \Val^{X_1}$ and $R_2\in \Val^{X_2}$;
the notation $R_1 \land R_2$ is defined as
$
  R_1 \land R_2 \is
    \set*{ \m{s} \in \Val^{X_1\union X_2}
      | \restr{\m{s}}{X_1} \in R_1 \land \restr{\m{s}}{X_2} \in R_2
    }.
$

\subsection{Weakest Precondition}

To reason about (hyper-)programs,
we introduce a \emph{weakest-precondition assertion}~(WP)
$\WP {\m{t}} {Q}$, which intuitively states:
given the current input distributions (at each index),
if we run the programs in $\m{t}$ at their corresponding index
we obtain output distributions that satisfy~$Q$;
furthermore, every frame is preserved.
We refer to the number of indices of $\m{t}$ as the \emph{arity} of the WP.

\begin{definition}[Weakest Precondition]
\label{def:wp}
For $a\in\Model_I$ and $\m{\prob} \of \Dist(\Full{\Hyp{\Store}})$
let $ a \raLeq \m{\prob} $ mean
$
  a \raLeq (\Full{\Hyp{\Store}},\m{\prob},\fun x.1).
$
\[
  \WP {\m{t}} {Q} \is
    \fun a.
      \forall \m{\prob}_0.
        \forall c \st
        (a \raOp c) \raLeq \m{\prob}_0
        \implies
        \exists b \st
        \bigl(
          (b \raOp c) \raLeq \sem{\m{t}}(\m{\prob}_0)
          \land
          Q(b)
        \bigr)
\]
\end{definition}
The assertion holds on the resources~$a$ such that
if, together with some frame~$c$,
they can be seen as a fragment of the global distribution
$\m{\prob}_0$, then it is possible to update the resource
to some~$b$ which still composes with the frame~$c$,
and~$b\raOp c$ can be seen as a fragment of the output distribution
$\sem{\m{t}}(\m{\prob}_0)$.
Moreover, such~$b$ needs to satisfy the postcondition~$Q$.

We discussed some of the WP rules of \thelogic\ in \cref{sec:overview};
the full set of rules is produced in \appendixref{sec:appendix:rules}.
Let us briefly mention the axioms for assignments:
\begin{proofrules}\small
    \infer*[lab=wp-samp]{}{
  \perm{\ip{x}{i} : 1}
  \proves
  \WP {\m[i: \code{x:~$\dist$($\vec{v}$)}]}
      {\distAs{\ip{x}{i}}{\dist(\vec{v})}}
}     \label{rule:wp-samp}

    \infer*[lab=wp-assign]{
  \p{x} \notin \pvar(\expr)
  \\
  \forall \p{y} \in \pvar(\expr) \st
    \m{\permap}(\ip{y}{i}) > 0
  \\
  \m{\permap}(\ip{x}{i})=1
}{
  (\m{\permap})
  \proves
  \WP {\m[i: \code{x:=}\expr]}[\big] {
    \sure{\p{x}\at{i} = \expr\at{i}}\withp{\m{\permap}}
  }
}
     \label{rule:wp-assign}
  \end{proofrules}
\Cref{rule:wp-samp} is the expected ``small footprint'' rule for
sampling; the precondition only requires full permission on the variable
being assigned, to forbid any frame to record information about it.
\Cref{rule:wp-assign} requires full permission on \p{x},
and non-zero permission on the variables on the RHS of the assignment.
This allows the postcondition to assert that \p{x} and the expression~$\expr$
assigned to it are equal with probability~1.
The condition $\p{x} \notin \pvar(\vec{\expr})$ ensures $\expr$ has the same
meaning before and after the assignment, but is not restrictive:
if needed the old value of \p{x} can be stored in a temporary variable,
or the proof can condition on \p{x} to work with its pure value.

The assignment rules are the only ones that impose constraints on the owned
permissions.
In proofs, this means that most manipulations simply thread through permissions
so that the needed ones can reach the applications of the assignment rules.
To avoid cluttering derivations with this bookkeeping,
we mostly omit permission information from assertions.
The appropriate permission annotations can be easily inferred,
as we show in the following example.

\begin{example}
\label{ex:perm-triples}
  Consider the following triple with permissions omitted:
  \[
    \distAs{\Ip{x}{1}}{\prob_1} *
    \sure{\Ip{x}{1}=\Ip{y}{1}} *
    \distAs{\Ip{z}{1}}{\prob_2}
    \proves
    \WP {\m[\I1: \code{x:=z}]} {
      \sure{\Ip{x}{1}=\Ip{z}{1}} *
      \distAs{\Ip{z}{1}}{\prob_2}
    }
  \]
  To be able to apply \cref{rule:wp-assign},
  we need to get $\perm{x:1, z:q}$ from the precondition,
  for some $q>0$.
  To do so, formally, we need to be more explicit about the permissions owned.
  The pattern is that whenever a triple is considered, the
  precondition should own full permissions on the variables
  assigned to in the program term, and non-zero permission on the other relevant variables.
  In our example we need permission~1 for $\Ip{x}{1}$ and arbitrary permissions
  $q_1,q_2>0$ for $\Ip{y}{1}$ and $\Ip{z}{1}$ respectively.
  Since we have two separated sub-assertions that refer to $\Ip{x}{1}$,
  we would split the full permission into two halves,
  obtaining the precondition:
  \[
    (\distAs{\Ip{x}{1}}{\mu_1})\withperm{\Ip{x}{1}:\onehalf} *
    \sure{\Ip{x}{1}=\Ip{y}{1}}\withperm{\Ip{x}{1}:\onehalf,\Ip{y}{1}:q_1} *
    (\distAs{\Ip{z}{1}}{\prob_2})\withperm{\Ip{z}{1}:q_2}
  \]
  To obtain the full permission on $\Ip{x}{1}$ we are now forced to consume
  both the first two resources, weakening the precondition to:
  \[
    \perm{\Ip{x}{1}:1} *
    \perm{\Ip{y}{1}:q_1} *
    (\distAs{\Ip{z}{1}}{\prob_2})\withperm{\Ip{z}{1}:q_2}
  \]
  This step in general forces the consumption of any frame
  recording information about the assigned variables.
  To obtain non-zero permission for $\Ip{z}{1}$ while still being able to
  frame $\distAs{\Ip{z}{1}}{\prob_2}$,
  we let~$q = q_2/2$ and weaken the precondition to:
  \[
    \perm{\Ip{x}{1}:1, \Ip{z}{1}:q} *
    \perm{\Ip{y}{1}:q_1} *
    (\distAs{\Ip{z}{1}}{\prob_2})\withperm{\Ip{z}{1}:q}
  \]
  Now an application of \ref{rule:wp-frame} and \ref{rule:wp-assign}
  would give us a postcondition:
  \[
    \sure{\Ip{x}{1}=\Ip{z}{1}}\withperm{\Ip{x}{1}:1, \Ip{z}{1}:q} *
    \perm{\Ip{y}{1}:q_1} *
    (\distAs{\Ip{z}{1}}{\prob_2})\withperm{\Ip{z}{1}:q}
  \]
  which is strong enough to imply the desired postcondition.
  In a fully expanded proof, one would keep the permissions
  owned in the postcondition so that they can be used in
  proofs concerning the continuation of the program.
\end{example}

 \section{Case Studies for \thelogic}
\label{sec:discussion}
\label{sec:examples}
Our evaluation of \thelogic\ is based on two main lines of enquiry: (1) Are high-level principles about probabilistic reasoning provable from the core constructs of \thelogic? (2) Does \thelogic, through enabling new reasoning patterns, expand the horizon for verification of probabilistic programs
\emph{beyond} what was possible before?
We include case studies that try to highlight the contribution of \thelogic\ each question, and sometimes both at the same time.
Specifically, our evaluation is guided by the following research questions:
\begin{enumerate}[label=\textbf{RQ\arabic*:},ref=\textbf{RQ\arabic*}]
\item Do joint conditioning and independence offer a good abstract interface over the underlying semantic model?
\label{rq:abstract}
\item Can known unary/relational principles be reconstructed from \thelogic's primitives?
\label{rq:known}
\item Can new unary/relational principles be discovered (as new lemmas) and proved from \thelogic's primitives?
\label{rq:new}
\item Can \thelogic's primitives be successfully incorporated in an effective \emph{program} logic?\label{rq:programs}
\end{enumerate}
We already demonstrated positive answers to some of these questions
in \cref{sec:overview}:
for example, the proof of the One-time pad addresses
  \labelcref{rq:abstract,rq:known},
the proof of \ref{rule:seq-swap} addresses \labelcref{rq:new,rq:programs}.
In this section we provide a more detailed evaluation
through a number of challenging examples.
The full proofs of the case studies and additional examples
are in \appendixref{sec:appendix:examples}.
Here, we summarize some highlights to frame the key contributions of \thelogic.

\subsection{pRHL-style Reasoning}
\label{sec:ex:prhl-style}

Our first example is an encoding of pRHL's judgments in \thelogic,
sketching how pRHL-style reasoning can be
effectively embedded and extended in \thelogic{}
(\labelcref{rq:abstract,rq:known,rq:new,rq:programs}).

In pRHL, the semantics of triples implicitly always conditions
on the input store,
so that programs are always seen as running from a pair of
\emph{deterministic} input stores satisfying the relational precondition.
Judgments in pRHL have the form
$
  \proves t_1 \sim t_2 : R_0 \implies R_1
$
where $R_0,R_1$ are two relations on states
(the pre- and postcondition, respectively)
and $t_1,t_2$ are the two programs to be compared.
Such a judgment can be encoded in \thelogic\ as:
\begin{equation}
  \cpl{R_0}
  \proves
  \E \prob.
    \CC\prob \m{s}.(
      \var{St}(\m{s}) \land
      \WP {\m[\I1:t_1,\I2:t_2]} {\cpl{R_1}}
    )
  \quad
  \text{where}
  \quad
  \var{St}(\m{s}) \is
    \sure{\ip{x}{i}=\m{s}(\ip{x}{i})}_{\ip{x}{i}\in I\times\Var}
  \label{triple:prhl}
\end{equation}
As the input state is always conditioned,
and the precondition is always a relational lifting,
one is always in the position of applying \ref{rule:c-cons}
to eliminate the implicit conditioning of the lifting and the one wrapping the
WP, reducing the problem to a goal where the input state is deterministic
(and thus where the primitive rules of WP laws apply without need for
further conditioning).
As noted in \cref{sec:overview:obox},
LHC-style WPs allow us to lift our unary WP rules
to binary with little effort.

An interesting property of the encoding in~\eqref{triple:prhl} is that
anything of the form $ \CC\prob \m{s}.(\var{St}(\m{s}) \land \dots) $
has ownership of the full store (as it conditions on every variable).
We observe that WPs (of any arity) which have this property
enjoy an extremely powerful rule.
Let $ \ownall \is \A \ip{x}{i} \in I\times\Var.\own{\ip{x}{i}} $.
The following is a valid (primitive) rule in \thelogic:
\begin{proofrule}
  \infer*[lab=c-wp-swap]{}{
  \CC\prob v.
    \WP{\m{t}}{Q(v)}
    \land \ownall
  \proves
  \WP{\m{t}}*{\CC\prob v.Q(v)}
}   \label{rule:c-wp-swap}
\end{proofrule}

\Cref{rule:c-wp-swap},
allows the shift of the conditioning on the input to the conditioning of the output.
This rule provides a powerful way to make progress in lifting
a conditional statement to an unconditional one.
To showcase \ref{rule:c-wp-swap},
consider the two programs in \cref{fig:cond-swap-code}, which
are equivalent:
if we couple the \p{x} in both programs,
the other two samplings can be coupled under conditioning on \p{x}.
Formally, let $ P \gproves Q \is P \land \ownall \proves Q \land \ownall $.
We process the two assignments to $\p{x}$, which we can couple
$
  \distAs{\Ip{x}{1}}{d_0} *
  \distAs{\Ip{x}{2}}{d_0}
  \proves
  \CC{d_0} v.(\sure{\Ip{x}{1}=v} \land \sure{\Ip{x}{2}=v})
$.
Then, let $t_1$ ($t_2$) be the rest of \p{prog1} (\p{prog2}).
We can then derive:
\par
\begin{derivation}
\infer*[right=\ref{rule:rl-convex}]{
\infer*[Right=\ref{rule:c-wp-swap}]{
\infer*[Right=\ref{rule:c-cons}]{
\infer*[Right=\ref{rule:rl-merge}]{
\infer*[Right=\ref{rule:coupling}]{
  \forall v\st
  \sure{\Ip{x}{1}=v} \land \sure{\Ip{x}{2}=v}
  \gproves
  \WP {\m[\I1: t_1, \I2: t_2]}*{
  \begin{matrix*}[l]
    \cpl{\Ip{x}{1} = \Ip{x}{2}} *
    \distAs{\Ip{y}{1}}{d_1(v)} *
    \distAs{\Ip{y}{2}}{d_1(v)} *
    {}\\
    \distAs{\Ip{z}{1}}{d_2(v)} *
    \distAs{\Ip{z}{2}}{d_2(v)}
  \end{matrix*}
  }
}{
  \forall v\st
  \sure{\Ip{x}{1}=v} \land \sure{\Ip{x}{2}=v}
  \gproves
  \WP {\m[\I1: t_1, \I2: t_2]} {
    \cpl{\Ip{x}{1} = \Ip{x}{2}} *
    \cpl{\Ip{y}{1} = \Ip{y}{2}} *
    \cpl{\Ip{z}{1} = \Ip{z}{2}}
  }
}}{
  \forall v\st
  \sure{\Ip{x}{1}=v} \land \sure{\Ip{x}{2}=v}
  \gproves
  \WP {\m[\I1: t_1, \I2: t_2]} {
    \cpl{\Ip{x}{1} = \Ip{x}{2} \land
         \Ip{y}{1} = \Ip{y}{2} \land
         \Ip{z}{1} = \Ip{z}{2}}
  }
}}{
  \CC{d_0} v.(\sure{\Ip{x}{1}=v} \land \sure{\Ip{x}{2}=v})
  \gproves
  \CC{d_0} v.
  \WP {\m[\I1: t_1, \I2: t_2]} {
    \cpl{\Ip{x}{1} = \Ip{x}{2} \land
         \Ip{y}{1} = \Ip{y}{2} \land
         \Ip{z}{1} = \Ip{z}{2}}
  }
}}{
  \CC{d_0} v.(\sure{\Ip{x}{1}=v} \land \sure{\Ip{x}{2}=v})
  \gproves
  \WP {\m[\I1: t_1, \I2: t_2]} {
  \CC{d_0} v.
    \cpl{\Ip{x}{1} = \Ip{x}{2} \land
         \Ip{y}{1} = \Ip{y}{2} \land
         \Ip{z}{1} = \Ip{z}{2}}
  }
}}{
  \CC{d_0} v.(\sure{\Ip{x}{1}=v} \land \sure{\Ip{x}{2}=v})
  \gproves
  \WP {\m[\I1: t_1, \I2: t_2]} {
    \cpl{\Ip{x}{1} = \Ip{x}{2} \land
         \Ip{y}{1} = \Ip{y}{2} \land
         \Ip{z}{1} = \Ip{z}{2}}
  }
}
\end{derivation}
\vspace{\belowdisplayskip}

Where the top triple can be easily derived using standard steps.
Reading it from bottom to top, we start by invoking convexity of
relational lifting to introduce a conditioning modality in the postcondition
matching the one in the precondition.
\Cref{rule:c-wp-swap} allows us to bring the whole WP under the modality,
allowing \cref{rule:c-cons} to remove it on both sides.
From then it is a matter of establishing and combining
the couplings on \p{y} and \p{z}.
Note that these couplings are only possible because the coupling
on \p{x} made the parameters of $d_1$ and of $d_2$ coincide on both indices.
In \cref{sec:ex:von-neumann} we show this kind of derivation can be useful
for unary reasoning too.

While the $\ownall$ condition is restricting,
without it the rule is unsound in the current model.
We leave it as future work to study whether there is a model
that validates this rule without requiring~$\ownall$.

\subsection{One Time Pad Revisited}
\label{sec:ex:one-time-pad}
In \cref{sec:overview}, we prove the \p{encrypt} program correct relationally
(missing details are in \appendixref{sec:appendix:examples:onetimerel}).
An alternative way of stating and proving the correctness of \p{encrypt}
is to establish that in the output distribution \p{c} and \p{m} are independent,
which can be expressed as the \emph{unary} goal (also studied in~\cite{barthe2019probabilistic}):
$
(\m{\permap})
  \proves
  \WP {\m[\I1: \code{encrypt()}]} {
    \distAs{\Ip{c}{1}}{\Ber{1/2}} *
    \distAs{\Ip{m}{1}}{\Ber{p}}
  }
$
(where $\m{\permap} = \m[\Ip{k}{1}:1,\Ip{m}{1}:1,\Ip{c}{1}:1]$).
The triple states that after running \p{encrypt},
the ciphertext \p{c} is distributed as a fair coin,
and---importantly---is \emph{not} correlated with the plaintext in \p{m}.
The PSL proof in~\cite{barthe2019probabilistic} performs some of the steps
within the logic, but needs to carry out some crucial entailments at the meta-level, which is a symptom of unsatisfactory abstractions~(\labelcref{rq:abstract}). The same applies to the Lilac proof in~\cite{lilac2} which requires  ad-hoc lemmas proven on the semantic model.
The stumbling block is proving the valid entailment:
\[
  \distAs{\Ip{k}{1}}{\Ber{\onehalf}} *
  \distAs{\Ip{m}{1}}{\Ber{p}} *
  \sure{\Ip{c}{1} = \Ip{k}{1} \xor \Ip{m}{1}}
  \proves
  \distAs{\Ip{m}{1}}{\Ber{p}} *
  \distAs{\Ip{c}{1}}{\Ber{\onehalf}}
\]
In \thelogic\ we can prove the entailment in two steps:
(1) we condition on \p{m} and \p{k} to compute the result of
the \p{xor} operation and obtain that \p{c} is distributed as $\Ber{\onehalf}$;
(2) we carefully eliminate the conditioning while preserving the independence
of \p{m} and \p{c}.

The first step starts by conditioning on \p{m} and \p{k} and proceeds as follows:
\begin{eqexplain}
  & \CC{\Ber{p}} m.
    \bigl(
      \sure{\Ip{m}{1}=m} *
      \CC{\Ber{\onehalf}} k.
        (\sure{\Ip{k}{1}=k} *
        \sure{\Ip{c}{1} = k \xor m})
    \bigr)
\whichproves
  \CC{\Ber{p}} m.
    \left(
    \sure{\Ip{m}{1}=m} *
    \begin{cases}
      \CC{\Ber{\onehalf}} k. \sure{\Ip{c}{1}=k} \CASE m=0
      \\
      \CC{\Ber{\onehalf}} k. \sure{\Ip{c}{1}=\neg k} \CASE m=1
    \end{cases}
    \right)
  \byrule{c-cons}
\whichproves
  \CC{\Ber{p}} m.
    \bigl(
      \sure{\Ip{m}{1}=m} *
      \CC{\Ber{\onehalf}} k. \sure{\Ip{c}{1}=k}
    \bigr)
  \byrule{c-transf}
\end{eqexplain}The crucial entailment is the application of \ref{rule:c-transf} to the $m=1$ branch,
by using negation as the bijection
(which satisfies the premises of the rules since $\Ber{\onehalf}$ is unbiased).

The second step uses the following primitive rule of \thelogic:
\begin{proofrule}
  \infer*[lab=prod-split]{}{
  \distAs{(\aexpr_1\at{i}, \aexpr_2\at{i})}{\prob_1 \otimes \prob_2}
  \proves
  \distAs{\aexpr_1\at{i}}{\prob_1} *
  \distAs{\aexpr_2\at{i}}{\prob_2}
}
   \label{rule:prod-split}
\end{proofrule}
with which we can prove:
\begin{eqexplain}
& \CC{\Ber{p}} m.
    \bigl(
      \sure{\Ip{m}{1}=m} *
      \CC{\Ber{\onehalf}} k. \sure{\Ip{c}{1}=k}
    \bigr)
\whichproves
  \CC{\Ber{p}} m.
  \CC{\Ber{\onehalf}} k.
    \sure{\Ip{m}{1}=m \land \Ip{c}{1}=k}
  \ifappendix \byrules{c-frame,sure-merge}\else \byrules{c-frame}\fi \whichproves
  \CC{\Ber{p} \pprod \Ber{\onehalf}} (m,k).
    \sure{(\Ip{m}{1},\Ip{c}{1})=(m,k)}
  \byrules{c-fuse}
\whichproves
  \distAs{(\Ip{m}{1},\Ip{c}{1})}{(\Ber{p} \pprod \Ber{\onehalf})}
  \byrule{c-unit-r}
\whichproves
  \distAs{\Ip{m}{1}}{\Ber{p}} *
  \distAs{\Ip{c}{1}}{\Ber{\onehalf}}
  \byrule{prod-split}
\end{eqexplain}

As this is a common manipulation needed to extract unconditional independence
from a conditional fact, we can formulate it as the more general
derived rule
\begin{proofrule}
  \infer*[lab=c-extract]{}{
  \CC{\prob_1} v_1. \bigl(
    \sure{\aexpr_1\at{i} = v_1} *
    \distAs{\aexpr_2\at{i}}{\prob_2}
  \bigr)
  \proves
  \distAs{\aexpr_1\at{i}}{\prob_1} *
  \distAs{\aexpr_2\at{i}}{\prob_2}
}   \label{rule:c-extract}
\end{proofrule}

\subsection{Markov Blankets}
\label{sec:ex:markov-blanket}

In probabilistic reasoning, introducing conditioning is easy,
but deducing unconditional facts from conditional ones is not immediate.
The same applies to the \supercond\ modality: by design, one cannot eliminate it for free.
Crucial to \thelogic's expressiveness is the inclusion of rules that can
soundly derive unconditional information from conditional assertions.

We use the concept of a \emph{Markov Blanket}---a very common tool in Bayesian reasoning
  for simplifying conditioning---to illustrate \thelogic's expressiveness (\labelcref{rq:abstract,rq:known}).
Intuitively, Markov blankets identify a set of variables that affect the distribution of a random variable \emph{directly}:
this is useful because by conditioning on those variables
we can remove conditional connection between the random variable
and all the variables on which it \emph{indirectly} depends.

For concreteness, consider the program
\code{x1:~$\dist_1$;
x2:~$\dist_2(\p{x1})$;
x3:~$\dist_3(\p{x2})$}.
The program describes a Markov chain of three variables.
One way of interpreting this pattern is that the joint output distribution
is described by the program as a product of conditional distributions:
the distribution over \p{x2} is described conditionally on \p{x1},
and the one of \p{x3} conditionally on \p{x2}.
This kind of dependencies are ubiquitous in, for instance, hidden Markov models and Bayesian network representations of distributions.

A crucial tool for the analysis of such models is the concept of a
Markov Blanket of a variable~\p{x}: the set of variables that are direct dependencies of~\p{x}.
Clearly~\p{x3} depends on~\p{x2} and, indirectly, on~\p{x1}.
However, Markov chains enjoy the memorylessness property:
when fixing a variable in the chain, the variables that follow it are independent from the variables that preceded it.
For our example this means that if we condition on~\p{x2}, then
\p{x1} and~\p{x3} are independent (\ie we can ignore the indirect dependencies).

In \thelogic\ we can characterize the output distribution with the assertion
\[
  \CC{\dist_1} v_1. \Bigl(
    \sure{\p{x1}=v_1} *
    \CC{\dist_2(v_1)} v_2. \bigl(
      \sure{\p{x2}=v_2} *
      \distAs{\p{x3}}{\dist_3(v_2)}
    \bigr)
  \Bigr)
\]
Note how this postcondition represents the output distribution
as implicitly as the program does.
We want to transform the assertion into:
\[
  \CC{\prob_2} v_2.
  \bigl(
    \sure{\p{x2}=v_2} *
    \distAs{\p{x1}}{\prob_1(v_2)} *
    \distAs{\p{x3}}{\dist_3(v_2)}
  \bigr)
\]
for appropriate $\prob_2$ and $\prob_1$.
This isolates the conditioning to the direct dependency of \p{x1}
and keeps full information about \p{x3},
available for further manipulation down the line.

In probability theory, the proof of memorylessness is an application
of Bayes' law: we are computing
the distribution of \p{x1} conditioned on \p{x2},
from the distribution of \p{x2} conditioned on \p{x1}.

In \thelogic\ we can produce the transformation using the \supercond\ rules,
in particular the right-to-left direction of \ref{rule:c-fuse}
and the primitive rule that is behind its left-to-right
direction:
\begin{proofrule}
  \infer*[lab=c-unassoc]{}{
  \CC{\bind(\prob,\krnl)} w.K(w)
  \proves
  \CC\prob v. \CC{\krnl(v)} w.K(w)
}   \label{rule:c-unassoc}
\end{proofrule}

Using these we can prove:
\begin{eqexplain}
  &
  \CC{\dist_1} v_1. \Bigl(
    \sure{\p{x1}=v_1} *
    \CC{\dist_2(v_1)} v_2. \bigl(
      \sure{\p{x1}=v_2} *
      \distAs{\p{x3}}{\dist_3(v_2)}
    \bigr)
  \Bigr)
\whichproves
  \CC{\dist_1} v_1. \Bigl(
    \CC{\dist_2(v_1)} v_2. \bigl(
      \sure{\p{x1}=v_1} *
      \sure{\p{x1}=v_2} *
      \distAs{\p{x3}}{\dist_3(v_2)}
    \bigr)
  \Bigr)
  \byrules{c-frame}
\whichproves
  \CC{\prob_0} (v_1,v_2). \bigl(
      \sure{\p{x1}=v_1} *
      \sure{\p{x2}=v_2} *
      \distAs{\p{x3}}{\dist_3(v_2)}
  \bigr)
  \byrules{c-fuse}
\whichproves
  \CC{\prob_2} v_2. \Bigl(
    \CC{\prob_1(v_2)} v_1.
    \bigl(
      \sure{\p{x1}=v_1} *
      \sure{\p{x2}=v_2} *
      \distAs{\p{x3}}{\dist_3(v_2)}
    \bigr)
  \Bigr)
  \byrules{c-unassoc}
\whichproves
  \CC{\prob_2} v_2. \Bigl(
    \sure{\p{x2}=v_2} *
    \CC{\prob_1(v_2)} v_1.
    \bigl(
      \sure{\p{x1}=v_1} *
      \distAs{\p{x3}}{\dist_3(v_2)}
    \bigr)
  \Bigr)
  \byrules{sure-str-convex}
\whichproves
  \CC{\prob_2} v_2. \bigl(
    \sure{\p{x2}=v_2} *
    \distAs{\p{x1}}{\prob_1(v_2)} *
    \distAs{\p{x3}}{\dist_3(v_2)}
  \bigr)
  \byrules{c-extract}
\end{eqexplain}
where
$
  \dist_1 \fuse \dist_2 = \prob_0 =
  \bind(\prob_2,\prob_1).
$
The existence of such $\prob_2$ and $\prob_1$ is a simple application
of Bayes' law:
$
  \prob_2(v_2) =
    \Sum_{v_1 \in \Val} \prob_0(v_1,v_2),
$
and
$
  \prob_1(v_2)(v_1) =
    \frac{\prob_0(v_1,v_2)}{\prob_2(v_2)}.
$
We see the ability of \thelogic\ to perform these manipulations
as evidence that \supercond\ and independence form a sturdy abstraction
over the semantic model~(\labelcref{rq:abstract}).
The amount of meta-reasoning required to manipulate the distributions
indexing the conditioning modality are minimal and localized,
and offer a good entry-point to inject facts about distributions
without interfering with the rest of the proof context.

\subsection{Multi-party Secure Computation}
\label{sec:ex:multiparty}

In \emph{multi-party secure computation}, the goal is to for~$N$
parties to compute a function~$f(x_1,\dots,x_N)$ of
some private data~$x_i$ owned by each party~$i$,
without revealing any more information about~$x_i$ than the output of~$f$
would reveal if computed centrally by a trusted party.
For example, if $f$ is addition, a secure computation of~$f$ can be used
to compute the total number of votes without revealing who voted positively:
some information would leak (e.g. if the total is non-zero then \emph{somebody} voted positively) but only what is revealed by knowing the total and nothing more.

To achieve this objective, multi-party secure addition~(\p{MPSAdd})
works by having the parties break their secret into~$N$ \emph{secret shares}
which individually look random, but the sum of which amounts to the original secret.
These secret shares are then distributed to the other parties so that each party knows an incomplete set of shares of the other parties.
Yet, each party can reliably compute the result of the function by computing a function of the received shares.

As it is very often the case, there is no single ``canonical'' way of specifying
this kind of security property.
For  \p{MPSAdd}, for instance, we can formalize security
(focusing on the perspective of party~1)
in two ways:
as a unary or as a relational specification.

The \emph{unary specification} says that,
  conditionally on the secret of party~$1$
  and the sum of the other secrets,
  all the values received by~$1$ (we call this the \emph{view} of~$1$)
  are independent from the secrets of the other parties.
Roughly:
\begin{equation*}
  \distAs{(\p{x}_1, \p{x}_2, \p{x}_3)\at{\I1}}{\prob_0}
  \proves
  \WP {\m[\I1: \p{MPSAdd}]}*{
    \E\prob.
    \CC {\prob} {(v, s)}.
    \begin{grp}
      \sure{\p x_1\at{\I1} = v \land (\p x_2 + \p x_3)\at{\I1}=s} * {}
      \\
      \own{\p{view}_1\at{\I1}} * \own{\p x_2\at{\I1},\p x_3\at{\I1}}
\end{grp}
  }
\end{equation*}
where $\prob_0$ is an arbitrary distribution of the three secrets.
Notice how conditioning nicely expresses that the acceptable leakage is just the sum.

The \emph{relational specification} says that
    when running the program from two initial states
    differing only in the secrets of the other parties,
    but not in their sum,
    the views of party~$i$ would be distributed in the same way.
Roughly:
\begin{equation*}
  \cpl*{
  \begin{conj*}
    \p x_1\at{\I1} = \p x_1\at{\I2}
    \land
    (\p x_2+\p x_3)\at{\I1} = (\p x_2+\p x_3)\at{\I2}
  \end{conj*}
  }
\proves
  \WP {\m<1:\p{MPSAdd},2:\p{MPSAdd}>}[\bigg]{
    \cpl*{
\p{view}_1\at{\I1} = \p{view}_1\at{\I2}
}
  }
\end{equation*}

The two specifications look quite different and also suggest quite different
proof strategies: the unary judgment suggests a proof by manipulating independence and conditioning; the relational one hints at a proof by relational lifting. Depending on the program, each of these strategies could have their merits.
As a first contribution, we show that \thelogic\ can not only specify in both styles (\labelcref{rq:abstract}), but also provide
\emph{proofs} in both styles (\labelcref{rq:programs}).

Having two very different specifications for the same security goal, however begs the question:
are they equivalent?
After all, as the \emph{prover} of the property one might prefer one proof style over the other, but as a \emph{consumer} of the specification the choice might
be dictated by the proof context that needs to use the specification for
proving a global goal.
To decouple the proof strategy from the uses of the specification,
we would need to be able to convert one specification into the other
\emph{within} the logic, thus sparing the prover from having to forsee
which specification a proof context might need in the future.

Our second key result is that in fact the equivalence between
the unary
and the relational specification
can be proven in \thelogic.
This is enabled by the powerful \supercond\ rules and the encoding
of relational lifting as \supercond.
This is remarkable as this type of result has always been justified
entirely at the level of the semantic model in other logics (\eg pRHL, Lilac).
This illustrates the fitness of \thelogic\ as a tool for abstract
meta-level reasoning (\labelcref{rq:abstract}).

In~\appendixref{sec:appendix:ex:multiparty} we provide \thelogic\ proofs for:
\begin{enumerate*}
  \item the unary specification;
  \item the relational specification (independently of the unary proof);
  \item the equivalence of the two specifications.
\end{enumerate*}
Although the third item would spare us from proving one of the first two,
we provide direct proofs in the two styles to provide a point of comparison
between them.

\subsection{Von Neumann Extractor}
\label{sec:ex:von-neumann}

A randomness extractor is a mechanism that transforms a stream of
``low-quality'' randomness sources into a stream of ``high-quality''
randomness sources.
The von Neumann extractor~\cite{vonNeumann}
is perhaps the earliest instance of such mechanism,
and it converts a stream of independent coins with the same bias~$p$
into a stream of independent \emph{fair} coins.
Verifying the correctness of the extractor requires careful reasoning
under conditioning, and showcases the use of \cref{rule:c-wp-swap} in a
unary setting (\labelcref{rq:known,rq:programs}).

We can model the extractor, up to $N \in \Nat$ iterations, in our language\footnote{While technically our language does not support arrays,
  they can be easily encoded as a collection of~$N$ variables.
}
as shown in \cref{fig:von-neumann}.
The program repeatedly flips two biased coins, and outputs the outcome of the first coin if the outcomes where different, otherwise it retries.
As an example, we prove in \thelogic~ that the bits produced in \p{out} are independent fair coin flips.
Formally, for $\ell$ produced bits, we want the following to hold:
\[
  \var{Out}_\ell \is
  \distAs{\p{out}[0]\at{\I1}}{\Ber{\onehalf}} *
  \dots *
  \distAs{\p{out}[\ell-1]\at{\I1}}{\Ber{\onehalf}}.
\]
To know how many bits were produced, however,
we need to condition on \p{len}
obtaining the specification
(recall $ P \gproves Q \is P \land \ownall \proves Q \land \ownall $):
\[
  \gproves \WP {\m[\I1: \p{vn}(N)]}*{
    \E \prob. \CC \prob \ell. \bigl(
      \sure{\Ip{len}{1} = \ell \leq N} *
      \var{Out}_\ell
    \bigr)
  }
\]

\begin{wrapfigure}[12]{R}{30ex}\begin{sourcecode*}[gobble=2,aboveskip=0pt,belowskip=0pt]
  def vn($N$):
    len := 0
    repeat $N$:
      coin_1 :~ Ber($p$)
      coin_2 :~ Ber($p$)
      if coin_1 != coin_2 then:
        out[len] := coin_1
        len := len+1
  \end{sourcecode*}
  \caption{Von Neumann extractor.}
  \label{fig:von-neumann}
\end{wrapfigure}

The postcondition straightforwardly generalizes to a loop invariant
\[
  P(i) =
  \E \prob. \CC \prob \ell. \bigl(
    \sure{\Ip{len}{1} = \ell \leq i} *
    \var{Out}_\ell
  \bigr)
\]
The main challenge in the example is handling the if-then statement.
Intuitively we want to argue that if $ \p{coin}_1 \ne \p{coin}_2 $,
the two coins would have either values $(0,1)$ or $(1,0)$,
and both of these outcomes have probability $ p(1-p) $;
therefore,
conditionally on the `if' guard being true,
$\p{coin}_1$ is a fair coin.

Two features of \thelogic\ are crucial to implement the above intuition.
The first is the ability of manipulating conditioning given by the \supercond\ rules.
At the entry point of the if-then statement in \cref{fig:von-neumann}
we obtain
$
  P(i) *
  \distAs{\p{coin}_1\at{\I1}}{\Ber{p}} *
  \distAs{\p{coin}_2\at{\I1}}{\Ber{p}}.
$
Using \thelogic's rules we can easily derive
$
  P(i) *
  \distAs{(\p{coin}_1 \ne \p{coin}_2,\p{coin}_1)\at{\I1}}{\prob_0}
$
for some $\prob_0$.
The main insight of the algorithm then can be expressed as the fact that
$\prob_0 = \beta \fuse \krnl$ for some $\beta\of\Dist(\set{0,1})$
which is the distribution of $\p{coin}_1 \ne \p{coin}_2$,
and some~$\krnl$ describing the distribution of the first coin in the two
cases, which we know is such that $\krnl(1)=\Ber{\onehalf}$.
Then, thanks to \ref{rule:c-unit-r} and \ref{rule:c-fuse}, we obtain:
\begin{align*}
&\distAs{(\p{coin}_1 \ne \p{coin}_2,\p{coin}_1)\at{\I1}}{(\beta \fuse \krnl)}
\\ {}\proves{}&
\CC \beta b. \bigl(
  \sure{(\p{coin}_1 \ne \p{coin}_2)\at{\I1} = b} *
  \pure{b=1} \implies
    \distAs{\p{coin}_1\at{\I1}}{\Ber{\onehalf}}
\bigr)
\end{align*}
Then, if we could reason about the `then' branch under conditioning,
since the guard $\p{coin}_1 \ne \p{coin}_2$ implies $b=1$ we would obtain
$ \distAs{\p{coin}_1\at{\I1}}{\Ber{\onehalf}} $, which is the key to the proof.
The ability of reasoning under conditioning is the second feature
of \thelogic\ which unlocks the proof.
In this case, the step is driven by \cref{rule:c-wp-swap},
which allows us to prove the if-then statement by case analysis on~$b$.

\subsection{Monte Carlo Algorithms}
\label{sec:ex:monte-carlo}

By elaborating on the Monte Carlo example of \cref{sec:intro},
we want to show the fitness of \thelogic\ as a program logic
(\labelcref{rq:programs}) and its specific approach for dealing
with the structure of a program.
Recall the example in Figure \ref{fig:between-code} and the goal
outlined in \cref{sec:intro} of comparing the accuracy of the two
Monte Carlo algorithms \p{BETW\_SEQ} and \p{BETW}.
This goal can be encoded as
\[
  \begin{conj}
  \sure{\Ip{l}{1}=\Ip{r}{1}=0} *{}\\
  \sure{\Ip{l}{2}=\Ip{r}{2}=0}
  \end{conj}
  \withp{\m{\permap}}
  \proves
  \WP {\m<
    \I1: \code{BETW_SEQ($x$, $S$)},
    \I2: \code{BETW($x$, $S$)}
  >}[\bigg]{
    \cpl{\Ip{d}{1} \leq \Ip{d}{2}}
  }
\]
(where $\m{\permap}$ contains full permissions for all the variables)
which, through the relational lifting, states that it is more likely
to get a positive answer from \p{BETW} than from \p{BETW\_SEQ}.
The challenge is implementing the intuitive relational argument
sketched in \cref{sec:intro},
in the presence of very different looping structures.
More precisely, we want to compare the sequential composition of two loops
$ l_1 = (\Loop{N}{\tA}\p;\Loop{N}{\tB}) $
with a single loop
$ l_2 = \Loop{(2N)}{t} $
considering the $N$ iterations of~$\tA$ in lockstep with the first~$N$ iterations of~$l_2$, and the $N$ iterations of $\tB$ with the remaining~$N$ iterations of $l_2$.
It is not possible to perform such proof purely in pRHL, which can only handle loops that are perfectly aligned, and tools based on pRHL overcome this limitation by offering a number of code transformations, proved correct externally to the logic, with which one can rewrite the loops so that they syntactically align. In this case such a transformation could look like
$ \Loop{(M+N)}{t} \equiv \Loop{M}{t}\p;\Loop{N}{t} $,
using which one can rewrite $l_2$ so it aligns with the two shorter loops.
What \thelogic\ can achieve is to avoid the use of such ad-hoc syntactic transformations, and produce a proof structured in two steps: first, one can prove, \emph{within the logic}, that it is sound to align the loops as described; and then proceed with the proof of the aligned loops.

The key idea is that the desired alignment of loops can be expressed
as a (derived) rule, encoding the net effect of the syntactic loop splitting,
without having to manipulate the syntax:
\begin{proofrule}
  \infer*[lab=wp-loop-split]{
  P_1(N_1) \proves P_2(0)
  \\\\
  \forall i < N_1 \st
    P_1(i) \proves \WP{\m[\I1: t_1, \I2: t]}{P_1(i+1)}
  \\\\
  \forall j < N_2 \st
    P_2(j) \proves \WP{\m[\I1: t_2, \I2: t]}{P_2(j+1)}
}{
  P_1(0) \proves
  \WP{\m[
    \I1: (\Loop{N_1}{t_1}\p;\Loop{N_2}{t_2}),
    \I2: \Loop{(N_1+N_2)}{t}
  ]}{P_2(N_2)}
}   \label{rule:wp-loop-split}
\end{proofrule}
The rule considers two programs: a sequence of two loops, and a single loop
with the same cumulative number of iterations.
It asks the user to produce two relational loop invariants~$P_1$ and~$P_2$
which are used to relate $N_1$ iterations of $t_1$ and $t$ together,
and $N_2$ iterations of $t_2$ and $t$ together.

Crucially,
such rule is \emph{derivable}
from the primitive rules of looping of \thelogic:
\begin{proofrules}\small
  \infer*[lab=wp-loop,right=$n\in\Nat$]{
  \forall i < n\st
  P(i) \proves
  \WP {\m[j: t]} {P(i+1)}
}{
  P(0) \proves
  \WP {\m[j: \Loop{n}{t}]} {P(n)}
}   \label{rule:wp-loop}

  \infer*[lab=wp-loop-unf]{}{
  {\begin{array}{@{}r@{}l@{}}
    &\WP {\m[i: \Loop{n}{t}]} {
      \WP {\m[i: t]} { Q }
    }
    \\\proves{}&
    \WP {\m[i: \Loop{(n+1)}{t}]} {Q}
  \end{array}}
}   \label{rule:wp-loop-unf}
\end{proofrules}
\Cref{rule:wp-loop} is a standard unary invariant-based rule;
\ref{rule:wp-loop-unf} simply reflects the
semantics of a loop in terms of its unfoldings.
Using these
we can prove \ref{rule:wp-loop-split}
avoiding semantic reasoning all together,
and fully generically on the loop bodies,
allowing it to be reused in any situation fitting the pattern.

In our example, we can prove our goal by instanting it with the loop invariants:
\begin{align*}
  P_1(i) &\is
    \cpl{
      \Ip{r}{1}\leq\Ip{r}{2}
      \land
      \Ip{l}{1}=0\leq\Ip{l}{2}
    }
  &
  P_2(j) &\is
    \cpl{
      \Ip{r}{1}\leq\Ip{r}{2}
      \land
      \Ip{l}{1}\leq\Ip{l}{2}
    }
\end{align*}

\begingroup \newcommand{\tBet}[1]{t_{\p{M}}^{#1}}

\begin{figure*}
  \lstset{gobble=2}\hfill
  \hbox{\begin{sourcecode*}
  def BETW_MIX($x$,$S$):
    repeat $N$:
      p :~ $\prob_S$; l := l || p <= $x$
      q :~ $\prob_S$; r := r || q >= $x$
    d := r && l
  \end{sourcecode*}}
  \hfill
  \begin{tabular}{cc}
  \begin{sourcecode*}
  def prog1:
    x :~ $d_0$
    y :~ $d_1$(x)
    z :~ $d_2$(x)
  \end{sourcecode*}
  &
  \begin{sourcecode*}
  def prog2:
    x :~ $d_0$
    z :~ $d_2$(x)
    y :~ $d_1$(x)
  \end{sourcecode*}
  \end{tabular}
  \hfill\null

  \begin{minipage}{.4\linewidth}
    \caption{A variant of the \p{BETW} program.}
    \label{fig:betw-mix-code}
  \end{minipage}\begin{minipage}{.4\linewidth}
    \caption{Conditional Swapping}
    \label{fig:cond-swap-code}
  \end{minipage}
\end{figure*}

This handling of structural differences as derived proof patterns
is more powerful than syntactic transformations:
it can, for example, handle transformations that are sound only under some
assumptions about state.
To show an instance of this,
we consider a variant of the previous example:
\p{BETW\_MIX} (in \cref{fig:betw-mix-code})
is another variant of \p{BETW\_SEQ}
which still makes~$2N$ samples but interleaves sampling
for the minimum and for the maximum.
We want to prove that this is equivalent to \p{BETW\_SEQ}.
Letting $\m{\permap}$ contain full permissions for the relevant variables,
the goal is \[
  P_0\withp{\m{\permap}}
  \proves
  \WP {\m[
    \I1: \code{BETW_SEQ($x, S$)},
    \I2: \code{BETW_MIX($x, S$)}
  ]} {
    \cpl{\Ip{d}{1} = \Ip{d}{2}}
  }
\]
with $P_0 = \sure{{\Ip{l}{1}=\Ip{r}{1}=0}}*\sure{{\Ip{l}{2}=\Ip{r}{2}=0}}$.

Call $\tBet{1}$ and $\tBet{2}$ the first and second half of the body of the loop
of \p{BETW\_MIX}, respectively.
The strategy
is to consider together one execution of $\tA$
(the body of the loop of \p{AboveMin}),
and $\tBet{1}$;
and one of~$\tB$ (of \p{BelowMax}),
and~$\tBet{2}$.
The strategy relies on the observation that every iteration of the three loops
is \emph{independent} from the others.
To formalize the proof idea we thus first prove a derived proof pattern
encoding the desired alignment, which we can state for generic~$t_1,t_2,t_1',t_2'$:
\begin{proofrule}
  \infer*[lab=wp-loop-mix]{
  \forall i < N \st
    P_1(i) \proves \WP{\m[\I1: t_1, \I2: t_1']}{P_1(i+1)}
  \\
  \forall i < N \st
    P_2(i) \proves \WP{\m[\I1: t_2, \I2: t_2']}{P_2(i+1)}
}{
  P_1(0) * P_2(0)
  \proves
  \WP{\m[
    \I1: (\Loop{N}{t_1}\p;\Loop{N}{t_2}),
    \I2: \Loop{N}{(t_1';t_2')}
  ]}{P_1(N) * P_2(N)}
}   \label{rule:wp-loop-mix}
\end{proofrule}
The rule matches on two programs: a sequence of two loops,
and a single loop with a body split into two parts.
The premises require a proof that $t_1$ together with $t_1'$ (the first half of the body of the second loop) preserve the invariant $P_1$;
and that the same is true for $t_2$ and $t_2'$ with respect to an invariant~$P_2$.
The precondition $P_1(0)*P_2(0)$ in the conclusion ensures that the two
loop invariants are independent.
The rule \ref{rule:wp-loop-mix} can be again entirely derived
from \thelogic's primitive rules.
We can then apply it to our example
using as invariants
$
  P_1 \is \cpl{\Ip{l}{1} = \Ip{l}{2}}
$ and $
  P_2 \is \cpl{\Ip{r}{1} = \Ip{r}{2}}.
$
Then, \ref{rule:rl-merge} closes the proof.

\endgroup

 \section{Related Work}
\label{sec:relwork}

Research on deductive verification of probabilistic programs has developed a
wide range of techniques that employ {\em unary} and  {\em relational} styles of reasoning. \thelogic\ advances the state of the art in both styles, by coherently unifying the strengths of both. We limit our comparison here to deductive techniques only, and focus most of our attention on explaining how \thelogic\ offers new  reasoning tools compared to these.

\paragraph{\bfseries Unary-style Reasoning.}
Early work in this line focuses more on analyzing marginal distributions and probabilities, and features like harnessing the power of  probabilistic independence and conditioning have been more recently added to make more expressive program logics~\cite{ramshaw1979formalizing,rand2015vphl,barthe2016ellora,barthe2019probabilistic,bao2022separation,lilac}.

Much work in this line has been inspired by {\em Separation Logic}~(SL),
a powerful tool
for reasoning about pointer-manipulating programs,
known for its support of \emph{local reasoning}
of separated program components~\cite{reynolds2000intuitionistic}.
PSL~\cite{barthe2019probabilistic} was the first logic to present a SL model for reasoning about the probabilistic independence of program variables, which facilitates modular reasoning about independent components within a probabilistic program.
In~\cite{bao2021bunched} and~\cite{bao2022separation} SL variants are used for reasoning about conditional independence and negative dependence, respectively;
both are used in algorithm analysis as relaxations of
independence.

\paragraph{Lilac}
Lilac~\cite{lilac} is the most recent addition to this group and introduces a new foundation of probabilistic separation logic based on measure theory.
It enables reasoning about independence and conditional independence uniformly in one logic and supports continuous distributions.
\thelogic\ also uses a measure-theory based model, similar to Lilac,
although limited to discrete distributions.
While \thelogic\ uses Lilac's independent product as a model of separating conjunction, it differs from Lilac in three aspects:
(1) the treatment of ownership,
(2) support for mutable state, and
(3) the model of conditioning.

Ownership as almost-measurability is required to support inferences like
$
  \own{\p{x}} * \sure{\p{x}=\p{y}}
  \proves
  \own{\p{y}},
$
which were implicitly used in the first version of Lilac,
but were not valid in its model.
\citet{lilac2} fixes the issue by changing the meaning of $\sure{\p{x}=\p{y}}$,
while our fix acts on the meaning of ownership
(and we see $\sure{E}$ assertions as an instance of regular ownership).

Lilac works with
immutable state~\cite{staton2020},
which simplifies reasoning in certain contexts
(e.g., the frame rule and the if rule).
\thelogic's model supports mutable state through a creative use of permissions,
obtaining a clean frame rule, at the cost of some predictable bookkeeping.

The more significant difference with Lilac is however in the definition of the conditioning modality.
Lilac's modality~$\LC{v}{E}\,P(v)$ is indexed by a random variable~$E$,
and roughly corresponds to the \thelogic\ assertion
$ \E \prob. \CC \prob v.(\sure{E=v} * P(v)) $.
The difference is not merely syntactic,
and requires changing the model of the modality.
For example, Lilac's modality satisfies
$
  {\LC{v}{E}\,P_1(v) \land \LC{v}{E}\,P_2(v)}
  \proves
  {\LC{v}{E}\,(P_1(v) \land P_2(v))},
$
but the analogous rule
$
  {\CC{\prob} v. K_1(v)
    \land
  \CC{\prob} v. K_2(v)}
  \proves
  {\CC\prob v.
    (K_1(v) \land K_2(v))}
$
(corresponding to \ref{rule:c-and} without the side condition)
is unsound in \thelogic:
The meaning of the modalities in the premise ensures
the \emph{existence} of two kernels $\krnl_1$ and $\krnl_2$ supporting
$K_1$ and $K_2$ respectively,
but the conclusion requires the existence of a \emph{single} kernel
supporting both~$K_1$ and~$K_2$.
Lilac's rule holds because when one conditions on a random variable,
the corresponding kernels are unique.
We did not find losing this rule limiting.
On the other hand,
Lilac's conditioning has two key disadvantages:
(i)  it does not record the distribution of~$E$,
     losing this information when conditioning,
(ii) it does not generalize to the relational setting.
Even considering only the unary setting,
having access to the distribution~$\prob$ in fact unlocks a number of new
rules (\eg \ref{rule:c-unit-r} and \ref{rule:c-fuse})
that are key to the increased expressivity of \thelogic.
In particular, the rules of \thelogic\ provide a wider arsenal of tools
that can convert a conditional assertion back into an unconditional one.
This is especially important when conditioning is used as a reasoning tool,
regardless of whether the end goal is a conditional statement.

\paragraph{\bfseries Relational Reasoning}
\citet{barthe2009formal} extend relational Hoare logic~\cite{benton2004simple} to reason about probabilistic programs in a logic called pRHL (probabilistic Relational Hoare Logic).
In pRHL, assertions on pairs of deterministic program states are lifted to assertions on pairs of distributions, and on the surface, the logic simply manipulates the deterministic assertions.
A number of variants of pRHL were successfully applied to proving various cryptographic protocols and differential privacy algorithms~\cite{barthe2009formal,barthe2015coupling,hsu2017probabilistic,wang2019proving, zhang2017lightdp}.
When a natural relational proof for an argument exists, these logics are simple and elegant to use. However, they fundamentally
trade expressiveness for ease of use.
A persisting problem with them has been that they rely on a strict structural alignment between the order of samples in the two programs. Recall our discussion in \cref{sec:overview:obox} for an example of this that \thelogic\ can handle.
\citet{gregersen2023asynchronous} recently proposed Clutch,
a logic to prove contextual refinement in a probabilistic higher-order language,
where ``out of order'' couplings between samplings are achieved by
using ghost code that pre-samples some assignments,
a technique inspired by \emph{prophecy variables}~\cite{jung2019future}.
In \cref{sec:overview} we showed how \thelogic\ can resolve the issue
without ghost code
(in the context of first-order imperative programs) by using framing and probabilistic independence creatively.
In contrast to \thelogic, Clutch can only express relational properties;
it also uses separation but with its classical interpretation as disjointness
of deterministic state.

Polaris~\cite{tassarotti2019polaris}, is an early instance of a probabilistic relational (concurrent) separation logic.
However, separation in Polaris is again classic disjointness of state.

Our \pre n-ary WP is inspired by LHC~\cite{d2022proving},
which shows how arity-changing rules (like \ref{rule:wp-nest})
can accommodate modular and flexible relational proofs of deterministic programs.

\paragraph{\bfseries Other Techniques.}
Expectation-based approaches, which reason about expected quantities of probabilistic programs via a weakest-pre-expectation operator that propagates information about expected values backwards through the program, have been classically used to verify randomized algorithms~\cite{kozen1983PDL,Morgan:1996,kaminski2016weakest,kaminski2019thesis,aguirre2021pre,Bartocci2022moment}.
These logics offer ergonomic dedicated principles for expectations, but do not aim at unifying principles for analyzing more general classes of properties or proof techniques, like we attempt here.
Ellora~\cite{barthe2016ellora} proposes an assertion-based logic (without separation nor conditioning) to overcome the limitation of working only with expectations.

 \section{Conclusions and Future Work}
\thelogic's journey started as a quest to integrate unary and relational
probabilistic reasoning and ended up uncovering \supercond{}
as a key fundational tool.
Remarkably, to achieve our goal we had to deviate from Lilac's previous
proposal in both the definition of conditioning,
  to enable the encoding of relational lifting,
and of ownership (with almost measurability),
  to resolve an issue with almost sure assertions
(recently corrected~\cite{lilac2} in a different way).
In addition, our model supports mutable state without sacrificing
expressiveness.
One limitation of our current model is lack of support for continuous
distributions.
Lilac's model and recent advances in it~\cite{LiAJ0H24} could suggest a pathway for a continuous extension of \thelogic,
but it is unclear if all our rules would be still valid;
for example \cref{rule:c-fuse}'s soundness hinges on properties of discrete distributions that we could not extend to the general case in an obvious way.
\thelogic's encoding of relational lifting and the novel proof principles it uncovered for it are a demonstration of the potential
of \supercond\ as a basis for deriving high-level logics on top of an ergonomic
core logic.
Obvious candidates for such scheme are approximate couplings~\cite{apRHL}
(which have been used for \eg differential privacy),
and expectation-based calculi (à la Ellora).
 
\begin{acks}
  We would like to thank Justin Hsu
  for connecting the authors and the many discussions.
  We also thank Derek Dreyer and Deepak Garg for the
  discussions and support.
  We are grateful to the POPL'25 reviewers for
  their constructive feedback.
  Jialu Bao was supported by the
  \grantsponsor{nsf}
    {NSF}
    {https://doi.org/10.13039/100000001}
  Award No.~\grantnum{nsf}{2153916}.
  Emanuele D'Osualdo was supported by a
  \grantsponsor{persist}
    {European Research Council (ERC)}
    {https://doi.org/10.13039/501100000781}
  Consolidator Grant for the project ``PERSIST'' under the European Union's Horizon 2020 research and innovation programme
  (grant No.~\grantnum{persist}{101003349}).
Azadeh Farzan was supported by the
  \grantsponsor{canada}
    {National Science and Engineering Research Council of Canada}
    {https://doi.org/10.13039/501100000038}
  Discovery Grant.
\end{acks}

\setlabel{LAST}
\label{paper-last-page}

\newcommand{\SortNoop}[1]{}

\appendix

\allowdisplaybreaks

\section{The Rules of \thelogic}
\label{sec:appendix:rules}

In this section we list all the rules of \thelogic,
including some omitted (but useful) rules in addition to those
that appear in the main text.
Since our treatment is purely semantic,
rules are simply lemmas that hold in the model.
Although we do not aim for a full axiomatization,
we try to identify the key proof principles that apply to each of our connectives.
For brevity, we omit the rules that apply to the basic connectives of separation logic, as they are well-known and have been proven correct for any model that is an RA. For those we refer to \cite{KrebbersJ0TKTCD18}.

In \cref{fig:assertions} we summarize the notation we use for
assertions over \thelogic's model.
Recall that \thelogic's assertions
$P \in \HAssrt_I \is \Model_I \ucto \Prop $
are the upward-closed predicates over elements of
the RA $\Model_I$.

\begin{figure}[h]
\adjustfigure \begin{align*}
    \pure{\varphi} &\is \fun \wtv. \varphi
    \\
    \Own{b} &\is \fun a. b \raLeq a
    \\
    P \land Q &\is \fun a.
        P(a) \land Q(a)
    \\
    P * Q &\is \fun a.
      \exists b_1,b_2 \st
        (b_1 \raOp b_2) \raLeq a \land
        P(b_1) \land
        Q(b_2)
    \\
    \E x \of X. K &\is \fun a.
      \exists x \of X \st
        K(x)(a)
    & (K\from X \to \HAssrt_I)
    \\
    \A x \of X. K &\is \fun a.
      \forall x \of X \st
        K(x)(a)
    & (K\from X \to \HAssrt_I)
    \\
    \Own{\m{\salg}, \m{\prob}} &\is
      \E \m{\permap}. \Own{\m{\salg}, \m{\prob}, \m{\permap}}
    \\
    \distAs{\aexpr\at{i}}{\prob} & \is
      \E \m{\salg},\m{\prob}.
      \Own{\m{\salg},\m{\prob}} *
      \pure{
        \almostM{\aexpr}{(\m{\salg}(i),\m{\prob}(i))}
        \land
        \prob = \m{\prob}(i) \circ \inv{\aexpr}
      }
    \\
    \CMod{\prob} K &\is
    \fun a.
      \begin{array}[t]{@{}r@{\,}l@{}}
        \E \m{\sigmaF}, \m{\mu}, \m{\permap}, \m{\krnl}.
        & (\m{\sigmaF}, \m{\mu}, \m{\permap}) \raLeq a
        \land
        \forall i\in I\st
        \m{\mu}(i) = \bind(\prob, \m{\krnl}(i))
       \\ & \land \;
        \forall v \in \psupp(\prob).
          K(v)(\m{\sigmaF}, \m{\krnl}(I)(v), \m{\permap})
      \end{array}
    & (\prob \of \Dist(A), K\from A \to \HAssrt_I)
    \\
    \WP {\m{t}} {Q} &\is
      \fun a.
        \forall \m{\prob}_0.
          \forall c \st
          (a \raOp c) \raLeq \m{\prob}_0
          \implies
          \exists b \st
          \bigl(
            (b \raOp c) \raLeq \sem{\m{t}}(\m{\prob}_0)
            \land
            Q(b)
          \bigr)
    \\
    \sure{\aexpr\at{i}} &\is
\distAs{(\aexpr \in \true)\at{i}}{\dirac{\True}}
    \\
    \own{\aexpr\at{i}} &\is
      \E \prob. \distAs{\aexpr\at{i}}{\prob}
    \\
    \perm{\ip{x}{i}:q} &\is
      \E\m{\psp},\m{\permap}.
        \Own{\m{\psp},\m{\permap}}
        * \pure{\m{\permap}(i)(\p{x}) = q}
    \\
    P\withp{\m{\permap}} &\is
      P \land \E\m{\psp}.\Own{\m{\psp}, \m{\permap}}
    \\
    \cpl{R} &\is
      \E \prob \of \Dist(\Val^{X}).
        \pure{\prob(R) = 1} *
        \CC\prob \m{v}.
          \sure{\ip{x}{i} = \m{v}(\ip{x}{i})}_{\ip{x}{i}\in X}
    & (R \subs \Val^{X}, X \subs I \times \Var)
  \end{align*}
\caption{The assertions used in \thelogic.}
\label{fig:assertions}
\end{figure}

\begin{proposition}[Upward-closure]
  All the assertions in \cref{fig:assertions} are upward-closed.
\end{proposition}
\begin{proof}
  Easy by inspection of the definitions.
  The definitions where upward-closedness is less obvious
  (\eg \supercond)
  are made upward-closed by construction by explicit use of
  the order $\raLeq$ in the definition.
\end{proof}

Although we adopt a ``shallow embedding'' approach to assertions
(and thus we do not provide separate syntax),
the rules of \thelogic\ provide an axiomatic treatment
of these assertions so that the user should never manipulate raw predicates
over the semantic model.
We consider the connectives listed above WP
(included) in \cref{fig:assertions} to be the ones that the user
should never need to unfold into their definitions and only manipulate
through rules.

We make a distinction between
``primitive'' and ``derived'' rules.
The primitive rules require proofs that manipulate the semantic model definitions directly; these are the ones we would consider part of a proper axiomatization.
The derived rules can be proved sound by staying at the level of the logic,
\ie by using the primitive rules of \thelogic.

\Cref{fig:primitive-rules} lists the primitive rules for distribution ownership assertions and for the \supercond\ modality.
\Cref{fig:wp-rules} lists the primitive rules for the weakest precondition modality.
In \cref{fig:derived-rules} we list some useful derived rules,
and in \cref{fig:derived-wp-rules} we show some derived rules for WP.

We provide proofs for each rule in the form of lemmas in \cref{sec:appendix:soundness}.
The name labelling each rule is a link to the proof of soundness of the rule.

\let\RuleName\RuleNameProofLink

\begin{figure}[btp]
\adjustfigure[\small]
  \rulesection*{Distribution ownership rules}
  \begin{proofrules}
    \infer*[lab=and-to-star]{
  \idx(P) \inters \idx(Q) = \emptyset
}{
  P \land Q \proves P \sepand Q
}     \relabel{rule:and-to-star}

    \infer*[lab=dist-inj]{}{
  \distAs{\aexpr\at{i}}{\prob}
  \land
  \distAs{\aexpr\at{i}}{\prob'}
  \proves
  \pure{\prob=\prob'}
}     \label{rule:dist-inj}

    \infer*[lab=sure-merge]{}{
  \sure{\aexpr_1\at{i}} *
  \sure{\aexpr_2\at{i}}
  \lequiv
  \sure{(\aexpr_1 \land \aexpr_2)\at{i}}
}
     \label{rule:sure-merge}

    \infer*[lab=sure-and-star]{
  \psinv(P, \pvar(E\at{i}))
}{
  \sure{E\at{i}} \land P
  \proves
  \sure{E\at{i}} \ast P
}
     \label{rule:sure-and-star}

    \infer*[lab=prod-split]{}{
  \distAs{(\aexpr_1\at{i}, \aexpr_2\at{i})}{\prob_1 \otimes \prob_2}
  \proves
  \distAs{\aexpr_1\at{i}}{\prob_1} *
  \distAs{\aexpr_2\at{i}}{\prob_2}
}
     \relabel{rule:prod-split}
  \end{proofrules}
\rulesection{\Supercond rules}
\begin{proofrules}
    \infer*[lab=c-true]{}{
  \proves \CC\prob \wtv.\True
}     \relabel{rule:c-true}

    \infer*[lab=c-false]{}{
  \CC\prob v.\False
  \proves
  \False
}     \label{rule:c-false}

    \infer*[lab=c-cons]{
  \forall v\st
  K_1(v) \proves K_2(v)
}{
  \CC\prob v.K_1(v)
  \proves
  \CC\prob v.K_2(v)
}     \relabel{rule:c-cons}

    \infer*[lab=c-frame]{}{
  P * \CC\prob v.K(v)
  \proves
  \CC\prob v.(P * K(v))
}     \relabel{rule:c-frame}

    \infer*[lab=c-unit-l]{}{
  \CC{\dirac{v_0}} v.K(v)
  \lequiv
  K(v_0)
}     \relabel{rule:c-unit-l}

    \infer*[lab=c-unit-r]{}{
  \distAs{\aexpr\at{i}}{\mu}
  \lequiv
  \CC\prob v.\sure{\aexpr\at{i}=v}
}     \relabel{rule:c-unit-r}

    \infer*[lab=c-assoc]{
\prob_0 = \bind(\prob,\fun v.(\bind(\krnl(v), \fun w.\return(v,w))))
}{
  \CC{\prob} v.\CC{\krnl(v)} w.K(v,w)
  \proves
  \CC{\prob_0} (v,w).K(v,w)
}     \label{rule:c-assoc}

    \infer*[lab=c-unassoc]{}{
  \CC{\bind(\prob,\krnl)} w.K(w)
  \proves
  \CC\prob v. \CC{\krnl(v)} w.K(w)
}     \relabel{rule:c-unassoc}

    \infer*[lab=c-and]{
  \idx(K_1) \inters \idx(K_2) = \emptyset
}{
  \CC{\prob} v. K_1(v)
    \land
  \CC{\prob} v. K_2(v)
  \proves
  \CC\prob v.
    (K_1(v) \land K_2(v))
}     \relabel{rule:c-and}

    \infer*[lab=c-skolem]{
  \prob \of \Dist(\Full{A})
}{
  \CC\prob v. \E x \of X. Q(v, x)
  \proves
  \E f \of A \to X. \CC\prob v. Q(v, f(v))
}     \relabel{rule:c-skolem}

    \infer*[lab=c-transf]{
f \from \psupp(\prob') \to \psupp(\prob)
  \;\text{ bijective}
  \\\\
  \forall b \in \psupp(\prob') \st
    \prob'(b) = \prob(f(b))
}{
  \CC\prob a.K(a)
  \proves
  \CC{\prob'} b.K(f(b))
}     \relabel{rule:c-transf}

    \infer*[lab=sure-str-convex]{}{
  \CC\prob v.(K(v) * \sure{\aexpr\at{i}})
  \proves
  \sure{\aexpr\at{i}} * \CC\prob v.K(v)
}     \relabel{rule:sure-str-convex}

    \infer*[lab=c-for-all]{}{
  \CC\prob v. \A x \of X. Q(v, x)
  \proves
  \A x \of X. \CC\prob v. Q(v, x)
}     \label{rule:c-for-all}

    \infer*[lab=c-pure]{}{
  \pure{\prob(\event)=1} * \CC\prob v.K(v)
  \lequiv
  \CC\prob v.(\pure{v \in \event} * K(v))
}     \relabel{rule:c-pure}
  \end{proofrules}
  \rulesectionend
\caption{The primitive rules of \thelogic.}
\label{fig:primitive-rules}
\end{figure}

\begin{figure}[tp]
\adjustfigure[\small]
  \rulesection*{Structural WP rules}
  \begin{proofrules}
    \infer*[lab=wp-cons]{
  Q \proves Q'
}{
  \wpc{\m{t}}{Q}
  \proves
  \wpc{\m{t}}{Q'}
}     \relabel{rule:wp-cons}

    \infer*[lab=wp-frame]{}{
  P \sepand \wpc{\hpt}{Q}
  \proves
  \wpc{\hpt}{\liftA{P} \sepand Q}
}     \relabel{rule:wp-frame}

    \infer*[lab=wp-nest]{}{
  \wpc{\m{t}_1}{
    \wpc{\m{t}_2}{Q}
  }
  \lequiv
  \wpc{(\m{t}_1 \m. \m{t}_2)}{Q}
}     \relabel{rule:wp-nest}

    \infer*[lab=wp-conj]{
  \idx(Q_1) \inters \supp{\m{t}_2} \subs \supp{\m{t}_1}
    \\
  \idx(Q_2) \inters \supp{\m{t}_1} \subs \supp{\m{t}_2}
}{
  \wpc{\m{t}_1}{Q_1}
  \land
  \wpc{\m{t}_2}{Q_2}
  \proves
  \wpc{(\m{t}_1 \m+ \m{t}_2)}{Q_1 \land Q_2}
}     \label{rule:wp-conj}

    \infer*[lab=c-wp-swap]{}{
  \CC\prob v.
    \WP{\m{t}}{Q(v)}
    \land \ownall
  \proves
  \WP{\m{t}}*{\CC\prob v.Q(v)}
}     \relabel{rule:c-wp-swap}
  \end{proofrules}
  \rulesection{Program WP rules}
  \begin{proofrules}
    \infer*[lab=wp-skip]{}{
  P \proves \WP {\m[i: \code{skip}]} {P}
}     \label{rule:wp-skip}

    \infer*[lab=wp-seq]{}{
  \WP {\m[i: t]}[\big]{
    \WP {\m*[i: \smash{t'}]} {Q}
  }
  \proves
  \WP {\m[i: (t\code{;}\ t')]} {Q}
}     \relabel{rule:wp-seq}

    \infer*[lab=wp-assign]{
  \p{x} \notin \pvar(\expr)
  \\
  \forall \p{y} \in \pvar(\expr) \st
    \m{\permap}(\ip{y}{i}) > 0
  \\
  \m{\permap}(\ip{x}{i})=1
}{
  (\m{\permap})
  \proves
  \WP {\m[i: \code{x:=}\expr]}[\big] {
    \sure{\p{x}\at{i} = \expr\at{i}}\withp{\m{\permap}}
  }
}
     \relabel{rule:wp-assign}

    \infer*[lab=wp-samp]{}{
  \perm{\ip{x}{i} : 1}
  \proves
  \WP {\m[i: \code{x:~$\dist$($\vec{v}$)}]}
      {\distAs{\ip{x}{i}}{\dist(\vec{v})}}
}     \relabel{rule:wp-samp}

    \infer*[lab=wp-if-prim]{}{
{\begin{array}{@{}r@{}l@{}}
  &\ITE{v}
    {\WP{\m[i: t_1]}{Q(1)}}
    {\WP{\m[i: t_2]}{Q(0)}}
\\\proves{}&
  \WP{
    \m[i: \code{if $\;v\;$ then $\;t_1\;$ else $\;t_2$}]
  }{Q(v)}
\end{array}}
}     \label{rule:wp-if-prim}

    \infer*[lab=wp-bind]{}{
  \sure{\expr\at{i}=v} *
    \WP{\m*[i: {\Ectxt[v]}]}{Q}
  \proves
  \WP{\m*[i: {\Ectxt[\expr]}]}{Q}
}     \label{rule:wp-bind}

    \infer*[lab=wp-loop-unf]{}{
  {\begin{array}{@{}r@{}l@{}}
    &\WP {\m[i: \Loop{n}{t}]} {
      \WP {\m[i: t]} { Q }
    }
    \\\proves{}&
    \WP {\m[i: \Loop{(n+1)}{t}]} {Q}
  \end{array}}
}     \relabel{rule:wp-loop-unf}

    \infer*[lab=wp-loop,right=$n\in\Nat$]{
  \forall i < n\st
  P(i) \proves
  \WP {\m[j: t]} {P(i+1)}
}{
  P(0) \proves
  \WP {\m[j: \Loop{n}{t}]} {P(n)}
}     \relabel{rule:wp-loop}
  \end{proofrules}
  \rulesectionend
\caption{The primitive WP rules of \thelogic.}
\label{fig:wp-rules}
\end{figure}

\begin{figure}[tp]
\adjustfigure[\small]
  \rulesection*{Ownership and distributions}
  \begin{proofrules}
    \infer*[lab=sure-dirac]{}{
  \distAs{\aexpr\at{i}}{\dirac{v}}
  \lequiv
  \sure{\aexpr\at{i}=v}
}     \label{rule:sure-dirac}

    \infer*[lab=sure-eq-inj]{}{
  \sure{\aexpr\at{i} = v}
  *
  \sure{\aexpr\at{i} = v'}
  \proves
  \pure{v=v'}
}
     \label{rule:sure-eq-inj}

    \infer*[lab=sure-sub]{}{
  \distAs{\aexpr_1\at{i}}{\prob}
  *
  \sure{(\aexpr_2 = f(\aexpr_1))\at{i}}
  \proves
  \distAs{\aexpr_2\at{i}}{\prob \circ \inv{f}}
}
     \label{rule:sure-sub}

    \infer*[lab=dist-fun]{}{
  \distAs{\aexpr\at{i}}{\prob}
  \proves
  \distAs{(f\circ\aexpr)\at{i}}{\prob \circ \inv{f}}
}
     \label{rule:dist-fun}

    \infer*[lab=dirac-dup]{}{
  \distAs{\aexpr\at{i}}{\dirac{v}}
  \proves
  \distAs{\aexpr\at{i}}{\dirac{v}} *
  \distAs{\aexpr\at{i}}{\dirac{v}}
}     \label{rule:dirac-dup}

    \infer*[lab=dist-supp]{}{
  \distAs{\aexpr\at{i}}{\prob}
  \proves
  \distAs{\aexpr\at{i}}{\prob} * \sure{\aexpr\at{i} \in \psupp(\prob)}
}
     \label{rule:dist-supp}

    \infer*[lab=prod-unsplit]{}{
  \distAs{\aexpr_1\at{i}}{\prob_1} *
  \distAs{\aexpr_2\at{i}}{\prob_2}
  \proves
  \distAs{(\aexpr_1\at{i}, \aexpr_2\at{i})}{\prob_1 \otimes \prob_2}
}
     \label{rule:prod-unsplit}
    \end{proofrules}
    \rulesection{\Supercond}
    \begin{proofrules}
      \infer*[lab=c-fuse]{}{
  \CC{\prob} v.
  \CC{\krnl(v)} w.
    K(v,w)
  \lequiv
  \CC{\prob \fuse \krnl} (v,w). K(v,w)
}       \relabel{rule:c-fuse}

      \infer*[lab=c-swap]{}{
  \CC{\prob_1} v_1.
    \CC{\prob_2} v_2.
      K(v_1, v_2)
  \proves
  \CC{\prob_2} v_2.
    \CC{\prob_1} v_1.
      K(v_1, v_2)
}       \label{rule:c-swap}

      \infer*[lab=sure-convex]{}{
  \CC\prob v.\sure{\aexpr\at{i}}
  \proves
  \sure{\aexpr\at{i}}
}       \label{rule:sure-convex}

      \infer*[lab=dist-convex]{}{
  \CC\prob v.\distAs{\aexpr\at{i}}{\prob'}
  \proves
  \distAs{\aexpr\at{i}}{\prob'}
}       \label{rule:dist-convex}

      \infer*[lab=c-sure-proj]{}{
  \CC\prob (v, w).\sure{\aexpr(v)\at{i}}
  \lequiv
  \CC{\prob\circ\inv{\proj_1}} v.\sure{\aexpr(v)\at{i}}
}       \label{rule:c-sure-proj}

      \infer*[lab=c-sure-proj-many]{}{
  \CC\prob (\m{v}, w).
    \sure{\ip{x}{i}=\m{v}(\ip{x}{i})}_{\ip{x}{i}\in X}
  \lequiv
  \CC{\prob\circ\inv{\proj_1}} \m{v}.
    \sure{\ip{x}{i}=\m{v}(\ip{x}{i})}_{\ip{x}{i}\in X}
}       \label{rule:c-sure-proj-many}

      \infer*[lab=c-extract]{}{
  \CC{\prob_1} v_1. \bigl(
    \sure{\aexpr_1\at{i} = v_1} *
    \distAs{\aexpr_2\at{i}}{\prob_2}
  \bigr)
  \proves
  \distAs{\aexpr_1\at{i}}{\prob_1} *
  \distAs{\aexpr_2\at{i}}{\prob_2}
}       \relabel{rule:c-extract}

      \infer*[lab=c-dist-proj]{}{
  \CC\prob (x, y).
    \distAs{\aexpr\at{i}(x)}{\prob(x)}
  \proves \CC{\prob\circ\inv{\proj_1}} x.
    \distAs{\aexpr\at{i}(x)}{\prob(x)}
}       \label{rule:c-dist-proj}
    \end{proofrules}
    \rulesection{Relational Lifting}
    \begin{proofrules}
    \infer*[lab=rl-cons]{
  R_1 \subs R_2
}{
  \cpl{R_1} \proves \cpl{R_2}
}
     \label{rule:rl-cons}

    \infer*[lab=rl-unary]{
  R \subs \Val^{\set{\p{x}_1\at{i},\dots,\p{x}_n\at{i}}}
}{
  \cpl{R} \proves \sure{R(\p{x}_1\at{i},\dots,\p{x}_n\at{i})}
}
     \label{rule:rl-unary}

    \infer*[lab=rl-eq-dist]{i \ne j}{
  \cpl{\ip{x}{i}=\ip{y}{j}}
  \proves
  \E \prob.
    \distAs{\ip{x}{i}}{\prob}
    *
    \distAs{\ip{y}{j}}{\prob}
}     \label{rule:rl-eq-dist}

    \infer*[lab=rl-convex]{}{
  \CC\prob \wtv.\cpl{R} \proves \cpl{R}
}
     \relabel{rule:rl-convex}

    \infer*[lab=rl-merge]{}{
  \cpl{R_1} * \cpl{R_2}
  \proves
  \cpl{R_1 \land R_2}
}
     \relabel{rule:rl-merge}

    \infer*[lab=rl-sure-merge]{
  R \subs \Val^{X}
  \\
  \pvar(\expr\at{i}) \subs X
}{
  \cpl{R} * \sure{\ip{x}{i} = \expr\at{i}}
  \proves
  \cpl{R \land \p{x}\at{i} = \expr\at{i}}
}     \label{rule:rl-sure-merge}

    \infer*[lab=coupling]{
  \prob \circ \inv{\proj_1} = \prob_1
  \\
  \prob \circ \inv{\proj_2} = \prob_2
  \\
  \prob(R) = 1
}{
  \distAs{\p{x}_1\at{\I1}}{\prob_1} *
  \distAs{\p{x}_2\at{\I2}}{\prob_2}
  \proves
  \cpl{R(\p{x}_1\at{\I1}, \p{x}_2\at{\I2})}
}
     \relabel{rule:coupling}\end{proofrules}
  \rulesectionend
\caption{Derived rules.}
\label{fig:derived-rules}
\end{figure}

\begin{figure}[tp]
\adjustfigure[\small]
\begin{proofrules}
    \infer*[lab=wp-loop-0]{}{
  P \proves \WP {\m[i: \Loop{0}{t}]} {P}
}     \label{rule:wp-loop-0}

    \infer*[lab=wp-loop-lockstep,right=$n\in\Nat$]{
  \forall k < n\st
    P(k) \proves \WP {\m[i: t, j: t']}{P(k+1)}
}{
  P(0) \proves
  \WP {\m[i: (\Loop{n}{t}), j: (\Loop{n}{t'})]} {P(n)}
}     \label{rule:wp-loop-lockstep}

    \infer*[lab=wp-rl-assign]{
  R \subs \Val^{X}
  \\
  \ip{x}{i} \notin \pvar(\expr\at{i}) \subs X
  \\
  \forall \p{y} \in \pvar(\expr) \st
    \m{\permap}(\ip{y}{i}) > 0
  \\
  \m{\permap}(\ip{x}{i})=1
}{
  \cpl{R}\withp{\m{\permap}}
  \proves
  \WP {\m[i: \code{x:=}\expr]}[\big] {
    \cpl{R \land \p{x}\at{i} = \expr\at{i}}\withp{\m{\permap}}
  }
}     \label{rule:wp-rl-assign}

    \infer*[lab=wp-if-unary]{
  P * \sure{\ip{e}{\I1} = 1}
  \gproves
  \WP {\m[\I1: t_1]}{Q(1)}
  \\
  P * \sure{\ip{e}{\I1} = 0}
  \gproves
  \WP {\m[\I1: t_2]}{Q(0)}
}{
  P * \distAs{\ip{e}{\I1}}{\beta}
  \gproves
  \WP {\m[\I1: \Cond{\p{e}}{t_1}{t_2}]}[\big]{\CC{\beta} b.Q(b)}
}     \label{rule:wp-if-unary}
  \end{proofrules}
\caption{Derived WP rules.}
\label{fig:derived-wp-rules}
\end{figure}
 \section{Auxiliary Definitions}
\label{sec:appendix:definition}

\begin{definition}[Bind and return]
  Let $A$ be a countable set and $\salg$ a \salgebra.
  We define the following functions:
  \begin{align*}
    \return &\from A \to \Dist(\Full{A})
    &
    \bind &\from
      \Dist(\Full{A}) \to
      (A \to \Dist(\salg))
        \to \Dist(\salg)
    \\
    \return&(a) \is \dirac{a}
    &
    \bind&(\prob, \krnl) \is
      \fun \event \in \salg.
\sum_{a\in A} \prob(a) \cdot \krnl(a)(\event)
  \end{align*}
  We will use throughout Haskell-style notation for monadic expressions,
  for instance:
  \[
    \bigl(\DO{x <- \prob; y <- f(x); \return(x+y)}\bigr)
    \equiv
    \bind(\prob,
      \fun x.
        \bind(f(x),
          \fun y.
            \return(x+y)
        )
    )
  \]
\end{definition}

The $\bind$ and $\return$ operators form a well-known monad with
$\Dist$, and thus satisfy the monadic laws:
\begin{align*}
  \bind(\prob,\fun x.\return(x)) &= \prob
  \tag{\textsc{unit-r}}
  \label{prop:bind-unit-r}
  \\
  \bind(\return(v),\krnl) &= \krnl(v)
  \tag{\textsc{unit-l}}
  \label{prop:bind-unit-l}
  \\
  \bind(\bind(\prob,\krnl_1),\krnl_2) &=
    \bind(\prob,\fun x.\bind(\krnl_1(x),\krnl_2))
  \tag{\textsc{assoc}}
  \label{prop:bind-assoc}
\end{align*}

By \cref{thm:countable-partition-generated},
for any sigma algebra $\sigmaF'$ on countable underlying set, there exists
a partition $S$ of the underlying space that generated it, so we can transform
any such $\sigmaF'$ to a full sigma algebra over $S$.
Since we are working with countable underlying set throughout,
the requirement of $\prob$ to be over the full sigma algebra $\Full{A}$ is not an extra restriction.

\begin{definition}[Fusion operation]
  Given $\prob \of \Dist(A)$ and $\krnl \from A \to \Dist(B)$,
  we define their fusion $(\prob \fuse \krnl) \of \Dist(A \times B)$
  by:
  \[
    \prob\fuse\krnl \is
      \fun(v,w). \prob(v)\krnl(v)(w).
  \]
\end{definition}

\begin{proposition}
  $
    \prob\fuse\krnl =
    \bind(\prob,\fun v.(\bind(\krnl(v), \fun w.\return(v,w)))).
  $
\end{proposition}

\begin{proof}
  By unfolding the definitions, we obtain
  \begin{align*}
    &\bind(\prob,\fun v.(\bind(\krnl(v), \fun w.\return(v,w))))(a,b)
    \\
    &=
    \Sum_{a'\in A}
      \prob(a') \cdot
      \Sum_{b'\in B}
        \krnl(a')(b')\cdot\dirac{(a',b')}(a,b)
    \\&= \prob(a)\krnl(a)(b)
    \qedhere
  \end{align*}
\end{proof}

\begin{lemma}\label{lm:fuse-split}
  For all $\prob \of \Dist(A\times B)$,
  there exists a $\krnl \from A \to \Dist(B)$
  such that
  $ \prob = (\prob \circ \inv{\proj_1})\fuse\krnl $.
\end{lemma}

\begin{proof}
  Let $\prob_1 = \prob \circ \inv{\proj_1}$.
  Then the result is immediate by letting
  $
    \krnl(a)(b) =
      \begin{cases}
        \frac{\prob_0(a,b)}{\prob_1(a)}
          \CASE \prob_1(a) > 0
        \\
        0 \OTHERWISE
      \end{cases}
  $.
\end{proof}

\subsection{Program Semantics}

We assume each primitive operator $\prim \in \set{\p+,\p-,\p<,\dots}$
has an associated arity $\arity(\prim) \in \Nat$, and
is given semantics as some function
$ \sem{\prim} \from \Val^{\arity(\prim)} \to \Val. $
Expressions $\expr \in \Expr$ are given semantics as a function
$ \sem{\expr} \from \Store \to \Val $
as standard:
\begin{align*}
  \sem{v}(s) &\is v
  &
  \sem{\p{x}}(s) &\is s(\p{x})
  &
  \sem{\prim(\expr_1,\dots,\expr_{\arity(\prim)})}(s) &\is
    \sem{\prim}(\sem{\expr_1},\dots,\sem{\expr_{\arity(\prim)}})
\end{align*}

\begin{definition}[Term semantics]
\label{def:semantics}
  Given $\term \in \Term$ we define its \emph{kernel semantics}
  $\Sem[K]{\term} \from \Store \to \Dist(\Full{\Store}) $
  as follows:
  \begin{align*}
    \Sem[K]{\code{skip}}(s) &\is
      \return(s)
    \\
    \Sem[K]{\code{x:=}\expr}(s) &\is
        \return(s\upd{\p{x}->\sem{\expr}(s)})
    \\
    \Sem[K]{\code{x:~$\dist$($\expr_1,\dots,\expr_n$)}}(s) &\is
      \DO{
        v <- \sem{\dist}(\sem{\expr_1}(s),\dots,\sem{\expr_n}(s));
        \return(s\upd{\p{x}->v})
      }
    \\
    \Sem[K]{\term_1\p;\term_2}(s) &\is
      \DO{
        s' <- \Sem[K]{\term_1}(s);
        \Sem[K]{\term_2}(s')
      }
    \\
    \Sem[K]{\code{if$\;\expr\;$then$\;\term_1\;$else$\;\term_2\;$}}(s) &\is
      \DO{
        \ITE{\sem{\expr}(s) \ne 0}
          {\Sem[K]{\term_1}(s)}
          {\Sem[K]{\term_2}(s)}
      }
    \\
    \Sem[K]{\Loop{\expr}{\term}}(s) &\is
      \var{loop}_{\term}(\sem{\expr}(s), s)
\end{align*}
  where $\var{loop}_{\term}$ simply iterates $\term$:
  \[
    \var{loop}_{\term}(n, s) \is
      \begin{cases}
        \return(s) \CASE n \leq 0 \\
        \DO{s' <- \var{loop}_{\term}(n-1, s); \Sem[K]{\term}(s')} \OTHERWISE
      \end{cases}
  \]
The semantics of a term is then defined as:
  \begin{align*}
    \sem{\term} &\from \Dist(\Full{\Store}) \to \Dist(\Full{\Store})
    \\
    \sem{\term}(\prob) &\is \DO{s <- \prob; \Sem[K]{\term}(s)}
  \end{align*}
\end{definition}

Evaluation contexts~$\Ectxt$ are defined by the following grammar:
\begin{grammar}
  \Ectxt \is
       \Assn{x}{\pr{\Ectxt}}
    | \Sample{x}{\dist}{\vec{\expr}_1,\pr{\Ectxt},\vec{\expr}_2}
| \Cond{\pr{\Ectxt}}{\term_1}{\term_2}
    | \Loop{\pr{\Ectxt}}{\term}
  \\
  \pr{\Ectxt} \is
      [\hole]
    | \prim(\vec{\expr}_1,\pr{\Ectxt},\vec{\expr}_2)
\end{grammar}

A simple property holds for evaluation contexts.

\begin{lemma}
  \label{lemma:context-binding}
  $\Sem[K]{\Ectxt[\expr]}(s) = \Sem[K]{\Ectxt[\sem{\expr}(s)]}(s).$
\end{lemma}

\begin{proof}
  Easy by induction on the structure of evaluation contexts.
\end{proof}

\subsection{Permissions}
\label{sec:appendix:permissions}

\Cref{rule:sure-and-star}
needs a side-condition on assertions which concerns how
an assertion constrains permission ownership.
In \thelogic, most manipulations do not concern permissions,
except for when a mutation takes place, where permissions are used
to make sure the frame forgets all information about the variable to be mutated.
The notion of \emph{permission-abstract assertion} we now define
characterises the assertions which are not chiefly concerned about permissions.

An assertion~$P \in \HAssrt_I$ is \emph{permission-abstract}
with respect to some~$X \subs I \times \Var$,
written $\psinv(P, X)$, if it is invariant under scaling of permission of~$X$;
that is:
\[
  \psinv(P, X) \is
    \forall \m{\salg},\m{\prob},\m{\permap}, q, n\in \Nat\setminus\set{0}.
      P(\m{\salg},\m{\prob},\m{\permap}\m[\ip{x}{i}: q])
      \implies
        P(\m{\salg},\m{\prob},\m{\permap}\m[\ip{x}{i}: q/n]).
\]
For example,
  fixing $X=\set{\ip{x}{i}}$ then
  $ \distAs{\ip{x}{i}}{\prob} $,
  $ \sure{\ip{x}{i}=v} $, and
  $ \perm{\ip{y}{i}: 1} $
  are permission-abstract,
  but $ \perm{\ip{x}{i}: \onehalf} $ is not. \section{Measure Theory Lemmas}

\paragraph{Notation}
In what follows,
given $n\in\Nat$ with~$n> 1$,
we write $\numlist{n}$
to denote the set $\set{1, \dots, n}$.
Moreover, for iterated summation we use the
notation
$
  \sum_{i \in I \mid \Phi(i)} f(i)
$
where~$I = \set{i_0, i_1, \dots}$ is countable
and $\Phi$ is a predicate on elements of~$I$,
to denote the sum $ f(j_0) + f(j_1) + \dots $
where $j_0, j_1, \dots$ is the sublist of $ i_0, i_1, \dots $ consisting
of the elements that satisfy~$\Phi$.
A similar convention is used for other commutative and associative operators,
\eg $\union$.
A countable partition of~$\Outcomes$ is a partition of $\Outcomes$,
$S \subs \powerset(\Outcomes)$,
with countably many sets.
For uniformity, we represent countable partitions as $S = \cpart{A}$
with the convention that when the partition has finitely many sets,
say~$n$, all the $A_i$ with $i \geq n$ are empty.

As mentioned, \thelogic\ is only concerned with discrete distributions,
\ie distributions over a countable set of outcomes.
The following lemma expresses the key property of \salgebra[s]
over countable outcomes that we exploit for proving the
other results.

\begin{lemma}
  \label{thm:countable-partition-generated}
 Let $\Outcomes$ be an countable set, and $\salg$ to be an arbitrary \salgebra{}
  on $\Outcomes$. Then there exists a countable partition $S$ of~$\Outcomes$
  such that $\salg = \closure{S}$.
\end{lemma}

\begin{proof}
  For every element $x \in \Outcomes$,
  we identify the smallest event $E_x \in \salg$ such that $x \in E_x$,
  and show that for $x, z \in \Outcomes$,
  either $E_x = E_z$ or $E_x \cap E_z = \emptyset$.
  Then the set $S = \set{E_x \mid x \in \Outcomes}$
  is a partition of~$\Outcomes$,
  and any event $E \in \salg$ can be represented as
  $\Union_{x \in E} E_x$, which suffices to show that $\salg = \closure{S}$.

  For every $x, y$, let
\begin{align*}
    A_{x, y} &=
    \begin{cases}
      \Outcomes \CASE \text{$\forall E \in \salg$, either $x , y$ both in~$E$ or $x , y$ both not in~$E$} \\
      E \OTHERWISE, \text{pick any $E \in \salg$ such that $x \in E$ and $y \notin E$}
    \end{cases}
  \end{align*}
Then we show that, for all $x$,
  $E_x = \cap_{y \in \Outcomes} A_{x, y}$ is the smallest
  event in $\salg$ such that $x \in E_x$ as follows.
  If there exists $E_x'$ such that $x \in E_x'$ and $E_x' \subset E_x$,
  then $E_x \setminus E_x'$ is not empty. Let $y$ be an element of
  $E_x \setminus E_x'$, and by the definition of $A_{x, y}$, we have
  $y \notin A_{x,y}$. Thus, $y \notin \cap_{y \in \Outcomes} A_{x, y} = E_x$,
  which contradicts with $y \in E_x \setminus E_x'$.

  Next, for any  $x, z \in \Outcomes$,
  since $E_x$ is the smallest event containing $x$ and
  $E_z$ is the smallest event containing $z$,
  the smaller event $E_z \setminus E_x$ is either equivalent to
  $E_z$ or not containing $z$.
  \begin{casesplit}
    \item
      If $E_z \setminus E_x = E_z$, then $E_x$ and $E_z$ are
      disjoint.
    \item If $z \not\in E_z \setminus E_x$, then it must $z \in E_x$,
      which implies that there exists \emph{no} $E \in \salg$ such that
      $x \in E$ and $z \notin E$. Because $\salg$ is closed under
      complement, then there exists \emph{no} $E \in \salg$ such that
      $x \notin E$ and $z \in E$ as well. Therefore,
      we have $x \in  \Inters_{y \in \Outcomes} A_{z, y} = E_z$ as well.
      Furthermore, because $E_z$ is the smallest event in
      $\salg$ that contains $z$ and $E_x$ also contains $z$,
      we have $E_z \subseteq E_x$; symmetrically, we have
      $E_x \subseteq E_z$.
      Thus, $E_x = E_z$.
    \end{casesplit}
    Hence,  the set $S = \set{E_x \mid x \in \Outcomes}$ is a
    countable partition of~$\Outcomes$.
\end{proof}

\begin{lemma}
\label{lemma:sigma-alg-representation}
  If $S = \cpart{A}$ is a partition of $\Outcomes$,
  and $\salg = \closure{S}$,
  then every event~$E\in\salg$ can be written as
  $E = \Dunion_{i \in I} A_i$ for some $I \subs \Nat$.
In other words,
  $
    \closure{S} = \set[\big]{ \Dunion_{i \in I} A_i | I \subseteq \Nat }.
  $
\end{lemma}

\begin{proof}
  Because \salgebra[s] are closed under countable union,
  for any $I \subseteq \Nat$,
  $ \Dunion_{i \in I} A_i \in \closure{S} $.
  Thus, $\closure{S} \supseteq \set[\big]{ \Dunion_{i \in I} A_i \mid I \subseteq \Nat }$.

  Also, $\set[\big]{ \Dunion_{i \in I} A_i \mid I \subseteq \Nat }$ is
  a \salgebra{}:
\begin{itemize}
    \item
      $\Outcomes = \Dunion_{i \in \Nat} A_i$.
    \item
      Given a countable sequences of events
      $E_1 = \Dunion_{i \in I_1} A_i$,
      $E_2 = \Dunion_{i \in I_2} A_i$,
      \dots,
      let
      $I = \Union_{j\in \Nat} I_j$;
      then we have
      $ \Union_{j\in\Nat} E_i = \Dunion_{i\in I} A_i $.
    \item
      If $E = \Dunion_{i \in I} A_i$, then
      the complement of $E$ is
      $
        (\Outcomes\setminus E) = \Dunion_{i \in (\Nat \setminus I)} A_i
      $.
  \end{itemize}
  Then, $\set{ \Dunion_{i \in I} A_i \mid I \subseteq \Nat }$ is
  a \salgebra{} that contains $S$.
  Therefore, $
    \closure{S} = \set[\big]{\Dunion_{i \in I} A_i \mid I \subseteq \Nat }
  $.
\end{proof}

\begin{lemma}
 \label{lemma:partition-order}
 Let $\Outcomes$ be a countable set.
 If\/ $S_1 = \cpart[i]{A}$ and
    $S_2 = \cpart[j]{B}$ are both countable partitions of $\Outcomes$,
 then $\closure{S_1} \subseteq \closure{S_2}$ implies that
 for any $B_j \in S_2$ with $B_j\ne \emptyset$,
 we can find a unique $A_i \in S_1$ such that $B_j \subseteq A_i$.
\end{lemma}

\begin{proof}
 For any $B_j \in S_2$ with $B_j\ne \emptyset$,
 pick an arbitrary element $s \in B_j$ and
 denote the unique element of $S_1$ that contains~$s$ as~$A_i$.
 Because $A_i \in S_1$ and $S_1 \subset  \closure{S_1} \subseteq \closure{S_2}$,  we have $A_i \in  \closure{S_2}$.
Note that $s \in B_j$ and $B_j$ is an element of the partition $S_2$
  that generates $\closure{S_2}$, $B_j$ must be the smallest event
  in $\closure{S_2}$ that contains $s$.
 Because $s \in A_i$ as well, $B_j$ being the smallest event containing $s$
  implies that $B_j \subseteq A_i$.
\end{proof}

\begin{lemma}
 \label{lemma:bind-extend}
 Assume we are given a \salgebra{} $\sigmaF_1$
 over a countable set $\Outcomes$,
 measure $\mu_1 \in \Dist(\sigmaF_1)$,
 a countable set $A$,
 a distribution $\mu \in \Full{A}$,
 and a function $\krnl_1 \colon A \to \Dist(\sigmaF_1)$
 such that $\mu_1 = \bind(\mu, \krnl_1)$.
 Then, for any probability space $(\sigmaF_2, \mu_2)$ such that
 $(\sigmaF_1, \mu_1) \extTo (\sigmaF_2, \mu_2)$,
 there exists $\krnl_2$ such that $\mu_2 = \bind(\mu, \krnl_2)$,
 and that for any $a \in \psupp(\mu)$,
 $(\sigmaF_1, \krnl_1(a)) \extTo (\sigmaF_2, \krnl_2(a))$.
\end{lemma}

\begin{proof}
  By~\cref{thm:countable-partition-generated},
  $\sigmaF_i = \closure{S_i}$ for some countable partition $S_i$.
  Also, $(\sigmaF_1, \mu_1) \extTo (\sigmaF_2, \mu_2)$ implies that
  $\sigmaF_1 \subseteq \sigmaF_2$.
  So we have $\closure{S_1} \subseteq \closure{S_2}$,
  which by~\cref{lemma:partition-order} implies that
  for any $B \in S_2$ with $B \ne \emptyset$,
  we can find a unique $A \in S_1$ such that $B \subseteq A$.
  Let~$f$ be the mapping associating to any $B\ne \emptyset$
  the corresponding $A = f(B)$, and $f(B)=\emptyset$ when $B=\emptyset$.

  Then, we define $\krnl_2$ as follows:
  for any $a \in A$, $E \in \sigmaF_2$,
  there exists $S \subseteq S_2$ such that
  $E = \Dunion_{B \in S} B$,
  then define
\begin{align*}
   \krnl_2(a)(E) = \sum_{B \in S} \krnl_1(a)(f(B)) \cdot h(B),
  \end{align*}
  where
  $h(B) = \mu_2(B) / \mu_2(f(B))$ if $ \mu_2(f(B)) \neq 0$
  and $h(B) = 0$ otherwise.

Then we calculate:
  \begin{align}
      &\bind(\mu, \krnl_2)(E) \notag\\
   &=  \sum_{a\in A} \mu(a) \cdot \krnl_2(E) \notag \\
   &=  \sum_{a\in A} \mu(a) \cdot \sum_{B \in S} \krnl_1(a)(f(B)) \cdot h(B) \notag \\
   &=  \sum_{B \in S}  \sum_{a\in A} \mu(a) \cdot \krnl_1(a)(f(B)) \cdot h(B)  \notag \\
   &=  \sum_{B \in S}  \bind(\mu, \krnl_1)(f(B)) \cdot h(B) \notag  \\
   &=  \sum_{B \in S}  \mu_1(f(B)) \cdot h(B) \notag \\
   &=  \sum_{B \in S \mid \mathrlap{\mu_2(f(B)) \neq 0}}  \mu_1(f(B)) \cdot \frac{\mu_2(B)}{ \mu_2(f(B))} \notag \\
   &=  \sum_{B \in S \mid \mathrlap{\mu_2(f(B)) \neq 0}}  \mu_2(f(B)) \cdot \frac{\mu_2(B)}{\mu_2(f(B))} \tag{$\mu_1(E') = \mu_2(E')$ for any $E' \in \sigmaF_1$}\\
   &=  \sum_{B \in S \mid \mathrlap{\mu_2(f(B)) \neq 0}}  \mu_2(B)  \notag \\
   &=  \sum_{B \in S \mid \mathrlap{\mu_2(f(B)) \neq 0}}  \mu_2(B)
       \quad+\quad
       \sum_{B \in S \mid  \mathrlap{\mu_2(f(B)) = 0}}  \mu_2(B)  \tag{Because $\mu_2(f(B)) = 0$ implies $\mu_2(B) = 0$} \\
   &= \sum_{B \in S}  \mu_2(B) \notag \\
   &= \mu_2(\Dunion_{B \in S} B ) \notag \\
   &= \mu_2(E) \notag
  \end{align}
Thus, $\bind(\mu, \krnl_2)= \mu_2$.

  Also, for any $a \in A_{\mu}$, for any $E \in \sigmaF_1$,
  there exists $S' \subseteq S_1$
  such that $E=\Dunion_{A \in S'} A$.
  \begin{align*}
   \krnl_2(a)(E)
   &= \krnl_2(a)\bigl(\Dunion_{A \in S'} A\bigr) \\
   &= \sum_{A \in S'} \krnl_2(a)(A) \\
   &= \sum_{A \in S'} \sum_{B \subseteq A \mid \mathrlap{B \in \sigmaF_2}} \krnl_2(a)(B)\\
&= \sum_{A \in S'}
        \sum_{B \subseteq A \mid \mathrlap{B \in \sigmaF_2, \mu_2(f(B)) \neq 0}}
          \quad\krnl_1(a)(f(B)) \cdot \frac{\mu_2(B)}{\mu_2(f(B))} \\
&= \sum_{A \in S' \mid \mathrlap{\mu_2(A) \neq 0}} \krnl_1(a)(A) \cdot \frac{\left(\sum_{B \subseteq A \mid B \in \sigmaF_2}  \mu_2(B) \right) }{\mu_2(A)} \\
   &= \sum_{A \in S' \mid \mathrlap{\mu_2(A) \neq 0}} \krnl_1(a)(A) \cdot \frac{ \mu_2(A)}{ \mu_2(A)} \\
   &= \sum_{A \in S' \mid \mathrlap{\mu_2(A) \neq 0}} \krnl_1(a)(A)  \\
   &= \sum_{A \in S'} \krnl_1(a)(A)  \\
   &= \krnl_1(a)\bigl(\Dunion_{A \in S'} A\bigr)  \\
   &= \krnl_1(a)(E)
  \end{align*}
Thus, for any $a$, $(\sigma_1, \krnl_1(a)) \extTo (\sigma_2, \krnl_2(a))$.
\end{proof}

\begin{lemma}
 \label{lemma:product-algebra}
 Given two \salgebra[s] $\sigmaF_1$ and $\sigmaF_2$
 over two countable underlying sets $\Outcomes_1, \Outcomes_2$,
 then a general element in the product \salgebra{}
 $\sigmaF_1 \otimes \sigmaF_2$ can
 be expressed as $\Dunion_{(i, j) \in I} (A_{i} \times B_{j})$
 for some $I \subseteq \Nat^2$ and
 $A_{i} \in \sigmaF_{1}, B_{j} \in \sigmaF_{2}$ for $(i,j) \in I$.
\end{lemma}

 \begin{proof}
By~\cref{thm:countable-partition-generated}, each \salgebra{}
  $\sigmaF_i$ is generated by a countable partition over $\Outcomes_i$.
  Let $S_1 = \cpart{A}$ be the countable partition that generates $\sigmaF_1$,
  $S_2 = \cpart{B}$ be the countable partition that generates $\sigmaF_2$.
    By \cref{lemma:sigma-alg-representation},
    a general element in $\sigmaF_1$ can be written as
  $\Dunion_{j \in J} A_{j}$ for some $J \subseteq \Nat$,
  and similarly,
   a general element in $\sigmaF_2$ can be written as
  $\Dunion_{k \in K} B_{k}$ for some $K \subseteq \Nat$.

  Note that $\{A_j \times B_k \}_{j, k \in \Nat}$ is a partition because:
  if $(A_j \times B_k)  \cap (A_{j'} \times B_{k'}) \neq \emptyset$ for some
  $j \neq j'$ and $k \neq k'$,
  then it must $A_j \cap A_{j'} \neq \emptyset$ and $B_k \cap B_{k'} \neq \emptyset$,
  and that imply that $A_j =A_{j'}$ and $B_j =B_{j'}$;
  therefore, $A_j \times B_k = A_{j'} \times B_{k'}$.

  We next show that $\sigmaF_1 \otimes \sigmaF_2$ is generated by
  the partition $\{A_j \times B_k \}_{j, k \in \Nat}$.
\begin{align*}
   \sigmaF_1 \otimes \sigmaF_2
   &= \closure*{\sigmaF_1 \times \sigmaF_2} \\
   &= \closure*{\set*{\Dunion_{j \in J_1} A_{j} \times \Dunion_{j \in J_2} B_{j} | J_1, J_2 \subseteq \Nat}} \\
   &= \closure*{\set*{\Dunion_{j \in J_1, k \in J_2}  (A_{j} \times B_{k}) | J_1, J_2 \subseteq \Nat }} \\
   &= \closure*{\set*{ A_{j} \times B_{k} | j, k\subseteq \Nat }}
   \end{align*}
Since each $A_j \in S_1 \subseteq \sigmaF_1$ and $B_k \in S_2 \subseteq \sigmaF_2$
   a general element in $\sigmaF_1 \otimes \sigmaF_2$ can
  be expressed as $\set*{\Dunion_{j, k \subseteq I} (A_{j} \times B_{k}) \mid
  A_{j} \in \sigmaF_{1}, B_{k} \in \sigmaF_{2}, I \subseteq \Nat^2}$
according to \cref{thm:countable-partition-generated}.
 \end{proof}

 \begin{lemma}
  \label{lemma:indep-prod-exists}
  Given two probability spaces
  $(\sigmaF_a, \mu_a), (\sigmaF_b, \mu_b) \in \ProbSp(\Outcomes)$,
  their independent product
  $(\sigmaF_a, \mu_a) \iprod (\sigmaF_b, \mu_b)$ exists
  if $\mu_a(E_a) \cdot \mu_b(E_b) = 0 $
  for any $E_a \in \sigmaF_a, E_b \in \sigmaF_b$ such that
  $E_a \cap E_b = \emptyset$.
 \end{lemma}

 \begin{proof}
We first define $\mu: \set{E_a \cap E_b \mid E_a \in \sigmaF_a, E_b \in \sigmaF_b}
  \to [0,1]$ by $\mu(E_a \cap E_b) = \mu_a(E_a) \cdot \mu_b(E_b)$
   for any $E_a \in \sigmaF_a, E_b \in \sigmaF_b$,
   and then show that $\mu$ could be extended to a probability
   measure on $\sigmaF_a \punion \sigmaF_b$.

  \begin{itemize}
   \item We first need to show that $\mu$ is \textbf{well-defined}.
    That is,
    $E_a \cap E_b = E_a' \cap E_b'$
    implies $\mu_a(E_a) \cdot \mu_b(E_b) = \mu_a(E'_a) \cdot \mu_b(E'_b)$.

     When $E_a \cap E_b = E_a' \cap E_b'$, it must $E_a \cap E_a' \supseteq
     E_a \cap E_b = E_a' \cap E_b'$, Thus, $E_a \setminus E_a' \subseteq E_a
     \setminus E_b$, and then $E_a \setminus E_a'$ is disjoint from $E_b$;
     symmetrically, $E_a' \setminus E_a$ is disjoint from $E_b'$.
     Since $E_a, E_a'$ are both in $\sigmaF_{a}$, we have $E_a \setminus E_a'$
     and $E_a' \setminus E_a$ both measurable in $\sigmaF_a$.
     Their disjointness and the result above implies that
     $\mu_a(E_a \setminus E_a') \cdot \mu_b(E_b) = 0$ and
     $\mu_a(E'_a \setminus E_a) \cdot \mu_b(E'_b) = 0$.
     Symmetric reasoning can also show that
     $E'_b \setminus E_b$ is disjoint from $E'_a \cap E_a$,
       and  $E_b \setminus E'_b$ is disjoint from $E'_a \cap E_a$,
       which implies
       $\mu_a(E_b \setminus E'_b) \cdot \mu_b(E'_a \cap E_a) = 0$ and
     $\mu_a(E'_b \setminus E_b) \cdot \mu_b(E'_a) = 0$.

     Then there are four possibilities:
     \begin{itemize}
      \item If $\mu_b(E_b) = 0$ and $\mu_b(E_b') = 0$,
       then $\mu_a(E_a) \cdot \mu_b(E_b) = 0 = \mu_a(E_a') \cdot \mu_b(E_b')$.
      \item If $\mu_a(E_a \setminus E'_a) = 0$ and $\mu_b(E'_a \setminus E_a) = 0$.
       Then
       \begin{align*}
        \mu_a(E_a) \cdot \mu_b(E_b) &=
        \mu_a((E'_a \setminus E_a) \disjunion (E'_a \cap E_a)) \cdot \mu_b(E_b) \\
          &= (\mu_a(E'_a \setminus E_a) + \mu_a(E'_a \cap E_a)) \cdot \mu_b(E_b) \\
          &= \mu_a(E'_a \cap E_a) \cdot \mu_b(E_b) \\
          &= (\mu_a(E_a \setminus E'_a) + \mu_a(E'_a \cap E_a)) \cdot \mu_b(E_b) \\
          &= \mu_a(E'_a) \cdot \mu_b(E_b)
       \end{align*}

       Thus,  either $\mu_a(E'_a \cap E_a) = 0$, which implies that
       \[
        \mu_a(E_a) \cdot \mu_b(E_b)  = (0 + 0) \cdot  \mu_b(E_b) = 0 =(0+0) \cdot \mu_b(E_b) = \mu_a(E'_a) \cdot \mu_b(E'_b),
       \]
       or we have both $\mu_b(E'_b \setminus E_b) = 0$ and $\mu_b(E_b \setminus E'_b) = 0$, which imply that
       \begin{align*}
        \mu_a(E_a) \cdot \mu_b(E_b)
          &= \mu_a(E'_a) \cdot \mu_b(E_b)\\
          &= \mu_a(E'_a) \cdot \mu_b((E_b \cap E'_b ) \disjunion (E_b \setminus E'_b)) \\
          &= \mu_a(E'_a) \cdot (\mu_b(E_b \cap E'_b ) + 0) \\
          &= \mu_a(E'_a) \cdot (\mu_b(E_b \cap E'_b ) + \mu_b(E'_b \setminus E_b)) \\
          &= \mu_a(E'_a) \cdot \mu_b(E'_b ).
       \end{align*}
       \item If $\mu_b(E'_b) = 0$ and $\mu_b(E_a \setminus E'_a) = 0$,
        then
        \begin{align*}
         \mu_a(E_a) \cdot \mu_b(E_b)
         &= (\mu_a(E_a \cap E'_a) + \mu_a(E_a \setminus E'_a)) \cdot (\mu_b(E_b \cap E'_b) + \mu_b(E_b \setminus E'_b)) \\
         &= \mu_a(E_a \cap E'_a)  \cdot \mu_b(E_b \setminus E'_b)
        \end{align*}
Because $\mu_a(E_b \setminus E'_b) \cdot \mu_b(E'_a \cap E_a) = 0$ and
     $\mu_a(E'_b \setminus E_b) \cdot \mu_b(E'_a) = 0$.
        Thus, $\mu_a(E_a) \cdot \mu_b(E_b) =0 = \mu_a(E'_a) \cdot \mu_b(E'_b)$.

       \item If $\mu_b(E_b) = 0$ and $\mu_b(E'_a \setminus E_a) = 0$,
        then symmetric as above.

     \end{itemize}

     In all these cases,
     $\mu_a(E_a) \cdot \mu_b(E_b) = \mu_a(E'_a) \cdot \mu_b(E'_b)$ as desired.

    \item Show that $\mu$ satisfy \textbf{countable additivity} in $\{E_a \cap E_b \mid E_a \in \sigmaF_a,
   E_b \in  \sigmaF_b\}$.

   We start with showing that $\mu$ is finite-additive.
   Suppose $E_a^n \cap E_b^n = \Disjunion_{i \in [n]}(A_i \cap B_i)$ where
   each $A_i \in \sigmaF_a$ and $B_i \in \sigmaF_b$.
  Fix any $A_i \cap B_i$, there is unique minimal $A \in \sigmaF_a$ containing
  $A_i \cap B_i$, because if $A \supseteq A_i \cap B_i$ and  $A' \supseteq
  A_i \cap B_i$, then  $A \cap A' \supseteq A_i \cap B_i$
  and $A \cap A' \in \sigmaF_A$ too, and $A \cap A'$ is smaller.
   Because we have shown that $\mu$ is well-defined, in the following proof,
   we can assume without loss of generality that $A_i$ is the smallest set in $\sigmaF_a$ containing $A_i \cap B_i$.
   Similarly, we let $B_i$ to be the smallest set in $\sigmaF_b$ containing $A_i \cap B_i$.
   Thus,  $E_a^n \cap E_b^n = \Disjunion_{i \in [n]}(A_i \cap B_i)$ implies
   every $A_i$ is smaller than $E_a^n$ and every $B_i$ is smaller than $E_b^n$.
   Therefore,
   $E_a^n \supseteq \cup_{i \in [n]} A_i$ and
   $E_b^n \supseteq \cup_{i \in [n]} B_i$,
   which implies that
   \[
    E_a^n \cap E_b^n \supseteq (\cup_{i \in [n]} A_i) \cap (\cup_{i \in [n]} B_i) \supseteq \cup_{i \in [n] } (A_i \cap B_i) = E_a^n \cap E_b^n,
   \]
   which implies that the $\supseteq$ in the inequalities all collapse to $=$.

   For any $I \subseteq [n]$, define
   $\alpha_I = \cap_{i \in I} A_i \setminus (\cup_{i \in [n] \setminus I} A_i)$, and $\beta_I = \cap_{i \in I} B_i \setminus (\cup_{i \in [n] \setminus I} B_i)$.
   For any $I \neq I'$, $\alpha_I \cap \alpha_{I'} = \emptyset$.
   Thus, $\{\alpha_I\}_{I \subseteq [n]}$ is a set of disjoint sets in $\cup_{i \in [n]} A_i$,
   and similarly, $\{\beta_I\}_{I \subseteq [n]}$ is a set of disjoint sets in $\cup_{i \in [n]} B_i$.
   Also, for any $i \in [n]$,
   we have
   $A_i = \cup_{I \subseteq [n] \mid i \in I} \alpha_I $
   and
   $B_i = \cup_{I \subseteq [n] \mid i \in I} \beta_I $.
   Furthermore, for any $I$,
\begin{align*}
    \alpha_I \cap \cup_{i \in [n]} B_i
    \subseteq (\cup_{i \in [n]} A_i) \cap (\cup_{i \in [n]} B_i)
    = \Dunion_{i \in [n]} A_i \cap B_i ,
   \end{align*}
and thus,
\begin{align}
    \alpha_I \cap \cup_{i \in [n]} B_i
    & = (\Dunion_{i \in [n]} A_i \cap B_i) \cap (\alpha_I \cap \cup_{i \in [n]} B_i) \notag \\
    & = \Dunion_{i \in [n]} \left( A_i \cap B_i \cap \alpha_I \cap \cup_{j \in [n]} B_j \right) \notag \\
    & = \Dunion_{i \in I} \left( A_i \cap B_i \cap \alpha_I \cap \cup_{j \in [n]} B_j \right)  \tag{$A_i \cap \alpha_I = \emptyset$ if $i \notin I$} \\
    & = \Dunion_{i \in I} \left( A_i \cap B_i \cap \alpha_I\right)
    \tag{$B_i \cap \cup_{j \in [n]} B_j = B_i$ for any $i$}\\
    & = \Dunion_{i \in I} \left( B_i \cap \alpha_I\right)
    \tag{$A_i \cap \alpha_I = \alpha_I$ for any $i \in I$ }\\
    & = \alpha_I \cap \cup_{i \in I } B_i
    \label{eq:finite-add-alpha}
   \end{align}

   Now,
  \begin{align}
   &\mu(E^n_a \cap E^n_b) \notag \\
   &= \mu((\cup_{i \in [n]} A_i) \cap (\cup_{i \in [n]} B_i)) \notag \\
   &= \mu((\Dunion_{I \subseteq [n]} \alpha_I) \cap (\cup_{i \in [n]} B_i)) \tag{By definition of $\alpha_I$}\\
   &= \mu_a(\Dunion_{I \subseteq [n]} \alpha_I) \cdot  \mu_b(\cup_{i \in [n]} B_i)  \tag{By definition of $\mu$} \\
   &= \left\lgroup\sum_{I \subseteq [n]} \mu_a (\alpha_I) \right\rgroup \cdot  \mu_b(\cup_{i \in [n]} B_i)  \tag{By finite-additivity of $\mu_a$} \\
   &= \sum_{I \subseteq [n]} \mu_a (\alpha_I) \cdot  \mu_b(\cup_{i \in [n]} B_i) \notag\\
   &= \sum_{I \subseteq [n]}  \mu(\alpha_I \cap (\cup_{i \in [n]} B_i)) \tag{By definition of $\mu$} \\
   &= \sum_{I \subseteq [n]}  \mu(\alpha_I \cap (\cup_{i \in I} B_i))
   \tag{By~\cref{eq:finite-add-alpha}}\\
   &= \sum_{I \subseteq [n]}  \mu_a(\alpha_I) \cdot \mu_b (\cup_{i \in I} B_i)\tag{By definition of $\mu$} \\
   &= \sum_{I \subseteq [n]}  \mu_a(\alpha_I) \cdot \mu_b (\cup_{i \in I} (\disjunion_{I' \subseteq [n] \mid i \in I'} \beta_{I'}))
   \tag{By definition of $\beta_I$} \\
   &= \sum_{I \subseteq [n]}  \mu_a(\alpha_I) \cdot \mu_b (\disjunion_{I' \subseteq [n] \mid I \cap I' \neq \emptyset} \beta_{I'}) \notag \\
   &= \sum_{I \subseteq [n]}  \mu_a(\alpha_I) \cdot \sum_{I' \subseteq [n] \mid \mathrlap{I \cap I' \neq \emptyset}} \mu_b (\beta_{I'}) \notag \\
   &= \sum_{I \subseteq [n]} \sum_{I' \subseteq [n] \mid \mathrlap{I \cap I' \neq \emptyset}}  \mu_a(\alpha_I) \cdot \mu_b (\beta_{I'}) \notag
  \end{align}
Meanwhile, for any $I, I'$, if $|I \cap I'| \geq 2$,
  then there exists some $j, k$ such that $j \in I \cap I'$ and $k \in I \cap I'$,
  so
\begin{align}
    \mu_{a}(\alpha_I) \cdot \mu_{b}(\beta_{I'})
    &= \mu_{a}(\cap_{i \in I} A_i \setminus (\cup_{i \in [n] \setminus I} A_i)) \cdot \mu_{b}(\cap_{i \in I} B_i \setminus (\cup_{i \in [n] \setminus I} B_i)) \notag \\
    &\leq \mu_{a}(A_j \cap A_k) \cdot \mu_{b}(B_j \cap B_k) \notag \\
    &= \mu(A_j \cap A_k \cap B_j \cap B_k) \notag \\
    &= \mu((A_j\cap B_j) \cap (A_k \cap B_k)) \notag \\
    &= \mu(\emptyset) \notag \\
    &= 0. \notag
  \end{align}
Thus, continuing our previous derivation,
\begin{align}
   &\mu(E^n_a \cap E^n_b) \notag \\
   &= \sum_{I \subseteq [n]} \sum_{I' \subseteq [n] \mid \mathrlap{I \cap I' \neq \emptyset}}  \mu_a(\alpha_I) \cdot \mu_b (\beta_{I'}) \notag \\
   &= \sum_{I \subseteq [n]}
        \sum_{I' \subseteq [n] \mid \mathrlap{1 = \card{I \cap I'}}}
          \mu_a(\alpha_I) \cdot \mu_b (\beta_{I'})
   \tag{Because $\mu_{a}(\alpha_I) \cdot \mu_{b}(\beta_{I'})=0$ if $|I \cap I'| \geq 2$}\\
   &= \sum_{i \in [n]}
        \sum_{I \subseteq [n] \mid i \in I}
          \sum_{I' \subseteq [n] \mid \mathrlap{I \cap I' = \set{i}} }
            \mu_a(\alpha_I) \cdot \mu_b (\beta_{I'})
      \notag \\
   &= \sum_{i \in [n]} \sum_{I \subseteq [n] \mid i \in I} \sum_{I' \subseteq [n] \mid \mathrlap{i \in I'}}  \mu_a(\alpha_I) \cdot \mu_b (\beta_{I'})
   \tag{Because $\mu_{a}(\alpha_I) \cdot \mu_{b}(\beta_{I'})=0$ if $|I \cap I'| \geq 2$}\\
   &= \sum_{i \in [n]}
         \left\lgroup
         \sum_{I \subseteq [n] \mid \mathrlap{i \in I}} \mu_a(\alpha_I)
         \cdot
         \sum_{I' \subseteq [n] \mid \mathrlap{i \in I'}} \mu_b (\beta_{I'})
         \right\rgroup
     \notag \\
   &= \sum_{i \in [n]} \mu_a(A_i) \cdot \mu_b(B_i) \notag \\
   &= \sum_{i \in [n]} \mu(A_i \cap B_i) \notag
  \end{align}

  Thus, we established the finite additivity.
  For countable additivity, suppose $E_a \cap E_b = \Disjunion_{i \in \Nat}(A_i \cap B_i)$. By the same reason as above, we also have
  \[
    E_a \cap E_b = (\cup_{i \in \Nat} A_i) \cap (\cup_{i \in \Nat} B_i) = \cup_{i \in \Nat} (A_i \cap B_i) = E_a \cap E_b.
  \]

Then,
  \begin{align}
   &\mu(E_a \cap E_b) \notag \\
   &= \mu((\cup_{i \in \Nat} A_i) \cap (\cup_{i \in \Nat} B_i)) \notag \\
   &= \mu_a(\cup_{i \in \Nat} A_i) \cdot \mu_b(\cup_{i \in \Nat} B_i) \notag \\
   &= \mu_a(\lim_{n \to \infty} \cup_{i \in [n]} A_i) \cdot \mu_b(\lim_{n \to \infty} \cup_{i \in [n]} B_i) \notag \\
   &= \lim_{n \to \infty} \mu_a(\cup_{i \in [n]} A_i) \cdot \lim_{n \to \infty} \mu_b( \cup_{i \in [n]} B_i) \tag{By continuity of $\mu_a$ and $\mu_b$} \\
   &= \lim_{n \to \infty} \mu_a(\cup_{i \in [n]} A_i) \cdot \mu_b( \cup_{i \in [n]} B_i) \tag{$\dagger$} \\
   &= \lim_{n \to \infty} \sum_{i \in [n]}  \mu_b ( B_i) \cdot \mu_a(A_i) \tag{By~\cref{eq:finite-add-alpha}} \\
   &= \sum_{i \in \Nat}  \mu_b ( B_i) \cdot \mu_a(A_i),
  \end{align}
where ($\dagger$) holds because that the product of limits equals to the limit of
  the product when both $\lim_{n \to \infty} \mu_a(\cup_{i \in [n]} A_i)$ and
  $\lim_{n \to \infty} \mu_b( \cup_{i \in [n]} B_i)$ are finite.
  Thus, we proved countable additivity as well.

 \item
  Next we show that we can \textbf{extend $\mu$ to a measure on
  $\sigmaF_a \punion \sigmaF_b$}.

  So far, we proved that $\mu$ is a sub-additive measure on the
  $\set{E_a \cap E_b | E_a \in \sigmaF_a, E_b \in \sigmaF_b}$,
  which forms a $\pi$-system.
  By a known theorem in probability theory
  (\eg \cite[Corollary 2.5.4]{rosenthal2006first}),
  we can extend a sub-additive measure on a
   $\pi$-system to the \salgebra{} it generates if the $\pi$-system is
   a semi-algebra.
   Thus, we can extend $\mu$ to a measure on $\closure{\{E_a \cap E_b \mid E_a \in \sigmaF_a,\ E_b \in \sigmaF_b\}}$ if we can prove $J = \{E_a \cap E_b \mid E_a \in \sigmaF_a,\ E_b \in \sigmaF_b\}$ is a semi-algebra.

   \begin{itemize}
    \item $J$ contains $\emptyset$ and $\Outcomes$: trivial.
    \item $J$ is closed under finite intersection:
     $(E_a \cap E_b) \cap (E'_a \cap E'_b) = (E_a \cap E'_a) \cap (E_b \cap E'_b)$, where $E_a \cap E'_a \in \sigmaF_a$, and $E_b \cap E'_b \in \sigmaF_b$.
    \item The complement of any element of $J$ is equal to a finite disjoint
     union of elements of $J$:
     \begin{align*}
      (E_a \cap E_b)^C &= E_a^C \cup E_b^C \\
             &= (E_a^C \cap \Outcomes) \disjunion (E_a \cap E_b^C)
     \end{align*}
     where $E_a^C, E_a \in \sigmaF_a$, and $E_b^C, \Outcomes \in \sigmaF_b$.

   \end{itemize}

     As shown in~\cite{lilac},
\begin{align}
      \closure{\{E_a \cap E_b \mid E_a \in \sigmaF_a, E_b \in \sigmaF_b\}}
     = \sigmaF_a \punion \sigmaF_b
     \end{align}
Thus, the extension of $\mu$ is a measure on $\sigmaF_a \punion \sigmaF_b$.

    \item Last, we show that $\mu$ is a \textbf{probability measure} on
     $\sigmaF_a \punion \sigmaF_b$:
     $\mu(\Outcomes) = \mu_a(\Outcomes) \cdot \mu_b(\Outcomes) = 1$.
    \qedhere
 \end{itemize}
 \end{proof}

 \begin{lemma}
  \label{lemma:fibre-prod-exists}
    Consider two probability spaces
  $(\sigmaF_1, \mu_1), (\sigmaF_2, \prob_2) \in \ProbSp(\Outcomes)$,
  and some other probability space $(\Full{A}, \prob)$ and kernel $\krnl$
  such that $\prob_1 = \bind(\prob, \krnl)$.

  Then, the independent product
  $(\sigmaF_1, \mu_1) \iprod (\sigmaF_2, \mu_2)$
  exists if and only if
   for any $a \in \psupp(\prob)$,
   the independent product
   $(\sigmaF_1, \krnl(a)) \iprod (\sigmaF_2, \prob_2)$ exists.
When they both exist,
  \[
   (\sigmaF_1, \prob_1) \iprod (\sigmaF_2, \prob_2)
  = (\sigmaF_1 \punion \sigmaF_2,
     \bind(\prob, \fun a. \krnl(a) \iprod \prob_2))
  \]
 \end{lemma}

 \begin{proof}
    We first show the backwards direction.
  By~\cref{lemma:indep-prod-exists},
  for any $a \in \psupp(\prob)$, to show that the independent product
  $(\sigmaF_1, \krnl(a)) \iprod (\sigmaF_1, \prob_1)$ exists,
  it suffices to show that for any $E_1 \in \sigmaF_1, E_2 \in \sigmaF_2$
  such that $E_1 \cap E_2 = \emptyset$,
  $\krnl(a)(E_1) \cdot \prob_2(E_2) = 0$.

  Fix any such $E_1, E_2$,
  because $(\sigmaF_1, \prob_1) \iprod (\sigmaF_2, \prob_2)$ is defined,
  we have $\prob_1(E_1) \cdot \prob_2(E_2) = 0$, then either $\prob_1(E_1) = 0$
  or $\prob_2(E_2) = 0$.
\begin{itemize}
   \item If $\prob_1(E_1) = 0$:
   Recall that
\[
    \prob_1(E_1) = \bind(\prob, \krnl)(E_1)
    = \sum_{a \in A} \prob(a) \cdot \krnl(a) (E_1)
    = \sum_{\mathclap{a \in \psupp(\prob)}} \prob(a) \cdot \krnl(a) (E_1)
   \]
Because all $\prob(a) > 0$ and $\krnl(a) (E_1) \geq 0$ for all $a \in \psupp(\prob)$
   $\sum_{a \in \psupp(\prob)} \prob(a) \cdot \krnl(a) (E_1) = 0$ implies that
   $\prob(a) \cdot \krnl(a) (E_1) = 0$ for all $a \in \psupp(\prob)$.
   Thus, for all $a \in \psupp(\prob)$, it must $\krnl(a) (E_1) = 0$.
   Therefore, $\krnl(a)(E_1) \cdot \prob_2(E_2) = 0$ for all  $a \in \psupp(\prob)$
   with this $E_1, E_2$.

  \item If $\prob_2(E_2) = 0$, then it is also clear that
    $\krnl(a)(E_1) \cdot \prob_2(E_2) = 0$ for all  $a \in \psupp(\prob)$.
  \end{itemize}
Thus, we have $\krnl(a)(E_1) \cdot \prob_2(E_2) = 0$ for any
  $E_1 \cap E_2 = \emptyset$ and $a \in \psupp(\prob)$.
  By~\cref{lemma:indep-prod-exists},
  the independent product $(\sigmaF_1, \krnl(a)) \iprod (\sigmaF_1, \prob_1)$ exists.

    For the forward direction:
    for any $E_1 \in \sigmaF_1$ and $E_2 \in \sigmaF_2$ such that $E_1 \cap E_2 = \emptyset$,
   the independent product $(\sigmaF_1, \krnl(a)) \iprod (\sigmaF_2, \mu_2)$ exists implies that
   \begin{align*}
     \krnl(a) (E_1) \cdot \mu_2(E_2) = 0.
   \end{align*}
Thus,
   \begin{align*}
    \mu_1(E_1) \cdot \mu_2(E_2)
    &= \bind(\mu, \krnl)(E_1)  \cdot \mu_2(E_2) \\
    &= \left\lgroup\sum_{a \in A} \mu(a) \cdot \krnl(a) (E_1) \right\rgroup  \cdot \mu_2(E_2) \\
    &= \sum_{a \in A_{\mu}} \mu(a) \cdot \left(\krnl(a) (E_1)  \cdot \mu_2(E_2) \right) \\
    &= \sum_{a \in A_{\mu}} \mu(a) \cdot 0
    = 0
   \end{align*}

   Thus, by~\cref{lemma:indep-prod-exists},
   the independent product $(\sigmaF_1, \mu_1) \iprod (\sigmaF_2, \mu_2)$ exists.
For any $E_1 \in \sigmaF_1$ and $E_2 \in \sigmaF_2$,
\begin{align*}
   &\bind(\prob, \fun a. \krnl(a) \iprod \prob_2 ) (E_1 \inters E_2)\\
  &= \sum_{\mathclap{a \in \psupp(\prob)}}
     \prob(a) \cdot
     \left(\krnl(a) \iprod  \prob_2\right)(E_1 \inters E_2) \\
  &= \sum_{\mathclap{a \in \psupp(\prob)}}
      \prob(a) \cdot \krnl(a)(E_1) \cdot \prob_2(E_2) \\
  &= \left\lgroup
     \sum_{a \in \psupp(\prob)}
      \prob(a) \cdot \krnl(a)(E_1)
    \right\rgroup \cdot \prob_2(E_2) \\
  &= \bind(\prob, \krnl)(E_1) \cdot \prob_2(E_2) \\
  &= \prob_1(E_1) \cdot \prob_2(E_2) \\
  &= (\prob_1 \iprod \prob_2)(E_1 \inters E_2)
  \end{align*}
Thus,
  $
   (\sigmaF_1, \prob_1) \iprod (\sigmaF_2, \prob_2)
  = (\sigmaF_1 \punion \sigmaF_2,
     \bind(\prob, \fun a. \krnl(a) \iprod \prob_2)).
  $
 \end{proof}

 \section{Model}
\label{sec:appendix:model}

\subsection{Basic Connectives}

The following are the definitions of the standard SL connectives
we use in \thelogic:
\begin{align*}
  \pure{\phi} &\is \fun \wtv. \phi
  &
  P * Q &\is \fun a.
    \exists b_1,b_2 \st
      (b_1 \raOp b_2) \raLeq a \land
      P(b_1) \land
      Q(b_2)
  \\
  \Own{b} &\is \fun a. b \raLeq a
  &
  P \wand Q &\is \fun a.
    \forall b\st
     \raValid(a \raOp b)
     \implies
     P(b)
     \implies
     Q(a \raOp b)
  \\
  P \land Q &\is \fun a. P(a) \land Q(a)
  &
  \A x \of X.P(x) &\is \fun a.
    \forall x \in X\st P(x)(a)
  \\
  P \lor Q &\is \fun a. P(a) \lor Q(a)
  &
  \E x \of X.P(x) &\is \fun a.
    \exists x \in X\st P(x)(a)
\end{align*}

\subsection{Construction of the \thelogic\ Model}

\begin{lemma}
  The structure $\PSpRA$ is an ordered unital resource algebra (RA) as defined
  in~\cref{def:ra}.
\end{lemma}

\begin{proof}
We defined $\raOp$ and $\raLeq$ the same way as in~\cite{lilac},
  and they have proved that $\raOp$ is associative and commutative,
  and $\raLeq$ is transitive and reflexive.
  We check the rest of conditions one by one.
  \begin{induction}
    \step[Condition~$a \raOp b = b \raOp a$]
      The independent product is proved to be commutative in ~\cite{lilac}.
    \step[Condition~$(a \raOp b) \raOp c = a \raOp (b \raOp c)$]
      The independent product is proved to be associative in~\cite{lilac}.
    \step[Condition~$a \raLeq b \implies b \raLeq c \implies a \raLeq c$]
      The order $\raLeq$ is proved to be transitive in~\cite{lilac}.
    \step[Condition~$a \raLeq a$]
      The order $\raLeq$ is proved to be reflexive in~\cite{lilac}.
    \step[Condition~$\raValid(a \raOp b) \implies \raValid(a)$]
      Pattern matching on $a \raOp b$,
      either there exists probability spaces $\psp_1, \psp_2$ such that
      $a = \psp_1$, $b = \psp_2$ and $\psp_1 \iprod \psp_2$ is defined,
      or $a \raOp b = \invalid$.
\begin{casesplit}
        \case[$a \raOp b = \invalid$] Note that
        $\raValid(a \raOp b)$ does not hold when $a \raOp b = \invalid$,
      so we can eliminate this case by ex falso quodlibet.
      \case[$a \raOp b = \psp_1 \iprod \psp_2$] Then
      $a = \psp_1$, and thus $\raValid(a)$.
      \end{casesplit}

    \step[Condition~$\raValid(\raUnit)$]
      Clear because $\raUnit \neq \invalid$.
\step[Condition~$a \raLeq b \implies \raValid(b) \implies \raValid(a)$]
      Pattern matching on $a$ and $b$,
      either there exists probability spaces $\psp_1, \psp_2$ such that
      $a = \psp_1$, $b = \psp_2$ and $\psp_1 \extTo \psp_2$ is defined,
      or $b = \invalid$.
      \begin{casesplit}
        \case[$b = \invalid$] Then $\raValid(b)$ does not hold,
      and we can eliminate this case by ex falso quodlibet.
        \case[$a = \psp_1$, $b = \psp_2$ and
        $\psp_1 \extTo \psp_2$]  We clearly have $\raValid(a)$.
      \end{casesplit}

    \step[Condition~$\raUnit \raOp a = a$]
      Pattern matching on $a$,
      either $a = \invalid$
      or there exists some probability space $\psp$ such that $a = \psp$.
      \begin{casesplit}
        \case[$a = \invalid$] Then $\raUnit \raOp a = \invalid = a$.
        \case[$a = \psp$] Then $\raUnit \raOp a = a$.
      \end{casesplit}

    \step[Condition~$a \raLeq b \implies a \raOp c \raLeq b \raOp c$]
      Pattern matching on $a$ and $b$.
      If $a \raLeq b $,
      then either $b = \invalid$ or there exists
       $\psp, \psp'$ such that $a = \psp$ and $b = \psp'$.
       \begin{casesplit}
         \case[$b = \invalid$] Then $b \raOp c = \invalid$ is the top element,
       and then $a \raOp c \raLeq b \raOp c$.
         \case[Otherwise]
       $a \raLeq b$ iff $\psp \raLeq \psp'$,
       then either $b \raOp c = \invalid$ and $a \raOp c \raLeq b \raOp c$ follows,
       or  $b \raOp c = \psp' \iprod \psp''$
       for some probability space $c = \psp''$.
       Then $\psp \raLeq \psp'$ implies that
       $\psp \raOp \psp''$ is also defined and
       $\psp \raOp \psp' \raLeq \psp \raOp \psp''$.
       Thus, $a \raOp c \raLeq b \raOp c$ too.
       \qedhere
      \end{casesplit}
  \end{induction}
\end{proof}

\begin{lemma} [RA composition preserves compatibility]
  \label{lemma:well-defined-psppmra-1}
   \[
  \salg_1\compat\permap_1
  \implies
  \salg_2\compat\permap_2
  \implies
    (\salg_1 \punion \salg_2)\compat(\permap_1 \raOp \permap_2)
  \]
\end{lemma}

\begin{proof}
  Let $S_1 = \set{x \in \Var \mid \permap_1(x) = 0}$,
  $S_2 = \set{x \in \Var \mid \permap_2(x) = 0}$.
  If $\salg_1\compat\permap_1$, then there exists
  $\psp_1' \in \ProbSp((\Var \setminus S_1) \to \Val)$
  such that
  $\psp_1 = \psp_1' \pprod \Triv{S_1 \to \Val}$
  In addition, if $\salg_2\compat\permap_2$, then there exists
  $\psp_2' \in \ProbSp((\Var \setminus S_2) \to \Val)$
  such that
  $\psp_2 = \psp_2' \pprod \Triv{S_2 \to \Val}$.
  Then,
  \begin{align*}
    \psp_1 \raOp \psp_2
    &= \psp_1 \iprod \psp_2 \\
    &= (\psp_1' \pprod \Triv{S_1 \to \Val}) \iprod
    (\psp_2' \pprod \Triv{S_2 \to \Val})
  \end{align*}

  Say $(\salg_1', \prob_1') = \psp_1'$,
  and  $(\salg_2', \prob_2') = \psp_2'$.
  Then the sigma algebra of $\psp_1 \raOp \psp_2$
  is
  \begin{align*}
  &\closure{\set{(E_1 \times S_1 \to \Val) \cap (E_2 \times S_2 \to \Val) \mid E_1 \in \salg_1', E_2 \in \salg_2'}} \\
    = & \closure{\set{\left( (E_1 \times (S_1 \setminus S_2) \to \Val) \cap (E_2 \times (S_2 \setminus E_1) \to \Val) \right) \times (S_1 \cap S_2) \mid E_1 \in \salg_1', E_2 \in \salg_2'}}
  \end{align*}

  Then, there exists $\psp'' \in \ProbSp((\Var \setminus (S_1 \cap S_2)) \to \Val)$ such that
  $\psp_1 \raOp \psp_2 = \psp'' \pprod  \Triv{(S_1 \cap S_2) \to \Val})$.
  Also,
  \begin{align*}
     &\set{x \in \Var \mid (\permap_1 \raOp \permap_2)(x) = 0} \\
    =&\set{x \in \Var \mid \permap_1(x) + \permap_2(x) = 0} \\
    =&\set{x \in \Var \mid \permap_1(x) = 0 \text{ and } \permap_2(x) = 0} \\
    =& S_1 \cap S_2
  \end{align*}
  Therefore, $\salg_1 \punion \salg_2$ is compatible with $\permap_1 \raOp \permap_2$
\end{proof}

\begin{lemma}
  The structure $(\Perm, \raLeq, \raValid, \raOp, \raUnit)$ is an ordered unital resource algebra (RA) as defined
  in~\cref{def:ra}.
\end{lemma}

\begin{proof}
  We check the conditions one by one.
  \begin{induction}
    \step[Condition~$a \raOp b = b \raOp a$]
      Follows from the commutativity of addition.

    \step[Condition~$(a \raOp b) \raOp c = a \raOp (b \raOp c)$]
      Follows from the associativity of addition.

    \step[Condition~$a \raLeq b \implies b \raLeq c \implies a \raLeq c$]
      $\raLeq$ is a point-wise lifting of the order $\leq$ on arithmetics,
      so it follows from the transitivity of $\leq$.

    \step[Condition~$a \raLeq a$]
      $\raLeq$ is a point-wise lifting of the order $\leq$ on arithmetics,
      so it follows from the reflexivity of $\leq$.

    \step[Condition~$\raValid(a \raOp b) \implies \raValid(a)$]
      By definition,
      \begin{align*}
        \raValid(a \raOp b)
        &\implies \forall x \in \Var, (a \raOp b)(x) \leq 1 \\
        &\implies \forall x \in \Var, a(x) + b(x) \leq 1 \\
        &\implies \forall x \in \Var, a(x) \leq 1 \\
        &\implies \raValid(a)
      \end{align*}

    \step[Condition~$\raValid(\raUnit)$]
      Note that $\raUnit = \fun \wtv. 0$ satisfies that
      $\forall x \in \Var, \raUnit(x) \leq 1$,
      so $\raValid(\raUnit)$.

    \step[Condition~$a \raLeq b \implies \raValid(b) \implies \raValid(a)$]
      By definition,
      $a \raLeq b$ means
      $\forall x \in \Var. a(x) \leq b(x)$,
      and $\raValid(b)$ means that
      $\forall x \in \Var. b(x) \leq 1$.
      Thus, $a \raLeq b$ and $\raValid(b)$ implies that
      $\forall x \in \Var. a(x) \leq b(x) \leq 1$,
      which implies $\raValid(a)$.

    \step[Condition~$\raUnit \raOp a = a$]
      By definition,
      \begin{align*}
        \raUnit \raOp a
        &= \fun x. (\fun \wtv. 0)(x) + a(x) \\
        &= \fun x. 0 + a(x) \\
        &= a.
      \end{align*}

      \step[Condition~$a \raLeq b \implies a \raOp c \raLeq b \raOp c$]
      By definition,
      \begin{align*}
        a \raLeq b
        &\iff \forall x \in \Var. a(x) \leq b(x) \\
        &\implies \forall x \in \Var. a(x) + c(x) \leq b(x) + c(x) \\
        &\implies a \raOp c \raLeq b \raOp c
        \tag*{\qedhere}
      \end{align*}
  \end{induction}
\end{proof}

\begin{lemma}
  \label{lemma:psppm-ra}
  The structure $\PSpPmRA$ is an ordered unital resource algebra (RA) as defined
  in~\cref{def:ra}.
\end{lemma}

\begin{proof}
  We want to check that $\PSpPmRA$ satisfies all the requirements to be
  an ordered unital resource algebra (RA).
  Because $\PSpPmRA$ is very close to a product of $\PSpRA$ and $\Perm$,
  the proof below is very close to the proof that product RAs
  are RA.

  First,~\cref{lemma:well-defined-psppmra-1} implies that $\raOp$
  is well-defined.

  Then we need to check all the RA axioms are satisfied.
  For any $a, b \in \PSpPmRA$ and any $\psp_1, \permap_1, \psp_2, \permap_2$ such
  that $a = (\psp_1, \permap_1), b = (\psp_2, \permap_2)$.

  We check the conditions one by one.
  \begin{induction}
\step[Condition~$\raValid(a \raOp b) \implies \raValid(a)$]
      By definition,
      $a \raOp b = (\psp_1, \permap_1) \raOp (\psp_2, \permap_2)
      = (\psp_1 \raOp \psp_2, \permap_1 \raOp \permap_2)$.
      And $\raValid(\psp_1 \raOp \psp_2, \permap_1 \raOp \permap_2)$
      implies that
      $\raValid(\psp_1 \raOp \psp_2)$ and
      $\raValid(\permap_1 \raOp \permap_2)$.
      Because $\PSpRA$ and $\Perm$ are both RAs,
      we have $\raValid(\psp_1)$ and $\raValid(\permap_1)$.
      Thus, $\raValid(\psp_1, \permap_1)$.
\step[Condition~$\raValid(\raUnit)$]
      Clear because $\raUnit = (\Triv{\Store}, \fun \p{x}. 0)$
      and $\Triv{\Store} \neq \invalid$, and
      $\forall x. \left(\fun \p{x}. 0\right)(x) \leq 1$.

    \step[Condition~$a \raLeq b \implies \raValid(b) \implies \raValid(a)$]
      $a \raLeq b $ implies that
      $\psp_1 \raLeq \psp_2$ and $\permap_1 \raLeq \permap_2$.
      $\raValid(b)$ implies that
       $\psp_2 \neq \invalid$, and
      $\forall x. \left(\permap_2 \right)(x) \leq 1$.
      Thus,  $\psp_1 \neq \invalid$, and
      $\forall x. \left(\permap_1 \right)(x) \leq 1$.
      And therefore, $\raValid(b)$.

    \step[Condition~$\raUnit \raOp a = a$]
$
        \raUnit \raOp a
        = (\Triv{\Store}, \fun \p{x}. 0) \raOp (\psp_1, \permap_1) \\
        =  (\Triv{\Store} \raOp \psp_1,  \fun \p{x}. 0 \raOp \permap_1) \\
        = (\psp_1, \permap_1) = a.
      $

    \step[Condition~$a \raLeq b \implies a \raOp c \raLeq b \raOp c$]
      $a \raLeq b $ implies that
      $\psp_1 \raLeq \psp_2$ and $\permap_1 \raLeq \permap_2$.

      Say $c = (\psp_3, \permap_3)$.
      Then $a \raOp c = (\psp_1 \raOp \psp_3, \permap_1 \raOp \permap_3)$ and $b \raOp c = (\psp_2 \raOp \psp_3, \permap_2 \raOp \permap_3)$.
      Because $\psp_1 \raLeq \psp_2$,
      $\psp_1 \raOp \psp_3 \raLeq \psp_2 \raOp \psp_3$;
      similarly,
      $\permap_1 \raOp \permap_3 \raLeq \permap_2 \raOp \permap_3$.
      Thus, $a \raOp c \raLeq b \raOp c$.
      \qedhere
  \end{induction}
\end{proof}

  \begin{lemma}
    \label{lemma:product-preserve-ra}
    If $M$ is an RA, then $\Hyp{M}$ is also an RA.
  \end{lemma}

  \begin{proof}
    RA is known to be closed under products, and
    $\Hyp{M}$ can be obtained as products of $M$,
    so we omit the proof.
  \end{proof}

  \begin{lemma}
  $\Model_I$ is an RA.
\end{lemma}

\begin{proof}
  By~\cref{lemma:psppm-ra}, $\PSpPmRA$ is an RA.
  By~\cref{lemma:product-preserve-ra},
  $\Model_I = \Hyp{\PSpPmRA}$ is also an RA.
\end{proof}

 \section{Characterizations of \SuperCond\ and Relational Lifting}
\label{sec:appendix:alt-cond-rl}

Interestingly, it is possible to characterize the conditioning modality
using the other connectives of the logic.
\begin{proposition}[Alternative Characterization of \Supercond]
\label{prop:cond-as-wand}
  The following is a logically equivalent characterization
  of the \supercond\ modality:
  \begin{align*}
    \CMod{\prob} K &\lequiv
    \begin{array}[t]{@{}r@{\,}l@{}}
      \E \m{\salg}, \m{\prob}, \m{\permap}, \m{\krnl}.
        & \Own{\m{\salg}, \m{\prob}, \m{\permap}} *
        \pure{\forall i\in I\st
        \m{\mu}(i) = \bind(\prob, \m{\krnl}(i))}
     \\ & * \;
     \A v \in \psupp(\prob).
     \Own{\m{\sigmaF}, \m{\krnl}(I)(v), \m{\permap}}
     \wand
     K(v)
    \end{array}
\end{align*}
\end{proposition}

\begin{proof}
	In the following, we sometimes abbreviate
	${\forall i\in I\st
                \m{\mu}(i) = \bind(\prob, \m{\krnl}(i))}$
	by writing just ${\m{\mu} = \bind(\prob, \m{\kappa})}$.

	We start with the embedding:

	\begin{align*}
				&\begin{array}[t]{@{}r@{\,}l@{}}
        \E \m{\salg}, \m{\prob}, \m{\permap}, \m{\krnl}.
        & \Own{\m{\salg}, \m{\prob}, \m{\permap}} *
        \pure{\forall i\in I\st
                \m{\mu}(i) = \bind(\prob, \m{\krnl}(i))}
        \\ & * \;
        \A a \in \psupp(\prob).
          \Own{\m{\sigmaF}, \m{\krnl}(I)(a), \m{\permap}}
          \wand
            K(a)
 			 \end{array}\\
{}\lequiv {} &
				\fun r.
        \exists \m{\sigmaF}, \m{\mu}, \m{\permap}, \m{\kappa}.
        \big(\Own{\m{\sigmaF}, \m{\mu'}, \m{\permap}} \sepand
				\pprop{\m{\mu} = \bind(\prob, \m{\kappa})} \sepand \\
					&\qquad \qquad
        (\forall a \in \psupp(\prob). \Own{\m{\sigmaF}, \m{\kappa} a, \m{\permap}} \sepimp K(a)) \big) (r) \\
{}\lequiv {} &
			 \fun r.
				\exists \m{\sigmaF}, \m{\mu},  \m{\permap},  \m{\kappa},
				\m{\sigmaF_1}, \m{\mu_1},  \m{\permap_1},
				\m{\sigmaF_2}, \m{\mu_2}, \m{\permap_2},
				\m{\sigmaF_3}, \m{\mu_3}, \m{\permap_3}, \\
				&\qquad \qquad
				r \extOf (\m{\sigmaF_1}, \m{\mu_1}, \m{\permap_1}) \raOp
				(\m{\sigmaF_2}, \m{\mu_2}, \m{\permap_2}) \raOp (\m{\sigmaF_3}, \m{\mu_3}, \m{\permap_3}) \land\\
				&\qquad \qquad
				(\m{\sigmaF_1}, \m{\mu_1}, \m{\permap_1}) \extOf (\m{\sigmaF}, \m{\mu}, \m{\permap})  \land
				\pprop{\m{\mu} = \bind(\prob, \m{\kappa})} \land \\
				&\qquad \qquad
				(\forall a \in \psupp(\prob). \forall r_1, r_2 \st
				r_1 \raOp (\m{\sigmaF_3}, \m{\mu_3}, \m{\permap_3}) = r_2 \land
				r_1 \extOf (\m{\sigmaF}, \m{\kappa} a, \m{\permap})
				\implies K(a)(r_2))\\
{}\lequiv {} &
				\fun r.
				\exists \m{\sigmaF}, \m{\mu}, \m{\permap},
				\m{\sigmaF_3}, \m{\mu_3}, \m{\permap_3}, \m{\kappa}. \\
				& \qquad \qquad
				r \extOf
				(\m{\sigmaF}, \m{\mu}, \m{\permap})  \raOp (\m{\sigmaF_3}, \m{\mu_3}) \land
				\pprop{\m{\mu} = \bind(\prob, \m{\kappa})} \land \\
				& \qquad \qquad
				(\forall a \in \psupp(\prob). \forall r_1, r_2 \st
				r_1 \raOp (\m{\sigmaF_3}, \m{\mu_3}, \m{\permap_3}) = r_2 \land
				r_1 \extOf (\m{\sigmaF}, \m{\kappa} a, \m{\permap})
				\implies K(a)(r_2))
	\end{align*}
For the last equivalence, the forward direction holds because
\begin{align*}
		r &\extOf (\m{\sigmaF_1}, \m{\mu_1}, \m{\permap_1}) \raOp
		(\m{\sigmaF_2}, \m{\mu_2}, \m{\permap_2}) \raOp (\m{\sigmaF_3}, \m{\mu_3}, \m{\permap_3})\\
			&\extOf  (\m{\sigmaF_1}, \m{\mu_1}, \m{\permap_1}) \raOp (\m{\sigmaF_3}, \m{\mu_3}, \m{\permap_3}) \\
			&\extOf  (\m{\sigmaF}, \m{\mu}, \m{\permap}) \raOp (\m{\sigmaF_3}, \m{\mu_3}, \m{\permap_3}).
	\end{align*}
The backward direction holds because we can pick
	$(\m{\sigmaF_1}, \m{\mu_1}, \m{\permap_1}) = (\m{\sigmaF}, \m{\mu}, \m{\permap})$,
	$(\m{\sigmaF_2}, \m{\mu_2})$ be the trivial probability space on $\store$ and
	$\m{\permap_2} = \fun \wtv. 0$.

	\begin{itemize}
		\item To show that the embedding implies the original assertion $\CMod{\prob} K $,
	we start with $\m{\mu}(i) \iprod \m{\mu_3}(i)$. For any $i$, we have
	$\m{\mu}(i) = \bind(\prob, \m{\kappa}(i))$, and thus
\begin{align*}
											\m{\mu }(i) \iprod \m{\mu_3}(i)
											&= \bind(\mu, \m{\kappa}(i)) \iprod \m{\mu_3}(i).
										\end{align*}
According to~\cref{lemma:fibre-prod-exists},
										$\m{\mu }(i) \iprod \m{\mu_3}(i)$ is defined implies that
										$\m{\kappa}(i)(a) \iprod  \m{\mu_3}(i)$ is defined for any $a \in $.
										Furthermore,
\begin{align*}
											\m{\mu }(i) \iprod \m{\mu_3}(i)
											&= \bind(\mu, \fun a. \m{\kappa}(i)(a) \iprod  \m{\mu_3}(i) )
										\end{align*}
										We abbreviate the hyperkernel $\m[i: \fun a. \m{\kappa}(i)(a) \iprod  \m{\mu_3}(i) \mid i \in I]$
										as $\m{\kappa'} $.
For any $a \in \psupp(\prob)$,
										the assertion
										\[
											\forall a \in \psupp(\prob). \forall r_1, r_2.
											r_1 \iprod (\m{\sigmaF_3}, \m{\mu_3}, \m{\permap_3}) = r_2 \land
											r_1 \extOf (\m{\sigmaF}, \m{\kappa}(I)a, \m{\permap})
											\implies K(a)(r_2)
										\]
										applies with the specific case
										$r_1 =  (\m{\sigmaF }, \m{\kappa}(I)(a), \m{\permap})$,
										gives us
										\[
K(a) ((\m{\sigmaF}, \m{\kappa}(I)(a), \m{\permap}) \raOp (\m{\sigmaF_3} , \m{\mu_3}, \m{\permap_3} )])
										\]
										By the definition of composition in our resource algebra,
										we have that $K(a)$ holds on $(\m{\sigmaF} \punion \m{\sigmaF_3},  \m{\kappa}'(I)(a), \m{\permap} + \m{\permap_3})$.

										For any $r$,
										\begin{itemize}
											\item If $\raValid(r)$, then there exists $\m{\sigmaF'}, \m{\mu'}, \m{\permap'}$ such that
												$r = (\m{\sigmaF'}, \m{\mu'}, \m{\permap'})$.
												Note that
												\begin{align*}
													r = (\m{\sigmaF'}, \m{\mu'}, \m{\permap'})
													&\extOf (\m{\sigmaF}, \m{\mu}, \m{\permap}) \raOp  (\m{\sigmaF_3} , \m{\mu_3}, \m{\permap_3})
													=  (\m{\sigmaF} \punion \m{\sigmaF_3} , \m{\mu} \iprod \m{\mu_3}, \m{\permap} + \m{\permap_3})
												\end{align*}

												By~\cref{lemma:bind-extend}, $\m{\mu} \iprod \m{\mu_3} = \bind(\mu, \m{\kappa'})$ implies
												that there exists $\m{\kappa''}$ such that
												$\m{\mu}(i) = \bind(\mu, \m{\kappa''}(i)) $, and that for any $a \in \psupp{\mu}$,
												$(\m{\sigmaF} \punion \m{\sigmaF_3}, \m{\kappa'}(I)(a)) \extTo (\m{\sigmaF'}, \m{\kappa''}(I)(a)) $.
												Thus, by monotonicity with respect to the extension order,
												that would imply 	$K(a)$ holds on $(\m{\sigmaF'}, \m{\kappa''}(I)(a), \m{\permap'})$.
												And $K(a)$ holds on $(\m{\sigmaF'}, \m{\kappa''}(I)(a), \m{\permap'})$ for any
												$a \in \psupp{\mu}$ together with
												$\m{\mu}(i) = \bind(\mu, \m{\kappa''}(i))$ implies that $r$ satisfy the original assertion
												of conditioning modality.
\item If not $\raValid(r)$, then $r$ satisfies any assertions, so $r$ satisfy the original
												assertion of conditioning modality.
										\end{itemize}

									\item
										To show the other direction that having the original assertion implies the
										embedded assertion.
										Assume $\CMod{\mu} K(r)$,
										that is,
										\begin{align*}
											    \begin{array}[t]{@{}r@{\,}l@{}}
    													\E \m{\sigmaF}, \m{\mu}, \m{\permap}, \m{\krnl}.
     													 & (\m{\sigmaF}, \m{\mu}, \m{\permap}) \raLeq r
     													 \land
     													   \forall i\in I\st
     													     \m{\mu}(i) = \bind(\prob, \m{\krnl}(i))
     													 \\ & \land \;
     													   \forall v \in \psupp(\prob).
     													     K(v)(\m{\sigmaF}, \m{\krnl}(I)(v), \m{\permap})
    											\end{array}
													(r)
										\end{align*}

										To show that $r$ also satisfy the embedding,
										we pick the witness for the existential quantifier as follows:
										let $(\m{\sigmaF_3}, \m{\mu_3})$ be the trivial probability space on
										$\Store$;
										let $\m{\permap_3} = \fun \wtv. 0$;
										pick $(\m{\sigmaF}_{\text{embd}}, \m{\mu}_{\text{embd}}, \m{\permap}_{\text{embd}})$
										be the $(\m{\sigmaF}_{\text{orig}}, \m{\mu}_{\text{orig}}, \m{\permap}_{\text{orig}})$
										that witness $\CMod{\mu} K (r)$,
										and $\m{\kappa}_{\text{embd}} = \m{\kappa}_{\text{orig}}$.

										Then:
										\begin{itemize}
											\item First we show
												\begin{align*}
													r
													&\raGeq
													(\m{\sigmaF}_{\text{orig}}, \m{\mu}_{\text{orig}}, \m{\permap}_{\text{orig}}) \\
													&= (\m{\sigmaF}_{\text{orig}}, \m{\mu}_{\text{orig}}, \m{\permap}_{\text{orig}}) \raOp (\m{\sigmaF_3}, \m{\mu_3}, \m{\permap_3}) \\
													&= (\m{\sigmaF}_{\text{embd}}, \m{\mu}_{\text{embd}}, \m{\permap}_{\text{embd}}) \raOp (\m{\sigmaF_3}, \m{\mu_3}, \m{\permap_3})
												\end{align*}
\item $\m{\mu}_{\text{orig}} = \bind(\mu, \m{\kappa}_{\text{orig}} (I) (a))$ 
											implies $\m{\mu}_{\text{embd}} = \bind(\mu, \m{\kappa}_{\text{embd}} (I) (a))$.
\item For any $r_1, r_2$,
										\[
											r_1 \raOp (\m{\sigmaF_3}, \m{\mu_3}, \m{\permap_3}) = r_2 \land
											r_1 \extOf (\m{\sigmaF}_{\text{embd}}, \m{\kappa}_{\text{embd}} (I)(a), \m{\permap}_{\text{embd}})
										\]
										implies that $r_2 = r_1 \extOf (\m{\sigmaF}_{\text{orig}}, \m{\kappa}_{\text{orig}} (I)(a), \m{\permap}_{\text{orig}})$.
										By the assumption that the orig assertion holds,
										we have $K(a) (\m{\sigmaF}_{\text{orig}}, \m{\kappa}_{\text{orig}} (I)(a), \m{\permap}_{\text{orig}})$,
                which implies $K(a)(r_2)$.
										\end{itemize}

										Therefore, $r$ also satisfy the embedding.\qedhere
								\end{itemize}
									\end{proof}

\begin{lemma}[Alternative Characterization of Relational Lifting]\label{lemma:cpl-is-cpl}
Given a relation $R$ over $\Store_1 \times \Store_2$,
	then $\cpl{R}(\m{\sigmaF}, \m{\mu}, \m{\permap})$ holds
iff there exists $\wh{\mu}$ over $\m{\sigmaF}(1) \otimes \m{\sigmaF}(2)$ such that
	$\wh{\mu}(R) = 1$,
  $\wh{\mu} \circ \pi_1^{-1} = \m{\mu}(1)$,
  and $\wh{\mu} \circ \pi_2^{-1} = \m{\mu}(2)$.
\end{lemma}

\begin{proof}
	We first unfold the definition of the coupling modality:
	\begin{align*}
		&  \cpl{R} (\m{\sigmaF}, \m{\mu}, \m{\permap})\\
{} \iff {} &
		\left( \exists \mu. \mu(R) = 1 \land
		\CC{\prob} \m{v}.
		\LAnd_{\ip{x}{i} \in X}
    \sure{\ip{x}{i} = \m{v}(\ip{x}{i})} \right)
    (\m{\sigmaF}, \m{\mu}, \m{\permap})
    \\
{} \iff {} &
		\exists \mu. \mu(R) = 1 \land
		\left( \CC{\prob} \m{v}.
		\LAnd_{\ip{x}{i} \in X}
    \sure{\ip{x}{i} = \m{v}(\ip{x}{i})} \right)
    (\m{\sigmaF}, \m{\mu}, \m{\permap})
    \\
{} \iff {}&
		\exists \prob. \prob(R) = 1 \land
    \exists \m{\sigmaF'}, \m{\mu'}, \m{\permap'}, \m{\krnl} \st
      (\m{\sigmaF}, \m{\mu}, \m{\permap}) \raGeq (\m{\sigmaF'}, \m{\mu'}, \m{\permap'}) \land
      \forall i \in I.\bind(\prob, \m{\krnl}(i) )= \m{\prob}(i) {}\land{} \\
		&\qquad \qquad 	\forall \m{v}\in \psupp(\prob) \st
    \left( \LAnd_{\ip{x}{i} \in X} \sure{\ip{x}{i} = \m{v}(\ip{x}{i})} \right)
    (\m{\sigmaF'}, \m{\mu'}, \m{\permap'})\\
{} \iff {} &
		\exists \prob. \prob(R) = 1 \land
		\exists \m{\sigmaF'}, \m{\mu'}, \m{\permap'}, \m{\krnl} \st
       (\m{\sigmaF}, \m{\mu}, \m{\permap}) \raGeq (\m{\sigmaF'}, \m{\mu'}, \m{\permap'}) \land
			\forall i \in I.\bind(\prob, \m{\krnl}(i) )= \prob_i {}\land{}
		  \\
			&\qquad \qquad \qquad \forall \m{v}\in \psupp(\prob) \st
      \LAnd_{i \in \{1,2\}}
		  \LAnd_{\ip{x}{i} \in X}
        \pure{\almostM{(\ip{x}{i} = \m{v}(\ip{x}{i}))}{(\sigmaF'(i), \m{\krnl}(i)(\m{v}))} }
        \land \\
      & \qquad \qquad \qquad \qquad \qquad \qquad \qquad  \pure{\m{\krnl}(i)(\m{v}) \circ \inv{(\ip{x}{i} = \m{v}(\ip{x}{i}))} = \delta_{\True} }
	\end{align*}

	Now, to show that
	$\cpl{R} (\m{\sigmaF}, \m{\mu}, \m{\permap})$
  implies there exists $\wh{\mu}$ over $\m{\sigmaF}(1) \otimes \m{\sigmaF}(2)$ such that
	$\wh{\mu}(R) = 1$,
	$\wh{\mu} \circ \pi_1^{-1} = \mu_1$,
	and $\wh{\mu} \circ \pi_2^{-1} = \mu_2$, we define $\wh{\mu}$ over
	$\sigmaF_1 \otimes \sigmaF_2$  as
	$\bind(\mu, \fun \m{v}. \m{\krnl}(1)(\m{v}) \pprod \m{\krnl}(2)(\m{v}))$.
	Then,
\begin{align*}
		\wh{\mu}(R) &= \bind(\mu, \m{v}. \m{\krnl}(1)(\m{v}) \pprod \m{\krnl}(2)(\m{v}))(R) \\
								&= \sum_{\mathclap{\m{v} \in \psupp(\prob)}} \mu(\m{v}) \cdot \left( \m{\krnl}(1)(\m{v}) \pprod \m{\krnl}(2)(\m{v}) \right)(R)
	\end{align*}
Since $\mu(R) = 1$, then for all $\m{v} \in \psupp_{\prob}$, and $\m{v} \in R$, by additivity:
	\begin{align}
		\left( \m{\krnl}(1)(\m{v}) \pprod \m{\krnl}(2)(\m{v}) \right)(R) &\geq \left( \m{\krnl}(1)(\m{v}) \pprod \m{\krnl}(2)(\m{v}) \right)(\m{v}) \notag \\
   &=  \m{\krnl}(1)(\m{v}) (\pi_1 (\m{v})) \cdot \m{\krnl}(2)(\m{v}) (\pi_2 (\m{v})) \label{eq:cpladequacy:explain}\\
	 &= \m{\krnl}(1)(\m{v})(\LAnd_{\ip{x}{1} \in X} \ip{x}{1} = \m{v}(\ip{x}{1})) \cdot \m{\krnl}(2)(\m{v})(\LAnd_{\ip{x}{2} \in X} \ip{x}{2} = \m{v}(\ip{x}{2})) \notag \\
	 &= 1		  \notag
	\end{align}
where~\cref{eq:cpladequacy:explain} is because $\m{v}$ as a singleton can also be thought of as a Cartesian product.
Thus,
	\begin{align*}
		\wh{\mu}(R) &= \sum_{\mathclap{\m{v} \in \psupp(\prob)}} \mu(\m{v}) \cdot 1 \cdot 1
								= \sum_{\mathclap{\m{v} \in \psupp(\prob)}} \mu(\m{v})
								= 1
	\end{align*}
Meanwhile, for $E_i \in \sigmaF_i$, and
  let $j = 2$ if $i = 1$ and $j = 1$ if $i = 2$,
\begin{align}
		(\wh{\mu} \circ \pi_i^{-1}) (E_i)
&= \wh{\mu} (E_i \times \Store_j) \\
		&= \bind(\mu, \fun \m{v} . \m{\krnl}(1)(\m{v}) \pprod \m{\krnl}(2)(\m{v}))(  E_i \times \Store_j) \notag \\
		&= \sum_{\mathclap{\m{v} \in \psupp{\mu}}} \left( \m{\krnl}(1)(\m{v}) \pprod \m{\krnl}(2)(\m{v}) \right) (E_i \times \Store_j) \notag \\
		&= \sum_{\mathclap{\m{v} \in \psupp{\mu}}} \m{\krnl}(i)(\m{v})(E_i) \cdot \m{\krnl}(j)(\m{v}) (\Store_j) \notag \\
		&= \sum_{\mathclap{\m{v} \in \psupp{\mu}}} \m{\krnl}(i)(\m{v})(E_i) \cdot 1 \notag \\
		&= \bind(\mu, \fun \m{v} . \m{\krnl}(i)) (E_i) \notag \\
		&= \mu_i(E_i) \notag
	\end{align}
Thus, we complete the forward direction.

	For the backwards direction,
	we pick $\mu = \wh{\mu}$,
  $\m{\mu'}(i) = \pi_{i}\wh{\mu}$ ($\m{\sigmaF'}(i)$ accordingly),
  $\m{\permap'}=\m{\permap'}$,
	and $\m{\krnl}(i) = \fun \m{v}. \dirac{\pi_{\ip{x}{i} \in X} \m{v}}$.
Then,
\begin{align*}
		\bind(\mu, \m{\krnl}(i))
		&= \bind(\wh{\mu}, \fun \m{v}. \dirac{\pi_{\ip{x}{i} \in X} \m{v}})\\
		&= \wh{\mu} \circ \pi_i^{-1} = \m{\mu}(i)
	\end{align*}
Also, by definition,
	$\m{k}(i)(\m{v}) = \dirac{\pi_{\ip{x}{i} \in X} \m{v}}$.
Thus,
  $\almostM{\{\ip{x}{i} = \m{v}(\ip{x}{i})\}}{(\m{\sigmaF}(i), \m{\krnl}(i)(\m{v}))}$ and
  $\m{\krnl}(i)(\m{v}) \circ \inv{\ip{x}{i} = \m{v}(\ip{x}{i})} = \dirac{\True}$.
\end{proof}

 \section{Soundness}
\label{sec:appendix:soundness}

\subsection{Soundness of Primitive Rules}
\label{sec:appendix:primitive-rules}

\subsubsection{Soundness of Distribution Ownership Rules}
\begin{lemma}
\label{proof:and-to-star}
  \Cref{rule:and-to-star} is sound.
\end{lemma}

\begin{proof}
  Assume a valid $a\in\Model_I$ is such that
  it satisfies $ P \land Q $.
  This means that for some $(\m{\salg},\m{\prob},\m{\permap}) \raLeq a$, both
    $P(\m{\salg},\m{\prob},\m{\permap})$ and
    $Q(\m{\salg},\m{\prob},\m{\permap})$
  hold.
  We want to prove $(P \ast Q)(a)$ holds.
  To this end, let
  $ (\m{\salg}_1, \m{\prob}_1, \m{\permap}_1) $ and
  $ (\m{\salg}_2, \m{\prob}_2, \m{\permap}_2) $
  be such that
  for every $i \in \idx(P)$:
  \begin{align*}
    \m{\salg}_1(i) &= \m{\salg}(i)
    &
    \m{\salg}_2(i) &= \set{\emptyset, \Outcomes}
    \\
    \m{\prob}_1(i) &= \m{\prob}(i)
    &
    \m{\prob}_2(i) &= \fun \event. \ITE{\event=\Outcomes}{1}{0}
    \\
    \m{\permap}_1(i) &= \m{\permap}(i)
    &
    \m{\permap}_2(i) &= \fun \wtv. 0
  \intertext{
    and for all $i\in I \setminus \idx(P)$:
  }
    \m{\salg}_2(i) &= \m{\salg}(i)
    &
    \m{\salg}_1(i) &= \set{\emptyset, \Outcomes}
    \\
    \m{\prob}_2(i) &= \m{\prob}(i)
    &
    \m{\prob}_1(i) &= \fun \event. \ITE{\event=\Outcomes}{1}{0}
    \\
    \m{\permap}_2(i) &= \m{\permap}(i)
    &
    \m{\permap}_1(i) &= \fun \wtv. 0
  \end{align*}
Clearly, by construction,
  $
    (\m{\salg}_1, \m{\prob}_1, \m{\permap}_1)
    \iprod
    (\m{\salg}_2, \m{\prob}_2, \m{\permap}_2)
    =
    (\m{\salg},\m{\prob},\m{\permap}).
  $
  and
  $P(\m{\salg}_1, \m{\prob}_1, \m{\permap}_1)$.
  Since $\idx(P) \inters \idx(Q) = \emptyset$,
  we also have
  $Q(\m{\salg}_2, \m{\prob}_2, \m{\permap}_2)$.
  Therefore,
  $(P \ast Q)(\m{\salg}, \m{\prob}, \m{\permap})$,
  and so $(P \ast Q)(a)$ by upward closure.
\end{proof} \begin{lemma}
\label{proof:dist-inj}
  \Cref{rule:dist-inj} is sound.
\end{lemma}

\begin{proof}
  Assume a valid $a\in\Model_I$ is such that both
  $ \distAs{E\at{i}}{\prob }(a) $ and
  $ \distAs{E\at{i}}{\prob'}(a) $
  hold.
  Let $ a = (\m{\salg}, \m{\prob}_0,\m{\permap}) $,
  then we know
  $ \prob = \m{\prob}_0 \circ \inv{E\at{i}} = \prob'$,
  which proves the claim.
\end{proof}
 \begin{lemma}
\label{proof:sure-merge}
  \Cref{rule:sure-merge} is sound.
\end{lemma}

\begin{proof}
  The proof for the forward direction is very similar to
  the one for~\cref{rule:sure-eq-inj}.
  For $a \in \Model_I$,
  if $(\sure{E_1\at{i}} \ast \sure{E_2 \at{i}})(a)$.
  Then there exists
  $a_1, a_2$ such that $a_1 \raOp a_2 \raLeq a$ and
  $\sure{E_1 \at{i}}(a_1)$,
  $\sure{E_2 \at{i}}(a_2)$.
  Say $a = (\m{\sigmaF}, \m{\mu}, \m{\permap})$,
  $a_1 = (\m{\sigmaF}_1, \m{\mu}_1, \m{\permap}_1)$
  and $a_2 = (\m{\sigmaF}_2, \m{\mu}_2, \m{\permap}_2)$.
  Then $\sure{E_1\at{i}}(a_1)$ implies that
  \begin{align*}
    \m{\mu}_1 (\inv{E_1\at{i}}(\True)) = 1
  \end{align*}
  And similarly,
  \begin{align*}
    \m{\mu}_2 (\inv{E_2\at{i} }(\True)) = 1
  \end{align*}
  Thus,
  \begin{align*}
    \m{\mu} (\inv{E_1\at{i}}(\True) \cap \inv{E_2\at{i}}(\True) )
    &= \m{\mu}_1 (\inv{E_1\at{i}}(\True))  \cdot \m{\mu}_2 (\inv{E_2\at{i}}(\True))
     = 1.
  \end{align*}
  Hence,
  \begin{align*}
    \m{\mu} (\inv{E_1\at{i} \land E_2\at{i}}(\True) )
    &= \m{\mu} (\inv{E_1\at{i}}(\True) \cap \inv{E_2\at{i}}(\True)) = 1
  \end{align*}
  Thus, $\sure{E_1\at{i} \land E_2\at{j}} (a)$.

  Now we prove the backwards direction:
  Say $a = (\m{\sigmaF}, \m{\mu}, \m{\permap})$.
  if  $\sure{E_1\at{i} \land E_2\at{j}} (a)$,
  then $\m{\mu} (\inv{E_1\at{i} \land E_2\at{i}}(\True)) = 1$,
  and then
  \begin{align*}
    \m{\mu} (\inv{E_1\at{i}}(\True))  &\geq \m{\mu} (\inv{E_1\at{i} \land E_2\at{i}}(\True)) = 1 \\
    \m{\mu} (\inv{E_2\at{i}}(\True))  &\geq \m{\mu} (\inv{E_1\at{i} \land E_2\at{i}}(\True)) = 1
  \end{align*}

  Let $\m{\sigmaF_1} = \closure{\inv{E_1\at{i}}(\True)}$
  and $\m{\sigmaF_2} = \closure{\inv{E_2\at{i}}(\True)}$.
  Then,
  \begin{gather*}
    \sure{E_1\at{i}} (\m{\sigmaF_1}, \constrain{\m{\mu}}{\m{\sigmaF_1}}, \fun \wtv. 0) \\
    \sure{E_2\at{i}} (\m{\sigmaF_2}, \constrain{\m{\mu}}{\m{\sigmaF_2}}, \fun \wtv. 0) \\
    (\m{\sigmaF_1}, \constrain{\m{\mu}}{\m{\sigmaF_1}}, \fun \wtv. 0) \ast  (\m{\sigmaF_2}, \constrain{\m{\mu}}{\m{\sigmaF_2}}, \fun \wtv. 0) \raLeq a
  \end{gather*}
  Thus, $\sure{E_1\at{i}} \ast \sure{E_2\at{i}}$ holds on $a$.
\end{proof} \begin{lemma}
\label{proof:sure-and-star}
  \Cref{rule:sure-and-star} is sound.
\end{lemma}

\begin{proof}
  Assume $a = (\m{\salg}, \m{\prob}, \m{\permap}) \in \Model_I$ and
  $(\sure{E\at{i}} \land P)(a)$ holds.
  We want to show that
  $(\sure{E\at{i}} \ast P)(a)$ holds.
  First note that:
  \begin{align*}
    (\sure{E\at{i}} \land P)(a)
    & \implies \sure{E\at{i}}(a) \land P(a) \\
    & \implies \almostM{E}{(\m{\salg}(i), \m{\prob}(i))}
    \land \m{\prob} \circ \inv{E\at{i}}(\true) = \dirac{\True}
    \land P(a)
  \end{align*}

  Define $\m{\salg}', \m{\permap}_{\aexpr}, \m{\permap}_{P}$ such that,
  for any $j \in I$:
  \begin{align*}
    \m{\salg}'(j) &=
    \begin{cases}
      \set{\emptyset, \Store} \CASE j \neq i\\
      \set{\emptyset, \Store, \inv{E\at{i}}(\True), \Store \setminus
      \inv{E\at{i}}(\True)} \OTHERWISE
    \end{cases}
    \\
    \m{\permap}_{\aexpr}(j) &=
    \begin{cases}
      \fun \wtv.0 \CASE j \ne i \\
      \fun \ip{x}{i}.
        \ITE{\p{x} \in \pvar(\aexpr)}{\m{\permap}(i)(\ip{x}{i})/2}{0}
    \CASE j = i
    \end{cases}
    \\
    \m{\permap}_{P}(j) &=
    \begin{cases}
      \m{\permap}(j) \CASE j \ne i \\
      \fun \ip{x}{i}.
        \ITE{\p{x} \in \pvar(\aexpr)}{\m{\permap}(i)(\ip{x}{i})/2}{\m{\permap}(i)(\ip{x}{i})}
    \CASE j = i
    \end{cases}
  \end{align*}
  By construction, we have
  $
    \m{\permap} = \m{\permap}_{\aexpr} \raOp \m{\permap}_{P}.
  $
  Now let:
  \begin{align*}
    b &= (\m{\salg}', \restr{\m{\prob}}{\m{\salg}'}, \m{\permap}_{\aexpr})
    &
    a' &= (\m{\salg}, \m{\prob}, \m{\permap}_{P})
  \end{align*}
  note that $\raValid(b)$ holds because $\m{\salg}'(i)$ can at best be non-trivial on $\pvar(\aexpr)$.
  The resource $a'$ is also valid, since $\m{\permap}_{P}$ has the same non-zero components as $\m{\permap}$.
  Then
  $\sure{E\at{i}}(b)$ holds because
  $\almostM{E}{( \m{\salg}'(i), \restr{\m{\prob}}{\m{\salg}'}(i) )}$
  and $\restr{\m{\prob}}{\m{\salg}'} \circ \inv{E\at{i}}
  = \m{\prob} \circ \inv{E\at{i}} = \dirac{\True}$.
  By applying~\cref{lemma:indep-prod-exists}, it is easy to show that
  $(\m{\salg}', \restr{\m{\prob}}{\m{\salg}'}) \iprod (\m{\salg}, \m{\prob})$
  is defined and is equal to $(\m{\salg}, \m{\prob})$.
  Therefore,
  $\raValid(b \raOp a)$ and $b \raOp a = a$.
  By the side condition $\psinv(P, \pvar(E\at{i}))$ and the fact that
  $\m{\permap}_{P}$ is a scaled down version of $\m{\permap}$,
  we obtain from $P(a)$ that $P(a')$ holds too.
  This proves that
  $(\sure{E\at{i}} \ast P)(a)$ holds, as desired.
\end{proof}
 \begin{lemma}
\label{proof:prod-split}
  \Cref{rule:prod-split} is sound.
\end{lemma}

\begin{proof}
  For any $(\m{\salg},\m{\prob}, \m{\permap})$ such that
  $(\distAs{(\aexpr_1\at{i}, \aexpr_2\at{i})}{\mu_1 \pprod \mu_2}) (\m{\salg},\m{\prob}, \m{\permap})$,
  by definition, it must
  \begin{align*}
      \E \m{\salg'},\m{\prob'}.
      (\Own{\m{\salg'},\m{\prob'}})(\m{\salg},\m{\prob}, \m{\permap}) *
    \almostM{(\aexpr_1, \aexpr_2)}{(\m{\salg'}(i),\m{\prob'}(i))}
    \land
    \mu_1 \pprod \mu_2 = \m{\prob'}(i) \circ \inv{(\aexpr_1, \aexpr_2)}
    .
  \end{align*}
  We can derive from it that
  \begin{align*}
    \E \m{\salg'},\m{\prob'}, \m{\permap'}.
      & (\m{\salg'},\m{\prob'}) \raLeq (\m{\salg},\m{\prob}, \m{\permap}) * \\
      & \Big( \forall a, b  \in A. \exists L_{a, b},U_{a,b} \in \salg'(i)  \st
        L_{a, b} \subs \inv{(\aexpr_1, \aexpr_2)}(a, b) \subs U_{a, b}
       \land
        \prob'(L_{a, b})=\prob'(U_{a, b}) \land  \\
      &
       \mu_1 \pprod \mu_2 (a, b)
= \m{\prob'}(i) (L_{a, b})
       = \m{\prob'}(i) (U_{a, b})
    \Big)
  \end{align*}
  Also, for any $a, b, a', b' \in A$ such that $a \neq a'$
  or $b \neq b'$, we have
  $L_{a,b}$ disjoint from $L_{a',b'}$ because on $L_{a,b} \inters L_{a',b'}$,
  the random variable $(\aexpr_1, \aexpr_2)$ maps to both
  $(a, b)$ and $(a',b')$.

  Define
  \[
    \m{\salg}_1(i) = \closure{\set{(\Union_{b \in A} L_{a, b} ) \mid a \in A} \union \set{(\Union_{b \in A} U_{a, b})  \mid a \in A}},
  \]
  and similarly define
    \[
      \m{\salg}_2(i) = \closure{\set{(\Union_{a \in A} L_{a, b} ) \mid b \in A} \union \set{(\Union_{a \in A} U_{a, b} )  \mid b \in A}}.
  \]
  Denote $\m{\prob'}$ restricted to $\m{\salg}_1$ as $\m{\prob'}_1$
  and $\m{\prob'}$ restricted to $\m{\salg}_2$ as $\m{\prob'}_2$.

  We want to show that
  $(\m{\salg}_1(i), \m{\prob'}_1(i)) \indepcomb (\m{\salg}_2(i), \m{\prob'}_2(i)) \extTo (\m{\salg'}(i), \m{\prob'}(i))$,
  which boils down to show that for any $\event_1 \in \m{\salg}_1(i)$, any
  $\event_2 \in \m{\salg}_2(i)$,
  \begin{align*}
    \m{\prob'}(\event_1 \inters \event_2) = \m{\prob'}_1(\event_1) \cdot  \m{\prob'}_2(\event_2)
  \end{align*}

      For convenience, we will
      denote $\union_{b \in A} L_{a, b}$ as $L_a$,
      denote $\union_{a \in A} L_{a, b}$ as $L_b$,
      denote $\union_{b \in A} U_{a, b}$ as $U_a$,
      and denote $\union_{a \in A} U_{a, b}$ as $U_b$.

      First, using a standard construction in measure theory proofs,
      we rewrite $\salg_1$ and $\salg_2$ as sigma algebra generated
      by sets of partitions.
      Specifically, $\salg_1$ is equivalent to
      \[
        \closure{\set{\Inters_{a \in S_1} L_a \inters \Inters_{a \in S_2} U_a \setminus (\Union_{a \in A \setminus S_1} L_a \union \Union_{a \in A \setminus S_2} U_a)  \mid S_1, S_2 \subseteq A}}
      \]
      and similarly, $\salg_2$ is equivalent to
      \[
        \closure{\set{\Inters_{b \in T_1} L_b \inters \Inters_{b \in T_2} U_b \setminus (\Union_{b \in A \setminus T_1} L_b \union \Union_{b \in A \setminus T_2} U_b)  \mid T_1, T_2 \subseteq A}}.
      \]
      Thus, by~\cref{lemma:sigma-alg-representation}, any event $\event_1$ in
      $\salg_1$ can be represented by
      \[
        \Dunion_{S_1 \in I_1, S_2 \in I_2}
        \Inters_{a \in S_1} L_a \inters \Inters_{a \in S_2} U_a \setminus (\Union_{a \in A \setminus S_1} L_a \union \Union_{a \in A \setminus S_2} U_a)
      \]
      for some $I_1, I_2 \subseteq \mathcal{P}(A)$, where
      $\mathcal{P}$ is the powerset over $A$.
      Similarly, any event $\event_2$ in $\salg_2$ can be represented by
      \[
        \Dunion_{S_3 \in I_3, S_4 \in I_4}
        \Inters_{b \in S_3} L_b \inters \Inters_{b \in S_4} U_b \setminus (\Union_{b \in A \setminus S_3} L_b \union \Union_{b \in A \setminus S_2} U_b)
      \]
      for some  $I_3, I_4 \subseteq \mathcal{P}(A)$.
      Thus, $\event_1 \inters \event_2$ can be represented as
      \begin{align*}
        \event_1 \inters \event_2
        &=(\Dunion_{S_1 \in I_1, S_2 \in I_2}
        \Inters_{a \in S_1} L_a \inters \Inters_{a \in S_2} U_a \setminus (\Union_{a \in A \setminus S_1} L_a \union \Union_{a \in A \setminus S_2} U_a) ) \\
        &\Inters
        (\Dunion_{S_3 \in I_3, S_4 \in I_4}
        \Inters_{b \in S_3} L_b \inters \Inters_{b \in S_4} U_b \setminus (\Union_{b \in A \setminus S_3} L_b \union \Union_{b \in A \setminus S_2} U_b) )\\
        = & \Dunion_{S_1 \in I_1, S_2 \in I_2, S_3 \in I_3, S_4 \in I_4} (\Inters_{a \in S_1} L_a \inters \Inters_{a \in S_2} U_a \setminus (\Union_{a \in A \setminus S_1} L_a \union \Union_{a \in A \setminus S_2} U_a) ) \\
          &\inters ( \Inters_{b \in S_3} L_b \inters \Inters_{b \in S_4} U_b \setminus (\Union_{b \in A \setminus S_3} L_b \union \Union_{b \in A \setminus S_2} U_b) )
      \end{align*}

      Because $L_{a,b}$ and $L_{a',b'}$ are disjoint as long as not
      $a = a'$ and $b = b'$,
      we have $L_a$ disjoint from $L_{a'}$ if $a \neq a'$.
      Thus,
      $\Inters_{a \in S_1} L_a \inters \Inters_{a \in S_2} U_a \setminus (\Union_{a \in A \setminus S_1} L_a \union \Union_{a \in A \setminus S_2} U_a)$
      is not empty only when $S_1$ is singleton and empty.
      \begin{itemize}
        \item If $S_1$ is empty,
      then
      \[
        \Inters_{a \in S_1} L_a \inters \Inters_{a \in S_2} U_a \setminus (\Union_{a \in A \setminus S_1} L_a \union \Union_{a \in A \setminus S_2} U_a)
        =  \Inters_{a \in S_2} U_a \setminus (\Union_{a \in A} L_a \union \Union_{a \in A \setminus S_2} U_a)
      \]
      has measure 0 because $\Union_{a \in A} L_a$ has measure 1.
        \item Otherwise, if $S_1$ is singleton, say $S_1 = \{a'\}$,
      then
      \begin{align*}
        \Inters_{a \in S_1} L_a \inters \Inters_{a \in S_2} U_a \setminus (\Union_{a \in A \setminus S_1} L_a \union \Union_{a \in A \setminus S_2} U_a)
        &= L_{a'} \inters \Inters_{a \in S_2} U_a \setminus \Union_{a \in A \setminus S_2} U_a).
      \end{align*}
Furthermore,
      \begin{align*}
        \m{\prob'}(\Inters_{a \in S_2} U_a)
        &= \m{\prob'}(\Inters_{a \in S_2} L_a \disjunion (U_a \setminus L_a)) \\
        &= \m{\prob'}(\Inters_{a \in S_2} L_a) + 0
      \end{align*}
      And $\Inters_{a \in S_2} L_a$ is non-empty only if
      $S_2$ is a singleton set or empty set.
      Thus, $L_{a'} \inters \Inters_{a \in S_2} U_a \setminus \Union_{a \in A \setminus S_2} U_a) \subseteq \Inters_{a \in S_2} U_a$ has non-zero measure only if
      $S_2$ is empty or a singleton set.
\begin{itemize}
        \item When $S_2$ is empty,
      \begin{align*}
        L_{a'} \inters \Inters_{a \in S_2} U_a \setminus \Union_{a \in A \setminus S_2} U_a
        &= L_{a'} \setminus \Union_{a \in A} U_a
        \subseteq L_{a'} \setminus  U_{a'}
        =\emptyset
      \end{align*}
        \item When $S_2 = \{a'\}$,
      \begin{align*}
        L_{a'} \inters \Inters_{a \in S_2} U_a \setminus \Union_{a \in A \setminus S_2} U_a
        &= L_{a'} \setminus \Union_{a \in A, a \neq a'} U_a .
      \end{align*}
        \item
      When $S_2 = \{a''\}$ for some $a'' \neq a'$
      \begin{align*}
        L_{a'} \inters \Inters_{a \in S_2} U_a \setminus \Union_{a \in A \setminus S_2} U_a
        &= L_{a'} \inters U_{a''} \setminus \Union_{a \in A, a \neq a''} U_a \\
        &= \emptyset
      \end{align*}
      \end{itemize}
      \end{itemize}

Thus,
      \begin{align*}
         \m{\prob'}(\event_1)
        = & \m{\prob'}\Big(\Union_{S_1 \in I_1, S_2 \in I_2} \Inters_{a \in S_1} L_a \inters \Inters_{a \in S_2} U_a \setminus (\Union_{a \in A \setminus S_1} L_a \union \Union_{a \in A \setminus S_2} U_a) \inters) \\
          = & \m{\prob'}\Big(\Union_{\{a'\} \in I_1, S_2 \in I_2} ( L_{a'} \inters \Inters_{a \in S_2} U_a \setminus \Union_{a \in A \setminus S_2} U_a) \Big) \\
          = & \m{\prob'}\Big(\Union_{\{a'\} \in I_1 \inters I_2} L_{a'} \inters U_{a'} \setminus \Union_{a \in A, a \neq a'} U_a \Big)  \\
          = & \m{\prob'}\Big(\Union_{\{a'\} \in I_1 \inters I_2} ( L_{a'} \setminus \Union_{a \in A, a \neq a'} U_a )  \Big) \\
          = & \m{\prob'}\Big(\Union_{\{a'\} \in I_1 \inters I_2} ( L_{a'} \setminus \Union_{a \in A, a \neq a'} (L_a \Union (U_a \setminus L_a)) )  \Big) \\
          = & \m{\prob'}\Big(\Union_{\{a'\} \in I_1 \inters I_2} ( L_{a'} \setminus \Union_{a \in A, a \neq a'} (L_a ) ) \Big) \\
          = & \m{\prob'}\Big(\Union_{\{a'\} \in I_1 \inters I_2} L_{a'} \Big)
      \end{align*}
      Denote $\Union_{\{a'\} \in I_1 \inters I_2} L_{a'}$
      as $\event'_1$.
      And $\event_1 \setminus \event'_1$ and $\event'_1 \setminus \event_1$ both have measure 0.

      Similar results hold for $\event_2$ as well, and we can show that
      \begin{align*}
        \m{\prob'}(\event_2)
         = & \m{\prob'}\Big(\Union_{\{b'\} \in I_3 \inters I_4} L_{b'} \Big)
      \end{align*}
      Denote $\Union_{\{b'\} \in I_3 \inters I_4} L_{b'}$
      as $\event'_2$.
      And $\event_2 \setminus \event'_2$ and $\event'_2 \setminus \event_2$ both have measure 0.

      Thus,
      \begin{align*}
         \m{\prob'}(\event_1 \inters \event_2)
        =& \m{\prob'}(\event_1 \inters \event_2 \inters \event'_1)
        + \m{\prob'}((\event_1 \inters \event_2) \setminus \event'_1)\\
        =& \m{\prob'}(\event_1 \inters \event_2 \inters \event'_1)
        + 0 \\
        =& \m{\prob'}(\event_1 \inters \event_2 \inters \event'_1 \inters \event'_2) + \m{\prob'}((\event_1 \inters \event_2 \inters \event'_1) \setminus \event'_2) + 0 \\
        =& \m{\prob'}(\event_1 \inters \event_2 \inters \event'_1 \inters \event'_2) + 0 + 0 \\
        =&  \m{\prob'}(\event_1 \inters \event_2 \inters \event'_1 \inters \event'_2) +  \m{\prob'}((\event_2 \inters \event'_1 \inters \event'_2 ) \setminus \event_1) \\
        =&  \m{\prob'}(\event_2 \inters \event'_1 \inters \event'_2)  \\
        =&  \m{\prob'}(\event_2 \inters \event'_1 \inters \event'_2) + \m{\prob'}((\event'_1 \inters \event'_2 ) \setminus \event_2) \\
        =&  \m{\prob'}(\event'_1 \inters \event'_2) \\
        =&  \m{\prob'}\left((\Union_{\{a'\} \in I_1 \inters I_2} L_{a'}) \inters (\Union_{\{b'\} \in I_3 \inters I_4} L_{b'})\right) \\
         =&  \m{\prob'}\left(\Union_{\{a'\} \in I_1 \inters I_2, \{b'\} \in I_3 \inters I_4} L_{a', b'}\right) \\
         =&  \sum_{\substack{\{a'\} \in I_1 \inters I_2 \\ \{b'\} \in I_3 \inters I_4}} \m{\prob'}(L_{a', b'})
       \end{align*}

Next we show that
       $         \m{\prob'}(i) (L_{a, b})
         = \m{\prob'}(i) (\event_1) \cdot \m{\prob'}(i) (\event_2)$.
         Note that
      $\m{\prob'}(L_a)  = \sum_{b} \m{\prob'}(L_{a,b}) = \m{\prob'}(\inv{\aexpr_1}(a))$,
      and
      $\m{\prob'}(L_b)  = \sum_{a} \m{\prob'}(L_{a,b}) = \m{\prob'}(\inv{\aexpr_2}(b))$.
      And $\mu_1 \pprod \mu_2 = \m{\prob'}(i) \circ \inv{(\aexpr_1, \aexpr_2)}$
      implies that
      \begin{align*}
        \m{\prob'}(i) (L_{a, b})
        &= \mu_1 \pprod \mu_2 (a, b)\\
        &=\mu_1(a) \cdot \mu_2(b)
      \end{align*}
      Then
      \begin{align*}
        \mu_1(a)
        &= \mu_1(a) \cdot \sum_{b \in A} \mu_2(b)\\
        &= \sum_{b \in A} \mu_1(a) \cdot \mu_2(b) \\
        &= \sum_{b \in A} \m{\prob'}(i) (L_{a, b}) \\
        &= \m{\prob'}(i) \left(\sum_{b \in A} L_{a, b}\right)  \\
        &= \m{\prob'}(i) (L_a),
      \end{align*}
      and similarly,
      \begin{align*}
        \mu_2(b)
        &= \left(\sum_{a \in A} \mu_1(a)\right) \cdot \mu_2(b)\\
        &= \sum_{a \in A} (\mu_1(a) \cdot \mu_2(b))\\
        &= \sum_{a \in A} \m{\prob'}(i) (L_{a, b}) \\
        &= \m{\prob'}(i) \left(\sum_{a \in A} L_{a, b}\right)  \\
        &= \m{\prob'}(i) (L_b).
      \end{align*}
      Thus,
      \begin{align*}
         \m{\prob'}(i) (L_{a, b})
         &=\mu_1(a) \cdot \mu_2(b)
         = \m{\prob'}(i) (L_a) \cdot \m{\prob'}(i) (L_b)
      \end{align*}

Therefore,
       \begin{align*}
         \m{\prob'}(\event_1 \inters \event_2)
         =&  \sum_{\substack{\{a'\} \in I_1 \inters I_2 \\ \{b'\} \in I_3 \inters I_4}} \m{\prob'}(L_{a', b'}) \\
         =&  \sum_{\substack{\{a'\} \in I_1 \inters I_2 \\ \{b'\} \in I_3 \inters I_4}} \m{\prob'}(L_{a'}) \cdot \m{\prob'}(L_{b'}) \\
       =&  \sum_{\mathclap{\{a'\} \in I_1 \inters I_2}}  \m{\prob'}(L_{a'}) \cdot \sum_{\mathclap{\{b'\} \in I_3 \inters I_4}} \m{\prob'}(L_{b'}) \\
       =& \m{\prob'}(\event_1) \cdot \m{\prob'}(\event_2)\\
       =& \m{\prob_1'}(\event_1) \cdot \m{\prob_2'}(\event_2)
      \end{align*}

  Thus we have
  $(\m{\salg}_1, \m{\prob'}_1) \iprod (\m{\salg}_2, \m{\prob'}_2) \extTo (\m{\salg'}, \m{\prob'})$.
  Let $\m{\permap_1} = \m{\permap_2} = \fun x. \m{\permap'}(x)/2$.

  Next we show that $\distAs{\aexpr_1}{\mu_1} (\m{\salg}_1, \m{\prob'}_1, \m{\permap_1}) $ and $\distAs{\aexpr_2}{\mu_2} (\m{\salg}_2, \m{\prob'}_2, \m{\permap_2})$.
  By definition,
  $\distAs{\aexpr_1}{\mu_1} (\m{\salg}_1, \m{\prob'}_1, \m{\permap_1})$
  is equivalent to
\begin{align*}
      \E \m{\salg''},\m{\prob''}.
      (\Own{\m{\salg''},\m{\prob''}})(\m{\salg}_1, \m{\prob'}_1, \m{\permap_1}) *
    \almostM{\aexpr_1}{(\m{\salg''}(i), \m{\prob''}(i))}
    \land
    \mu_1 = \m{\prob''}(i) \circ \inv{\aexpr_1}
    ,
  \end{align*}
which is equivalent to
\begin{multline*}
      \E \m{\salg''},\m{\prob''}.
      (\m{\salg''},\m{\prob''}) \raLeq (\m{\salg}_1, \m{\prob'}_1) *
      \bigl(\forall a \in A. \exists S_a, T_a \in \ \m{\salg''}(i).\\
      S_a \subseteq \inv{\aexpr_1}(a) \subseteq T_a  \land
      \m{\prob''}(i)(S_a) =  \m{\prob''}(i)(S_a) \land
      \mu_1(a) = \m{\prob''}(i)(S_a) = \m{\prob''}(i)(T_{a})
      \bigr)
  \end{multline*}
We can pick the existential witness to be
  $\m{\salg}_1, \m{\prob'}_1$.
  For any $a \in A$,
  $ \inv{\aexpr_1}(a) = \Union_{b \in A}\inv{(\aexpr_1, \aexpr_2)} (a, b)$.
  Because we have $L_{a, b} \subseteq \inv{(\aexpr_1, \aexpr_2)} (a, b) \subseteq U_{a,b}$,
  then
  \[
   \Union_{b \in A} L_{a, b} \subseteq
   \inv{\aexpr_1}(a) = \Union_{b \in A}\inv{(\aexpr_1, \aexpr_2)} (a, b)
   \subseteq \Union_{b \in A} U_{a, b} .
 \]
  By definition, for each $a$,
  $\Union_{b \in A} L_{a, b} \in \m{\salg}_1(i)$ and
  $\Union_{b \in A} U_{a, b} \in \m{\salg}_1(i)$,
  and we also have
  \begin{align*}
     \m{\prob'}_1(i) (\Union_{b \in A} L_{a, b})
     &= \sum_{b \in A} \m{\prob'}_1(i) (L_{a,b})\\
     &= \sum_{b \in A} \m{\prob'}_1(i) (U_{a,b})\\
     &= \m{\prob'}_1(i) \bigl(\Union_{b \in A} U_{a, b}\bigr)\\
     &= \mu_1(a)
  \end{align*}
  Thus, $S_a = \Union_{b \in A} L_{a, b}$ and
  $T_a = \Union_{b \in A} U_{a, b}$ witnesses the conditions needed
  for
  $\distAs{\aexpr_1}{\mu_1} (\m{\salg}_1, \m{\prob'}_1, \m{\permap_1}) $.
  And similarly, we have $\distAs{\aexpr_2}{\mu_2} (\m{\salg}_2, \m{\prob'}_2, \m{\permap_2}) $.
\end{proof}
 
\subsubsection{Soundness of Conditioning Rules}
\begin{lemma}
\label{proof:c-true}
  \Cref{rule:c-true} is sound.
\end{lemma}

\begin{proof}
  Let $\m{{\raUnit}} = (\m{\salg}_{\m{{\raUnit}}}, \m{\prob}_{\m{{\raUnit}}}, \m{\permap}_{\m{{\raUnit}}}) \in \Model_I$
  be the unit of $\Model_I$ and
  $\m{\krnl} = \fun v. \m{\prob}_{\m{{\raUnit}}}$.
  Then,
  \begin{eqexplain}
    \True
\whichproves*
    \Own{\m{\salg}_{\m{{\raUnit}}}, \m{\prob}_{\m{{\raUnit}}}}
\whichproves
    \Own{\m{\salg}_{\m{{\raUnit}}}, \m{\prob}_{\m{{\raUnit}}}} *
      \pure{
        \forall i\in I\st
          \m{\prob}_{\m{{\raUnit}}}(i) = \bind(\prob, \m{\krnl}(i))
      }
\whichproves
      \Own{\m{\salg}_{\m{{\raUnit}}}, \m{\prob}_{\m{{\raUnit}}}}
      * \pure{
        \forall i\in I\st
          \m{\prob}_{\m{{\raUnit}}}(i) = \bind(\prob, \m{\krnl}(i))
      }
      * \True
\whichproves
    \E \m{\salg}_{\m{{\raUnit}}}, \m{\prob}_{\m{{\raUnit}}}, \m{\krnl}.
      \Own{\m{\salg}_{\m{{\raUnit}}}, \m{\prob}_{\m{{\raUnit}}}}
      \begin{array}[t]{@{}>{{}}l}
      * \pure{
        \forall i\in I\st
          \m{\prob}_{\m{{\raUnit}}}(i) = \bind(\prob, \m{\krnl}(i))
      } \\
      * (
        \forall v \in \psupp(\prob).
         \Own{\m{\salg}_{\m{{\raUnit}}}, \m{\krnl}(I)(v), \m{\permap}_{\m{{\raUnit}}}} \wand \True
      )
      \end{array}
\whichproves
    \CMod{\prob} \wtv. \True
  \qedhere
  \end{eqexplain}
\end{proof} \begin{lemma}
\label{proof:c-false}
  \Cref{rule:c-false} is sound.
\end{lemma}

\begin{proof}
  Assume $a \in \Model_I$ is such that
  $\raValid(a)$ and that it satisfies
  $ \CC{\prob} v.\False $.
  By definition, this means that,
  for some $ \m{\salg}_0, \m{\prob}_0,\m{\permap}_0,$ and $ \m{\krnl}_0 $:
  \begin{gather}
    (\m{\salg}_0, \m{\prob}_0,\m{\permap}_0) \raLeq a
    \label{c-false:prob0}
    \\
    \forall i\in I\st
      \m{\prob}_0(i) = \bind(\prob, \m{\krnl}_0(i))
    \label{c-false:prob0-bind}
    \\
    \forall v \in \psupp(\prob) \st
      \False(\m{\salg}_0, \m{\krnl}_0(I)(v), \m{\permap}_0)
    \label{c-false:False-krnl0}
  \end{gather}
  Let $v_0 \in \psupp(\prob)$---we know one exists because $\prob$
  is a (discrete) probability distribution.
  Then by \eqref{c-false:False-krnl0} on~$v_0$
  we get $\False(\m{\salg}_0, \m{\krnl}_0(I)(v_0), \m{\permap}_0)$ holds.
  Since $\False(\wtv)$ is by definition false,
  we get $\False(a)$ holds \emph{ex falso}.
\end{proof} \begin{lemma}
\label{proof:c-cons}
  \Cref{rule:c-cons} is sound.
\end{lemma}

\begin{proof}
  Assume $a \in \Model_I$ is such that
  $\raValid(a)$ and that it satisfies
  $ \CC{\prob} v.K(v) $.
  By definition, this means that,
  for some $ \m{\salg}_0, \m{\prob}_0,\m{\permap}_0,$ and $ \m{\krnl}_0 $:
  \begin{gather}
    (\m{\salg}_0, \m{\prob}_0,\m{\permap}_0) \raLeq a
    \label{c-cons:prob0}
    \\
    \forall i\in I\st
      \m{\prob}_0(i) = \bind(\prob, \m{\krnl}_0(i))
    \label{c-cons:prob0-bind}
    \\
    \forall v \in \psupp(\prob) \st
      K(v)(\m{\salg}_0, \m{\krnl}_0(I)(v), \m{\permap}_0)
    \label{c-cons:K-krnl0}
  \end{gather}
  Then by the premise $\forall v\st K(v) \proves K'(v)$
  and \eqref{c-cons:K-krnl0} we obtain
  \begin{equation}
    \forall v \in \psupp(\prob) \st
      K'(v)(\m{\salg}_0, \m{\krnl}_0(I)(v), \m{\permap}_0)
    \label{c-cons:K'-krnl0}
  \end{equation}
  By
  \eqref{c-cons:prob0}, \eqref{c-cons:prob0-bind}, and \eqref{c-cons:K'-krnl0}
  we get $ \CC{\prob} v.K'(v) $ as desired.
\end{proof} \begin{lemma}
\label{proof:c-frame}
  \Cref{rule:c-frame} is sound.
\end{lemma}

\begin{proof}
  Assume $a \in \Model_I$ is such that
  $\raValid(a)$ and that it satisfies
  $ P * \CC{\prob} v.K(v) $.
  By definition, this means that
  there exist some
  $(\m{\salg}_1, \m{\prob}_1, \m{\permap}_1)$,
  $(\m{\salg}_2, \m{\prob}_2, \m{\permap}_2)$,
  and $\m{\krnl}$
  such that
  \begin{gather}
    (\m{\salg}_1, \m{\prob}_1, \m{\permap}_1)
    \raOp
    (\m{\salg}_2, \m{\prob}_2, \m{\permap}_2)
    \raLeq a
    \\
    P(\m{\salg}_1, \m{\prob}_1, \m{\permap}_1)
    \\
\forall i \in I.
      \m{\prob}_2(i) = \bind(\prob, \m{\krnl}(i))
    \\
    \forall v\in \psupp(\prob) \st
      K(v)(\m{\salg}_2, \m{\krnl}(I)(v), \m{\permap}_2)
  \end{gather}
Now let:
  \begin{align*}
    (\m{\salg}',\m{\prob}',\m{\permap}')
    &=
    (\m{\salg}_1(i), \m{\prob}_1(i)) \iprod (\m{\salg}_2(i), \m{\prob}_2(i))
  &
    \m{\krnl}'(i) &= \fun v. \m{\prob}_1(i) \iprod \m{\krnl}(i)(v)
  \end{align*}
  By~\cref{lemma:fibre-prod-exists}, for each~$i\in I$:
  \begin{align*}
    (\m{\salg}',\m{\prob}',\m{\permap}')
    &= (\m{\salg}_1(i), \m{\prob}_1(i)) \iprod (\m{\salg}_2(i), \m{\prob}_2(i))
    \\
    &= (\m{\salg}_1(i) \punion \m{\salg}_2(i),
       \bind(\prob, \fun v. \m{\prob}_1(i) \iprod \m{\krnl}(i)(v)))
   \tag*{(By \cref{lemma:fibre-prod-exists})}
    \\
    &= (\m{\salg}_1(i) \punion \m{\salg}_2(i),
       \bind(\prob, \m{\krnl}'(i)))
  \end{align*}
  Notice that $ \m{\krnl}'(I)(v) = \m{\prob}_1 \iprod \m{\krnl}(I)(v) $.
  Thus we obtain:
  \begin{gather}
    (\m{\salg}',\m{\prob}',\m{\permap}')
    \raLeq a
    \\
\forall i \in I.
      \m{\prob}'(i) = \bind(\prob, \m{\krnl}'(i))
    \\
    \intertext{and for all $v \in \psupp(\prob)$,}
(\m{\salg}_1, \m{\prob}_1, \m{\permap}_1)
      \iprod
    (\m{\salg}_2, \m{\krnl}(I)(v), \m{\permap}_2)
    =
    (\m{\salg}', \m{\prob}_1 \iprod \m{\krnl}(I)(v), \m{\permap}')
    \raLeq
    (\m{\salg}', \m{\krnl}'(I)(v),\m{\permap}')
    \\
    P(\m{\salg}_1, \m{\prob}_1, \m{\permap}_1)
    \\
    K(v)(\m{\salg}_2, \m{\krnl}(I)(v), \m{\permap}_2)
  \end{gather}
  which gives us that $a$ satisfies
  $ \CC{\prob} v.(P * K(v)) $ as desired.
\end{proof} \begin{lemma}
\label{proof:c-unit-l}
  \Cref{rule:c-unit-l} is sound.
\end{lemma}

\begin{proof}
  Straightforward.
\end{proof} \begin{lemma}
\label{proof:c-unit-r}
  \Cref{rule:c-unit-r} is sound.
\end{lemma}

\begin{proof}
  We prove the two directions separately.

  \begin{casesplit}
  \case*[Forward direction $\distAs{\aexpr\at{i}}{\prob} \proves \CC\prob v.\sure{\aexpr\at{i} = v}$]
    By unfolding the assumption $\distAs{\aexpr\at{i}}{\prob}$ we get
    that there exist $\m{\salg},\m{\prob}$ such that:
    \[
      \Own{\m{\salg},\m{\prob}}
      *
      \pure{
        \almostM{\aexpr}{(\m{\salg}(i),\m{\prob}(i))}
      }
      *
      \pure{
        \prob = \m{\prob}(i) \circ \inv{\aexpr}
      }
    \]
    holds.
    Let
    \begin{align*}
      \m{\krnl} &\is
      \fun j.
        \begin{cases}
          \fun v. \m{\prob}(j) \CASE j \ne i
          \\
          \fun v. \gamma_v     \CASE j=i
        \end{cases}
      &
      \gamma_v &\is
        \fun \event \of \m{\salg}(i).
          \frac{\m{\prob}(i)(\event \inters \inv{(\aexpr=v)})}
               {\m{\prob}(i)(\inv{(\aexpr=v)})}
    \end{align*}
    That is, $\m{\krnl}(j)$ maps every~$v$ to $\m{\prob}(j)$ when $i\ne j$,
    while when $i=j$ it maps~$v$ to the distribution $\m{\prob}(i)$ conditioned on $\aexpr=v$.
    Note that $\m{\krnl}$ is well defined because
    \begin{enumerate*}
      \item
        although the events
        $\event \inters \inv{(\aexpr=v)}$ and
        $\inv{(\aexpr=v)}$
        might not belong to $\m{\salg}(i)$,
        their probability is uniquely determined
        by almost measurability of $\aexpr$;
      \item
        we are only interested in the cases where~$v \in \psupp(\prob)$,
        which implies that the denominator is not zero:
        $\m{\prob}(i)(\inv{(\aexpr=v)}) = \prob(v) > 0$.
    \end{enumerate*}
    By construction we obtain that
    \begin{gather}
      \forall j \in I\st
        \m{\prob}(j) = \bind(\prob, \m{\krnl}(j))
      \label{c-unit-r:bind}
      \\
      \forall v\in \psupp(\prob)\st
        \m{\krnl}(i)(v)(\inv{(E=v)}) = 1
      \label{c-unit-r:prob1}
    \end{gather}
    From \eqref{c-unit-r:prob1} we get that
    $\sure{\aexpr\at{i} = v}$ holds on
    $(\m{\salg}(i), \m{\krnl}(i)(v), \m{\permap}(i))$,
    from which it follows that:
    \[
      \Own{\m{\salg}, \m{\krnl}(I)(v), \m{\permap}}
      \wand \sure{\aexpr\at{i} = v}
    \]
    Therefore we obtain
    \begin{align*}
      & \E \m{\salg},\m{\prob}, \m{\krnl}, \m{\permap}.
          \Own{\m{\salg},\m{\prob}, \m{\permap}} *
          \pure{\forall j \in I. \m{\prob}(j) = \bind(\prob, \m{\krnl}(j))} \\
      & \qquad \qquad  *
          (
            \forall v \in A_{\prob}.
              \Own{\m{\salg}, \m{\krnl}(I)(v), \m{\permap}}
              \wand \sure{\aexpr\at{i} = v}
          )
    \end{align*}
    which gives us $ \CC\prob v.\sure{\aexpr\at{i} = v} $
    by \cref{prop:cond-as-wand}.

\case*[Backward direction $\CC\prob v.\sure{\aexpr\at{i} = v} \proves \distAs{\aexpr\at{i}}{\prob}$]
    First note that
    \begin{align*}
      \sure{\aexpr\at{i} = v} &(\m{\salg}, \m{\krnl}(v), \m{\permap})
      \\
      &\iff
        \bigl(\distAs{((\aexpr = v) \in \true)\at{i}}{\delta_{\True}}\bigr)
          (\m{\salg}, \m{\krnl}(I)(v), \m{\permap})
      \\
      &\iff
        \almostM{((\aexpr = v) \in \true)}{(\m{\salg}(i), \m{\krnl}(i)(v))}
        \land
        \delta_{\True} =
          \m{\krnl}(i)(v) \circ \inv{((\aexpr = v) \in \true)}
      \\
      &\iff
        \almostM{((\aexpr = v) \in \true)}{(\m{\salg}(i), \m{\krnl}(i)(v))}
        \land
        \delta_{v} = \m{\krnl}(i)(v) \circ \inv{\aexpr}
    \end{align*}

for some $\m{\krnl}$.
    This implies
    $\pure{\almostM{E}{\m{\salg}(i), \m{\krnl}(i)(v)}}$.
    Then, for any value $v \in \psupp(\prob)$,
    \begin{align*}
      \m{\prob}(i) \circ \inv{\aexpr}(v)
      &=(\bind(\prob, \m{\krnl}(i) ) \circ \inv{\aexpr})(v)\\
      &=  \bind(\prob, \m{\krnl}(i) )  (\inv{\aexpr}(v)) \\
      &=  \sum_{\mathclap{v'\in\psupp(\prob)}}  \prob(v') \cdot \m{\krnl}(i)(v') (\inv{\aexpr}(v)) \\
      &=  \sum_{\mathclap{v'\in\psupp(\prob)}}  \prob(v') \cdot (\m{\krnl}(i)(v') \circ \inv{\aexpr}) (v) \\
      &=  \sum_{\mathclap{v'\in\psupp(\prob)}}  \prob(v') \cdot \dirac{v'} (v) \\
&=  \prob(v)
    \end{align*}
    This implies the pure facts that
    $ \almostM{\aexpr}{(\m{\salg}(i),\m{\prob}(i))}$ and
    $\prob = \m{\prob}(i) \circ \inv{\aexpr}$.
    Therefore:
    \begin{eqexplain}
      \CC\prob v. \sure{\aexpr\at{i} = v} \notag
      \whichproves*
      \E \m{\salg},\m{\prob}, \m{\krnl}, \m{\permap}.
          \Own{\m{\salg},\m{\prob}, \m{\permap}} *
          \pure{\forall j \in I. \m{\prob}(j) = \bind(\prob, \m{\krnl}(j))} \notag \\
          & \qquad \qquad  *
          (\forall v \in A_{\prob}.
          \Own{\m{\salg}, \m{\krnl}(I)(v), \m{\permap}}
          \wand
          \sure{\aexpr\at{i} = v}
          )
      \whichproves
      \E \m{\salg},\m{\prob}.
        \Own{\m{\salg},\m{\prob}} *
          \pure{
            \almostM{\aexpr}{(\m{\salg}(i),\m{\prob}(i))}}
            \ast
          \pure{
            \prob = \m{\prob}(i) \circ \inv{\aexpr}
        }
       \whichproves
       \distAs{\aexpr\at{i}}{\prob}
       \qedhere
    \end{eqexplain}
  \end{casesplit}
\end{proof}
 \begin{lemma}
\label{proof:c-assoc}
  \Cref{rule:c-assoc} is sound.
\end{lemma}

\begin{proof}
  Define $\krnl' = \fun v . \bind(\krnl(v), \fun w . \return(v, w))$.
  We start by rewriting the assumption $\CC{\prob} v.\CC{\krnl(v)} w.K(v,w)$ so that $k'$ is used and~$K$ depends only on the binding of the innermost modality:
  \begin{eqexplain}
    \CC{\prob} v.\CC{\krnl(v)} w.K(v,w)
    \whichproves*
    \CC{\prob} v.\CC{\krnl'(v)} (v',w).K(v,w)
    \byrules{c-transf,c-cons}
    \whichproves
    \CC{\prob} v.\CC{\krnl'(v)} (v',w).K(v',w)
    \byrules{c-pure,c-cons}
  \end{eqexplain}
  \Cref{rule:c-transf} is applied to the innermost modality
  by using the bijection $f_v(w) = (v,w)$.
  Then, since $(v',w) \in \psupp(k'(v)) \implies v=v'$,
  we can replace~$v'$ for~$v$ in~$K$.

  Our goal is now to prove:
  \[
    \CC{\prob} v.\CC{\krnl'(v)} (v',w).K(v',w)
    \proves
    \CC{\bind(\prob,\krnl')} (v',w).K(v',w)
  \]

  Let $a\in\Model_I$ be such that $ \raValid(a) $ and that it satisfies
  $ \CC{\prob} v.\CC{\krnl'(v)} (v',w).K(v',w). $
  From this assumption we know that,
  for some $ \m{\salg}_0, \m{\prob}_0,\m{\permap}_0,$ and $ \m{\krnl}_0 $:
  \begin{gather}
    (\m{\salg}_0, \m{\prob}_0,\m{\permap}_0) \raLeq a
    \label{c-assoc:prob0}
    \\
    \forall i\in I\st
      \m{\prob}_0(i) = \bind(\prob, \m{\krnl}_0(i))
    \label{c-assoc:prob0-bind}
  \end{gather}
  such that $\forall v \in \psupp(\prob)$,
  there are some
  $ \m{\salg}_1^v, \m{\prob}_1^v,\m{\permap}_1^v,$ and $ \m{\krnl}_1^v $
  satisfying:
  \begin{gather}
    (\m{\salg}_1^v, \m{\prob}_1^v,\m{\permap}_1^v)
    \raLeq
    (\m{\salg}_0, \m{\krnl}_0(I)(v),\m{\permap}_0)
    \label{c-assoc:prob1}
    \\
    \forall i\in I\st
      \m{\prob}_1^v(i) = \bind(\krnl'(v), \m{\krnl}_1^v(i))
    \label{c-assoc:prob1-bind}
    \\
    \forall (v',w) \in \psupp(\krnl'(v)) \st
      K(v',w)(\m{\salg}_1^v, \m{\krnl}_1^v(I)(v',w), \m{\permap}_1^v)
    \label{c-assoc:K-krnl1}
  \end{gather}

  Our goal is to prove
  $ \CC{\bind(\prob,\krnl')} (v',w).K(v',w) $ holds on $a$.
  To this end, we want to show that
  there exists $\m{\krnl}_2'$ such that:
\begin{gather}
    \forall i\in I\st
      \m{\prob}_0(i) = \bind(\bind(\prob,\krnl'), \m{\krnl}_2'(i))
    \label{c-assoc:goal1}
    \\
    \forall (v',w) \in \psupp(\bind(\prob,\krnl')) \st
      K(v', w)(\m{\salg}_0 , \m{\krnl}_2'(I) (v'), \m{\permap}_0)
    \label{c-assoc:goal2}
  \end{gather}

  Now let
  \[
    \m{\krnl}_2(i) = \fun (v', w). \m{\krnl}_1^{v'}(i)(v', w).
  \]
  which by construction and \cref{c-assoc:prob1-bind} gives us
  \[
    \m{\prob}_1^v(i)
    = \bind(\krnl'(v), \m{\krnl}_1^v(i))
    = \bind(\krnl'(v), \m{\krnl}_2(i))
  \]
  Therefore, by \cref{c-assoc:prob1}, we can apply \cref{lemma:bind-extend}
  and obtain that there exists a $\m{\krnl}_2'$ such that
  \begin{gather}
    \m{\krnl}_0(i)(v) = \bind(\krnl'(v), \m{\krnl}_2'(i))
    \label{c-assoc:k0-bind2}
    \\
    \bigl(\m{\salg}_0, \m{\krnl}_2'(i)(v',w)\bigr)
    \extOf
    \bigl(\m{\salg}_1^{v'}, \m{\krnl}_2(i)(v',w)\bigr)
    =
    \bigl(\m{\salg}_1^{v'},\m{\krnl}_1^{v'}(i)(v',w)\bigr)
    \label{c-assoc:kgeq}
  \end{gather}
  By \cref{c-assoc:prob0-bind,c-assoc:k0-bind2}
  we have:
  \begin{align*}
    \m{\prob}_0(i)
    &= \bind(\prob, \m{\krnl}_0(i)) \\
    &= \bind(\prob, \fun v.\bind(\krnl'(v), \m{\krnl}_2'(i)))
      &\text{By \eqref{prop:bind-assoc}}\\
    &= \bind(\bind(\prob, \krnl'), \m{\krnl}_2'(i))
  \end{align*}
  which proves \cref{c-assoc:goal1}.

  Finally, to prove \cref{c-assoc:goal2}, we can observe that
  $(v',w) \in \psupp(\bind(\prob,\krnl'))$ implies $v'\in \psupp(\prob)$;
  therefore, by \eqref{c-assoc:K-krnl1}, upward closure of $K(v',w)$, and
  \eqref{c-assoc:kgeq} and \eqref{c-assoc:prob1},
  we can conclude~$K(v',w)$ holds on
  $(\m{\salg}_0 , \m{\krnl}_2'(I) (v'), \m{\permap}_0)$,
  as desired.
\end{proof} \begin{lemma}
\label{proof:c-unassoc}
  \Cref{rule:c-unassoc} is sound.
\end{lemma}

\begin{proof}
  Assume $a \in \Model_I$ is such that
  $\raValid(a)$ and that it satisfies
  $ \CC{\bind(\prob,\krnl)} w.K(w) $.
  By definition, this means that,
  for some $ \m{\salg}_0, \m{\prob}_0,\m{\permap}_0,$ and $ \m{\krnl}_0 $:
  \begin{gather}
    (\m{\salg}_0, \m{\prob}_0,\m{\permap}_0) \raLeq a
    \label{c-unassoc:prob0}
    \\
    \forall i\in I\st
      \m{\prob}_0(i) = \bind(\bind(\prob, \krnl), \m{\krnl}_0(i))
    \label{c-unassoc:prob0-bind}
    \\
    \forall w \in \psupp(\bind(\prob, \krnl)) \st
      K(w)(\m{\salg}_0, \m{\krnl}_0(I)(w), \m{\permap}_0)
    \label{c-unassoc:K-krnl0}
  \end{gather}
  Our goal is to show that~$a$ satisfies
  $\CC\prob v. \CC{\krnl(v)} w.K(w)$,
  for which it would suffice to show that there is a $\m{\krnl}_1$
  such that:
  \begin{gather}
    \forall i\in I\st
      \m{\prob}_0(i) = \bind(\prob, \m{\krnl}_1(i))
    \label{c-unassoc:prob0-bind1}
    \\
    \intertext{
      and for all $v \in \psupp(\prob)$
      there is a $\m{\krnl}_2^v$ with
    }
    \forall i\in I\st
      \m{\krnl}_1(i)(v) = \bind(\krnl(v), \m{\krnl}_2^v(i))
    \label{c-unassoc:krnl1-bind}
    \\
    \forall w \in \psupp(\krnl(v)) \st
      K(w)(\m{\salg}_0, \m{\krnl}_2^v(I)(w), \m{\permap}_0)
    \label{c-unassoc:K-krnl1}
  \end{gather}

  To prove this we let
  \begin{align*}
    \m{\krnl}_1(i) &= \fun v.\bind(\krnl(v), \m{\krnl}_0(i))
    &
    \m{\krnl}_2^v(i) &= \m{\krnl}_0(i)
  \end{align*}

  By \eqref{prop:bind-assoc} we have
  \[
    \m{\prob}_0(i)
    = \bind(\bind(\prob, \krnl), \m{\krnl}_0(i))
    = \bind(\prob, \fun v.\bind(\krnl(v), \m{\krnl}_0(i)))
    = \bind(\prob, \m{\krnl}_1(i))
  \]
  which proves \eqref{c-unassoc:prob0-bind1}.
  By construction,
  \[
    \m{\krnl}_1(i)(v)
    = \bind(\krnl(v), \m{\krnl}_0(i))
    = \bind(\krnl(v), \m{\krnl}_2^v(i))
  \]
  proving \eqref{c-unassoc:krnl1-bind}.
  Finally,
  $v \in \psupp(\prob)$ and $w \in \psupp(\krnl(v))$
  imply $w \in \psupp(\bind(\prob, \krnl))$,
  so by \eqref{c-unassoc:K-krnl0} we proved
  \eqref{c-unassoc:K-krnl1}, concluding the proof.
\end{proof}
 \begin{lemma}
\label{proof:c-and}
  \Cref{rule:c-and} is sound.
\end{lemma}

\begin{proof}
  Let~$I_1 = \idx(K_1)$ and $I_2 = I \setminus I_1$;
  by $\idx(K_1) \inters \idx(K_2) = \emptyset$
  we have $I_2 \sups \idx(K_2)$.
  Assume $a\in \Model_I$ is such that~$\raValid(a)$ holds and
  that it satisfies $
    \CC{\prob} v. K_1(v)
      \land
    \CC{\prob} v. K_2(v)
  $.
  This means that
  for each $j \in \set{1,2}$,
  for some $ \m{\salg}_j, \m{\prob}_j,\m{\permap}_j,$ and $ \m{\krnl}_j $:
  \begin{gather}
    (\m{\salg}_j, \m{\prob}_j,\m{\permap}_j) \raLeq a
    \label{c-and:prob}
    \\
    \forall i\in I\st
      \m{\prob}_j(i) = \bind(\prob, \m{\krnl}_j(i))
    \label{c-and:prob-bind}
    \\
    \forall v \in \psupp(\prob) \st
      K_j(v)(\m{\salg}_j, \m{\krnl}_j(I)(v), \m{\permap}_j)
    \label{c-and:K-krnl}
  \end{gather}
  Now let
  \begin{align*}
    \hat{\m{\salg}} &=
    \begin{cases}
      \m{\salg}_1(i) \CASE i \in I_1 \\
      \m{\salg}_2(i) \CASE i \in I_2
    \end{cases}
    &
    \hat{\m{\prob}} &=
    \begin{cases}
      \m{\prob}_1(i) \CASE i \in I_1 \\
      \m{\prob}_2(i) \CASE i \in I_2
    \end{cases}
    &
    \hat{\m{\permap}} &=
    \begin{cases}
      \m{\permap}_1(i) \CASE i \in I_1 \\
      \m{\permap}_2(i) \CASE i \in I_2
    \end{cases}
    &
    \hat{\m{\krnl}}(i) &=
    \begin{cases}
      \m{\krnl}_1(i) \CASE i \in I_1 \\
      \m{\krnl}_2(i) \CASE i \in I_2
    \end{cases}
  \end{align*}
  By construction, we have:
  \begin{gather*}
    (\hat{\m{\salg}}, \hat{\m{\prob}},\hat{\m{\permap}}) \raLeq a
    \\
    \forall i\in I\st
      \hat{\m{\prob}}(i) = \bind(\prob, \hat{\m{\krnl}}(i))
  \end{gather*}
  Moreover, for any~$v \in \psupp(\prob)$ and any $j \in \set{1,2}$,
  since $I_j \sups \idx(K_j)$, condition \eqref{c-and:K-krnl} implies
  \[
    K_j(v)(\hat{\m{\salg}}, \hat{\m{\krnl}}(I)(v), \hat{\m{\permap}})
  \]
  This means
  $(\hat{\m{\salg}}, \hat{\m{\krnl}}(I)(v), \hat{\m{\permap}})$
  satisfies
  $(K_1(v) \land K_2(v))$,
  and thus~$a$ satisfies
  $ \CC\prob v. (K_1(v) \land K_2(v)) $,
  as desired.
\end{proof} \begin{lemma}
\label{proof:c-skolem}
  \Cref{rule:c-skolem} is sound.
\end{lemma}

\begin{proof}
  For any resource $r = (\m{\sigmaF}, \m{\mu}, \m{\permap})$,
  \begin{align*}
    &
    \left(\CC\prob v. \E x \of \Var. Q(v, x) \right) (\m{\sigmaF}, \m{\mu}, \m{\permap}) \\
    {}\iff {} &
    \exists \m{\krnl} \st
             \forall i \in I.
             \m{\mu}(i) = \bind(\prob, \m{\krnl}(i))
    {}\land{}
    \forall v\in  \psupp(\prob) \st
     (\E x \of X. Q(v, x))(\m{\sigmaF}, \m{\krnl}(I)(v), \m{\permap})
  \end{align*}

   For all $v\in \psupp(\prob)$,
   $\E x \of X. Q(v, x)$ holds on $(\m{\sigmaF}, \m{\krnl}(I)(v), \m{\permap})$.
   Thus,
   $Q(v, x_v)(\m{\sigmaF}, \m{\krnl}(I)(v), \m{\permap})$
   holds for some $x_v$.
   Then define $f: A \to \Var$ by letting $f(v) = x_v$ for $v \in \psupp(\mu)$.
   Then,
   \begin{align*}
   \exists \m{\krnl} \st
             \forall i \in I.
             \m{\mu}(i) = \bind(\prob, \m{\krnl}(i))
               {}\land{}
              \forall v\in  \psupp(\prob) \st
      Q(v, f(v))(\m{\sigmaF}, \m{\krnl}(I)(v), \m{\permap})
  \end{align*}
  And therefore $\m{\sigmaF}, \m{\mu}, \m{\permap}$
  satisfies ${\E f \of A \to \Var. \CC\prob v. Q(v, x)}$.
\end{proof} \begin{lemma}
\label{proof:c-transf}
  \Cref{rule:c-transf} is sound.
\end{lemma}

\begin{proof}
  For any resource $a = (\m{\sigmaF}, \m{\mu}, \m{\permap})$,
  if $ \model{\CC\prob v.K(v)}{(\m{\sigmaF}, \m{\mu}, \m{\permap})}$,
  then
  \begin{align*}
    \begin{array}[t]{@{}r@{\,}l@{}}
      \E \m{\krnl}.
      & (\m{\sigmaF}, \m{\mu}, \m{\permap}) \raLeq a
      \land
        \forall i\in I\st
        \m{\mu}(i) = \bind(\prob, \m{\krnl}(i))
      \\ & \land \;
        \forall v \in \psupp(\prob).
          \model{K(v)}{ (\m{\sigmaF}, \m{\krnl}(I)(v), \m{\permap}) }
    \end{array}
  \end{align*}

  $\m{\mu} = \bind(\prob, \m{\krnl})$  says that for any
  $E \in \m{\sigmaF}$,
  \begin{align}
    \m{\mu}(E)
    &= \sum_{\mathclap{v \in \psupp(\mu)}} \prob(v) \cdot \m{\krnl}(I)(v)(E) \notag \\
    &= \sum_{\mathclap{v \mid f(v) \in \psupp(\mu)}} \prob(f(v)) \cdot \m{\krnl}(I)(f(v))(E) \tag{Because $f$ is bijective}\\
    &= \sum_{\mathclap{v \in \psupp(\mu')}} \mu'(v) \cdot \m{\krnl}(I)(f(v))(E) \tag{Because $\mu'(v) = \mu(f(v))$} \\
    &= \bind(\prob', \fun v . \m{\krnl}(I)(f(v)))(E) \notag
  \end{align}
  Thus, $\m{\mu} = \bind(\prob', \fun v . \m{\krnl}(I)(f(v)))$.
  Furthermore, $\model{K(f(v))}{ (\m{\sigmaF}, \m{\krnl}(I)(f(v)), \m{\permap}) }$.

  Thus, if we denote $\fun v . \m{\krnl}(I)(f(v))$ as $\m{\krnl'}$, it satisfies
  \begin{align*}
    \begin{array}[t]{@{}r@{\,}l@{}}
      & (\m{\sigmaF}, \m{\mu}, \m{\permap}) \raLeq a
      \land
        \forall i\in I\st
        \m{\mu}(i) = \bind(\prob', \m{\krnl'}(i))
      \\ & \land \;
        \forall v \in \psupp(\prob).
          \model{K(v)}{ (\m{\sigmaF}, \m{\krnl'}(I)(v), \m{\permap})}
    \end{array}
  \end{align*}
  Thus, $ \model{\CC\prob' v.K(f(v))}{(\m{\sigmaF}, \m{\mu}, \m{\permap})}$.
\end{proof}
 \begin{lemma}
\label{proof:sure-str-convex}
  \Cref{rule:sure-str-convex} is sound.
\end{lemma}

\begin{proof}
  Assume $a \in \Model_I$ is a valid resource that
  satisfies $\CMod{\prob} v. (K(v) \ast  \sure{\aexpr\at{i}})$.
  Then, by definition, we know that,
  for some $ (\m{\salg}_0, \m{\prob}_0, \m{\permap}_0) $ and $\m{\krnl}_0$:
  \begin{gather}
    (\m{\salg}_0, \m{\prob}_0, \m{\permap}_0) \raLeq a
    \label{sure-str-convex:a}
    \\
    \forall i\in I\st
       \m{\prob}_0(i) = \bind(\prob, \m{\krnl}_0(i))
    \label{sure-str-convex:bind-a}
    \\
    \intertext{and, for all $v \in \psupp(\prob)$, there are
      $(\m{\salg}^v_1, \m{\prob}^v_1, \m{\permap}^v_1)$,
      $(\m{\salg}^v_2, \m{\prob}^v_2, \m{\permap}^v_2)$
      such that}
(\m{\salg}^v_1, \m{\prob}^v_1, \m{\permap}^v_1)
    \raOp
    (\m{\salg}^v_2, \m{\prob}^v_2, \m{\permap}^v_2)
    \raLeq
      (\m{\salg}_0, \m{\krnl}_0(I)(v), \m{\permap}_0)
    \label{sure-str-convex:star}
    \\
    K(v)(\m{\salg}^v_1, \m{\prob}^v_1, \m{\permap}^v_1)
    \label{sure-str-convex:K}
    \\
    \sure{\aexpr\at{i}}(\m{\salg}^v_2, \m{\prob}^v_2, \m{\permap}^v_2)
    \label{sure-str-convex:sure}
  \end{gather}
  From \eqref{sure-str-convex:sure} we know that for all~$v \in \psupp(\prob)$
  there are $L^v_1,L^v_0,U^v_1,U^v_0 \in \m{\salg}^v_2(i)$ such that:
  \begin{align*}
    L^v_0 &\subs \inv{\aexpr}(\False) \subs U^v_0
    &
    \m{\prob}^v_2(L^v_0) &= \m{\prob}^v_2(U^v_0) = 0
    \\
    L^v_1 &\subs \inv{\aexpr}(\True) \subs U^v_1
    &
    \m{\prob}^v_2(L^v_1) &= \m{\prob}^v_2(U^v_1) = 1
  \end{align*}
  Without loss of generality, all $L^v_0, L^v_1, U^v_0, U^v_1$ can be assumed
  to be only non-trivial on $\pvar(\aexpr)$.
  Consequently, we can also assume that
  $\m{\permap}^v_2(\ip{x}{j})<1$ for every $\ip{x}{j}$,
and in addition
  $\m{\permap}^v_2(\ip{x}{j})>0$
  if and only if $\p{x}\in \pvar{\aexpr}$ and $j=i$.
  From these components we can construct a new resource:
  \begin{align*}
    \m{\salg}_3(j) &\is
      \begin{cases}
        \sigcl*{\set*{
          \Inters_{v \in \psupp(\prob)} L^v_1,
          \Union_{v \in \psupp(\prob)} U^v_1
        }}
        \CASE j  =  i \\
        \set{\Store, \emptyset}
        \CASE j \ne i
      \end{cases}
    \\
    \m{\prob}_3 &\is \restr{\m{\prob}_0}{\m{\salg}_3}
    \\
    \m{\permap}_3 &\is
      \fun \ip{x}{j}.
      \begin{cases}
        \min
          \set{\m{\permap}^v_2(\ip{x}{i}) | v \in \psupp(\prob)}
        \CASE j=i \land \p{x} \in \pvar(\aexpr) \\
        0 \OTHERWISE
      \end{cases}
  \end{align*}
  By construction we obtain that
  $ \forall j\in I\st \m{\salg}_3(j) \subs \m{\salg}_0(j) $,
and that
  $\raValid(\m{\salg}_3, \m{\prob}_3, \m{\permap}_3)$.
Now letting
  $
   \m{\permap}_1' = {\m{\permap}_0-\m{\permap}_3}
  $,
  we obtain a valid resource
  $(\m{\salg}_0, \m{\prob}_0, \m{\permap}_1')$.

  Moreover,
  we have
  $
    \m{\salg}_0 = \m{\salg}_0 \punion \m{\salg}_3
  $
  and
  $
    \forall j\in I\st
    \forall\event\in\m{\salg}_3(j)\st
      \m{\prob}_3(\event)\in\set{0,1}
  $,
  which means that
  for any $X \in \m{\salg}_3$ and $Y \in \m{\salg}_0$,
  $\m{\prob}_3(X) \cdot \m{\prob}_0(Y) = \m{\prob}_0(X \cap Y)$.
  Then, by \eqref{sure-str-convex:bind-a}:
  \[
    (\m{\salg}_0, \bind(\prob, \m{\krnl}_0), \m{\permap}_1')
    \iprod
    (\m{\salg}_3, \m{\prob}_3, \m{\permap}_3)
    \raLeq
    (\m{\salg}_0, \m{\prob}_0, \m{\permap}_0) = a
  \]
  To close the proof it would then suffice to show that
  $ \CMod{\prob} v. K(v) $
  holds on
  $(\m{\salg}_0, \bind(\prob, \m{\krnl}_0), \m{\permap}_1')$
  and that
  $ \sure{\aexpr\at{i}} $
  holds on
  $ (\m{\salg}_3(j),\m{\prob}_3,\m{\permap}_3) $.
  The latter is obvious.
  The former follows from
  the fact that $ \restr{\m{\krnl}_0(j)(v)}{\m{\salg}^v_1} = \m{\prob}^v_1(j) $;
  by upward-closure and \eqref{sure-str-convex:K}
  this means that, for all $v \in \psupp(\prob)$:
  \[
    K(v)(\m{\salg}^v_1, \m{\prob}^v_1, \m{\permap}^v_1)
    \implies
    K(v)(\m{\salg}_0, \m{\krnl}_0(I)(v), \m{\permap}_1')
  \]
  which proves our claim.
\end{proof}
 \begin{lemma}
\label{proof:c-for-all}
  \Cref{rule:c-for-all} is sound.
\end{lemma}

\begin{proof}
  By unfolding the definitions,
\begin{align*}
      &\CMod\prob \m{v}. \forall x: X. Q(\m{v}) \\
    \iff &
    \begin{array}[t]{r@{\,}l}
    \E \m{\sigmaF}, \m{\mu}_0, \m{\krnl}.
      &\Own{(\m{\sigmaF}, \m{\mu}_0)} *
      \pure{
        \forall i\in I\st
          \m{\mu}_0(i) = \bind(\prob, \m{\krnl}(i))
      }\\ & * \; (
        \forall a \in A_{\prob}.
\Own{(\m{\sigmaF}, \m[i: \m{\krnl}(i)(a) | i \in I])} \wand \forall x: X. Q(\m{v})
          )
    \end{array}
    \\
    \implies &
    \begin{array}[t]{r@{\,}l}
    \forall x: X.
    \E \m{\sigmaF}, \m{\mu}_0, \m{\krnl}.
      &\Own{(\m{\sigmaF}, \m{\mu}_0)} *
      \pure{
        \forall i\in I\st
          \m{\mu}_0(i) = \bind(\prob, \m{\krnl}(i))
      }\\ & * \; (
        \forall a \in A_{\prob}.
\Own{(\m{\sigmaF}, \m[i: \m{\krnl}(i)(a) | i \in I])} \wand Q(\m{v})
          )
    \end{array}\\
    \iff & \forall x: X. \CMod\prob \m{v}. Q(\m{v})
  \end{align*}
\end{proof} \begin{lemma}
\label{proof:c-pure}
  \Cref{rule:c-pure} is sound.
\end{lemma}

\begin{proof}
  We first prove the forward direction:
  For any $a \in \Model_I$,
  if $\model{\pure{\mu(X) = 1} \ast \CMod{\mu}. K(v)}{(a)}$,
  then there exists some
  $\m{\salg}_0, \m{\prob}_0,\m{\permap}_0,$ and $ \m{\krnl}_0 $:
  \begin{gather*}
    (\m{\salg}_0, \m{\prob}_0,\m{\permap}_0) \raLeq a
    \\
    \forall i\in I\st
      \m{\prob}_0(i) = \bind(\prob, \m{\krnl}_0(i)) \\
      \forall v \in \psupp(\prob) \st
      \model{K(v)}{(\m{\salg}_0, \m{\krnl}_0(I)(v), \m{\permap}_0)}
  \end{gather*}

   The pure fact $\pure{\mu(X) = 1}$ implies that
   $X \supseteq \psupp(\mu)$ , and thus
   for every $v \in \psupp(\mu)$, $\pure{v \in X}$.
   Therefore, $\model{K(v)}{(\m{\salg}_0, \m{\krnl}_0(I)(v), \m{\permap}_0)}$,
   which witnesses that
   $\model{\CMod{\mu}. \pure{v \in X} \ast K(v)}{(a)}$.

  We then prove the backward direction:
  if $\CMod{\mu}. \pure{v \in X} \ast K(v) $,
  then there exists
    $\m{\salg}_0, \m{\prob}_0,\m{\permap}_0,$ and~$ \m{\krnl}_0 $:
  \begin{gather*}
    (\m{\salg}_0, \m{\prob}_0,\m{\permap}_0) \raLeq a
    \\
    \forall i\in I\st
      \m{\prob}_0(i) = \bind(\prob, \m{\krnl}_0(i)) \\
      \forall v \in \psupp(\prob) \st
      \model{\pure{v \in X} \ast K(v)}{(\m{\salg}_0, \m{\krnl}_0(I)(v), \m{\permap}_0)}
  \end{gather*}
  Then it must $X \supseteq \psupp(\mu)$,
  which implies that $\pure{\mu(X) = 1}$.
  Meanwhile, $\pure{v \in X} \ast K(v)$ holding on ${(\m{\salg}_0, \m{\krnl}_0(I)(v), \m{\permap}_0)}$
  implies that $K(v)$ holds on
  ${(\m{\salg}_0, \m{\krnl}_0(I)(v), \m{\permap}_0)}$
  Therefore, $\pure{\mu(X) = 1} \ast \CMod{\mu}. K(v)$ holds on $a$.
\end{proof}

\subsection{Soundness of Primitive WP Rules}
\label{sec:appendix:wp-rules}

\subsubsection{Structural Rules}
\begin{lemma}
\label{proof:wp-cons}
  \Cref{rule:wp-cons} is sound.
\end{lemma}

\begin{proof}
  For any resource $a$,
  if $(\WP {\m{t}} {Q})(a)$, then
  \[
    \forall \m{\prob}_0 \st
    \forall c \st
    (a \raOp c) \raLeq \m{\prob}_0
    \implies
    \exists b \st
    \bigl(
      (b \raOp c) \raLeq \sem{\m{t}}(\m{\prob}_0)
      \land
      \model{Q}{(b)}
    \bigr)
  \]
  From the premise $Q \proves Q'$, and the fact that $b$ must be valid for $ (b \raOp c) \raLeq \sem{\m{t}}(\m{\prob}_0) $ to hold, we have that $Q(b)$ implies $Q'(b)$.
  Thus, it must
  \[
    \forall \m{\prob}_0 \st
    \forall c \st
    (a \raOp c) \raLeq \m{\prob}_0
    \implies
    \exists b \st
    \bigl(
      (b \raOp c) \raLeq \sem{\m{t}}(\m{\prob}_0)
      \land
      Q'(b)
    \bigr),
  \]
  which says $(\WP {\m{t}} {Q'})(a)$.
\end{proof} \begin{lemma}
\label{proof:wp-frame}
  \Cref{rule:wp-frame} is sound.
\end{lemma}

\begin{proof}
  Let $a\in\Model_I$ be a valid resource such that
  it satisfies $P \ast \WP {\m{t}} {Q}$.
  By definition, this means that, for some $a_1,a_2$:
  \begin{gather}
    a_1 \raOp a_2 \raLeq a
    \\
    P(a_1)
    \label{wp-frame:pa1}
    \\
    \forall \m{\prob}_0, c \st
      (a_2 \raOp c) \raLeq \m{\prob}_0
      \implies
        \exists b \st
        \bigl(
          (b \raOp c) \raLeq \sem{\m{t}}(\m{\prob}_0)
          \land
          Q(b)
        \bigr)
    \label{wp-frame:assumed-wp}
  \end{gather}
  Our goal is to prove $a$ satisfies
  $\WP {\m{t}} {P * Q}$, which,
  by unfolding the definitions, amounts to:
  \begin{equation}
    \exists a' \raLeq a \st
    \forall \m{\prob}_0, c' \st
      (a' \raOp c') \raLeq \m{\prob}_0
      \implies
      \exists b_1,b \st
        ((b_1 \raOp b) \raOp c') \raLeq \sem{\m{t}}(\m{\prob}_0)
        \land
        P(b_1) \land Q(b)
    \label{wp-frame:goal}
  \end{equation}

  Our goal can be proven by instantiating
  $a' = (a_1 \raOp a_2)$ and $b_1 = a_1$,
  from which we reduce the goal
  to proving, for all $\m{\prob}_0, c'$:
  \begin{equation}
      ((a_1 \raOp a_2) \raOp c') \raLeq \m{\prob}_0
      \implies
      \exists b \st
        ((a_1 \raOp b) \raOp c') \raLeq \sem{\m{t}}(\m{\prob}_0)
        \land
        P(a_1) \land Q(b)
  \end{equation}
  We have that $P(a_1)$ holds by \eqref{wp-frame:pa1}.
  By associativity and commutativity of the RA operation,
  we reduce the goal to:
  \begin{equation}
      (a_2 \raOp (a_1 \raOp c')) \raLeq \m{\prob}_0
      \implies
      \exists b \st
        (b \raOp (a_1 \raOp c')) \raLeq \sem{\m{t}}(\m{\prob}_0)
        \land
        Q(b)
  \end{equation}
  This follows by applying assumption~\eqref{wp-frame:assumed-wp}
  with $c = (a_1 \raOp c')$.
\end{proof}
 \begin{lemma}
\label{proof:wp-nest}
  \Cref{rule:wp-nest} is sound.
\end{lemma}

\begin{proof}
  Because $\m{t_1 \m. t_2}$ is defined,
  it must $\supp{\m{t_1}} \cap \supp{\m{t_2}} = \emptyset$.
  By definition
because $\supp{\m{t_1}} \cap \supp{\m{t_2}} = \emptyset$,
  we have
$
    \sem{(\m{t_1 \m. t_2)}}(\m{\prob})
    = \sem{\m{t_2}} (\sem{\m{t_1}} (\m{\prob}))
  $.

  \begin{itemize}
    \item For the $\vdash$ case:
Assume $a$ is a valid resource such that
  $(\WP {\m{t_1}} {\WP {\m{t_2}} {Q}})(a)$ holds.
  Our goal is to prove $(\WP {\m{t_1 \m. t_2}} {Q})(a_0)$ holds,
  which unfolds by definition of WP into:
  \begin{equation}
    \forall \m{\prob}_0\st
    \forall c_0 \st
    (a_0 \raOp c_0) \raLeq \m{\prob}_0
    \implies
      \exists a_2 \st
      \bigl(
        (a_2 \raOp c_0) \raLeq \sem{\m{t_1 \m. t_2}}(\m{\prob}_0)
      \land Q(a_2)
      \bigr)
    \label{wp-nest:goal}
  \end{equation}

  Take an arbitrary $\m{\prob}_0$ and $c_0$ such that
  $ (a_0 \raOp c_0) \raLeq \m{\prob}_0 $.
  By unfolding the WPs in the assumption,
  we have that there exists a
  $a_1 \in \Model_I$ such that:
  \begin{gather}
    (a_1 \raOp c_0) \raLeq \sem{\m{t_1}}(\m{\prob}_0)
    \label{wp-nest:a1}
    \\
    \forall \m{\prob}_1\st
    \forall c_1 \st
      (a_1 \raOp c_1) \raLeq \m{\prob}_1
      \implies
      \exists a_2 \st
      (a_2 \raOp c_1) \raLeq \sem{\m{t_2}}(\m{\prob}_1)
      \land  Q(a_2)
    \label{wp-nest:a2}
  \end{gather}
  We can apply \eqref{wp-nest:a2} to \eqref{wp-nest:a1}
  by instantiating $\m{\prob}_1$ with $\sem{\m{t_1}}(\m{\prob}_0)$,
  and $c_1$ with $c_0$,
  obtaining:
  \[
    \exists a_2 \st
    ((a_2 \raOp c_0) \raLeq \sem{\m{t_1}}(\sem{\m{t_2}}(\m{\prob}_0))
    \land  Q(a_2))
  \]
  When $\m{t_1 \m. t_2}$ is defined,
  the terms $\m{t_1}$ and $\m{t_2}$ are on disjoint indices,
  and thus
  $ \sem{ \m{t_1} \m. \m{t_2}}(\prob_0) = \sem{\m{t_1}}(\sem{\m{t_2}}(\prob_0)) $,
  we obtain the goal \eqref{wp-nest:goal} as desired.

    \item For the $\dashv$ case:
      For any resource $a$,
      \begin{align}
      &
      \WP {(\m{t_1} \m. \m{t_2})} {Q} (a) \notag \\
      {}\iff{} &
        \forall \m{\prob}_0 \st
        \forall c \st
        (a \raOp c) \raLeq \m{\prob}_0
        \implies
        \exists b \st
        \bigl(
          (b \raOp c) \raLeq \sem{\m{t_1} \m. \m{t_2}}(\m{\prob}_0)
          \land
          \model{Q}{(b)}
       \bigr)
       \label{helper:wp-nest:1}
      \end{align}

      Since $(b \raOp c) \raLeq \sem{\m{t_1} \m. \m{t_2}}(\m{\prob}_0)
      $, we have $\raValid(b \raOp c)$ and thus $\raValid(b)$. Say
      \begin{align*}
        b &= \m[i: \m{\sigmaF_b}(i), \m{\mu_b}(i), \m{\permap_b}(i)] \\
        c &= \m[i: \m{\sigmaF_c}(i), \m{\mu_c}(i), \m{\permap_c}(i)]
\end{align*}
      Let
      \begin{align*}
        b' &=
        \begin{cases}
          i:  (\m{\sigmaF_b}(i), \m{\mu_b}(i), \m{\permap_b}(i)) \CASE \text{ if $i \in \supp{\m{t_1}}$} \\
          i : (\m{\sigmaF}(i), \m{\mu}(i), \m{\permap}(i)) \CASE \text{ if $i \notin \supp{\m{t_1}}$ }
        \end{cases}
      \end{align*}
      Since $\raValid(b \raOp c)$ and  $\raValid(a \raOp c)$,
      on each index $i \in I$, we have $\raValid(b'(i) \raOp c(i))$  .
      Thus, $\raValid(b' \raOp c)$.
      Also, for each $i \in I$,
      $(b'(i) \raOp c(i)) \raLeq \sem{\m{t_1}(i)}(\m{\prob}_0(i))$,
      which implies that
      \[
          (b' \raOp c) \raLeq \sem{\m{t_1}}(\m{\prob}_0)
      \]
      We want to show next that
      $\model{\WP {\m{t_2}} {Q}}{b'}$.
      For any $c' = \m[i: \m{\sigmaF'_c}(i), \m{\mu'_c}(i), \m{\permap'_c}(i)]$
      such that $\raValid(b' \raOp c')$,
      it must
      \begin{align*}
        &\raValid((\m{\sigmaF_b}(i), \m{\mu_b}(i), \m{\permap_b}(i)) \raOp (\m{\sigmaF'_c}(i), \m{\mu'_c}(i), \m{\permap'_c}(i)))  &\text{if } i \in \supp{\m{t_1}} \\
        &\raValid((\m{\sigmaF}(i), \m{\mu}(i), \m{\permap}(i)) \raOp (\m{\sigmaF'_c}(i), \m{\mu'_c}(i), \m{\permap'_c}(i)))  &\text{if } i  \in \supp{\m{t_2}}
      \end{align*}
      By~\cref{helper:wp-nest:1},
      $\raValid((\m{\sigmaF}(i), \m{\mu}(i), \m{\permap}(i)) \raOp (\m{\sigmaF'_c}(i), \m{\mu'_c}(i), \m{\permap'_c}(i)))$
      also implies
      \[
        \raValid((\m{\sigmaF_b}(i), \m{\mu_b}(i), \m{\permap_b}(i)) \raOp (\m{\sigmaF'_c}(i), \m{\mu'_c}(i), \m{\permap'_c}(i))).
      \]
      Thus,
      $(b \raOp c) \raLeq \sem{\m{t_1} \m. \m{t_2}}(\m{\prob}_0)
          \land
          Q(b)$
          witnesses that
           $\WP {\m{t_2}} {Q}(b')$.
  \end{itemize}
\end{proof}
 \begin{lemma}
\label{proof:wp-conj}
  \Cref{rule:wp-conj} is sound.
\end{lemma}

\begin{proof}
     For any resource $a$,
      \begin{align*}
        & \bigl(\WP {\m{t_1}} {Q_1} \land   \WP {\m{t_2}} {Q_2}\bigr)(a)
        \\
        \iff &
        \forall \m{\prob}_0 \st
        \forall c \st
        (a \raOp c) \raLeq \m{\prob}_0
        \implies \\
             & \exists b \st
        \bigl(
          (b \raOp c) \raLeq \sem{\m{t_1}}(\m{\prob}_0)
          \land
          Q(b)
       \bigr) \land
         \exists b' \st
        \bigl(
          (b' \raOp c) \raLeq \sem{\m{t_2}}(\m{\prob}_0)
          \land
          Q(b')
       \bigr)
      \end{align*}

      Define $b''$ such that
      \begin{align*}
        b''(i) & =
        \begin{cases}
          b(i) \CASE i \in \idx(Q_1) \\
          b'(i) \OTHERWISE
        \end{cases}
      \end{align*}

      For any $c$, $\raValid(a \raOp c)$ implies
       $\raValid(b'' \raOp c)$ because
       $\raValid(b'(i) \raOp c(i))$ and
       $\raValid(b(i) \raOp c(i))$ for all $i$.
       Furthermore,
       $b''(i) = b(i)$ for $i \in \idx(Q_1)$ implies
       that $Q_1(b'')$.
       Also,
       $\idx(Q_2) \inters \supp{\m{t}_1} \subs \supp{\m{t}_2}$
       implies
       that $Q_2(b'')$.
       Therefore, $(Q_1 \land Q_2)(b'')$, witnessing
       $\model{\WP {\m{t_1} + \m{t_2}} {Q_1 \land Q_2}}{(a)} $.
\end{proof} \begin{lemma}
\label{proof:c-wp-swap}
  \Cref{rule:c-wp-swap} is sound.
\end{lemma}

\begin{proof}
    By the meaning of conditioning modality and weakest precondition transformer,
    \begin{align*}
    &(\ownall \land \CMod\prob v. \WP {\m{t}} {Q(v)})(a) \\
    {}\iff{}
    &\ownall(a) \land
    \begin{array}[t]{@{}r@{\,}l@{}}
    \E \m{\sigmaF}, \m{\mu}, \m{\permap}, \m{\krnl}.
      & (\m{\sigmaF}, \m{\mu}, \m{\permap}) \raLeq a
      \land
        \forall i\in I\st
        \m{\mu}(i) = \bind(\prob, \m{\krnl}(i))
      \\ & \land \;
        \forall v \in \psupp(\prob).
          (\WP {\m{t}} {Q(v)})(\m{\sigmaF}, \m{\krnl}(I)(v), \m{\permap})
    \end{array}
  \end{align*}
Intuitively, for each $v$,
  running $\m{t}$ on each fibre $(\m{\sigmaF}, \m{\krnl}(I)(v), \m{\permap})$
  gives a output resource that satisfies $Q(v)$.

  Assume $\raValid(a)$ holds and let
  $a = (\m{\sigmaF}_a, \m{\mu}_a, \m{\permap}_a)$.
  By~\cref{lemma:bind-extend},
  when $ (\m{\sigmaF}, \m{\mu}, \m{\permap}) \raLeq a$,
  $\m{\mu} = \bind(\prob, \m{\krnl})$ iff that there exists $\m{\krnl}''$ such that
  $\m{\mu}_a = \bind(\prob, \m{\krnl}'')$  and $\m{\krnl}(I)(v) \extTo \m{\krnl}''(I)(v)$ for every $v$.
  Thus,
\begin{align*}
    (\CMod\prob v. \WP {\m{t}} {Q(v)})(\m{\sigmaF}_a, \m{\mu}_a, \m{\permap}_a)
    \iff
    \begin{array}[t]{@{}r@{\,}l@{}}
    \E \m{\krnl}.
      &\forall i\in I\st
        \m{\mu}_a (i) = \bind(\prob, \m{\krnl}''(i))
        \\ & \land \;
        \forall v \in \psupp(\prob).
          (\WP {\m{t}} {Q(v)})(\m{\sigmaF}, \m{\krnl}(I)(v), \m{\permap})
    \end{array}
  \end{align*}

  We want to show that
  \[
      \WP {\m{t}}{\CMod\prob v. {Q(v)}} (a)
    \]
  which is equivalent to
  \[
        \A \m{\prob'}.
        \forall c\st
        a\raOp c \raLeq \m{\prob'}
          \implies
            \exists a'\st
            a'\raOp c \raLeq {\sem{\m{t}} (\m{\prob'})}
            \land
            (\CMod \prob Q(v))(a).
  \]
  Let's fix an arbitrary $ \m{\prob'}, c $ that satisfy
  $\raValid(a \raOp c) \land  a\raOp c \raLeq a_{\m{\prob'}}$,
  we try to construct a corresponding $a'$.
  The high-level approach that we will take is to show that running $\m{t}$ on $a$ takes us to a resource that is equivalent to bind the set of output resource satisfying $Q(v)$ to $\mu$.

  Recall that $a  = (\m{\sigmaF}_a, \m{\mu}_a, \m{\permap}_a)$ also satisfies
  $\ownall$, which says $\m{\sigmaF}_a =  \Full{\Var}$.
  We claim that
  $ a\raOp c  \raLeq (\Full{\Var}, \m{\mu'}, \permap_1)$ holds implies that
  the probability space $c$ is trivial.
  Say $c = (\m{\sigmaF}_c, \m{\mu}_c, \m{\permap}_c)$, then for any $E \in \m{\sigmaF}_c$,
  the event $E$ must also in $\m{\sigmaF}_a$ and $\Full{\Var}$  because they are
  the full sigma algebra.
  By definition of $ a\raOp c  \raLeq (\Full{\Var}, \m{\mu'}, \permap_1)$,
  we have
\begin{align}
    &\m{\mu}_c(E) \cdot \m{\mu}_a(E) =  \m{\mu'}(E \cap E)  = \m{\mu'}(E).
    \label{eq:c-wp-swap-c-trivial}
  \end{align}
  Another implication of  $ a\raOp c  \raLeq (\Full{\Var}, \m{\mu'}, \permap_1)$ is that
  we have $\m{\mu}_c(E)  = \m{\mu'}(E)$ and  $\m{\mu}_a(E)  = \m{\mu'}(E)$.
  Combining with~\cref{eq:c-wp-swap-c-trivial}, we can conclude
  \begin{align*}
    \m{\mu'}(E) \cdot \m{\mu'}(E) =  \m{\mu'}(E),
  \end{align*}
  which implies that
  $\m{\mu}_c(E) = \m{\mu'}(E)  \in \{0,1\}$.
  Therefore, $c$ is a trivial probability space and
  \begin{align*}
    (\m{\salg}_a, \m{\krnl}(I)(v), \m{\permap}_a) \raOp c \raLeq
    (\m{\salg}_a, \m{\krnl}(I)(v), \m{\permap}_a)
  \end{align*}

  Furthermore, for every $v \in \psupp(\mu)$, we have
  $(\WP {\m{t}} {Q(v)})(\m{\sigmaF}, \m{\krnl}(I)(v), \m{\permap})$
  which implies
  \begin{align}
    \label{eq:c-wp-swap:helper}
        \A \m{\krnl'}. &
(\m{\salg}_a, \m{\krnl}(I)(v), \m{\permap}_a) \raOp c
          \raLeq \m{\krnl'}(I)(v)
        \\
        & {} \implies {}
        \exists a_v\st
(a_v \raOp c \raLeq \sem{\m{t}} (\m{\krnl'}(I)(v))) \land Q(v)(a_v) .
      \end{align}
Therefore,
  \begin{align}
   & a\raOp c \raLeq a_{\m{\prob'}}
  {}\implies{} \forall v \in \psupp(\mu). (\raValid((\m{\salg}_a, \m{\krnl}(I)(v), \m{\permap}_a) \raOp c) \land  (\m{\salg}_a, \m{\krnl}(I)(v), \m{\permap}_a) \raOp c \raLeq (\Full{\Var}, \m{\krnl}(I)(v), \fullp) \tag{By~\ref{lemma:fibre-prod-exists} and~\ref{lemma:bind-extend}}\\
  {}\implies{} & \forall v \in \psupp(\mu).  \exists a_v\st
        \raValid(a_v \raOp c) \land  \left(a_v \raOp c \raLeq (\Full{\Var}, \sem{\m{t}} (\m{\krnl}(I)(v)), \fullp)\right)  \land  Q(v)(a_v)
        \tag{By~\cref{eq:c-wp-swap:helper}} \\
  {}\implies{} & \forall v \in \psupp(\mu). \m{\permap}_{a_v} + \m{\permap}_c \raLeq \fullp \land    Q(v)(\Full{\Var}, \sem{\m{t}} (\m{\krnl'}(I)(v)), \fullp).
        \tag{By upwards closure}
  \end{align}

  Let $a'_v =  (\Full{\Var}, \sem{\m{t}} (\m{\krnl'}(I)(v)), \m{\permap}_{a})$.
  Because $\mu_c(E) \in \{0, 1\}$ for any $E \in \sigmaF_c$, for every $v$,
  we have $ (\Full{\Var}, \sem{\m{t}} (\m{\krnl'}(I)(v)))  \raOp (\m{\sigmaF}_c, \m{\mu}_c)$
  defined and thus
  $a'_v \raOp c$ valid.
  Define
  \[
    a' =
     (\Full{\Var},
      \bind(\mu, \fun v . \sem{\m{t}} (\m{\krnl'}(I)(v)),
      \m{\permap}_a)
  \]
  By~\cref{lemma:fibre-prod-exists},
  $ \raValid(a'_v \raOp c)$ for all $v \in \psupp_{\mu}$ implies
  $\raValid(a' \raOp c)$.
  Also, because  $Q(v)(a_v)$ for all $v \in A_{\mu}$,
  $(\CMod{\prob} v. Q(v))(a') $.
Thus, $(\WP {\m{t}}{\CMod\prob v. {Q(v)}}) (a)$.
\end{proof}
 
\subsubsection{Program Rules}
\begin{lemma}
\label{proof:wp-skip}
  \Cref{rule:wp-skip} is sound.
\end{lemma}

\begin{proof}
  Assume $a \in \Model_I$ is valid and such that~$P(a)$ holds.
  By unfolding the definition of WP, we need to prove
  \[
    \forall \m{\prob}_0.
      \forall c \st
      (a \raOp c) \raLeq \m{\prob}_0
      \implies
      \exists b \st
      \bigl(
        (b \raOp c) \raLeq \sem{\m{t}}(\m{\prob}_0)
        \land
        P(b)
      \bigr)
  \]
  which follows trivially
  by $\sem{\m[i:\code{skip}]}(\m{\prob}_0)=\m{\prob}_0$ and picking $b=a$.
\end{proof}
 \begin{lemma}
\label{proof:wp-seq}
  \Cref{rule:wp-seq} is sound.
\end{lemma}

\begin{proof}
Assume $a_0 \in \Model_I$ is a valid resource such that
  $(\WP {\m[i: t]} {\WP {\m[i: t']} {Q}})(a_0)$ holds.
  Our goal is to prove $(\WP {(\m[i: t\p; t'])} {Q})(a_0)$ holds,
  which unfolds by definition of WP into:
  \begin{equation}
    \forall \m{\prob}_0\st
    \forall c_0 \st
    (a_0 \raOp c_0) \raLeq \m{\prob}_0
    \implies
      \exists a_2 \st
      \bigl(
      (a_2 \raOp c_0) \raLeq \sem{\m[i: t\p; t']}(\m{\prob}_0)
      \land Q(a_2)
      \bigr)
    \label{wp-seq:goal}
  \end{equation}

  Take an arbitrary $\m{\prob}_0$ and $c_0$ such that
  $ (a_0 \raOp c_0) \raLeq \m{\prob}_0 $.
  By unfolding the WPs in the assumption,
  we have that there exists a
  $a_1 \in \Model_I$ such that:
  \begin{gather}
    (a_1 \raOp c_0) \raLeq \sem{\m[i: t]}(\m{\prob}_0)
    \label{wp-seq:a1}
    \\
    \forall \m{\prob}_1\st
    \forall c_1 \st
      (a_1 \raOp c_1) \raLeq \m{\prob}_1
      \implies
      \exists a_2 \st
      ((a_2 \raOp c_1) \raLeq \sem{\m[i: t']}(\m{\prob}_1)
      \land  Q(a_2))
    \label{wp-seq:a2}
  \end{gather}
  We can apply \eqref{wp-seq:a2} to \eqref{wp-seq:a1}
  by instantiating $\m{\prob}_1$ with $\sem{\m[i: t]}(\m{\prob}_0)$,
  and $c_1$ with $c_0$,
  obtaining:
  \[
    \exists a_2 \st
    ((a_2 \raOp c_0) \raLeq \sem{\m[i: t']}(\sem{\m[i: t]}(\m{\prob}_0))
    \land  Q(a_2))
  \]
  Since by definition,
  $ \sem{t\p; t'}(\prob_0) = \sem{t'}(\sem{t}(\prob_0)) $,
  we obtain the goal \eqref{wp-seq:goal} as desired.
\end{proof} \begin{lemma}
\label{proof:wp-assign}
  \Cref{rule:wp-assign} is sound.
\end{lemma}

\begin{proof}
  \newcommand{\specialevent}{A}
  Let $a \in \Model_I$ be a valid resource,
  and let $a(i) = (\salg, \prob, \permap)$.
  By assumption we have
  $\permap(\p{x}) = 1$ and
  $\permap(\p{y}) > 0$ for all $ \p{y} \in \FV(\expr)$.
  We want to show that $a$ satisfies
  $\WP {\m[i: \code{x:=}\expr]} {\sure{\ip{x}{i} = \expr\at{i}}}$.
This is equivalent to
\[
    \forall \m{\prob}_0 \st
      \forall c\st
      (a\raOp c \raLeq  \m{\prob}_0)
      \implies
        \exists b\st (
          b\raOp c \raLeq  \sem{\m[i: \code{x:=}\expr]}(\m{\prob}_0) \land
          \sure{\ip{x}{i} = \expr\at{i}}(b)
        )
  \]
  We show this holds by picking~$b$ as follows:
  \begin{align*}
    b &\is a\m[i: {(\salg_b, \prob_b, \permap)}]
    &
    \salg_b &\is \set{
      \Store,\emptyset,\specialevent,\Store\setminus\specialevent
    }
    &
    \specialevent &\is \set{
      \store\upd{\p{x}->\sem{\expr}(\store)} | \store \in \Store
}
  \end{align*}
  where $\prob_b$ is determined by setting $ \prob_b(\specialevent)=1. $

  By construction we have that $\sure{\ip{x}{i} = \expr\at{i}}(b)$ holds.
  To close the proof we then need to show that
$(b \raOp c) \raLeq \sem{\m[i: \code{x:=}\expr]}(\m{\prob}_0)$.

  Let $c(i) = (\salg_c,\prob_c,\permap_c)$.
  Observe that by the assumptions on $\permap$,
  we have $\raValid(b)$ since $\salg_b$ is
  only non-trivial on $\pvar(\expr)\union\set{\p{x}}$;
  moreover, by the assumption $\raValid(a \raOp c)$
  we have that $\raValid(\permap + \permap_c)$ holds,
  which means that $\permap_c(\p{x})=0$,
  and thus $\salg_c$ is trivial on \p{x}.

  \newcommand{\Pre}{\operatorname{pre}}
Let us define the function
  $
    \Pre \from \powerset(\Store) \to \powerset(\Store)
  $
  as:
  \[
    \Pre(\event) \is
      \set{ \store
          | \store\upd{\p{x}->\sem{\expr}(\store)} \in \event
        }.
  \]
  That is, $\Pre(\event)$ is the weakest precondition (in the standard sense)
  of the assignment.
  By construction, we have:
  \begin{align*}
    \Pre(\specialevent) &= \Store
    &
    \Pre(\event_1 \inters \event_2) &=
      \Pre(\event_1) \inters \Pre(\event_2)
    \\
    \Pre(\Store \setminus \specialevent) &= \emptyset
    &
    \Pre(\event_c) &= \event_c
      \text{ for all } \event_c \in \salg_c
  \end{align*}
  In particular, the latter holds because $\salg_c$ is trivial in \p{x}.

  By unfolding the definition of $\sem{\hole}$,
  it is easy to check that for every $\event \in \Full{\Store}$:
  \[
    \sem{\code{x:=}\expr}(\mu_0)(\event) = \mu_0(\Pre(\event))
  \]

  We are now ready to show
  $(b \raOp c) \raLeq \sem{\m[i: \code{x:=}\expr]}(\m{\prob}_0)$
  by showing that
  $
    (\salg_b,\prob_b)\iprod(\salg_c,\prob_c)
    =
    (\salg_b \punion \salg_c,
     \restr{\sem{\code{x:=}\expr}(\prob_0)}{(\salg_b \punion \salg_c)})
  $
  where $\prob_0 = \m{\prob}_0(i)$.
  To show this it suffices to prove that
  for every $\event_b \in \salg_b$ and every $\event_c \in \salg_c$,
  $
    \sem{\code{x:=}\expr}(\prob_0)(\event_b \inters \event_c)
    = \prob_b(\event_b) \cdot \prob_c(\event_c).
  $
  We proceed by case analysis on $\event_b$:
  \begin{casesplit}
  \case[$\event_b = \specialevent$]
    Then:
    \begin{align*}
      \sem{\code{x:=}\expr}(\prob_0)(\specialevent \inters \event_c)
        &= \prob_0(\Pre(\specialevent \inters \event_c))
      \\&= \prob_0(\Pre(\specialevent) \inters \Pre(\event_c))
      \\&= \prob_0(\Store \inters \Pre(\event_c))
      \\&= \prob_0(\Pre(\event_c))
      \\&= \prob_b(\specialevent) \cdot \prob_0(\event_c)
      \\&= \prob_b(\specialevent) \cdot \prob_c(\event_c)
    \end{align*}
  \case[$\event_b = \Store \setminus \specialevent$]
    Then:
    \begin{align*}
      \sem{\code{x:=}\expr}(\prob_0)
          (\Store \setminus \specialevent \inters \event_c)
        &= \prob_0(\Pre((\Store \setminus \specialevent) \inters \event_c))
      \\&= \prob_0(\Pre(\Store \setminus \specialevent) \inters \Pre(\event_c))
      \\&= \prob_0(\emptyset \inters \Pre(\event_c))
      \\&= 0
      \\&= \prob_b(\Store \setminus \specialevent) \cdot \prob_c(\event_c)
    \end{align*}
  \case[$\event_b = \Store$ or $\event_b = \emptyset$]
    Analogous to the previous cases.
    \qedhere
  \end{casesplit}
\end{proof}
 \begin{lemma}
\label{proof:wp-samp}
  \Cref{rule:wp-samp} is sound.
\end{lemma}

\begin{proof}
  \newcommand{\specialevent}{A}
  Assume $a \in \Model_I$ is valid and such that
  $ a(i) = (\salg, \prob, \permap)$, with $\permap(x) = 1$.
  Our goal is to show that~$a$ satisfies
  $
    \WP {\m[i: \code{x:~$\dist$($\vec{v}$)}]} {\distAs{\ip{x}{i}}{d(\vec{v})}}
  $
which is equivalent to proving, for all $\m{\prob}_0$ and for all~$c$:
\begin{equation}
    (a\raOp c \raLeq  \m{\prob}_0)
    \implies
      \exists b\st \bigl(
        b\raOp c \raLeq \sem{\m[i: \code{x:~$\dist$($\vec{v}$)}]}(\m{\prob}_0)
        \land
        (\distAs{x}{d(\vec{v})})(b)
      \bigr)
    \label{wp-samp:goal}
  \end{equation}

  Let $\prob_0 = \m{\prob}_0(i)$ and
      $\prob_1 = \sem{\code{x:~$\dist$($\vec{v}$)}}(\prob_0)$.
  Moreover, let $ c(i) = (\salg_c,\prob_c,\permap_c) $.
  Observe that by the assumptions on $\permap$ and validity of~$a \raOp c$,
  we have $\permap_c(\p{x})=0$,
  which means $\salg_c$ is trivial on \p{x}.
  We aim to prove~\eqref{wp-samp:goal} by letting
  \begin{align*}
    b &\is a\m[i: {(\salg_b, \prob_b, \permap_b)}]
    &
    \prob_b &\is
      \restr{\prob_1}{\salg_b}
    \\
    \salg_b &\is
      \sigcl*{
        \set[\big]{
          \set{ \store \in \Store | \store(\p{x}) = v }
        | v\in \Val
        }
      }
    &
    \permap_b &\is
      \perm{\p{x}: 1}
\end{align*}
  Note that by construction $ \raValid(\permap_b + \permap_c) $,
  and $\raValid(b)$ since $\salg_b$ is only non-trivial in \p{x}.
  \newcommand{\Pre}{\operatorname{pre}}
  Similarly to the proof for~\cref{rule:wp-assign},
  we define the function
  $
    \Pre \from \powerset(\Store) \to \powerset(\Store)
  $
  as:
  \[
    \Pre(\event) \is
      \set{ \store
          | \exists v \in \Val \st \store\upd{\p{x}->v} \in \event
        }.
  \]
  Since $\salg_c$ is trivial on \p{x},
  for all $ \event_c \in \salg_c $,
  $\Pre(\event_c) = \event_c$.
  Moreover,
  for all $ \event_b \in \salg_b \setminus \set{\emptyset} $,
  $ \Pre(\event_b) = \Store $,
  since $\event_b$ is trivial on every variable except~\p{x}.

  By unfolding the definitions, we have:
  \begin{align*}
    \mu_1(\event) &=
      \sem{\code{x:~$\dist$($\vec{v}$)}}(\prob_0)(\event)
\\&=
      \Sum_{\store \in \event}
        \prob_0(\Pre(\store)) \cdot \sem{\dist}(\vec{v})(\store(\p{x}))
  \end{align*}
  We now show that
  $
    (\salg_b, \prob_b) \iprod (\salg_c, \prob_c)
    =
    (\salg_b \punion \salg_c, \restr{\prob_1}{(\salg_b \punion \salg_c)})
  $
  by showing that
  for all $\event_b \in \salg_b$ and
          $\event_c \in \salg_c$:
  $
    \prob_1(\event_b \inters \event_c)
    =
    \prob_b(\event_b) \cdot \prob_c(\event_c).
  $
  \newcommand{\valOf}{\operatorname{V}}
  To prove this we first
  define $\valOf \from \powerset(\Store) \to \powerset(\Val)$ as
  $
    \valOf(\event) \is
      \set{ \store(\p{x}) | \store \in \event },
  $
  and $ S_w \is \set{ \store | \store(\p{x}) = w } $.
  We observe that
  $
    \event_b   = \Dunion_{w \in \valOf(\event_b)} S_w
  $, and thus
  $
    \event_b \inters \event_c
    = \Dunion_{w \in \valOf(\event_b)} (\event_c \inters S_w)
  $;
  moreover,
  $
    \Pre(\event_c \inters S_w)
    = \set{ \store | \store\upd{x->w} \in \event_c } = \event_c
  $.
  Thus, we can calculate:
  \begin{align*}
    \prob_1(\event_b \inters \event_c)
      &=
    \sum_{\mathclap{\store \in \event_b \inters \event_c}}
      \prob_0(\Pre(\store)) \cdot \sem{\dist}(\vec{v})(\store(\p{x}))
    \\&=
    \sum_{w \in \valOf(\event_b)}
    \sum_{\store \in \event_c \inters S_{w}}
      \prob_0(\Pre(\store)) \cdot
      \sem{\dist}(\vec{v})(w)
    \\&=
    \sum_{w \in \valOf(\event_b)}
    \left(
      \sem{\dist}(\vec{v})(w) \cdot
      \sum_{\mathclap{\store \in \event_c \inters S_{w}}}
      \prob_0(\Pre(\store))
    \right)
    \\&=
    \left(
    \sum_{w \in \valOf(\event_b)}
      \sem{\dist}(\vec{v})(w)
      \cdot \prob_0(\Pre(\event_c \inters S_w))
    \right)
    \\&=
    \left(
    \sum_{w \in \valOf(\event_b)}
      \sem{\dist}(\vec{v})(w)
    \right)
    \cdot \prob_0(\event_c)
    \\&=
    \prob_b(\event_b) \cdot  \prob_c(\event_c)
  \end{align*}
The last equation is given by
  $ a\raOp c \raLeq  \m{\prob}_0 $ which implies that $ \prob_c = \restr{\prob_0}{\salg_c} $, and by:
  \begin{align*}
    \prob_b(\event_b)
      =
    \prob_1(\event_b)
      &=
    \sum_{\store \in \event_b}
      \prob_0(\Pre(\store)) \cdot \sem{\dist}(\vec{v})(\store(\p{x}))
    \\&=
    \sum_{w \in \valOf(\event_b)}
      \sum_{\store \in S_w}
        \prob_0(\Pre(\store)) \cdot \sem{\dist}(\vec{v})(w)
    \\&=
    \sum_{\mathclap{w \in \valOf(\event_b)}}
      \sem{\dist}(\vec{v})(w)
  \end{align*}

  Finally, we need to show $ (\distAs{x}{d(\vec{v})})(b) $
  which amounts to proving
  $\almostM{\p{x}}{(\salg_b,\prob_b)}$
  and
  $\sem{\dist}(\vec{v}) = \prob_b \circ \inv{\p{x}}$.
  The former holds because by construction \p{x} is measurable in $\salg_b$.
  For the latter, for all $W \subs \Val$:
  \[
    (\prob_b \circ \inv{\p{x}})(W)
    =
    \prob_b (\inv{\p{x}}(W))
    =
    \sum_{\mathclap{w \in \valOf(\inv{\p{x}}(W))}}
      \sem{\dist}(\vec{v})(w)
    =
    \sum_{w \in W}
      \sem{\dist}(\vec{v})(w)
    =
    \sem{\dist}(\vec{v})(W).
    \qedhere
  \]
\end{proof}
 \begin{lemma}
\label{proof:wp-if-prim}
  \Cref{rule:wp-if-prim} is sound.
\end{lemma}

\begin{proof}
  For any valid resource $a$,
  \begin{align*}
    & \left(\ITE{v}{\WP{\m[i: t_1]}{Q(1)}}{\WP{\m[i: t_2]}{Q(0)}}\right)(a)\\
    {}\iff {} &
    \begin{cases}
      \left(\WP{\m[i: t_1]}{Q(1)}\right)(a) \CASE v \doteq 1 \\
      \left(\WP{\m[i: t_2]}{Q(0)}\right)(a)\OTHERWISE \\
    \end{cases} \\
    {}\iff {}&
  \forall \m{\prob}_0.
    \forall c \st
    (a \raOp c) \raLeq \m{\prob}_0
    \implies
      \begin{cases}
          \exists b \st (b \raOp c) \raLeq \sem{i: \m{t_1}}(\m{\prob}_0)
      \land Q(1)(b)
      \CASE v \doteq 1 \\
          \exists b \st (b \raOp c) \raLeq \sem{i: \m{t_2}}(\m{\prob}_0)
      \land Q(0)(b)
      \OTHERWISE
      \end{cases} \\
    {}\iff {} &
    \forall \m{\prob}_0.
    \forall c \st
    (a \raOp c) \raLeq \m{\prob}_0
    {}\implies {}
    \exists b \st (b \raOp c) \raLeq \sem{i: \ITE{v}{\m{t_1}}{\m{t_2}}}(\m{\prob}_0)
      \land  Q(v \doteq 1)(b) \\
    \implies &  \left(\WP{\m[i: \ITE{v}{\m{t_1}}{\m{t_2}}]}{Q(v \doteq 1)}\right)(a)
  \end{align*}
\end{proof} \begin{lemma}
\label{proof:wp-bind}
  \Cref{rule:wp-bind} is sound.
\end{lemma}

\begin{proof}
  For any resource $a = (\m{\sigmaF}, \m{\mu}, \m{\permap})$,
  $(\sure{\expr\at{i}=v} * \WP{\m*[i: {\Ectxt[v]}]}{Q})(\m{\sigmaF}, \m{\mu}, \m{\permap})$ iff
    there exists $(\m{\sigmaF_1}, \m{\mu_1}, \m{\permap_1}), (\m{\sigmaF_2}, \m{\mu_2}, \m{\permap_2})$ such that
\begin{gather*}
      (\sure{\expr\at{i}=v})(\m{\sigmaF_1}, \m{\mu_1}, \m{\permap_1})\\
      (\WP{\m*[i: {\Ectxt[v]}]}{Q})(\m{\sigmaF_2}, \m{\mu_2}, \m{\permap_2})\\
      (\m{\sigmaF_1}, \m{\mu_1}, \m{\permap_1}) \raOp  (\m{\sigmaF_2}, \m{\mu_2}, \m{\permap_2})
      \raLeq (\m{\sigmaF}, \m{\mu}, \m{\permap})
    \end{gather*}
By the upwards closure, we also have
    \begin{gather*}
      (\sure{\expr\at{i}=v})(\m{\sigmaF}, \m{\mu}, \m{\permap}) \\
      (\WP{\m*[i: {\Ectxt[v]}]}{Q})(\m{\sigmaF}, \m{\mu}, \m{\permap})
    \end{gather*}
The fact that $(\sure{\expr\at{i}=v})(\m{\sigmaF_1}, \m{\mu_1}, \m{\permap_1})$
    implies that $\m{\mu_1}(\inv{(\expr\at{i}=v)} (\True)) = 1$,
    which implies that
    $\sem{\expr}(s) = v$ for all $s \in \psupp(\m{\mu}_1(i))$.

    By~\cref{lemma:context-binding}, we have for any $s \in \Store$,
    \begin{align*}
      \Sem[K]{\Ectxt[\expr]}(s) = \Sem[K]{\Ectxt[\sem{\expr}(s)]}(s),
    \end{align*}
    which implies that for any $\mu_0$ over $\Full{\Store}$
    \begin{align*}
      \sem{\Ectxt[\expr]}(\mu_0)
      &= \DO{
        s <- \m{\mu_0};
        \Sem[K]{\Ectxt[\expr]}(s)
      }\\
      &= \DO{
        s <- \m{\mu_0};
        \Sem[K]{\Ectxt[\sem{\expr}(s)]}(s)
      }\\
      &= \DO{
        s <- \m{\mu_0};
        \Sem[K]{\Ectxt[v]}(s)
      }\\
      &= \sem{\Ectxt[v]}(\mu_0).
    \end{align*}
Define $\mu'_0 = \Sem{\m[i: {\Ectxt[v]}]} \mu_0$.
    Thus, $(\WP{\m*[i: {\Ectxt[v]}]}{Q})(a)$ iff
\begin{align*}
  \forall \mu_0 \st
    \forall c\st
    (\raValid(a \raOp c) \land  a\raOp c \raLeq  a_{\mu_0})
      \implies
        \exists a'\st
        (\raValid(a' \raOp c) \land  a'\raOp c \raLeq  a_{\mu'_0} \land Q(a'))
\end{align*}
iff
\begin{align*}
  \forall \mu_0 \st
    \forall c\st
    (\raValid(a \raOp c) \land  a\raOp c \raLeq  a_{\mu_0})
      \implies
        \exists a'\st
        (\raValid(a' \raOp c) \land  a'\raOp c \raLeq  a_{\mu'_0} \land Q(a'))
\end{align*}
iff $\model{\WP{\m*[i: {\Ectxt[\expr]}]}{Q}}{(a)}$.
\end{proof}
 \begin{lemma}
\label{proof:wp-loop-unf}
  \Cref{rule:wp-loop-unf} is sound.
\end{lemma}

\begin{proof}
  By definition,
  \begin{align*}
    \sem{\Loop{(n+1)}{t}}(\prob)
    &= \bigl(
         \DO{s <- \prob; s' <- \var{loop}_{\term}(n,s); \Sem[K]{t}(s')}
       \bigr) \\
    &= \sem{(\Loop{n}{t})\p;t}(\prob)
  \end{align*}
  thus the rule follows from the argument of \cref{proof:wp-seq}.
\end{proof} \begin{lemma}
\label{proof:wp-loop}
  \Cref{rule:wp-loop} is sound.
\end{lemma}

\begin{proof}
  By induction on~$n$.
  \begin{induction}
    \step[Base case~$n=0$]
      Analogously to \cref{proof:wp-skip}
      since, by definition,
      $\Sem{\Loop{0}{t}}(\mu_0) = \mu_0$.

    \step[Induction step~$n>0$]
      By induction hypothesis
      $P(0) \proves \WP{\m[j: \Loop{(n-1)}{t}]}{P(n-1)}$ holds,
      and we want to show that
      $P(0) \proves \WP{\m[j: \Loop{n}{t}]}{P(n)}$.
      By \cref{proof:wp-loop-unf},
      it suffices to show
      $ P(0) \proves \WP{\m[j: \Loop{(n-1)}{t}]}{\WP {\m[j:t]}{P(n)}} $.
      By applying the induction hypothesis and \cref{proof:wp-cons} we are left
      with proving
      $ P(n-1) \proves \WP {\m[j:t]}{P(n)} $
      which is implied by the premise of the rule with $i=n-1 < n$.
    \qedhere
  \end{induction}
\end{proof}
 
\subsection{Soundness of Derived Rules}
\label{sec:appendix:derived-rules}

In this section we provide derivations for the rules we claim
are derivable in \thelogic.

\subsubsection{Ownership and Distributions}
\begin{lemma}
\label{proof:sure-dirac}
  \Cref{rule:sure-dirac} is sound.
\end{lemma}

\begin{proof}
  \begin{eqexplain}
    \distAs{E\at{i}}{\delta_v}
\whichisequiv*
    \E \m{\sigmaF}, \m{\prob}.
      \Own{(\m{\sigmaF}, \m{\prob})}
      *
      \pure{\m{\prob} \circ \inv{E\at{i}} = \dirac{v}}
\whichisequiv
    \E \m{\sigmaF}, \m{\prob}.
      \Own{(\m{\sigmaF}, \m{\prob})}
      *
      \pure{\m{\prob} \circ \inv{(E\at{i} = v)} = \dirac{\True}}
\whichisequiv
    \sure{E\at{i} = v}
\qedhere
  \end{eqexplain}
\end{proof} \begin{lemma}
\label{proof:sure-eq-inj}
  \Cref{rule:sure-eq-inj} is sound.
\end{lemma}

\begin{proof}
\begin{eqexplain}
  \sure{\aexpr\at{i} = v}
  *
  \sure{\aexpr\at{i} = v'}
\whichproves*
  \distAs{\aexpr\at{i}}{\dirac{v}}
  *
  \distAs{\aexpr\at{i}}{\dirac{v'}}
  \byrule{sure-dirac}
\whichproves
  \distAs{\aexpr\at{i}}{\dirac{v}}
  \land
  \distAs{\aexpr\at{i}}{\dirac{v'}}
\whichproves
  \pure{\dirac{v}=\dirac{v'}}
  \byrule{dist-inj}
\whichproves
  \pure{v=v'}
\qedhere
\end{eqexplain}
\end{proof} \begin{lemma}
\label{proof:sure-sub}
  \Cref{rule:sure-sub} is sound.
\end{lemma}

\begin{proof}
  \begin{eqexplain}
\distAs{\aexpr_1\at{i}}{\prob}
    *
    \sure{(\aexpr_2 = f(\aexpr_1))\at{i}}
  \whichproves*
    \CC\prob v.
      \sure{\aexpr_1\at{i} = v}
      *
      \sure{(\aexpr_2 = f(\aexpr_1))\at{i}}
  \byrules{c-unit-r,c-frame}
\whichproves
    \CC\prob v.
      \sure{\aexpr_1\at{i} = v \land \aexpr_2\at{i} = f(\aexpr_1\at{i})}
  \byrule{sure-merge}
\whichproves
    \CC\prob v. \sure{\aexpr_2\at{i} = f(v)}
  \byrule{c-cons}
\whichproves
    \CC\prob v. \CC{\dirac{f(v)}} \pr{v}.\sure{\aexpr_2\at{i} = \pr{v}}
  \byrule{c-unit-l}
\whichproves
    \CC{\pr{\prob}} \pr{v}.\sure{\aexpr_2\at{i} = \pr{v}}
  \byrules{c-assoc,c-sure-proj}
  \end{eqexplain}
  where $\pr{\prob} = \bind(\prob, \fun x. \dirac{f(x)}) = \prob \circ \inv{f}$.
  By \ref{rule:c-unit-r} we thus get
  $\distAs{\aexpr_2\at{i}}{\prob \circ \inv{f}}$.
\end{proof} \begin{lemma}
\label{proof:dist-fun}
  \Cref{rule:dist-fun} is sound.
\end{lemma}

\begin{proof}
  Assume $E\from \Store \to A$ and $ f \from A \to B $, then:
  \begin{eqexplain}
    \distAs{\aexpr\at{i}}{\prob}
\whichproves*
      \CC\prob v. \sure{(\aexpr = v)\at{i}}
    \byrules{c-unit-r}
\whichproves
      \CC\prob v. \sure{(f\circ\aexpr)\at{i} = f(v)}
    \byrules{c-cons}
\whichproves
      \CC\prob v. \CC{\dirac{f(v)}} \pr{v}.\sure{(f\circ\aexpr)\at{i} = \pr{v}}
    \byrule{c-unit-l}
\whichproves
      \CC{\pr{\prob}} \pr{v}.\sure{(f\circ\aexpr)\at{i} = \pr{v}}
    \byrules{c-assoc,c-sure-proj}
  \end{eqexplain}
  where $\pr{\prob} = \bind(\prob, \fun x. \dirac{f(x)}) = \prob \circ \inv{f}$.
  By \ref{rule:c-unit-r} we thus get
  $\distAs{(f\circ\aexpr)\at{i}}{\prob \circ \inv{f}}$.
\end{proof} \begin{lemma}
\label{proof:dirac-dup}
  \Cref{rule:dirac-dup} is sound.
\end{lemma}

\begin{proof}
  \begin{eqexplain}
    \distAs{E\at{i}}{\dirac{v}}
    \whichproves*
    \sure{E\at{i} = v}
    \byrule{sure-dirac}
    \whichproves
    \sure{E\at{i} = v} \ast \sure{E\at{i} = v}
    \byrule{sure-merge}
    \whichproves
    \distAs{E\at{i}}{\dirac{v}} * \distAs{E\at{i}}{\dirac{v}}
    \byrule{sure-dirac}
  \end{eqexplain}
\end{proof} \begin{lemma}
\label{proof:dist-supp}
  \Cref{rule:dist-supp} is sound.
\end{lemma}

\begin{proof}
  \begin{eqexplain}
    \distAs{E\at{i}}{\mu}
    \whichproves*
    \CMod{\mu} v. \sure{E\at{i} = v}
    \byrule{c-unit-r}
    \whichproves
    \pure{\mu(\psupp(\mu)) = 1} \ast
    \CMod{\mu} v. \sure{E\at{i} = v}
    \whichproves
    \CMod{\mu} v. \big(\pure{v \in \psupp(\mu)} \ast \sure{E\at{i} = v}\big)
    \byrule{c-pure}
    \whichproves
    \CMod{\mu} v. \big(\sure{E\at{i} = v} \ast \sure{E\at{i} \in \psupp(\mu)} \big)
\whichproves
    \big(\CMod{\mu} v.\sure{E\at{i} = v}\big)
    \ast \sure{E\at{i} \in \psupp(\mu)}
    \byrule{sure-str-convex}
    \whichproves
    \distAs{E\at{i}}{\mu}
    \ast \sure{E\at{i} \in \psupp(\mu)}
    \byrule{c-unit-r}
  \end{eqexplain}
\end{proof} \begin{lemma}
\label{proof:prod-unsplit}
  \Cref{rule:prod-unsplit} is sound.
\end{lemma}

\begin{proof}
  \begin{eqexplain}
    \distAs{\aexpr_1\at{i}}{\prob_1} *
    \distAs{\aexpr_2\at{i}}{\prob_2}
    \whichproves*
    \CC{\prob_1} v_1.
    \CC{\prob_2} v_2.
    \bigl(
      \sure{\aexpr_1\at{i} = v_1} *
      \sure{\aexpr_2\at{i} = v_2}
    \bigr)
    \byrules{c-unit-r,c-frame}
\whichproves
    \CC{\prob_1} v_1.
    \CC{\prob_2} v_2.
      \sure{(\aexpr_1, \aexpr_2)\at{i} = (v_1, v_2)}
    \byrules{sure-merge}
\whichproves
    \CC{\prob_1 \pprod \prob_2} (v_1,v_2).
      \sure{(\aexpr_1, \aexpr_2)\at{i} = (v_1, v_2)}
    \byrules{c-assoc}
\whichproves
    \distAs{(\aexpr_1\at{i}, \aexpr_2\at{i})}{\prob_1 \otimes \prob_2}
    \byrule{c-unit-r}
  \end{eqexplain}
\end{proof}
 
\subsubsection{\Supercond}
\begin{lemma}
\label{proof:c-fuse}
  \Cref{rule:c-fuse} is sound.
\end{lemma}

\begin{proof}
  Recall that
  $
    \prob \fuse \krnl
    \is
    \fun(v,w). \prob(v)\krnl(v)(w).
  $
  which can be reformulated as
  $
    \prob \fuse \krnl =
    \bind(\prob,\fun v.(\bind(\krnl(v), \fun w.\return(v,w)))).
  $

  The $(\proves)$ direction is an instance of \ref{rule:c-assoc}.

  The $(\provedby)$ direction follows from \ref{rule:c-unassoc}:
  \begin{eqexplain}
    \CC{\prob \fuse \krnl} (v',w'). K(v',w')
\whichproves*
    \CC \prob v.
      \CC{\bind(\krnl(v), \fun w.\dirac{(v,w)})} (v',w'). K(v',w')
    \byrule{c-unassoc}
\whichproves
    \CC \prob v.
      \CC{\krnl(v)} w.
        \CC {\dirac{(v,w)}} (v',w'). K(v',w')
    \byrule{c-unassoc}
\whichproves
    \CC{\prob} v.
    \CC{\krnl(v)} w.
      K(v,w)
    \byrule{c-unit-l}
  \end{eqexplain}
\end{proof}
 \begin{lemma}
\label{proof:c-swap}
  \Cref{rule:c-swap} is sound.
\end{lemma}

\begin{proof}
  \begin{eqexplain}
    \CC{\prob_1} v_1.
      \CC{\prob_2} v_2.
        K(v_1, v_2)
\whichproves*
      \CC{\prob_1 \pprod \prob_2}
        (v_1,v_2).
          K(v_1, v_2)
    \byrule{c-fuse}
\whichproves
    \CC{\prob_2} v_2.
      \CC{\prob_1} v_1.
          K(v_1, v_2)
    \byrule{c-fuse}
  \end{eqexplain}
  Where
  \[
    \prob_1 \pprod \prob_2
    =
    \prob_1 \fuse (\fun \wtv.\prob_2)
=
    \prob_2 \fuse (\fun \wtv.\prob_1)
\]
  justifies the applications of \ref{rule:c-fuse}.
\end{proof}
 \begin{lemma}
\label{proof:sure-convex}
  \Cref{rule:sure-convex} is sound.
\end{lemma}

\begin{proof}
  By \ref{rule:sure-str-convex} with $K = \True$.
\end{proof}
 \begin{lemma}
\label{proof:dist-convex}
  \Cref{rule:dist-convex} is sound.
\end{lemma}

\begin{proof}
  \begin{eqexplain}
    \CMod{\prob} v. \distAs{E\at{i}}{\mu'}
    \whichproves*
    \CMod{\mu} v.  \CMod{\mu'} w. \sure{E\at{i} = w}
    \byrule{c-unit-r}
    \whichproves
    \CMod{\mu'} w. \CMod{\mu} v.  \sure{E\at{i} = w}
    \byrule{c-swap}
    \whichproves
    \CMod{\mu'} w.  \sure{E\at{i} = w}
    \byrule{sure-convex}
    \whichproves
    \distAs{E\at{i}}{\mu'}
    \byrule{c-unit-r}
  \end{eqexplain}
\end{proof}
 \begin{lemma}
\label{proof:c-proj}
  The following rule is sound:
  \begin{proofrule}
  \infer{
    \forall (v,\wtv)\in\psupp(\prob).
    \forall \prob'.
    \CC{\prob'} w. P(v) \proves P(v)
  }{
    \CC \prob (v,w). P(v) \lequiv
    \CC {\prob\circ\inv{\proj}} v. P(v)
  }
  \end{proofrule}
\end{lemma}
\begin{proof}
  Assume that for all $(v,\wtv)\in\psupp(\prob)$,
  $\forall \prob'. \CC{\prob'} w. P(v) \proves P(v)$
  (\ie $P(v)$ is convex).
  By \cref{lm:fuse-split} there is some~$\krnl$ such that
  $ \prob=(\prob\circ\inv{\proj})\fuse\krnl $.
  Then:
  \begin{eqexplain}
    \CC \prob (v,w). P(v)
    \whichisequiv*
    \CC {\prob\circ\inv{\proj}} v.
    \CC {\krnl(v)} w. P(v)
    \byrule{c-fuse}
    \whichisequiv
    \CC {\prob\circ\inv{\proj}} v. P(v)
  \end{eqexplain}
  The last step is justified by the convexity assumption in the~$(\proves)$
  direction,
  and by \ref{rule:c-true} and \ref{rule:c-frame} in the~$(\provedby)$ direction.
\end{proof}

\begin{lemma}
\label{proof:c-sure-proj}
  \Cref{rule:c-sure-proj} is sound.
\end{lemma}
\begin{proof}
  By \cref{proof:c-proj} and \cref{proof:sure-convex}.
\end{proof}
 \begin{lemma}
\label{proof:c-sure-proj-many}
  \Cref{rule:c-sure-proj-many} is sound.
\end{lemma}

\begin{proof}
  Let $
    X_i \is \set{ \p{x} | \ip{x}{i} \in X }
  $ for every $i\in I$.
  Then:
  \begin{eqexplain}
    \CC\prob (\m{v}, w).
      \sure{\ip{x}{i}=\m{v}(\ip{x}{i})}_{\ip{x}{i}\in X}
\whichisequiv*
    \CC{\prob} (\m{v}, w).
      \LAnd_{i\in I}
        \sure{\LAnd_{\p{x}\in X_i} \ip{x}{i}=\m{v}(\ip{x}{i})}
\whichisequiv
    \LAnd_{i\in I}
      \CC{\prob} (\m{v}, w).
        \sure{\LAnd_{\p{x}\in X_i} \ip{x}{i}=\m{v}(\ip{x}{i})}
    \byrule{c-and}
\whichisequiv
    \LAnd_{i\in I}
      \CC{\prob\circ\inv{\proj_1}} \m{v}.
        \sure{\LAnd_{\p{x}\in X_i} \ip{x}{i}=\m{v}(\ip{x}{i})}
    \byrule{c-sure-proj}
\whichisequiv
    \CC{\prob\circ\inv{\proj_1}} \m{v}.
      \LAnd_{i\in I}
        \sure{\LAnd_{\p{x}\in X_i} \ip{x}{i}=\m{v}(\ip{x}{i})}
    \byrule{c-and}
\whichisequiv
    \CC{\prob\circ\inv{\proj_1}} \m{v}.
      \sure{\ip{x}{i}=\m{v}(\ip{x}{i})}_{\ip{x}{i}\in X}
  \end{eqexplain}
  Note that the (iterated) applications of \ref{rule:c-and}
  satisfy the side condition
  because the inner assertions are by construction on disjoint indices.
  The backward direction of \ref{rule:c-and} holds by
  the standard laws of conjunction.
\end{proof}
 \begin{lemma}
\label{proof:c-extract}
  \Cref{rule:c-extract} is sound.
\end{lemma}

\begin{proof}
  \begin{eqexplain}
    \CC{\prob_1} v_1. \bigl(
      \sure{\aexpr_1\at{i} = v_1} *
      \distAs{\aexpr_2\at{i}}{\prob_2}
    \bigr)
\whichproves*
      \CC{\prob_1} v_1.
        \bigl(
          \sure{\aexpr_1\at{i}=v_1} *
          \CC{\prob_2} v_2. \sure{\aexpr_2\at{i}=v_2}
        \bigr)
      \byrule{c-unit-r}
    \whichproves
      \CC{\prob_1} v_1.
      \CC{\prob_2} v_2.
        \bigl(
          \sure{\aexpr_1\at{i}=v_1} *
          \sure{\aexpr_2\at{i}=v_2}
        \bigr)
      \byrules{c-frame}
    \whichproves
      \CC{\prob_1} v_1.
      \CC{\prob_2} v_2.
        \sure{\aexpr_1\at{i}=v_1 \land \aexpr_2\at{i}=v_2}
      \byrules{sure-merge}
    \whichproves
      \CC{\prob_1 \pprod \prob_2} (v_1,v_2).
        \sure{(\aexpr_1\at{i},\aexpr_2\at{i})=(v_1,v_2)}
      \byrules{c-assoc}
    \whichproves
      \distAs{(\aexpr_1\at{i},\aexpr_2\at{i})}{(\prob_1 \pprod \prob_2)}
      \byrule{c-unit-r}
    \whichproves
      \distAs{\aexpr_1\at{i}}{\prob_1} *
      \distAs{\aexpr_2\at{i}}{\prob_2}
      \byrule{prod-split}
  \end{eqexplain}
\end{proof}
 \begin{lemma}
\label{proof:c-dist-proj}
  \Cref{rule:c-dist-proj} is sound.
\end{lemma}

\begin{proof}
  By \cref{proof:c-proj} and \cref{proof:dist-convex}.
\end{proof}
 
\subsubsection{Relational Lifting}
\begin{lemma}
\label{proof:rl-cons}
  \Cref{rule:rl-cons} is sound.
\end{lemma}

\begin{proof}
  \begin{eqexplain}
    \cpl{R_1}
\whichis*
      \E \prob.
        \pure{\prob(R_1) = 1} *
        \CC\prob \m{v}.
          \sure{\ip{x}{i} = \m{v}(\ip{x}{i})}_{\ip{x}{i}\in X}
\whichproves
      \E \prob.
        \pure{\prob(R_2) = 1} *
        \CC\prob \m{v}.
          \sure{\ip{x}{i} = \m{v}(\ip{x}{i})}_{\ip{x}{i}\in X}
    \by{$R_1 \subseteq R_2$}
\whichis \cpl{R_2}
  \qedhere
  \end{eqexplain}
\end{proof} \begin{lemma}
\label{proof:rl-unary}
  \Cref{rule:rl-unary} is sound.
\end{lemma}

\begin{proof}
  \begin{eqexplain}
    \cpl{R}
\whichis*
    \E \prob.
      \pure{\prob(R) = 1} *
      \CC\prob \m{v}.
        \sure{\ip{x}{i} = \m{v}(\p{x}{i})}_{\p{x}{i}\in X}
\whichproves
    \E \prob.
      \CC\prob \m{v}.
        \pure{\m{v} \in R}
        * \sure{\ip{x}{i} = \m{v}(\p{x}{i})}_{\p{x}{i}\in X}
    \byrule{c-pure}
\whichproves
    \E \prob.
      \CC\prob \m{v}.
        \sure{R(\p{x}_1\at{i}, \dots , \p{x}_n\at{i})}
\whichproves
    \E \prob.
      \sure{R(\p{x}_1\at{i}, \dots, \p{x}_n\at{i})}
    \byrule{sure-convex}
\whichproves
    \sure{R(\p{x}_1\at{i}, \dots, \p{x}_n\at{i})}
  \qedhere
  \end{eqexplain}
\end{proof}

 \begin{lemma}
\label{proof:rl-eq-dist}
  \Cref{rule:rl-eq-dist} is sound.
\end{lemma}

\begin{proof}
  \begin{eqexplain}
    \cpl{x\at{i} = y\at{j}}
\proves{}&
    \E \prob'.
     \CC{\prob'} (v_1,v_2).\bigl(
      \sure{\Ip{x}{i} = v_1} \land
      \sure{\Ip{y}{j} = v_2} \land
      \pure{v_1=v_2}
    \bigr)
\whichproves
    \E \prob'.
     \CC{\prob'} (v_1, v_2).\bigl(
      \sure{\Ip{x}{i} = v_1} \land
      \sure{\Ip{y}{j} = v_1 }
    \bigr)
    \byrule{c-cons}
\whichproves
    \E \prob'.
     \CC{\prob'\circ\inv{\proj_1}} v_1.\bigl(
      \sure{\Ip{x}{i} = v_1} \land
      \sure{\Ip{y}{j} = v_1 }
    \bigr)
    \byrule{c-sure-proj}
\whichproves
    \E \prob.
     \CC{\prob} v_1.\bigl(
      \sure{\Ip{x}{i} = v_1} \land
      \sure{\Ip{y}{j} = v_1 }
    \bigr)
    \by{$ \prob = \prob'\circ\inv{\proj_1} $}
\whichproves
    \E \prob.
      \bigl(
      \CC\prob v_1.
       \sure{\Ip{x}{i} = v_1}
      \bigr)
      \land
      \bigl(
      \CC\prob v_1.
        \sure{\Ip{y}{j} = v_1}
      \bigr)
\whichproves
    \E \prob.
    \distAs{\Ip{x}{i}}{\mu}
    \land
    \distAs{\Ip{y}{j}}{\mu}
    \byrule{c-unit-r}
\whichproves
    \E \prob.
    \distAs{\Ip{x}{i}}{\mu}
    \ast
    \distAs{\Ip{y}{j}}{\mu}
    \byrule{and-to-star}
  \end{eqexplain}
\end{proof} \begin{lemma}
  \Cref{rule:rl-convex} is sound.
\end{lemma}

\begin{proof}
 \begin{eqexplain}
   \CMod{\prob} a \st \cpl{R}
\whichis*
\CMod{\prob} a \st
     \E \prob'.
       \pure{\prob'(R) = 1} *
       \bigl(
         \CC{\prob'} \m{v}.
           \sure{\ip{x}{i} = \m{v}(\ip{x}{i})}_{\ip{x}{i}\in X}
       \bigr)
\whichproves
     \E \krnl.
     \CMod{\prob} a \st
       \bigl(
         \CC{\krnl(a)} \m{v}.
           \sure{\ip{x}{i} = \m{v}(\ip{x}{i})}_{\ip{x}{i}\in X}
           * \pure{R(\m{v})}
       \bigr)
   \byrules{c-pure,c-skolem}
\whichproves
     \E \hat{\prob}.
     \CMod{\hat{\prob}} (a,\m{v}) \st
       \sure{\ip{x}{i} = \m{v}(\ip{x}{i})}_{\ip{x}{i}\in X}
       * \pure{R(\m{v})}
   \byrule{c-fuse}
\whichproves
     \E \hat{\prob}.
     \pure{\hat{\prob} \circ \inv{\proj_2}(R) = 1} *
     \CMod{\hat{\prob}} (a,\m{v}) \st
       \sure{\ip{x}{i} = \m{v}(\ip{x}{i})}_{\ip{x}{i}\in X}
   \byrule{c-pure}
\whichproves
     \E \hat{\prob}.
     \pure{\hat{\prob} \circ \inv{\proj_2}(R) = 1} *
     \CMod{\hat{\prob} \circ \inv{\proj_2}} \m{v} \st
       \sure{\ip{x}{i} = \m{v}(\ip{x}{i})}_{\ip{x}{i}\in X}
   \byrule{c-sure-proj-many}
\whichproves
     \E \hat{\prob}'.
     \pure{\hat{\prob}'(R) = 1} *
     \CMod{\hat{\prob}'} \m{v} \st
       \sure{\ip{x}{i} = \m{v}(\ip{x}{i})}_{\ip{x}{i}\in X}
\whichis \cpl{R}
 \end{eqexplain}
 In the derivation we use
 $ \hat{\prob} = \prob \fuse \krnl $,
 and
 $\hat{\prob}' = \hat{\prob} \circ \inv{\proj_2}$.
\end{proof}
 \begin{lemma}
\label{proof:rl-merge}
  \Cref{rule:rl-merge} is sound.
\end{lemma}

\begin{proof}
  Let $R_1\in \Val^{X_1}$ and $R_2\in \Val^{X_2}$
  and let
  $ X = X_1 \inters X_2 $,
  $ Y_1 = X_1 \setminus X $, and
  $ Y_2 = X_2 \setminus X $, so that
  $ X_1 \union X_2 = Y_1 \dunion X \dunion Y_2 $.

  By definition, $\cpl{R_1} * \cpl{R_2}$ entails that for some
  $ \prob_1,\prob_2 $ with $ \prob_1(R_1)=1 $ and $ \prob_1(R_2)=1 $:
  \begin{eqexplain}
&
    \CC{\prob_1} \m{v}_1.
      (\sure{\ip{x}{i} = \m{v}_1(\ip{x}{i})}_{\ip{x}{i}\in X_1}) *
\CC{\prob_2} \m{v}_2.
      (\sure{\ip{x}{i} = \m{v}_2(\ip{x}{i})}_{\ip{x}{i}\in X_2})
\whichproves
  \CC{\prob_1} (\m{w}_1,\m{v}_1). (
    \sure{\ip{y}{i} = \m{w}_1(\ip{y}{i})}_{\ip{y}{i}\in Y_1}
    \land
    \sure{\ip{x}{i} = \m{v}_1(\ip{x}{i})}_{\ip{x}{i}\in X}
    * \pure{(\m{w}_1\m{v}_1) \in R_1}
  ) *
  {} \\ &
  \CC{\prob_2} (\m{w}_2,\m{v}_2). (
    \sure{\ip{y}{i} = \m{w}_2(\ip{y}{i})}_{\ip{y}{i}\in Y_2}
    \land
    \sure{\ip{x}{i} = \m{v}_2(\ip{x}{i})}_{\ip{x}{i}\in X}
    * \pure{(\m{w}_2\m{v}_2) \in R_2}
  )
  \byrule{c-pure}
\whichproves
  \CC{\prob_1} (\m{w}_1,\m{v}_1).
    \CC{\prob_2} (\m{w}_2,\m{v}_2).
    \begin{pmatrix*}[l]
    \sure{\ip{y}{i} = \m{w}_1(\ip{y}{i})}_{\ip{y}{i}\in Y_1}
          \land
          \sure{\ip{x}{i} = \m{v}_1(\ip{x}{i})}_{\ip{x}{i}\in X} * {}
    \\
    \sure{\ip{y}{i} = \m{w}_2(\ip{y}{i})}_{\ip{y}{i}\in Y_2}
          \land
          \sure{\ip{x}{i} = \m{v}_2(\ip{x}{i})}_{\ip{x}{i}\in X} * {}
    \\
      \pure{(\m{w}_1\m{v}_1) \in R_1}
    * \pure{(\m{w}_2\m{v}_2) \in R_2}
    \end{pmatrix*}
  \byrule{c-frame}
\whichproves
  \CC{\prob_1} (\m{w}_1,\m{v}_1).
    \CC{\prob_2} (\m{w}_2,\m{v}_2).
    \begin{pmatrix*}[l]
    \sure{\ip{y}{i} = \m{w}_1(\ip{y}{i})}_{\ip{y}{i}\in Y_1}
          \land
          \sure{\ip{x}{i} = \m{v}_1(\ip{x}{i})}_{\ip{x}{i}\in X} * {}
    \\
    \sure{\ip{y}{i} = \m{w}_2(\ip{y}{i})}_{\ip{y}{i}\in Y_2}
          \land
          \sure{\ip{x}{i} = \m{v}_2(\ip{x}{i})}_{\ip{x}{i}\in X} * {}
    \\
      \pure{(\m{w}_1\m{v}_1) \in R_1}
    * \pure{(\m{w}_2\m{v}_2) \in R_2}
    * \pure{\m{v}_1=\m{v}_2}
    \end{pmatrix*}
  \byrule{sure-eq-inj}
\whichproves
  \CC{\prob_1} (\m{w}_1,\m{v}_1).
  \CC{\prob_2} (\m{w}_2,\m{v}_2).
    \begin{pmatrix*}[l]
    \sure{\ip{y}{i} = \m{w}_1(\ip{y}{i})}_{\ip{y}{i}\in Y_1} \land {}
    \\
    \sure{\ip{x}{i} = \m{v}_1(\ip{x}{i})}_{\ip{x}{i}\in X} \land {}
    \\
    \sure{\ip{y}{i} = \m{w}_2(\ip{y}{i})}_{\ip{y}{i}\in Y_2} * {}
    \\
    \pure{(\m{w}_1\m{v}_1) \in R_1 \land (\m{w}_2\m{v}_1) \in R_2}
    \end{pmatrix*}
  \byrule{c-cons}
\whichproves
  \CC{\prob_1} (\m{w}_1,\m{v}_1).
  \CC{\prob_2 \circ \inv{\pi_1}} (\m{w}_2).
    \begin{pmatrix*}[l]
    \sure{\ip{y}{i} = \m{w}_1(\ip{y}{i})}_{\ip{y}{i}\in Y_1} \land {}
    \\
    \sure{\ip{x}{i} = \m{v}_1(\ip{x}{i})}_{\ip{x}{i}\in X} \land {}
    \\
    \sure{\ip{y}{i} = \m{w}_2(\ip{y}{i})}_{\ip{y}{i}\in Y_2} * {}
    \\
    \pure{(\m{w}_1\m{v}_1) \in R_1 \land (\m{w}_2\m{v}_1) \in R_2}
    \end{pmatrix*}
  \byrule{c-sure-proj}
  \end{eqexplain}

  Thus by letting $
  \prob = \prob_1 \pprod (\prob_2 \circ \inv{\pi_1})
  =\bind(\prob_1, \krnl_2)
  $ where
  \[
    \krnl_2 = \fun (\m{w}_1\m{v}_1).(
      \bind(\prob_2,\fun (\m{w}_2,\m{v}_2).
                        \return (\m{w}_1\m{w}_2\m{v}_1))
    )
  \]
  we obtain:
  \begin{eqexplain}
  &
  \CC{\prob_1} (\m{w}_1',\m{v}_1').
    \CC{\krnl_2(\m{w}_1',\m{v}_1')} (\m{w}_1,\m{w}_2,\m{w}).
    \begin{pmatrix*}[l]
    \sure{\ip{y}{i} = \m{w}_1(\ip{y}{i})}_{\ip{y}{i}\in Y_1} \land {}
    \\
    \sure{\ip{x}{i} = \m{w}(\ip{x}{i})}_{\ip{x}{i}\in X} \land {}
    \\
    \sure{\ip{y}{i} = \m{w}_2(\ip{y}{i})}_{\ip{y}{i}\in Y_2} * {}
    \\
    \pure{(\m{w}_1\m{w}) \in R_1 \land (\m{w}_2\m{w}) \in R_2}
    \end{pmatrix*}
\whichproves
  \CC{\prob} (\m{w}_1,\m{w}_2,\m{w}).
    \begin{pmatrix*}[l]
    \sure{\ip{y}{i} = \m{w}_1(\ip{y}{i})}_{\ip{y}{i}\in Y_1} \land {}
    \\
    \sure{\ip{x}{i} = \m{w}(\ip{x}{i})}_{\ip{x}{i}\in X} \land {}
    \\
    \sure{\ip{y}{i} = \m{w}_2(\ip{y}{i})}_{\ip{y}{i}\in Y_2}
    \\
    \pure{(\m{w}_1\m{w}) \in R_1 \land (\m{w}_2\m{w}) \in R_2}
    \end{pmatrix*}
  \byrules{c-assoc,c-sure-proj}
\whichproves
  \CC{\prob} \m{v}.
    \sure{\ip{x}{i} = \m{v}(\ip{x}{i})}_{\ip{x}{i}\in (X_1\union X_2)}
    * \pure{(\restr{\m{v}}{X_1}) \in R_1 \land (\restr{\m{v}}{X_2}) \in R_2}
  \end{eqexplain}
  The result gives us $ \cpl{R_1 \land R_2} $
  by \ref{rule:c-pure} and \cref{def:rel-lift}.
\end{proof} \begin{lemma}
\label{proof:rl-sure-merge}
  \Cref{rule:rl-sure-merge} is sound.
\end{lemma}

\begin{proof}
  \begin{eqexplain}
    \cpl{R} * \sure{\ip{x}{i} = \expr\at{i}}
  \whichproves*
    \E\prob.
      \CC\prob \m{v}. \bigl(
        \sure{\ip{y}{i} = \m{v}(\ip{y}{i})}_{\ip{y}{i}\in X}
        * \pure{R(\m{v})}
      \bigr)
      * \sure{\ip{x}{i} = \expr\at{i}}
  \bydef
\whichproves
    \E\prob.
      \CC\prob \m{v}. \bigl(
        \sure{\ip{y}{i} = \m{v}(\ip{y}{i})}_{\ip{y}{i}\in X}
        * \pure{R(\m{v})}
        * \sure{\ip{x}{i} = \expr\at{i}}
      \bigr)
  \byrule{c-frame}
\whichproves
    \E\prob.
      \CC\prob \m{v}. \bigl(
        \sure{\ip{y}{i} = \m{v}(\ip{y}{i})}_{\ip{y}{i}\in X}
        * \pure{R(\m{v})}
        * \sure{\ip{x}{i} = \sem{\expr\at{i}}(\m{v})}
      \bigr)
  \by{$\pvar(\expr\at{i}) \subs X$}
\whichproves
    \E\prob.
      \CC\prob \m{v}. \bigl(
        \sure{\ip{y}{i} = \m{v}(\ip{y}{i})}_{\ip{y}{i}\in X}
        * \pure{R(\m{v})}
        * \CC{\dirac{\sem{\expr\at{i}}(\m{v})}} w.\sure{\ip{x}{i} = w}
      \bigr)
  \byrule{c-unit-l}
\whichproves
    \E\prob.
      \CC\prob \m{v}.
      \CC{\dirac{\sem{\expr\at{i}}(\m{v})}} w.
      \bigl(
        \sure{\ip{y}{i} = \m{v}(\ip{y}{i})}_{\ip{y}{i}\in X}
        * \pure{R(\m{v})}
        * \sure{\ip{x}{i} = w}
      \bigr)
  \byrule{c-frame}
\whichproves
    \E\prob'.
      \CC{\prob'} \m{v}'.
      \begin{grp}
        \sure{\ip{y}{i} = \m{v}'(\ip{y}{i})}_{\ip{y}{i}\in X}
        * \sure{\ip{x}{i} = \m{v}'(\ip{x}{i})} \\ {}
        * \pure{R(\m{v}') \land \sem{\expr\at{i}}(\m{v}') = \m{v}'(\ip{x}{i})}
      \end{grp}
  \byrules{c-pure,c-assoc}
\whichproves
    \cpl{R \land \ip{x}{i} = \expr\at{i}}
  \end{eqexplain}
  where we let
  $
    \prob' \is
      \left(
      \DO{
        \m{v} <- \prob;
        \return(\m{v}\upd{\ip{x}{i} -> \sem{\expr\at{i}}(\m{v})})
      }
      \right).
  $
\end{proof} \begin{lemma}
\label{proof:coupling}
  \Cref{rule:coupling} is sound.
\end{lemma}

\begin{proof}
  Assuming
  $\prob \circ \inv{\proj_1} = \prob_1$,
  $\prob \circ \inv{\proj_2} = \prob_2$, and
  $\prob(R) = 1$, we have:
  \begin{eqexplain}
    \distAs{\p{x}_1\at{\I1}}{\prob_1} *
    \distAs{\p{x}_2\at{\I2}}{\prob_2}
\whichproves*
    \CC{\prob_1} v. \sure{x_1\at{1} = v} *
    \CC{\prob_2} w. \sure{x_2\at{2} = w}
    \byrule{c-unit-r}
\whichproves
    \CC{\prob} (v, w). \sure{x_1\at{1} = v} *
    \CC{\prob} (v, w). \sure{x_2\at{2} = w}
    \byrule{c-sure-proj}
\whichproves
    \CC{\prob} (v, w). \sure{x_1\at{1} = v} \land
    \CC{\prob} (v, w). \sure{x_2\at{2} = w}
    \byrule{and-to-star}
\whichproves
    \CC{\prob} (v, w).
      (\sure{x_1\at{1} = v} \land
      \sure{x_2\at{2} = w})
    \byrule{c-and}
\whichproves
    \cpl{R(x_1\at{1}, x_2\at{2})}
    \by{$\prob(R) = 1$}
  \end{eqexplain}
\end{proof}
 
\subsubsection{Weakest Precondition}
\begin{lemma}
\label{proof:wp-loop-0}
  \Cref{rule:wp-loop-0} is sound.
\end{lemma}

\begin{proof}
  Special case of \ref{rule:wp-loop} with $n=0$,
  which makes the premises trivial.
\end{proof} \begin{lemma}
\label{proof:wp-loop-lockstep}
  \Cref{rule:wp-loop-lockstep} is sound.
\end{lemma}

\begin{proof}
  We derive the following rule:
  \[
  \infer*{
    \forall k < n\st
      P(k) \proves \WP {\m[i: t, j: t']}{P(k+1)}
  }{
    P(0) \proves
    \WP {\m[i: (\Loop{n}{t}), j: (\Loop{n}{t'})]} {P(n)}
  }
  \]
  (for $n\in \Nat$ and $i \ne j$)
  from the standard \ref{rule:wp-loop},
  as follows.
  Let
  \[
    P'(k) \is
    \WP {\m[j: \Loop{k}{t'}]} {P(k)}
  \]
  Note that
  $P(0) \proves P'(0)$
  by \ref{rule:wp-loop-0}.
  Then we can apply the \ref{rule:wp-loop} using $P'$ as a loop invariant
  \begin{derivation}
    \infer*[Right=\ref{rule:wp-nest}]{
    \infer*{
    \infer*[Right=\ref{rule:wp-loop}]{
    \infer*{
    \infer*[Right=\ref{rule:wp-nest}]{
    \infer*[Right=\ref{rule:wp-loop-unf}]{
    \infer*[Right=\ref{rule:wp-cons}]{
    \infer*[Right=\ref{rule:wp-nest}]{
    \infer*{}{
      \forall k \leq n \st
      P(k)
      \proves
      \WP {\m[i: t, j: t']} {P(k+1)}
    }}{
      \forall k \leq n \st
      P(k)
      \proves
        \WP {\m[j: t']} {
          \WP {\m[i: t]} {P(k+1)}
      }
    }}{
      \forall k \leq n \st
      \WP {\m[j: \Loop{k}{t'}]} {P(k)}
      \proves
      \WP {\m[j: \Loop{k}{t'}]}[\big]{
        \WP {\m[j: t']} {
          \WP {\m[i: t]} {P(k+1)}
        }
      }
    }}{
      \forall k \leq n \st
      \WP {\m[j: \Loop{k}{t'}]} {P(k)}
      \proves
      \WP {\m[j: \Loop{(k+1)}{t'}]}[\big]{
        \WP {\m[i: t]} {P(k+1)}
      }
    }}{
      \forall k \leq n \st
      \WP {\m[j: \Loop{k}{t'}]} {P(k)}
      \proves
      \WP {\m[i: t]}[\big]{
        \WP {\m[j: \Loop{(k+1)}{t'}]} {P(k+1)}
      }
    }}{
      \forall k \leq n \st
        P(k)' \proves
          \WP {\m[i: t]} {P'(k+1)}
    }}{
      P'(0) \proves
      \WP {\m[i: (\Loop{n}{t})]}{
        P'(n)
      }
    }}{
      P(0) \proves
      \WP {\m[i: (\Loop{n}{t})]}{
        \WP{\m[j: (\Loop{n}{t'})]} {P(n)}
      }
    }}{
      P(0) \proves
      \WP {\m[i: (\Loop{n}{t}), j: (\Loop{n}{t'})]} {P(n)}
    }
  \end{derivation}
  From bottom to top,
    we focus on component $i$ using \ref{rule:wp-nest};
    then we use $P(0) \proves P'(0)$ and transitivity of entailment
    to rewrite the goal using the invariant $P'$;
    we then use \ref{rule:wp-loop} and unfold the invariant;
    using \ref{rule:wp-nest} twice we can swap the two components so that
    component~$j$ is the topmost WP in the assumption and conclusion;
    using \ref{rule:wp-loop-unf} we break off the first $k$ iterations at~$j$;
    finally, using \ref{rule:wp-cons} we can eliminate the topmost
    WP on both sides of the entailments.

  It is straightforward to adapt the argument for any number of components
  looping the same number of times.
\end{proof} \begin{lemma}
\label{proof:wp-rl-assign}
  \Cref{rule:wp-rl-assign} is sound.
\end{lemma}

\begin{proof}
  Define
    $\m{\permap}_R \is (\m{\permap} \setminus \ip{x}{i})/2$ and
    $ \m{\permap}_{\p{x}} \is \m{\permap} - \m{\permap}_R $;
  note that by $\m{\permap}(\ip{x}{i})=1$
  we have $\m{\permap}_{\p{x}}(\ip{x}{i})=1$.
  We first show that the following hold:
  \begin{align}
    \cpl{R}\withp{\m{\permap}}
    &\proves
    \cpl{R}\withp{\m{\permap}_R} * (\m{\permap}_{\p{x}})
    \\
    \cpl{R}\withp{\m{\permap}_R}
      * \sure{\p{x}\at{i} = \expr\at{i}}\withp{\m{\permap}_{\p{x}}}
    &\proves
    \cpl{R \land \p{x}\at{i} = \expr\at{i}}\withp{\m{\permap}}
  \end{align}
  The first entailment holds because $\cpl{R}$ is permission-scaling-invariant
  (see \cref{sec:appendix:permissions})
  and by the assumption that $\p{x}\notin \pvar(R)$.
  The second entailment holds by \ref{rule:rl-sure-merge}.

  We can then derive:
  \[
  \begin{derivation}
    \infer*[Right=\ref{rule:wp-cons}]{
    \infer*[Right=\ref{rule:wp-frame}]{
    \infer*[Right=\ref{rule:wp-assign}]{ }{
      (\m{\permap}_{\p{x}})
      \proves
      \WP {\m[i: \code{x:=}\expr]}[\big] {
          \sure{\p{x}\at{i} = \expr\at{i}}\withp{\m{\permap}_{\p{x}}}
      }
    }}{
      \cpl{R}\withp{\m{\permap}_R}
      * (\m{\permap}_{\p{x}})
      \proves
      \WP {\m[i: \code{x:=}\expr]}[\big] {
        \cpl{R}\withp{\m{\permap}_R}
          * \sure{\p{x}\at{i} = \expr\at{i}}\withp{\m{\permap}_{\p{x}}}
      }
    }}{
      \cpl{R}\withp{\m{\permap}}
      \proves
      \WP {\m[i: \code{x:=}\expr]}[\big] {
        \cpl{R \land \p{x}\at{i} = \expr\at{i}}\withp{\m{\permap}}
      }
    }
  \end{derivation}
  \qedhere
  \]
\end{proof} \begin{lemma}
\label{proof:wp-if-unary}
  \Cref{rule:wp-if-unary} is sound.
\end{lemma}

\begin{proof}
  From the premises, we derive:
  \[
  \begin{derivation}
  \infer*[right=\ref{rule:c-wp-swap}]{
  \infer*[Right={\ref{rule:c-unit-r},\ref{rule:c-frame}}]{
  \infer*[Right=\ref{rule:c-cons}]{
  \infer*[Right=\ref{rule:wp-bind}]{
  \infer*[Right=\ref{rule:wp-if-prim}]{
  \infer*{
    P * \sure{\Ip{e}{1}=1} \gproves \WP {\m[\I1:t_1]} {Q(1)}
    \\
    P * \sure{\Ip{e}{1}=0} \gproves \WP {\m[\I1:t_2]} {Q(0)}
  }{\forall b\in\set{0,1}\st
    P * \sure{\Ip{e}{1}=1}
    \gproves
    \ITE{b}{\WP{\m[\I1:t_1]}{Q(1)}}{\WP{\m[\I1:t_2]}{Q(0)}}
  }}{\forall b\in\set{0,1}\st
    P * \sure{\Ip{e}{1}=b}
    \gproves
    \WP {\m[\I1:
        (\code{if $\;b\;$ then $\;t_1\;$ else $\;t_2$})
      ]}{
      Q(b\beq 1)
    }
  }}{\forall b\in\set{0,1}\st
    P * \sure{\Ip{e}{1}=b}
    \gproves
    \WP {\m[\I1:
      (\code{if e then $\;t_1\;$ else $\;t_2$})
    ]}{
      Q(b\beq 1)
    }
  }}{\CC{\beta} b.(P * \sure{\Ip{e}{1}=b})
    \gproves
    \CC{\beta} b.
    \WP {\m[\I1:
      (\code{if e then $\;t_1\;$ else $\;t_2$})
    ]}{
      Q(b\beq 1)
    }
  }}{P * \distAs{\Ip{e}{1}}{\beta}
    \gproves
    \CC{\beta} b.
    \WP {\m[\I1:
      (\code{if e then $\;t_1\;$ else $\;t_2$})
    ]}{
      Q(b\beq 1)
    }
  }}{P * \distAs{\Ip{e}{1}}{{\beta}}
    \gproves
    \WP {\m[\I1:
        (\code{if e then $\;t_1\;$ else $\;t_2$})
      ]}{
      \CC{\beta} b.Q(b\beq 1)
    }
  }
  \end{derivation}
  \qedhere
  \]
\end{proof}  \section{Case Studies}
\label{sec:appendix:examples}

\subsection{pRHL-style Reasoning}
\label{sec:appendix:ex:prhl}

Here we elaborate on the conditional swap example that appeared in~\cref{sec:ex:prhl-style}.
By~\cref{rule:wp-samp}, for each index $i \in \{1, 2\}$, we have
\begin{align*}
  \gproves \WP{\m[i: \code{x:~}d_0]}{\distAs{\ip{x}{i}}{d_0}}
\end{align*}
By~\cref{rule:wp-conj}, we can combine the two programs together and derive
\begin{align*}
  \gproves \WP{\m[1: \code{x:~}d_0, 2: \code{x:~}d_0]}{\distAs{\Ip{x}{1}}{d_0} \land \distAs{\Ip{x}{2}}{d_0}}
\end{align*}
By~\cref{rule:c-unit-r},
\[\distAs{\Ip{x}{1}}{d_0} \land \distAs{\Ip{x}{2}}{d_0} \proves
  \CC{d_0} v. \sure{\Ip{x}{1} = v} \land  \CC{d_0} v. \sure{\Ip{x}{2} = v}
\]
Then, we can apply~\cref{rule:c-and}, which implies
\[
  \CC{d_0} v. \sure{\Ip{x}{1} = v} \land  \CMod{d_0} v. \sure{\Ip{x}{2} = v}
  \proves
  {\CC{d_0} v. (\sure{\Ip{x}{1} = v} \land  \sure{\Ip{x}{2} = v})}
\]
from which we can derive:
\begin{align*}
  \gproves \WP{\m[\I1: \code{x:~}d_0, \I2: \code{x:~}d_0]}{\CC{d_0} v.(\sure{\Ip{x}{1}=v} \land \sure{\Ip{x}{2}=v})}.
\end{align*}

For the rest of \p{prog1} (\p{prog2}):
similarly, by~\cref{rule:wp-samp}, for each index $i \in \{1, 2\}$, we have
\begin{align}
  &\gproves \WP{\m[i: \code{y:~}d_1(v)]}{\distAs{\Ip{y}{i}}{d_1(v)}}
  \label{helper:case:cond-swap:1}\\
  &\gproves \WP{\m[i: \code{z:~}d_2(v)]}{\distAs{\Ip{z}{i}}{d_2(v)}} .
  \label{helper:case:cond-swap:2}
\end{align}
By~\cref{rule:wp-frame},
\begin{align}
  &\distAs{\Ip{z}{i}}{d_2(v)} \gproves \WP{\m[i: \code{y:~}d_1(v)]}{\distAs{\Ip{z}{i}}{d_2(v)} * \distAs{\Ip{y}{i}}{d_1(v)}}
  \label{helper:case:cond-swap:3}\\
  &\distAs{\Ip{y}{i}}{d_1(v)} \gproves \WP{\m[i: \code{z:~}d_2(v)]}{\distAs{\Ip{y}{i}}{d_1(v)} * \distAs{\Ip{z}{i}}{d_2(v)}} .
  \label{helper:case:cond-swap:4}
\end{align}
Thus, applying~\cref{rule:wp-seq} to combine~\cref{helper:case:cond-swap:1}
and~\cref{helper:case:cond-swap:4}, we get
\begin{align}
  \gproves \WP{\m[\I1:  \code{y:~}d_1(v); \code{z:~}d_2(v)]}{\distAs{\Ip{y}{1}}{d_1(v)} * \distAs{\Ip{z}{1}}{d_2(v)}};
  \label{helper:case:cond-swap:5}
\end{align}
By applying~\cref{rule:wp-seq} to combine~\cref{helper:case:cond-swap:2}
and~\cref{helper:case:cond-swap:3}, we get
\begin{align}
  \gproves \WP{\m[\I2: \code{z:~}d_2(v); \code{y:~}d_1(v)]}{\distAs{\Ip{y}{2}}{d_1(v)} * \distAs{\Ip{z}{2}}{d_2(v)}} .
  \label{helper:case:cond-swap:6}
\end{align}
Then, by~\cref{rule:wp-bind}, we can derive
\begin{align*}
\sure{\Ip{x}{1}=v}
  &\gproves \WPv{\m[\I1:  \code{y:~}d_1(x); \code{z:~}d_2(x)]}{
  \distAs{\Ip{y}{1}}{d_1(v)} * \distAs{\Ip{z}{1}}{d_2(v)}} \\
\sure{\Ip{x}{2}=v}
  &\gproves \WP{\m[\I2: \code{z:~}d_2(v); \code{y:~}d_1(v)]}{
  \distAs{\Ip{y}{2}}{d_1(v)} * \distAs{\Ip{z}{2}}{d_2(v)}}
\end{align*}
Then, applying~\cref{rule:wp-conj} to combine the program at index 1 and 2,
we get
\begin{align*}
  \sure{\Ip{x}{1}=v} \land \sure{\Ip{x}{2}=v}
  &\gproves \WPv{\m[\I1: \code{y:~}d_1(x); \code{z:~}d_2(x);
                    \I2: \code{z:~}d_2(v); \code{y:~}d_1(v)]}{
  \distAs{\Ip{y}{1}}{d_1(v)} * \distAs{\Ip{z}{1}}{d_2(v)} *
  \distAs{\Ip{y}{2}}{d_1(v)} * \distAs{\Ip{z}{2}}{d_2(v)}}
\end{align*}
Also, we have
\begin{align*}
  \sure{\Ip{x}{1}=v} \land \sure{\Ip{x}{2}=v}
  &\gproves \WPv{\m[\I1: \code{y:~}d_1(x); \code{z:~}d_2(x);
                    \I2: \code{z:~}d_2(v); \code{y:~}d_1(v)]}{
   \sure{\Ip{x}{1}=v} \land \sure{\Ip{x}{2}=v}
 }
\end{align*}
by~\cref{rule:wp-frame},
where $\sure{\Ip{x}{1}=v} \land \sure{\Ip{x}{2}=v} \proves \sure{\Ip{x}{1}=\Ip{x}{2}}$. Therefore,
we have
\begin{align*}
  \sure{\Ip{x}{1}=v} \land \sure{\Ip{x}{2}=v}
  &\gproves \WPv{\m[\I1: \code{y:~}d_1(x); \code{z:~}d_2(x);
                    \I2: \code{z:~}d_2(v); \code{y:~}d_1(v)]}{
   \sure{\Ip{x}{1} = \Ip{x}{2}}
 }
\end{align*}
By~\cref{rule:wp-conj},
\begin{align*}
  &\sure{\Ip{x}{1}=v} \land \sure{\Ip{x}{2}=v}\\
  {}\gproves {}&\WPv{\m[\I1: \code{y:~}d_1(x); \code{z:~}d_2(x);
                        \I2: \code{z:~}d_2(v); \code{y:~}d_1(v)]}{
   \sure{\Ip{x}{1} = \Ip{x}{2}}
   \land \left(\distAs{\Ip{y}{1}}{d_1(v)} * \distAs{\Ip{z}{1}}{d_2(v)}
   * \distAs{\Ip{y}{2}}{d_1(v)} * \distAs{\Ip{z}{2}}{d_2(v)}\right)
 }
\end{align*}
By~\cref{rule:sure-and-star}, we get
\begin{align*}
  &\sure{\Ip{x}{1}=v} \land \sure{\Ip{x}{2}=v} \\
  {} \gproves {} &\WPv{\m[\I1: \code{y:~}d_1(x); \code{z:~}d_2(x);
                          \I2: \code{z:~}d_2(v); \code{y:~}d_1(v)]}{
   \sure{\Ip{x}{1} = \Ip{x}{2}}
   * \distAs{\Ip{y}{1}}{d_1(v)} * \distAs{\Ip{z}{1}}{d_2(v)}
   * \distAs{\Ip{y}{2}}{d_1(v)} * \distAs{\Ip{z}{2}}{d_2(v)}
 }
\end{align*}

Now, we can proceed with the derivation explained in~\cref{sec:ex:prhl-style}.
\begin{derivation}
\infer*[right=\ref{rule:rl-convex}]{
\infer*[Right=\ref{rule:c-wp-swap}]{
\infer*[Right=\ref{rule:c-cons}]{
\infer*[Right=\ref{rule:rl-merge}]{
\infer*[Right=\ref{rule:coupling}]{
  \forall v\st
  \sure{\Ip{x}{1}=v} \land \sure{\Ip{x}{2}=v}
  \gproves
  \WP {\m[\I1: t_1, \I2: t_2]}*{
  \begin{matrix*}[l]
    \cpl{\Ip{x}{1} = \Ip{x}{2}} *
    \distAs{\Ip{y}{1}}{d_1(v)} *
    \distAs{\Ip{y}{2}}{d_1(v)} *
    {}\\
    \distAs{\Ip{z}{1}}{d_2(v)} *
    \distAs{\Ip{z}{2}}{d_2(v)}
  \end{matrix*}
  }
}{
  \forall v\st
  \sure{\Ip{x}{1}=v} \land \sure{\Ip{x}{2}=v}
  \gproves
  \WP {\m[\I1: t_1, \I2: t_2]} {
    \cpl{\Ip{x}{1} = \Ip{x}{2}} *
    \cpl{\Ip{y}{1} = \Ip{y}{2}} *
    \cpl{\Ip{z}{1} = \Ip{z}{2}}
  }
}}{
  \forall v\st
  \sure{\Ip{x}{1}=v} \land \sure{\Ip{x}{2}=v}
  \gproves
  \WP {\m[\I1: t_1, \I2: t_2]} {
    \cpl{\Ip{x}{1} = \Ip{x}{2} \land
         \Ip{y}{1} = \Ip{y}{2} \land
         \Ip{z}{1} = \Ip{z}{2}}
  }
}}{
  \CC{d_0} v.(\sure{\Ip{x}{1}=v} \land \sure{\Ip{x}{2}=v})
  \gproves
  \CC{d_0} v.
  \WP {\m[\I1: t_1, \I2: t_2]} {
    \cpl{\Ip{x}{1} = \Ip{x}{2} \land
         \Ip{y}{1} = \Ip{y}{2} \land
         \Ip{z}{1} = \Ip{z}{2}}
  }
}}{
  \CC{d_0} v.(\sure{\Ip{x}{1}=v} \land \sure{\Ip{x}{2}=v})
  \gproves
  \WP {\m[\I1: t_1, \I2: t_2]} {
  \CC{d_0} v.
    \cpl{\Ip{x}{1} = \Ip{x}{2} \land
         \Ip{y}{1} = \Ip{y}{2} \land
         \Ip{z}{1} = \Ip{z}{2}}
  }
}}{
  \CC{d_0} v.(\sure{\Ip{x}{1}=v} \land \sure{\Ip{x}{2}=v})
  \gproves
  \WP {\m[\I1: t_1, \I2: t_2]} {
    \cpl{\Ip{x}{1} = \Ip{x}{2} \land
         \Ip{y}{1} = \Ip{y}{2} \land
         \Ip{z}{1} = \Ip{z}{2}}
  }
}
\end{derivation}

Last, with~\cref{rule:wp-seq}, we have
\begin{align*}
    \gproves
    \WP {\m[\I1: \code{prog1}, \I2: \code{prog2}]} {
    \cpl{\Ip{x}{1} = \Ip{x}{2} \land
         \Ip{y}{1} = \Ip{y}{2} \land
         \Ip{z}{1} = \Ip{z}{2}}
  }
\end{align*} 

\subsection{One-time Pad (Relational)}
\label{sec:appendix:examples:onetimerel}
\label{sec:appendix:ex:otp-rel}

  To wrap up the proof of \Cref{sec:overview}
we first observe that the assertion~$P$
of~(\ref{ex:xor:start}) can be easily obtained by
using the WP rules for assignments and sequencing, proving:
\[
    \True\withp{\m{\permap}}
    \proves
    \WP {\m<
      \I1: \code{encrypt()},
      \I2: \code{c:~Ber(1/2)}
    >}*{
      \begin{pmatrix}
        \distAs{\Ip{k}{1}}{\Ber{\onehalf}}
          *
        \distAs{\Ip{m}{1}}{\Ber{p}}
          *
        \distAs{\Ip{c}{2}}{\Ber{\onehalf}}
          * {}\\
        \sure{\Ip{c}{1} = \Ip{k}{1} \xor \Ip{m}{1}}
      \end{pmatrix}
      \withp{\m{\permap}}
    }
\]
where $\m{\permap} = \m[\Ip{k}{1}:1,\Ip{m}{1}:1,\Ip{c}{1}:1,\Ip{c}{2}:1]$
(\ie we have full permissions on the variables we modify).

We can prove the entailment:
\[
  \CC{\Ber{p}} v.
  \left(
    \sure{\Ip{m}{1}=v} *
    \begin{pmatrix}
    \distAs{\Ip{k}{1}}{\Ber{\onehalf}}
    \\ {}*
    \distAs{\Ip{c}{2}}{\Ber{\onehalf}}
    \end{pmatrix}
  \right)
  \proves
  \CC{\Ber{p}} v.
    \left(
      \sure{\Ip{m}{1}=v}
      *
      \begin{cases}
        \cpl{ \Ip{k}{1} = \Ip{c}{2} }     \CASE v=0 \\
        \cpl{ \Ip{k}{1} = \neg\Ip{c}{2} } \CASE v=1
      \end{cases}
    \right)
\]
by using \ref{rule:c-cons}, which asks us to prove that the
two assertions inside the conditioning are in the entailment relation
for each value of~$v$.
This leads to these two cases:
\begin{align*}
  \sure{\Ip{m}{1}=0}
  * \distAs{\Ip{k}{1}}{\Ber{\onehalf}}
  * \distAs{\Ip{c}{2}}{\Ber{\onehalf}}
  & \proves
  \sure{\Ip{m}{1}=0} * \cpl{ \Ip{k}{1} = \Ip{c}{2} }
\\
  \sure{\Ip{m}{1}=1}
  * \distAs{\Ip{k}{1}}{\Ber{\onehalf}}
  * \distAs{\Ip{c}{2}}{\Ber{\onehalf}}
  & \proves
  \sure{\Ip{m}{1}=1} * \cpl{ \Ip{k}{1} = \neg\Ip{c}{2} }
\end{align*}
which are straightforward consequences of the two couplings we proved
in~\eqref{ex:xor-two-cpl}.

Finally, the assignment to \p{c} in \p{encrypt} generated the fact
$\sure{\Ip{c}{1} = \Ip{k}{1} \xor \Ip{m}{1}}$.
By routine propagation of this fact
(using \ref{rule:c-frame} and \ref{rule:sure-merge})
we can establish:
\begin{eqexplain}
  &
  \CC{\Ber{p}} v.
    \left(
      \sure{\Ip{m}{1}=v}
      *
      \begin{cases}
        \cpl{ \Ip{k}{1} = \Ip{c}{2} }     \CASE v=0 \\
        \cpl{ \Ip{k}{1} = \neg\Ip{c}{2} } \CASE v=1
      \end{cases}
    \right)
    *
    \sure{\Ip{c}{1} = \Ip{k}{1} \xor \Ip{m}{1}}
\whichproves
  \CC{\Ber{p}} v.
    \left(
      \sure{\Ip{m}{1}=v}
      *
      \begin{cases}
        \cpl{ \Ip{k}{1} = \Ip{c}{2} } * \sure{ \Ip{c}{1} = \Ip{k}{1} \xor 0 }
          \CASE v=0 \\
        \cpl{ \Ip{k}{1} = \neg\Ip{c}{2} } * \sure{ \Ip{c}{1} = \Ip{k}{1} \xor 1 }
          \CASE v=1
      \end{cases}
    \right)
\whichproves
  \CC{\Ber{p}} v.
    \left(
      \sure{\Ip{m}{1}=v}
      *
      \begin{cases}
        \cpl{ \Ip{c}{1} = \Ip{c}{2} } \CASE v=0 \\
        \cpl{ \Ip{c}{1} = \Ip{c}{2} } \CASE v=1
      \end{cases}
    \right)
\whichproves
  \CC{\Ber{p}} v.
        \cpl{ \Ip{c}{1} = \Ip{c}{2} }
\whichproves
  \cpl{ \Ip{c}{1} = \Ip{c}{2} }
  \byrule{rl-merge}
\end{eqexplain}

In particular, the entailments
\begin{align*}
\cpl{ \Ip{k}{1} = \Ip{c}{2} } * \sure{ \Ip{c}{1} = \Ip{k}{1} \xor 0 }
&\proves
\cpl{ \Ip{c}{1} = \Ip{c}{2} }
\\
\cpl{ \Ip{k}{1} = \neg\Ip{c}{2} } * \sure{ \Ip{c}{1} = \Ip{k}{1} \xor 1 }
&\proves
\cpl{ \Ip{c}{1} = \Ip{c}{2} }
\end{align*}
can be proved by applying \ref{rule:rl-sure-merge} and \ref{rule:rl-cons}.

\subsection{One-time Pad (Unary)}
\label{sec:appendix:ex:otp-unary}

  Similarly to the relational version, we can, using \ref{rule:wp-seq},\ref{rule:wp-assign} and \ref{rule:wp-samp}, easily show that:
\begin{equation*}
  \True\withp{\m{\permap}}
  \proves
  \WP {
   \m[\I1: \code{encrypt}()]
  }{
     \distAs{\Ip{k}{1}}{\Ber{\onehalf}} *
     \distAs{\Ip{m}{1}}{\Ber{p}} *
     \sure{\Ip{c}{1} = \Ip{k}{1} \xor \Ip{m}{1}}
  }.
\end{equation*}

We then show the crucial derivation of \cref{sec:ex:one-time-pad} in more detail.
\begin{eqexplain}
  &
  \distAs{\Ip{k}{1}}{\Ber{\onehalf}} *
  \distAs{\Ip{m}{1}}{\Ber{p}} *
  \sure{\Ip{c}{1} = \Ip{k}{1} \xor \Ip{m}{1}}
\whichproves
  \CC{\Ber{p}} m.
  \bigl(
    \sure{\Ip{m}{1}=m} *
    \distAs{\Ip{k}{1}}{\Ber{\onehalf}} *
    \sure{\Ip{c}{1} = \Ip{k}{1} \xor \Ip{m}{1}}
  \bigr)
  \byrules{c-unit-r,c-frame}
\whichproves
  \CC{\Ber{p}} m.
  \bigl(
    \sure{\Ip{m}{1}=m} *
    \distAs{\Ip{k}{1}}{\Ber{\onehalf}} *
    \sure{\Ip{c}{1} = \Ip{k}{1} \xor m}
  \bigr)
  \byrules{sure-merge}
\whichproves
  \CC{\Ber{p}} m.
  \Bigl(
    \sure{\Ip{m}{1}=m} *
    \CC{\Ber{\onehalf}} k.
    \bigl(
      \sure{\Ip{k}{1}=k} *
      \sure{\Ip{c}{1} = \Ip{k}{1} \xor m}
    \bigr)
  \Bigr)
  \byrules{c-unit-r,c-frame}
\whichproves
  \CC{\Ber{p}} m.
    \bigl(
      \sure{\Ip{m}{1}=m} *
      \CC{\Ber{\onehalf}} k.
        \sure{\Ip{k}{1}=k \land \Ip{c}{1} = k \xor m}
    \bigr)
  \byrules{sure-merge}
\whichproves
  \CC{\Ber{p}} m.
    \left(
    \sure{\Ip{m}{1}=m} *
    \begin{cases}
      \CC{\Ber{\onehalf}} k. \sure{\Ip{c}{1}=k} \CASE m=0
      \\
      \CC{\Ber{\onehalf}} k. \sure{\Ip{c}{1}=\neg k} \CASE m=1
    \end{cases}
    \right)
  \byrule{c-cons}
\whichproves
  \CC{\Ber{p}} m.
    \left(
    \sure{\Ip{m}{1}=m} *
    \begin{cases}
      \CC{\Ber{\onehalf}} k. \sure{\Ip{c}{1}=k} \CASE m=0
      \\
      \CC{\Ber{\onehalf}} k. \sure{\Ip{c}{1}=k} \CASE m=1
    \end{cases}
    \right)
  \byrule{c-transf}
\whichproves
  \CC{\Ber{p}} m.
    \bigl(
      \sure{\Ip{m}{1}=m} *
      \CC{\Ber{\onehalf}} k. \sure{\Ip{c}{1}=k}
    \bigr)
\whichproves
  \CC{\Ber{p}} m.
  \CC{\Ber{\onehalf}} k.
    \bigl(
      \sure{\Ip{m}{1}=m} *
      \sure{\Ip{c}{1}=k}
    \bigr)
  \byrules{c-frame}
\whichproves
  \CC{\Ber{p}} m.
  \CC{\Ber{\onehalf}} k.
    \sure{\Ip{m}{1}=m \land \Ip{c}{1}=k}
  \byrules{sure-merge}
\whichproves
  \CC{\Ber{p} \pprod \Ber{\onehalf}} (m,k).
    \sure{(\Ip{m}{1},\Ip{c}{1})=(m,k)}
  \byrules{c-assoc}
\whichproves
  \distAs{(\Ip{m}{1},\Ip{c}{1})}{(\Ber{p} \pprod \Ber{\onehalf})}
  \byrule{c-unit-r}
\whichproves
  \distAs{\Ip{m}{1}}{\Ber{p}} *
  \distAs{\Ip{c}{1}}{\Ber{\onehalf}}
  \byrule{prod-split}
\end{eqexplain}

The application of \ref{rule:c-transf} to the case with $m=1$ is as follows:
\[
\infer{
  \forall b \in \set{0,1}.
    \Ber{\onehalf}(b)=\Ber{\onehalf}(\neg b)
}{
  \CC{\Ber{\onehalf}} k. \sure{\Ip{c}{1}=\neg k}
  \proves
  \CC{\Ber{\onehalf}} k. \sure{\Ip{c}{1}=k}
}
\] 
\begin{figure*}
  \adjustfigure[\small]\setlength\tabcolsep{0pt}\begin{tabular*}{\textwidth}{
    @{\extracolsep{\fill}}
    *{4}{p{\textwidth/4}}@{}
  }
\begin{sourcecode*}
def BelowMax($x$,$S$):
  repeat $N$:
    q:~$\prob_S$
    r':=r
    r := r' || q >= $x$
\end{sourcecode*}
&
\begin{sourcecode*}
def AboveMin($x$,$S$):
  repeat $N$:
    p:~$\prob_S$
    l':=l
    l := l' || p <= $x$
\end{sourcecode*}
&
\begin{sourcecode*}
def BETW_SEQ($x$, $S$):
  BelowMax($x$,$S$);
  AboveMin($x$,$S$);
  d := r && l
\end{sourcecode*}
\\
\begin{sourcecode*}
def BETW($x$,$S$):
  repeat $2 N$:
    s:~$\prob_S$
    l':=l
    l := l' || s <= $x$
    r':=r
    r := r' || s >= $x$
  d := r && l
\end{sourcecode*}
&
\begin{sourcecode*}
def BETW_MIX($x$, $S$):
  repeat $N$:
    p:~$\prob_S$
    l':=l
    l := l' || p <= $x$
    q:~$\prob_S$
    r':=r
    r := r' || q >= $x$
  d := r && l
\end{sourcecode*}
&
\begin{sourcecode*}
def BETW_N($x$,$S$):
  repeat $N$:
    s:~$\prob_S$
    l':=l
    l := l' || s <= $x$
    r':=r
    r := r' || s >= $x$
  d := r && l
\end{sourcecode*}
\end{tabular*}   \caption{Stochastic dominance examples: composing Monte Carlo algorithms in different ways. All variables are initially 0.}
  \label{fig:between-code-repeat}
\end{figure*}

\subsection{Markov Blanket and Variable Elimination}
\label{sec:appendix:ex:markov-blanket}

  In probabilistic reasoning, introducing conditioning is easy,
but deducing unconditional facts from conditional ones is not immediate.
The same applies to the \supercond\ modality: by design, one cannot eliminate it for free.
Crucial to \thelogic's expressiveness is the inclusion of rules that can
soundly derive unconditional information from conditional assertions.

In this example we show how \thelogic\ is able to derive a common
tool used in Bayesian reasoning to simplify conditioning as much as possible
through the concept of a \emph{Markov Blanket}.

For concreteness, consider the program:
\begin{center}
\code{x1:~$\dist_1$;
x2:~$\dist_2(\p{x1})$;
x3:~$\dist_3(\p{x2})$}
\end{center}
The program describes a Markov chain of three variables.
One way of interpreting this pattern is that the joint output distribution
is described by the program as a product of conditional distributions:
the distribution over \p{x2} is described conditionally on \p{x1},
and the one of \p{x3} conditionally on \p{x2}.
This kind of dependencies are ubiquitous in, for instance, hidden Markov models and Bayesian network representations of distributions.

A crucial tool for the analysis of such models is the concept of a
\emph{Markov Blanket} of a variable \p{x}: the set of variables that are direct dependencies of \p{x}.
Clearly \p{x3} depends on \p{x2} and, indirectly, on \p{x1}.
However, Markov chains enjoy the memorylessness property:
when fixing a variable in the chain, the variables that follow it are independent from the variables that preceded it.
For our example this means that if we condition on \p{x2},
\p{x1} and \p{x3} are independent (\ie we can ignore the indirect dependencies).

In \thelogic\ we can characterize the output distribution with the assertion
\[
  \CC{\dist_1} v_1. \Bigl(
    \sure{\p{x1}=v_1} *
    \CC{\dist_2(v_1)} v_2. \bigl(
      \sure{\p{x2}=v_2} *
      \distAs{\p{x3}}{\dist_3(v_2)}
    \bigr)
  \Bigr)
\]
Note how this postcondition represents the output distribution
as implicitly as the program does.
We want to transform the assertion into:
\[
  \CC{\prob_2} v_2.
  \bigl(
    \sure{\p{x2}=v_2} *
    \distAs{\p{x1}}{\prob_1(v_2)} *
    \distAs{\p{x3}}{\dist_3(v_2)}
  \bigr)
\]
for appropriate $\prob_2$ and $\prob_1$.
This isolates the conditioning to the direct dependency of \p{x1}
and keeps full information about \p{x3},
available for further manipulation down the line.

In probability theory, the proof of memorylessness is an application
of Bayes' law: we are computing
the distribution of \p{x1} conditioned on \p{x2},
from the distribution of \p{x2} conditioned on \p{x1}.

In \thelogic\ we can produce the transformation using the joint conditioning rules:
\begin{eqexplain}
  &
  \CC{\dist_1} v_1. \Bigl(
    \sure{\p{x1}=v_1} *
    \CC{\dist_2(v_1)} v_2. \bigl(
      \sure{\p{x1}=v_2} *
      \distAs{\p{x3}}{\dist_3(v_2)}
    \bigr)
  \Bigr)
\whichproves
  \CC{\dist_1} v_1. \Bigl(
    \CC{\dist_2(v_1)} v_2. \bigl(
      \sure{\p{x1}=v_1} *
      \sure{\p{x1}=v_2} *
      \distAs{\p{x3}}{\dist_3(v_2)}
    \bigr)
  \Bigr)
  \byrules{c-frame}
\whichproves
  \CC{\prob_0} (v_1,v_2). \bigl(
      \sure{\p{x1}=v_1} *
      \sure{\p{x2}=v_2} *
      \distAs{\p{x3}}{\dist_3(v_2)}
  \bigr)
  \byrules{c-assoc}
\whichproves
  \CC{\prob_2} v_2. \Bigl(
    \CC{\prob_1(v_2)} v_1.
    \bigl(
      \sure{\p{x1}=v_1} *
      \sure{\p{x2}=v_2} *
      \distAs{\p{x3}}{\dist_3(v_2)}
    \bigr)
  \Bigr)
  \byrules{c-unassoc}
\whichproves
  \CC{\prob_2} v_2. \Bigl(
    \sure{\p{x2}=v_2} *
    \CC{\prob_1(v_2)} v_1.
    \bigl(
      \sure{\p{x1}=v_1} *
      \distAs{\p{x3}}{\dist_3(v_2)}
    \bigr)
  \Bigr)
  \byrules{sure-str-convex}
\whichproves
  \CC{\prob_2} v_2. \bigl(
    \sure{\p{x2}=v_2} *
    \distAs{\p{x1}}{\prob_1(v_2)} *
    \distAs{\p{x3}}{\dist_3(v_2)}
  \bigr)
  \byrules{c-extract}
\end{eqexplain}
where
$
  (\dist_1 \fuse \dist_2) = \prob_0 =
  \bind(\prob_2,\prob_1).
$
The existence of such $\prob_2$ and $\prob_1$ is a simple application
of Bayes' law:
$
  \prob_2(v_2) =
    \Sum_{v_1 \in \Val} \prob_0(v_1,v_2),
$
and
$
  \prob_1(v_2)(v_1) =
    \frac{\prob_0(v_1,v_2)}{\prob_2(v_2)}.
$

\subsection{Multi-party Secure Computation}
\label{sec:appendix:ex:multiparty}

  The idea of \emph{multi-party secure computation} is to allow~$N$
parties to compute a function~$f(x_1,\dots,x_N)$ of
some private data~$x_i$ owned by each party~$i$,
without revealing any more information about~$x_i$ than the output of~$f$
would reveal if computed centrally by a trusted party.
When $f$ is addition, a secure computation of~$f$ is useful, for example,
to compute the total number of votes without revealing who voted positively:
some information would leak (e.g., if the total is non-zero then \emph{somebody} voted positively) but only what is revealed by knowing the total and nothing more.

To achieve this objective, multi-party secure addition~(MPSAdd)
works by having the parties break their secret into~$N$ \emph{secret shares}
which individually look random, but the sum of which amounts to the original secret.
These secret shares are then distributed to the other parties so that each party knows an incomplete set of shares of the other parties.
Yet, each party can reliably compute the result of the function by computing a function of the received shares.

\Cref{fig:mpsadd} shows the MPSAdd algorithm, for the case~$N=3$,
as modelled in~\cite{barthe2019probabilistic}.
The algorithm works as follows.
Each party~$i$ knows its secret~$x_i$.
All the parties agree on an appropriate prime number~$p$ to use,
and want to compute the sum of all the secrets (modulo~$p$).
The algorithm goes through three phases:
\begin{itemize}
  \item\emph{Phase~1:}
    Each party~$i$ computes its three secret shares in row \p{r[$i$][-]}
    by generating two independent random
    numbers~\p{r[$i$][1]}, \p{r[$i$][2]} in~$\Zp$
    drawing from the \emph{uniform} distribution~$\UnifZp$.
The third share \p{r[$i$][3]} is chosen so that the sum of the shares
    amounts to the secret~\p{x}$_i$.

    Then columns are communicated so that each party~$i$
    receives column \p{r[-][$j$]} for every $j\ne i$.
    Each party now knows 2 out of 3 of the shares of each other party.
  \item\emph{Phase~2:}
    Each party~$i$
    computes, for every $j\ne i$,
    \p{s[$j$]} as the sum of the column \p{r[-][$j$]},
    and sends it to the other parties.
  \item\emph{Phase~3:}
    Each party now knows \p{s[-]},
    the sum of which is the sum of the secrets.
\end{itemize}

\citet{barthe2019probabilistic} provide a partial proof of the example:
\begin{enumerate*}
\item
  They only verify uniformity and independence from the input of the
  secret shares \p{r} (with a proof that involves ad-hoc axioms);
\item
  As PSL can only represent (unconditional) independence,
  the proof cannot state anything
  useful about the other values \p{s} that are circulated to compute the sum;
  in principle they can also leak too much information,
  but will not be independent of the secret since they are supposed to
  disclose their sum.
\end{enumerate*}

In this section we set out to prove the stronger property that,
by the end of the computation, nothing is revealed by the values
party~$i$ received from the other parties other than the sum.
We focus on party~$1$ as the specifications and proofs for the other parties
are entirely analogous.

\begin{figure}
  \centering \begin{tabular}{c}
  \begin{sourcecode}[gobble=2]
  def MPSAdd:
    // Phase 1
    for i in [1,2,3]:
      r[i][1] :~ $\UnifZp$ // Uniform sample from $\color{codecomment}\Zp$
      r[i][2] :~ $\UnifZp$
      r[i][3] := x$_i$ - r[i][1] - r[i][2] mod $p$
    // Phase 2
    for i in [1,2,3]:
      s[i] := r[1][i] + r[2][i] + r[3][i] mod $p$
    // Phase 3
    sum := s[1] + s[2] + s[3] mod $p$
  \end{sourcecode}\end{tabular}
  \caption{Multi-party secure addition.}
  \label{fig:mpsadd}
\end{figure}

We observe that the above goal can be formalized using two
very different judgments, one unary and one relational.

\smallskip

The \textbf{\emph{unary specification}} says that,
conditionally on the secret of party~$i$, and the sum of the other secrets,
all the values received by~$i$ (we call this the \emph{view} of~$i$)
are independent from the secrets of the other parties;
moreover the learned components of~\p{r} are uniformly distributed.

Formally, the view of party~$1$ is the vector:
\[
  \p{view}_1 = (
    \p{r[1][-]},\p{r[2][2]}, \p{r[2][3]},
    \p{r[3][2]}, \p{r[3][3]},
    \p{s[-]},\p{sum}
  )
\]
The unary specification would then assume an arbitrary distribution~$\prob_0$
of the secrets (making no assumption about their independence),
and asserting in the postcondition that $\p{view}_1$ and $(\p x_1,\p x_2)$ are
independent conditionally on $\p x_1$ and $\p x_2+\p x_3$
(\ie conditionally on the secret of party~1 and the total sum of the secrets).
\begin{equation}
  (\distAs{(\p{x}_1, \p{x}_2, \p{x}_3)}{\prob_0})
  \withp{\m{\permap}}
  \proves
  \WP {\p{MPSAdd}}*{
    \CC {\tilde{\prob}_0} {(v_1, v_{23})}.
    \begin{grp}
      \sure{\p x_1 = v_1 \land (\p x_2 + \p x_3)=v_{23}} * {}
      \\
      \distAs{\p{view}_1}{U(v_1,v_{23})} *
      \E \prob_{23}.\distAs{(\p x_2,\p x_3)}{\prob_{23}}
    \end{grp}
  }
\label{multiparty:unary:goal}
\end{equation}
For readability we omit the indices on the term and variables
as they are $\at{\I1}$ everywhere;
we also implicitly interpret equalities and sums to be modulo~$p$.
Here $U(v_1,v_{23})$ distributes
the components of \p{r} in $\p{view}_1$ uniformly (except for \p{r[1][3]}),
and $ {
  \tilde{\prob}_0 = \left(\DO{(x_1,x_2,x_3) <- \prob_0; \return (x_1,x_2+x_3)}\right)
} $;
moreover $\m{\permap}$ contains all the necessary permissions for the
assignments in \p{MPSAdd}.
Note how the conditioning on the sum represents the expected leakage of
the multi-party computation.

\smallskip

The \textbf{\emph{relational specification}} says that
    when running the program from two initial states
    differing only in the secrets of the other parties,
    but not in their sum,
    the views of party~$i$ would be distributed in the same way.
As a binary judgement, the specification can be formalized as:
\begin{equation}
  \cpl*{
  \begin{conj*}
    \p x_1\at{\I1} = \p x_1\at{\I2}
    \land
    (\p x_2+\p x_3)\at{\I1} = (\p x_2+\p x_3)\at{\I2}
  \end{conj*}
  }
  \withp{\m{\permap}}
  \proves
  \WP {\m<1:\p{MPSAdd},2:\p{MPSAdd}>}*{
    \cpl*{
    \begin{conj*}
      \p x_1\at{\I1} = \p x_1\at{\I2}
      \land
      (\p x_2+\p x_3)\at{\I1} = (\p x_2+\p x_3)\at{\I2}
      \land
      \p{view}_1\at{\I1} = \p{view}_1\at{\I2}
    \end{conj*}
    }
  }
\label{multiparty:rel:goal}
\end{equation}

\medskip
As a first result, we show that \thelogic{} can produce a proof for both specifications, one in unary style, the other in a relational-lifting style.
This demonstrates how \thelogic{} supports both styles uniformly,
making it possible to pick and choose the one that fits the prover's intuition
best, or that turns out to be easier to carry out.

Having two very different specification for the same property,
however, begs the question of whether the two specifications are really
equivalent; moreover, we would like to make the choice of how to \emph{prove}
the property independent of how the property is represented.
This is important, for example, because one proof strategy (\eg the relational)
might be more convenient for some program, but we might then want to reuse
the proven specification in a larger proof that might be conducted in a different style (\eg unary).
Our second key result is that we show how we can derive one spec from the other
\emph{within} \thelogic{}, thus showing that picking a style for proving the program does not inhibit the reuse of the result in a different style.
This also illustrates the fitness of \thelogic{} as a tool for abstract
meta-level reasoning:
in pRHL or PSL/Lilac,
the adequacy of a specification (or conversion between styles)
needs to be proven outside of the logic.
In \thelogic{} this can happen within the same logic that supports the proofs of the programs.

The rest of the section is devoted to substantiating these claims,
providing \thelogic{} proofs for:
\begin{enumerate}
  \item the unary specification;
  \item the relationa specification (independently of the unary proof);
  \item the equivalence of the two specifications.
\end{enumerate}
Although the third item would spare us from proving one of the first two,
we provide direct proofs in the two styles to provide a point of comparison
between them.

\subsubsection{Proof of the unary specification}
As a first step, we can apply the rules for loops and assignments to obtain
the postcondition~$Q$:
\begin{align*}
  Q &= X * R_{12} * R_3 * S * \var{Sum}
  &
  R_{12} &=
    \Sep_{i\in\set{1,2,3}} (
      \distAs{\p{r[$i$][1]}}{\UnifZp}
      *
      \distAs{\p{r[$i$][2]}}{\UnifZp}
    )
  \\
  X &= \distAs{(\p{x}_1, \p{x}_2, \p{x}_3)}{\prob_0}
  &
  R_{3} &=
    \Sep_{i\in\set{1,2,3}}
      \sure[\big]{\p{r[$i$][3]} = \p x_i-\p{r[$i$][1]}-\p{r[$i$][2]}}
  \\
  \var{Sum} &= \sure{ \p{sum} = \p{s[1]}+\p{s[2]}+\p{s[3]} }
  &
  S &= \sure*{
        \LAnd_{i\in\set{1,2,3}}
          \p{s[$i$]} = \p{r[1][$i$]}+\p{r[2][$i$]}+\p{r[3][$i$]}
      }
\end{align*}
Now the goal is to show that $Q$ entails the postcondition of~\eqref{multiparty:unary:goal}.
As a first step we transform~$X$ into $\distAs{(\p{x}_1, \p{x}_2 + \p{x}_2)}{\prob}$ by \ref{rule:dist-fun}.
Then we condition on $(\p{x}_1, \p{x}_2 + \p{x}_3, \p{x}_2 \p{x}_3)$ and
the variables in $R_{12}$, obtaining:
\[
  \CC{\prob'} (v_1,v_{23},v_2,v_3).
  \begin{grp}
    \sure{\p{x}_1 = v_1 \land (\p{x}_2 + \p{x}_3) = v_{23} \land (\p x_2,\p x_3) = (v_2, v_3)} *
    {}\\
    \CC{\UnifZp} u_{11}. \CC{\UnifZp} u_{12}.
    \CC{\UnifZp} u_{21}. \CC{\UnifZp} u_{22}.
    \CC{\UnifZp} u_{31}. \CC{\UnifZp} u_{32}.
    {}\\\qquad
      \sure*{
        \begin{conj*}
\p{r[1][1]} = u_{11} \land*
          \p{r[1][2]} = u_{12} \land*
            \p{r[1][3]} = v_1 - u_{11} - u_{12}
          \land
\p{r[2][2]} = u_{22} \land*
            \p{r[2][3]} = v_2 - u_{21} - u_{22}
          \land
\p{r[3][2]} = u_{32} \land*
            \p{r[3][3]} = (v_{23}-v_2) - u_{31} - u_{32}
          \land
          \p{s[1]} = u_{11} + u_{21} + u_{31}
          \land
          \p{s[2]} = u_{12} + u_{22} + u_{32}
          \land
          \p{s[3]} = v_1 - u_{11} - u_{12} + v_2 - u_{21} - u_{22} + (v_{23} - v_2) - u_{31} - u_{32}
          \land
          \p{sum} = \p{s[1]} + \p{s[2]} + \p{s[3]}
        \end{conj*}
      }
  \end{grp}
\]
Here $\prob' = \DO{ (x_1,x_2,x_3) <- \prob_0; \return (x_1,x_2+x_3,x_2) } $.
We already weakened the assertion by forgetting the information about
\p{r[2][1]} and \p{r[3][1]}, which are not part of $\p{view}_1$.

Now we perform a change of variables thanks to \cref{rule:c-transf},
to express our equalities in terms of
$u_{21}' = u_{21}-v_2$ instead of $u_{21}$ and
$u_{31}' = u_{31}-(v_{23}-v_2)$ instead of $u_{31}$.
To justify the change we simply observe that, for all $n\in\Zp$,
the function $ f_n(u) = u-n \mod p $ is a bijection and
$ \UnifZp \circ \inv{f_n} = \UnifZp $.
This gives us, with some simple arithmetic simplifications:
\[
  \CC{\prob'} (v_1,v_{23},v_2,v_3).
  \begin{grp}
    \sure{\p{x}_1 = v_1 \land (\p{x}_2 + \p{x}_3) = v_{23} \land (\p x_2,\p x_3) = (v_2, v_3)} *
    {}\\
    \CC{\UnifZp} u_{11}. \CC{\UnifZp} u_{12}.
    \CC{\UnifZp} u_{21}'. \CC{\UnifZp} u_{22}.
    \CC{\UnifZp} u_{31}'. \CC{\UnifZp} u_{32}.
    {}\\\qquad
      \sure*{
        \begin{conj*}
          \p{r[1][1]} = u_{11} \land*
          \p{r[1][2]} = u_{12} \land*
            \p{r[1][3]} = v_1 - u_{11} - u_{12}
          \land
          \p{r[2][2]} = u_{22} \land*
            \p{r[2][3]} = - u_{21}' - u_{22}
          \land
          \p{r[3][2]} = u_{32} \land*
            \p{r[3][3]} = - u_{31}' - u_{32}
          \land
          \p{s[1]} = u_{11} + u_{21}' + u_{31}'+v_{23}
          \land
          \p{s[2]} = u_{12} + u_{22} + u_{32}
          \land
          \p{s[3]} = v_1 - u_{11} - u_{12} - u_{21}' - u_{22} - u_{31}' - u_{32}
          \land
          \p{sum} = \p{s[1]} + \p{s[2]} + \p{s[3]}
        \end{conj*}
      }
  \end{grp}
\]
In particular we removed all dependencies on $v_2$ from the inner formula.
We can now apply \ref{rule:c-assoc} to collapse all the inner conditioning
into a single one:
\[
  \CC{\prob'} (v_1,v_{23},v_2,v_3).
  \begin{grp}
    \sure{\p{x}_1 = v_1 \land (\p{x}_2 + \p{x}_3) = v_{23} \land (\p x_2,\p x_3) = (v_2, v_3)} *
    {}\\
    \CC{U(v_1,v_{23})} \m{u}.
      \sure*{\p{view}_1 = \m{u }}
  \end{grp}
\]
where $U(v_1,v_{23}) = (\DO{\m{v}<-\UnifZp \pprod \dots \pprod \UnifZp; \return g(\m{v})})$ takes the six independent samples from $\UnifZp$ and returns
the values for each of the components of $\p{view}_1$ (which justifies the dependency on $v_1$ and $v_{23}$).
Finally, we split $\prob' = \bind(\prob, \krnl)$ obtaining:
\begin{eqexplain}
  &
  \CC{\prob'} (v_1,v_{23},v_2,v_3).
  \begin{grp}
    \sure{\p{x}_1 = v_1 \land (\p{x}_2 + \p{x}_3) = v_{23} \land (\p x_2,\p x_3) = (v_2, v_3)} *
    {}\\
    \CC{U(v_1,v_{23})} \m{u}.
      \sure*{\p{view}_1 = \m{u }}
  \end{grp}
  \whichproves
\CC{\prob'} (v_1,v_{23},v_2,v_3).
  \begin{grp}
    \sure{\p{x}_1 = v_1 \land (\p{x}_2 + \p{x}_3) = v_{23}} *
    {}\\
    \sure{(\p{x}_2, \p{x}_3) = (v_2, v_3)} *
    {}\\
      \distAs{\p{view}_1}{U(v_1,v_{23})}
  \end{grp}
  \byrules{sure-merge,c-unit-r}
  \whichproves
\CC{\tilde{\prob}_0} (v_1,v_{23}).
  \begin{grp}
  \sure{\p{x}_1 = v_1 \land (\p{x}_2 + \p{x}_3) = v_{23}} *
    {}\\
  \CC{\krnl(v_1,v_{23})} (v_2,v_3).
    {}\\\quad
  \bigl(
    \sure{(\p{x}_2, \p{x}_3) = (v_2, v_3)} *
    \distAs{\p{view}_1}{U(v_1,v_{23})}
  \bigr)
  \end{grp}
  \byrules{c-unassoc,sure-str-convex}
  \whichproves
\CC{\tilde{\prob}_0} (v_1,v_{23}).
  \begin{grp}
  \sure{\p{x}_1 = v_1 \land (\p{x}_2 + \p{x}_3) = v_{23}} *
    {}\\
    \distAs{(\p{x}_2, \p{x}_3)}{\krnl(v_1,v_{23})} *
    \distAs{\p{view}_1}{U(v_1,v_{23})}
  \end{grp}
  \byrules{c-extract}
  \whichproves
\CC {\tilde{\prob}_0} {(v_1, v_{23})}.
  \begin{grp}
    \sure{\p x_1 = v_1 \land (\p x_2 + \p x_3)=v_{23}} * {}
    \\
    \distAs{\p{view}_1}{U(v_1,v_{23})} *
    \E \prob_{23}.\distAs{(\p x_2,\p x_3)}{\prob_{23}}
  \end{grp}
\end{eqexplain}

\noindent
This gets us the desired postcondition, and concludes the proof.

\subsubsection{Proof of the relational specification}

We now want to prove the goal~\eqref{multiparty:rel:goal}
using the relational lifting technique,
\ie inducing a suitable coupling between the variables of the two components.
Starting from precondition
$
  \cpl*{
  \begin{conj*}
    \p x_1\at{\I1} = \p x_1\at{\I2}
    \land*
    (\p x_2+\p x_3)\at{\I1} = (\p x_2+\p x_3)\at{\I2}
  \end{conj*}
  },
$
we can proceed just like in the unary proof, blindly applying the loop and assignment rules obtaining (reusing the assertions of the previous section):
\begin{align*}
  \cpl*{
  \begin{conj*}
    \p x_1\at{\I1} = \p x_1\at{\I2}
    \land
    (\p x_2+\p x_3)\at{\I1} = (\p x_2+\p x_3)\at{\I2}
  \end{conj*}
  }
  *
  \begin{grp}
    R_{12}\at{\I1} * R_3\at{\I1} * S\at{\I1} * \var{Sum}\at{\I1}
    *{}\\
    R_{12}\at{\I2} * R_3\at{\I2} * S\at{\I2} * \var{Sum}\at{\I2}
  \end{grp}
\end{align*}
Now the main task is to couple the sources of randomness
in $R_{12}$ so that the views of~1 coincide in the two components.
The idea is that if we can induce a coupling where
$
  r[2][1]\at{\I1} =
    r[2][1]\at{\I2} + (\p{x}_2\at{\I1} - \p{x}_2\at{\I2})
$ and $
  r[3][1]\at{\I1} =
    r[3][1]\at{\I2} + (\p{x}_3\at{\I1} - \p{x}_3\at{\I2})
$
then we would obtain $\p{view}_1\at{\I1}=\p{view}_1\at{\I2}$.
To implement the idea we combine two observations:
\begin{enumerate*}
  \item we can induce the desired coupling if we are under conditioning
        of $\p{x}_2$ and $\p{x}_3$ on both indices;
        a conditioning that is already happening in the relational lifting
        of the precondition;
  \item the coupling is valid because addition of some constant modulo~$p$
        to a uniformly distributed variable, gives a uniformly distributed variable (an observation we also exploited in the unary proof).
\end{enumerate*}

Formally, we first unfold the definition of the relational lifting
to reveal the \supercond\ on $\p{x}_1$, $\p{x}_2$, and $\p{x}_3$,
and move the other resources inside the conditioning by \ref{rule:c-frame}:
\begin{eqexplain}
  &
  \cpl*{
  \begin{conj*}
    \p x_1\at{\I1} = \p x_1\at{\I2}
    \land
    (\p x_2+\p x_3)\at{\I1} = (\p x_2+\p x_3)\at{\I2}
  \end{conj*}
  }
  *
  \begin{grp}
    R_{12}\at{\I1} * R_3\at{\I1} * S\at{\I1} * \var{Sum}\at{\I1}
    *{}\\
    R_{12}\at{\I2} * R_3\at{\I2} * S\at{\I2} * \var{Sum}\at{\I2}
  \end{grp}
  \whichproves
\E \hat\prob.
  \CC{\hat\prob} \svec{
    v_1,&v_2,&v_3\\w_1,&w_2,&w_3
  }.
  \begin{grp}
    \pure{v_1=w_1 \land v_2+v_3=w_2+w_3}     * {}\\
    \sure{\p x_i\at{\I1} = v_i}_{i\in\set{1,2,3}} * {}\\
    \sure{\p x_i\at{\I2} = w_i}_{i\in\set{1,2,3}}
  \end{grp}
  *
  \begin{grp}
    R_{12}\at{\I1} * R_3\at{\I1} * S\at{\I1} * \var{Sum}\at{\I1}
    *{}\\
    R_{12}\at{\I2} * R_3\at{\I2} * S\at{\I2} * \var{Sum}\at{\I2}
  \end{grp}
  \whichproves
\E \hat\prob.
  \CC{\hat\prob} \svec{
    v_1,&v_2,&v_3\\w_1,&w_2,&w_3
  }.
  \begin{grp}
    \pure{v_1=w_1 \land v_2+v_3=w_2+w_3}* {}\\
    \sure{\p x_i\at{\I1} = v_i}_{i\in\set{1,2,3}} *
    R_{12}\at{\I1} * R_3\at{\I1} * S\at{\I1} * \var{Sum}\at{\I1}
    * {}\\
    \sure{\p x_i\at{\I2} = w_i}_{i\in\set{1,2,3}} *
    R_{12}\at{\I2} * R_3\at{\I2} * S\at{\I2} * \var{Sum}\at{\I2}
  \end{grp}
  \byrule{c-frame}
\end{eqexplain}

Now by using \ref{rule:c-cons}, we can induce a coupling inside the condtioning.
By using \ref{rule:coupling} we obtain:
\begin{eqexplain}
  R_{12}\at{\I1} * R_{12}\at{\I2}
  \whichproves*
\begin{grp}
  \Sep_{i\in\set{1,2,3}} (
    \distAs{\p{r[$i$][1]}\at{\I1}}{\UnifZp}
    *
    \distAs{\p{r[$i$][2]}\at{\I1}}{\UnifZp}
  )
  * {}\\
  \Sep_{i\in\set{1,2,3}} (
    \distAs{\p{r[$i$][1]}\at{\I2}}{\UnifZp}
    *
    \distAs{\p{r[$i$][2]}\at{\I2}}{\UnifZp}
  )
  \end{grp}
  \whichproves
\cpl*{
    \begin{conj*}
      \p{r[1][1]}\at{\I1}=\p{r[1][1]}\at{\I2}
      \land
      \LAnd_{i\in\set{1,2,3}}
        \p{r[i][2]}\at{\I1}= \p{r[i][2]}\at{\I2}
      \land
      \p{r[2][1]}\at{\I1}=\p{r[2][1]}\at{\I2}+(v_2-w_2)
      \land
      \p{r[3][1]}\at{\I1}=\p{r[3][1]}\at{\I2}+(v_3-w_3)
    \end{conj*}
  }
  \byrule{coupling}
\end{eqexplain}
The application of \ref{rule:coupling} is supported by the coupling
\[
  \prob = \DO{
    (u_1,\dots,u_6) <- \UnifZp^{(6)};
    \return
      \begin{grp}
      (u_1,u_2,u_3+(v_2-w_2),u_4,u_5+(v_3-w_3),u_6), \\
      (u_1,u_2,u_3,u_4,u_5,u_6)
      \end{grp}
  }
\]
where $ \UnifZp^{(6)} $ is the independent product of 6 $\UnifZp$.
The joint distribution~$\prob$ satisfies
$\prob \circ \inv{\proj_1} = \UnifZp^{(6)}$,
$\prob \circ \inv{\proj_2} = \UnifZp^{(6)}$;
note that $\UnifZp^{(6)}$ is the distribution of
$ (\p{r[1][1]},\p{r[1][2]},\p{r[2][1]},\p{r[3][1]},\p{r[3][2]}) $
at both indices $\I1$ and $\I2$ (by $R_{12}$).

Now we can merge the relational lifting with the other almost sure
facts we have under conditioning, simplify and apply \ref{rule:rl-convex}
to obtain the desired unconditional coupling:
\begin{eqexplain}
  &
  \CC{\hat\prob} \svec{
    v_1,&v_2,&v_3\\w_1,&w_2,&w_3
  }.
  \begin{grp}
    \pure{v_1=w_1 \land v_2+v_3=w_2+w_3}* {}\\
    \sure{\p x_i\at{\I1} = v_i}_{i\in\set{1,2,3}} *
    R_{12}\at{\I1} * R_3\at{\I1} * S\at{\I1} * \var{Sum}\at{\I1}
    * {}\\
    \sure{\p x_i\at{\I2} = w_i}_{i\in\set{1,2,3}} *
    R_{12}\at{\I2} * R_3\at{\I2} * S\at{\I2} * \var{Sum}\at{\I2}
  \end{grp}
  \whichproves
\CC{\hat\prob} \svec{
    v_1,&v_2,&v_3\\w_1,&w_2,&w_3
  }.
  \begin{grp}
    \pure{v_1=w_1 \land v_2+v_3=w_2+w_3}* {}\\
    \sure{\p x_i\at{\I1} = v_i}_{i\in\set{1,2,3}} *
     R_3\at{\I1} * S\at{\I1} * \var{Sum}\at{\I1}
    * {}\\
    \sure{\p x_i\at{\I2} = w_i}_{i\in\set{1,2,3}} *
     R_3\at{\I2} * S\at{\I2} * \var{Sum}\at{\I2}
   * {}\\
   \cpl*{
     \begin{conj*}
       \p{r[1][1]}\at{\I1}=\p{r[1][1]}\at{\I2}
       \land
       \LAnd_{i\in\set{1,2,3}}
         \p{r[i][2]}\at{\I1}= \p{r[i][2]}\at{\I2}
       \land
       \p{r[2][1]}\at{\I1}=\p{r[2][1]}\at{\I2}+(v_2-w_2)
       \land
       \p{r[3][1]}\at{\I1}=\p{r[3][1]}\at{\I2}+(v_3-w_3)
     \end{conj*}
   }
  \end{grp}
  \whichproves
\CC{\hat\prob} \svec{
    v_1,&v_2,&v_3\\w_1,&w_2,&w_3
  }.
   \cpl*{
    \begin{conj*}
      v_1=w_1 \land* v_2+v_3=w_2+w_3
      \land
      \p{r[1][1]}\at{\I1}=\p{r[1][1]}\at{\I2}
      \land
      \LAnd_{i\in\set{1,2,3}}
        \p{r[i][2]}\at{\I1}= \p{r[i][2]}\at{\I2}
      \land
      \p{r[2][1]}\at{\I1}=\p{r[2][1]}\at{\I2}+(v_2-w_2)
      \land
      \p{r[3][1]}\at{\I1}=\p{r[3][1]}\at{\I2}+(v_3-w_3)
      \land
      \LAnd_{i\in\set{1,2,3}}
       \p x_i\at{\I1} = v_i \land*
       \p x_i\at{\I2} = w_i
      \land
      \LAnd_{i\in\set{1,2,3}}
        (\p{r[$i$][3]} = \p x_i-\p{r[$i$][1]}-\p{r[$i$][2]})\at{\I1}
      \land \dots
    \end{conj*}
   }
  \byrule{rl-sure-merge}
  \whichproves
\CC{\hat\prob} \svec{
    v_1,&v_2,&v_3\\w_1,&w_2,&w_3
  }.
    \cpl*{
    \begin{conj*}
      \p x_1\at{\I1} = \p x_1\at{\I2}
      \land
      (\p x_2+\p x_3)\at{\I1} = (\p x_2+\p x_3)\at{\I2}
      \land
      \p{view}_1\at{\I1} = \p{view}_1\at{\I2}
    \end{conj*}
    }
  \byrule{rl-cons}
  \whichproves
\cpl*{
  \begin{conj*}
    \p x_1\at{\I1} = \p x_1\at{\I2}
    \land
    (\p x_2+\p x_3)\at{\I1} = (\p x_2+\p x_3)\at{\I2}
    \land
    \p{view}_1\at{\I1} = \p{view}_1\at{\I2}
  \end{conj*}
  }
  \byrule{rl-convex}
\end{eqexplain}

\subsubsection{Proof of equivalence of the two specifications}

Now we prove that one can prove the relational specification from the unary one,
and vice versa.

\paragraph{From unary to relational}
We want to show that
assuming the unary specification \eqref{multiparty:unary:goal} holds,
we can derive the relational specification \eqref{multiparty:rel:goal}.
We first need to bridge the gap between having one component
in the assumption and two in the consequence.
This is actually easy: the proof of \eqref{multiparty:unary:goal} is
completely parametric in the index chosen for the program,
so the same proof can prove two specifications, one where the index of the term is \I1 and one where it is \I2.
As a side note, this ``reindexing'' argument can be made into a rule of the logic following the LHC approach~\cite{d2022proving} but we do not do this
in this paper so we can focus on the novel aspects of the logic.
More formally, let~$P(\prob_0)$ and~$Q(\prob_0)$ be the precondition and the postcondition
of the unary specification \eqref{multiparty:unary:goal}.
Furthermore, let~$ \cpl{R_1} $ and $ \cpl{R_2} $ be the precondition and postcondition of the relational specification~\eqref{multiparty:rel:goal}.

First we can infer a binary spec from two unary instances of the unary spec,
for arbitrary $\prob_1,\prob_2 \of \Dist(\Zp^3)$:
\begin{equation}
  \infer* {P(\prob_1)\at{\I1}
    \proves
    \WP {\m[\I1: \p{MPSAdd}]} {
      Q(\prob_1)\at{\I1}
    }
    \\\\
    P(\prob_2)\at{\I2}
    \proves
    \WP {\m[\I2: \p{MPSAdd}]} {
      Q(\prob_2)\at{\I2}
    }
  }{
    \begin{conj}
    P(\prob_1)\at{\I1} \land P(\prob_2)\at{\I2}
    \end{conj}
    \proves
    \WP {\m<1: \p{MPSAdd}, 2: \p{MPSAdd}>}*{
      \begin{conj*}
      Q(\prob_1)\at{\I1} \land Q(\prob_2)\at{\I2}
      \end{conj*}
    }
  }
  \label{multiparty:u2b:dup}
\end{equation}

Recalling from \eqref{multiparty:unary:goal}
that $ {
  \tilde{\prob} = \left(\DO{(x_1,x_2,x_3) <- \prob; \return (x_1,x_2+x_3)}\right)
} $,
 the proof works by showing:
\begin{gather}
  \cpl{R_1}
  \proves
  \E \prob_1,\prob_2.
  (P(\prob_1)\at{\I1} \land P(\prob_2)\at{\I2})
  *\pure{\tilde\prob_1=\tilde\prob_2}
  \label{multiparty:u2b:prec}
  \\
Q(\prob_1)\at{\I1} \land Q(\prob_2)\at{\I2}
  *\pure{\tilde\prob_1=\tilde\prob_2}
\proves
  \cpl{R_2}
  \label{multiparty:u2b:post}
\end{gather}
From \eqref{multiparty:u2b:prec} we obtain $\prob_1$ and $\prob_2$
with which to instantiate \eqref{multiparty:u2b:dup},
and that the precondition of \eqref{multiparty:u2b:dup} holds.
Then, by applying \ref{rule:wp-cons} to the conclusion of
\eqref{multiparty:u2b:dup} and \eqref{multiparty:u2b:post}
we obtain the desired relational specification \eqref{multiparty:rel:goal}.

Entailment~\eqref{multiparty:u2b:prec} is obtained in the same way one proves
\ref{rule:rl-eq-dist}:
\begin{eqexplain}
  \cpl*{
  \begin{conj*}
    \p x_1\at{\I1} = \p x_1\at{\I2}
    \land
    (\p x_2+\p x_3)\at{\I1} = (\p x_2+\p x_3)\at{\I2}
  \end{conj*}
  }
  \whichproves*
\E {\hat\prob}.
  \pure{\hat\prob(R_1)=1} *
  \CC{\hat\prob} \left(
    \svec{
    v_1,&v_2,&v_3\\w_1,&w_2,&w_3
    }
  \right).
    \begin{conj}
      \sure{\p x_i\at{\I1}=v_i}_{i\in\set{1,2,3}} \land
      \sure{\p x_i\at{\I2}=w_i}_{i\in\set{1,2,3}}
    \end{conj}
  \whichproves
\E {\hat\prob}.
  \pure{\hat\prob(R_1)=1} *
  \begin{conj}
    \CC{\hat\prob} \left(
      \svec{
      v_1,&v_2,&v_3\\w_1,&w_2,&w_3
      }
    \right).
      \sure{\p x_i\at{\I1}=v_i}_{i\in\set{1,2,3}}
    \land
    \CC{\hat\prob} \left(
      \svec{
      v_1,&v_2,&v_3\\w_1,&w_2,&w_3
      }
    \right).
      \sure{\p x_i\at{\I2}=w_i}_{i\in\set{1,2,3}}
    \end{conj}
  \whichproves
\E {\hat\prob}.
  \pure{\hat\prob(R_1)=1} *
  \begin{conj}
    \CC{\hat\prob\circ\inv{\proj_1}} (v_1,v_2,v_3).
      \sure{\p x_i\at{\I1}=v_i}_{i\in\set{1,2,3}}
    \land
    \CC{\hat\prob\circ\inv{\proj_2}} (w_1,w_2,w_3).
      \sure{\p x_i\at{\I2}=w_i}_{i\in\set{1,2,3}}
  \end{conj}
  \whichproves
\E \prob_1,\prob_2.
  \begin{conj}
    \distAs{(\p{x}_1, \p{x}_2, \p{x}_3)\at{\I1}}{\prob_1}
    \land
    \distAs{(\p{x}_1, \p{x}_2, \p{x}_3)\at{\I2}}{\prob_2}
  \end{conj}
  * \pure{\tilde\prob_1=\tilde\prob_2}
\end{eqexplain}
The last step is justified by letting
$\prob_1 = \hat\prob\circ\inv{\proj_1}$ and
$\prob_2 = \hat\prob\circ\inv{\proj_2}$,
noting that $ \hat\prob(R_1)=1 $ implies $\tilde\prob_1=\tilde\prob_2$.

Finally, we prove \eqref{multiparty:u2b:post}.
Under the assumption $\tilde\prob_1=\tilde\prob_2$
we know that there is some coupling $\hat\prob$
such that
$ \hat\prob\circ\inv{\proj_1} = \tilde\prob_1 =
  \tilde\prob_2 = \hat\prob\circ\inv{\proj_2} $, and
$ \hat\prob(R)=1 $ where
$ R = \set{ ((v_1,v_{23}), (v_1,v_{23})) | v_1,v_{23} \in \Zp } $.
\begin{eqexplain}
  &
  \begin{grp}
    \CC {\tilde{\prob}_1} {(v_1, v_{23})}.
    \begin{grp}
      \sure{\p x_1\at{\I1} = v_1 \land (\p x_2 + \p x_3)\at{\I1}=v_{23}} * {}
      \\
      \distAs{\p{view}_1\at{\I1}}{U(v_1,v_{23})} *
      \E \prob_{23}.\distAs{(\p x_2,\p x_3)\at{\I1}}{\prob_{23}}
    \end{grp}
    \\
    \CC {\tilde{\prob}_2} {(w_1, w_{23})}.
    \begin{grp}
      \sure{\p x_1\at{\I2} = w_1 \land (\p x_2 + \p x_3)\at{\I2}=w_{23}} * {}
      \\
      \distAs{\p{view}_1\at{\I2}}{U(w_1,w_{23})} *
      \E \prob_{23}.\distAs{(\p x_2,\p x_3)\at{\I2}}{\prob_{23}}
    \end{grp}
  \end{grp}
  \whichproves
\begin{grp}
    \CC {\hat\prob} \svec{
      v_1,& v_{23}\\w_1,& w_{23}
    }.
    \begin{grp}
      \sure{\p x_1\at{\I1} = v_1 \land (\p x_2 + \p x_3)\at{\I1}=v_{23}} * {}
      \\
      \distAs{\p{view}_1\at{\I1}}{U(v_1,v_{23})} *
      \E \prob_{23}.\distAs{(\p x_2,\p x_3)\at{\I1}}{\prob_{23}}
    \end{grp}
    \\
    \CC {\hat\prob} \svec{
      v_1,& v_{23}\\w_1,& w_{23}
    }.
    \begin{grp}
      \sure{\p x_1\at{\I2} = w_1 \land (\p x_2 + \p x_3)\at{\I2}=w_{23}} * {}
      \\
      \distAs{\p{view}_1\at{\I2}}{U(w_1,w_{23})} *
      \E \prob_{23}.\distAs{(\p x_2,\p x_3)\at{\I2}}{\prob_{23}}
    \end{grp}
  \end{grp}
  \byrule{c-sure-proj}
  \whichproves
\CC {\hat\prob} \svec{
    v_1,& v_{23}\\w_1,& w_{23}
  }.
  \begin{grp}
    \sure{\p x_1\at{\I1} = v_1 \land (\p x_2 + \p x_3)\at{\I1}=v_{23}}
    \\
    \sure{\p x_1\at{\I2} = w_1 \land (\p x_2 + \p x_3)\at{\I2}=w_{23}}
    \\
    \distAs{\p{view}_1\at{\I1}}{U(v_1,v_{23})} *
    \E \prob_{23}.\distAs{(\p x_2,\p x_3)\at{\I1}}{\prob_{23}}
    \\
    \distAs{\p{view}_1\at{\I2}}{U(w_1,w_{23})} *
    \E \prob_{23}.\distAs{(\p x_2,\p x_3)\at{\I2}}{\prob_{23}}
    \\
    \pure{v_1=w_1 \land v_{23}=w_{23}}
  \end{grp}
  \byrule{c-and}
  \whichproves
\E \krnl_1, \krnl_2.
  \CC {\hat\prob} \svec{
    v_1,& v_{23}\\w_1,& w_{23}
  }.
  \begin{grp}
    \sure{\p x_1\at{\I1} = v_1 \land (\p x_2 + \p x_3)\at{\I1}=v_{23}}
    \\
    \sure{\p x_1\at{\I2} = w_1 \land (\p x_2 + \p x_3)\at{\I2}=w_{23}}
    \\
\cpl{\p{view}_1\at{\I1}=\p{view}_1\at{\I2}}
    \\
\distAs{(\p x_2,\p x_3)\at{\I1}}{\krnl_1(v_1, v_{23},w_1,w_{23})}
    \\
    \distAs{(\p x_2,\p x_3)\at{\I2}}{\krnl_2(v_1, v_{23},w_1,w_{23})}
    \\
    \pure{v_1=w_1 \land v_{23}=w_{23}}
  \end{grp}
  \byrules{c-skolem,coupling}
  \whichproves
\E \krnl_1, \krnl_2.
  \CC {\hat\prob} \svec{
    v_1,& v_{23}\\w_1,& w_{23}
  }.
  \begin{grp}
    \sure{\p x_1\at{\I1} = v_1 \land (\p x_2 + \p x_3)\at{\I1}=v_{23}}
    \\
    \sure{\p x_1\at{\I2} = w_1 \land (\p x_2 + \p x_3)\at{\I2}=w_{23}}
    \\
    \cpl{\p{view}_1\at{\I1}=\p{view}_1\at{\I2}}
    \\
    \CC{\krnl_1(v_1, v_{23},w_1,w_{23})} (v_2,v_3).
    \sure{(\p x_2,\p x_3)\at{\I1} = (v_1,v_2)}
    \\
    \CC{\krnl_2(v_1, v_{23},w_1,w_{23})} (w_2,w_3).
    \sure{(\p x_2,\p x_3)\at{\I2} = (w_1,w_2)}
    \\
    \pure{v_1=w_1 \land v_{23}=w_{23}}
  \end{grp}
  \byrules{c-unit-r}
  \whichproves
\E \hat\prob_1.
  \CC {\hat\prob_1} \svec{
    v_1,& v_{23},& v_2,& v_3\\w_1,& w_{23},& w_2,& w_3
  }.
  \begin{grp}
    \sure{\p x_1\at{\I1} = v_1 \land (\p x_2 + \p x_3)\at{\I1}=v_{23}}
    \\
    \sure{\p x_1\at{\I2} = w_1 \land (\p x_2 + \p x_3)\at{\I2}=w_{23}}
    \\
    \cpl{\p{view}_1\at{\I1}=\p{view}_1\at{\I2}}
    \\
    \sure{(\p x_2,\p x_3)\at{\I1} = (v_1,v_2)}
    \\
    \sure{(\p x_2,\p x_3)\at{\I2} = (w_1,w_2)}
    \\
    \pure{v_1=w_1 \land v_{23}=w_{23}}
  \end{grp}
  \byrules{c-frame,c-assoc}
  \whichproves
\E \hat\prob_2.
  \CC {\hat\prob_2} \svec{
    v_1,& v_2,& v_3\\w_1,& w_2,& w_3
  }.
  \begin{grp}
    \sure{\p x_i\at{\I1} = v_i}_{i\in\set{1,2,3}}
    \land
    \sure{\p x_i\at{\I2} = w_i}_{i\in\set{1,2,3}}
    \\
    \cpl{\p{view}_1\at{\I1}=\p{view}_1\at{\I2}}
    \\
    \pure{v_1=w_1 \land v_2+v_3=w_2+w_3}
  \end{grp}
  \byrules{c-transf}
  \whichproves
\E \hat\prob_2.
  \CC {\hat\prob_2} (\m{v},\m{w}).
  \begin{grp}
    \sure{\p x_i\at{\I1} = v_i}_{i\in\set{1,2,3}}
    \land
    \sure{\p x_i\at{\I2} = w_i}_{i\in\set{1,2,3}}
    \\
    \E {\hat\nu}.
      \CC{\hat\nu} (\m{u}_1,\m{u}_2).
        \sure{\p{view}_1\at{\I1}=\m{u}_1}\land
        \sure{\p{view}_1\at{\I2}=\m{u}_2}\land
        \pure{\m{u}_1=\m{u}_2}
    \\
    \pure{v_1=w_1 \land v_2+v_3=w_2+w_3}
  \end{grp}
  \bydef
  \whichproves
\E \hat\prob_3.
  \CC {\hat\prob_3} (\m{v},\m{w},\m{u}_1,\m{u}_2).
  \begin{grp}
    \sure{\p x_i\at{\I1} = v_i}_{i\in\set{1,2,3}}
    \land
    \sure{\p x_i\at{\I2} = w_i}_{i\in\set{1,2,3}}
    \\
    \sure{\p{view}_1\at{\I1}=\m{u}_1}\land
    \sure{\p{view}_1\at{\I2}=\m{u}_2}\land
    \\
    \pure{v_1=w_1 \land v_2+v_3=w_2+w_3 \land \m{u}_1=\m{u}_2}
  \end{grp}
  \byrules{c-skolem,c-assoc}
  \whichproves
\cpl*{
    \begin{conj*}
      \p x_1\at{\I1} = \p x_1\at{\I2}
      \land
      (\p x_2+\p x_3)\at{\I1} = (\p x_2+\p x_3)\at{\I2}
      \land
      \p{view}_1\at{\I1} = \p{view}_1\at{\I2}
    \end{conj*}
  }
  \bydef
\end{eqexplain}

\paragraph{From relational to unary}
The proof in the reverse direction follows the same strategy as the unary to relational direction, except the entailments between preconditions and postconditions are reversed; their proof steps in the previous section are however reversible, so we do not repeat the proof here.
There is one niggle:
what we can derive from the relational specification is the judgment
$
  P(\prob_0)\at{\I1} \land P(\prob_0)\at{\I2}
  \proves
  \WP {\m[\I1: \p{MPSAdd}, \I2: \p{MPSAdd}]}*{
    Q(\prob_0)\at{\I1} \land Q(\prob_0)\at{\I2}
  }
$
which in fact implies the unary specification (by ignoring component \I2),
but we have not provided \thelogic\ with rules to eliminate components.
The issue has been solved in LHC~\cite{d2022proving} with
the \emph{projection modality}
$ \P i. P = \fun a. \exists b\st P(a\m[i: b]) $,
which supports the principle
\[
  \infer*[lab=wp-proj]{
    P \proves \WP {\m{t}} {Q}
  }{
    \P i. P \proves \WP {(\m{t}\setminus i)} {\P i.Q}
  }
\]
The rule is sound in \thelogic\ model and could be used to fill this (small) gap
since
$
  P(\prob_0)\at{\I1}
  \proves
  \P\I2.\bigl(P(\prob_0)\at{\I1} \land P(\prob_0)\at{\I2}\bigr)
$ and
$
  \P\I2.\bigl(Q(\prob_0)\at{\I1} \land Q(\prob_0)\at{\I2}\bigr)
  \proves
  Q(\prob_0)\at{\I1}
$.
We did not emphasize projection in this paper to avoid distracting from the main novel contributions.

\subsection{Von Neumann Extractor}
\label{sec:appendix:ex:von-neumann}

  A randomness extractor is a mechanism that transforms a stream of
``low-quality'' randomness sources into a stream of ``high-quality''
randomness sources.
The Von Neumann extractor is perhaps the earliest instance of such mechanism,
and it converts a stream of independent coins with the same bias~$p$
into a stream of independent \emph{fair} coins.

In our language we can model the extractor, up to $N \in \Nat$ iterations,
as shown in \cref{fig:von-neumann}.
The program repeatedly flips two biased coins, and outputs the outcome of the first coin if the outcomes where different, otherwise it retries.
What we can prove is that the bits produced in \p{out} are independent fair coin flips: if we produced~$\ell$ bits we should be able to prove the postcondition
\[
  \var{Out}_\ell \is
  \distAs{\p{out}[0]\at{\I1}}{\Ber{\onehalf}} *
  \dots *
  \distAs{\p{out}[\ell-1]\at{\I1}}{\Ber{\onehalf}}.
\]
To know how many bits were produced, however,
we need to condition on \p{len}
obtaining the specification:
\[
  \gproves \WP {\m[\I1: \p{vn}(N)]} {
    \E \prob. \CC \prob \ell. \bigl(
      \sure{\Ip{len}{1} = \ell \leq N} *
      \var{Out}_\ell
    \bigr)
  }
\]
(Recall $ P \gproves Q \is P \land \ownall \proves Q \land \ownall $)

\begin{mathfig}[\small]
  \begin{proofoutline}
  \PREC{\ownall}\\
  \CODE{len:=0}\\
  \ASSR{\sure{\Ip{len}{1} = 0}}\\
  \ASSR{\E \prob. \CC \prob \ell. \bigl(
      \sure{\Ip{len}{1} = \ell \leq 0}
    \bigr)}
    \TAG[von-neumann:P0]
  \\
  \begin{proofjump}[rule:wp-loop,"invariant $P(i) =
    \E \prob. \CC \prob \ell. \bigl(
      \sure{\Ip{len}{1} = \ell \leq i} *
      \var{Out}_\ell
    \bigr)$"
  ]
  \CODE{repeat\ $N$:}\\
  \begin{proofindent}
\ASSR{P(i)}\\
  \CODE{coin_1 :~ Ber($p$)}\\
  \CODE{coin_2 :~ Ber($p$)}\\
  \ASSR{
    P(i)
    * \distAs{\p{coin}_1\at{\I1}}{\Ber{p}}
    * \distAs{\p{coin}_2\at{\I1}}{\Ber{p}}
  }\\
  \ASSR{
    \CC \prob \ell. \CC \beta b.
    \begin{pmatrix}
      \sure{\Ip{len}{1}=\ell\leq i} *
      \sure{(\p{coin}_1 \ne \p{coin}_2)\at{\I1} = b}
      \\ {}
      * \var{Out}_\ell *
      (\pure{b=1} \implies \distAs{\p{coin}_1\at{\I1}}{\Ber{\onehalf}})
    \end{pmatrix}
  }
  \TAG[von-neumann:ClCb]
  \\
  \begin{proofjump}[rule:c-wp-elim]
    \ASSR{
      \sure{\Ip{len}{1}=\ell\leq i} *
      \sure{(\p{coin}_1 \ne \p{coin}_2)\at{\I1} = b}
      \\ {}
      * \var{Out}_\ell *
      (\pure{b=1} \implies \distAs{\p{coin}_1\at{\I1}}{\Ber{\onehalf}})
    }\\
    \CODE{if coin_1 != coin$_2\ $ then:}\\
    \begin{proofindent}
      \ASSR{
        \sure{\Ip{len}{1}=\ell\leq i} *
        \sure{(\p{coin}_1 \ne \p{coin}_2)\at{\I1}}
        \\ {}
        * \var{Out}_\ell *
        \distAs{\p{coin}_1\at{\I1}}{\Ber{\onehalf}}
      }\\
      \CODE{out[len] := coin_1}\\
      \CODE{len := len+1}\\
      \ASSR{
        \sure{\Ip{len}{1}=\ell+1\leq i+1} *
        \sure{(\p{coin}_1 \ne \p{coin}_2)\at{\I1}}
        \\ {}
        * \var{Out}_\ell *
        \distAs{\p{coin}_1\at{\I1}}{\Ber{\onehalf}}
        * \sure{(\p{out[len]} = \p{coin}_1)\at{\I1}}
      }\\
    \end{proofindent}
    \\
    \ASSR{
      \sure{(\p{coin}_1 \ne \p{coin}_2)\at{\I1}=b}
      * \var{Out}_\ell \\{} *
      \begin{cases}
        \sure{\Ip{len}{1}=\ell+1\leq i+1} *
        \distAs{\p{out[len]}\at{\I1}}{\Ber{\onehalf}}
        \CASE b=1
        \\
        \sure{\Ip{len}{1}=\ell\leq i+1}
        \CASE b=0
      \end{cases}
    }
  \end{proofjump}
  \\
  \ASSR{
    \CC \prob \ell.\CC \beta b.
      \sure{(\p{coin}_1 \ne \p{coin}_2)\at{\I1}=b}
      * \var{Out}_\ell * \dots
  }
  \\
  \ASSR{
    \CC {\prob'} {\ell'}. \bigl(
      \sure{\Ip{len}{1} = \ell' \leq i+1} *
      \var{Out}_{\ell'}
    \bigr)
  }
  \quad\TAG[von-neumann:Pi+1]
  \end{proofindent}
  \end{proofjump}
  \\
  \POST{\E \prob. \CC \prob \ell. \bigl(
        \sure{\Ip{len}{1} = \ell \leq N} *
        \var{Out}_\ell
      \bigr)}
  \end{proofoutline}
  \caption{Proof outline of the Von Neumann extractor example.}
  \label{fig:von-neumann-outline}
\end{mathfig}

The \thelogic\ proof of this specification is shown
in the outline in \cref{fig:von-neumann-outline}.
The postcondition straightforwardly generalizes to a loop invariant
\[
  P(i) =
  \E \prob. \CC \prob \ell. \bigl(
    \sure{\Ip{len}{1} = \ell \leq i} *
    \var{Out}_\ell
  \bigr)
\]
At step~\eqref{von-neumann:P0} we show, by
using \ref{rule:c-unit-l} and the definition of $\sure{\hole}$,
that we can obtain the loop invariant with $i=0$:
$P(0) = \E \prob. \CC \prob \ell. \bigl(
      \sure{\Ip{len}{1} = \ell \leq 0} *
      \var{Out}_0
    \bigr) = \E \prob. \CC \prob \ell. \bigl(
      \sure{\Ip{len}{1} = \ell \leq 0}
    \bigr). $

For the proof of the body of the loop we can assume~$P(i)$ and we need to prove
the postcondition~$P(i+1)$.
After sampling the two coins,
we reach the point where we apply the fundamental insight behind
the extractor, at step~\eqref{von-neumann:ClCb}.
The key idea is that with some probability~$q$ the two coins will be different,
in which case the outcomes of the two coins can be either $(0,1)$ or $(1,0)$,
which both have the same probability $p(1-p)$.
Therefore, if the coins are different, $\p{coin}_1=0$ and $\p{coin}_1=1$
have the same probability, \ie $\p{coin}_1$ looks like a fair coin.

\thelogic\ is capable of representing this reasoning as follows.

We start with two independent biased coins, which we can combine
into a random variable $(\p{coin}_1 \ne \p{coin}_2,\p{coin}_1)$
recording whether the two outcomes where different and the outcome
of the first coin;
it is easy to derive (using \ref{rule:prod-unsplit} and \ref{rule:dist-fun}):
\[
  \distAs{\p{coin}_1\at{\I1}}{\Ber{p}} *
  \distAs{\p{coin}_2\at{\I1}}{\Ber{p}}
  \proves
  \distAs{(\p{coin}_1 \ne \p{coin}_2,\p{coin}_1)\at{\I1}}{\prob_0}
\]
where
$
  \prob_0 \is \left(\DO{
    c_1 <- \Ber{p};
    c_2 <- \Ber{p};
    \return (c_1 \ne c_2, c_1)
  }\right).
$
Now the crucial observation of the extractor can be phrased as
a reformulation of~$\prob_0$:
\begin{align*}
  \prob_0 &= \beta \fuse \krnl
  &
  \beta &\is \Ber{q}
  &
  \krnl(1) &\is \Ber{\onehalf} &
  \krnl(0) &\is \Ber{q'}\end{align*}
Here one first determines
(with some probability~$q$ which is a function of~$p$)
whether the two coins will be different or equal,
and then generates $c_1$ accordingly:
in the ``different'' branch ($b=1$) the first coin is distributed as $\Ber{\onehalf}$ while in the ``equal'' branch ($b=0$) the first coin is distributed with some bias~$q'$ (also a function of~$p$).

So using $ \prob_0 = \beta \fuse \krnl $ we derive:
\begin{eqexplain}
  &
  \distAs{(\p{coin}_1 \ne \p{coin}_2,\p{coin}_1)\at{\I1}}{(\beta \fuse \krnl)}
\whichproves
\CC {\beta \fuse \krnl} {(b,c_1)}.
    \sure{(\p{coin}_1 \ne \p{coin}_2)\at{\I1} = b \land \p{coin}_1\at{\I1}=c_1}
  \byrules{c-unit-r}
\whichproves
\CC \beta b.\CC {\krnl(b)} {c_1}.
    \sure{(\p{coin}_1 \ne \p{coin}_2)\at{\I1} = b} *
    \sure{\p{coin}_1\at{\I1}=c_1}
  \byrules{c-fuse}
\whichproves
\CC \beta b. \bigl(
    \sure{(\p{coin}_1 \ne \p{coin}_2)\at{\I1} = b} *
    \CC {\krnl(b)} {c_1}.
      \sure{\p{coin}_1\at{\I1}=c_1}
  \bigr)
  \byrules{sure-str-convex}
\whichproves
\CC \beta b. \bigl(
    \sure{(\p{coin}_1 \ne \p{coin}_2)\at{\I1} = b} *
    \pure{b=1} \implies
    \CC {\Ber{\onehalf}} {c_1}.
      \sure{\p{coin}_1\at{\I1}=c_1}
  \bigr)
  \byrules{c-cons}
\whichproves
\CC \beta b. \bigl(
    \sure{(\p{coin}_1 \ne \p{coin}_2)\at{\I1} = b} *
    \pure{b=1} \implies
      \distAs{\p{coin}_1\at{\I1}}{\Ber{\onehalf}}
  \bigr)
  \byrules{c-unit-r}
\end{eqexplain}

The application of \ref{rule:c-fuse}
allows us to first condition on $\p{coin}_1 \ne \p{coin}_2$,
and then the first coin.
We can then weaken the case where $b=0$ and only record that
if $b=1$ then $\p{coin}_1$ is a fair coin.

This takes us through step~\eqref{von-neumann:ClCb} of \cref{fig:von-neumann-outline}.
Now the precondition of the if statement is conditional on \p{len} and
$\p{coin}_1 \ne \p{coin}_2$.
Intuitively, we want to evaluate the effects of the if statement
in the two possible outcomes and put together the results.
This is precisely the purpose of the \ref{rule:c-wp-swap} rule,
which together with \ref{rule:c-cons} gives us the derived rule:

\begin{proofrule}
  \infer*[lab=c-wp-elim]{
    \forall v\in\psupp(\prob) \st
    P(v) \gproves \WP {\m{t}} {Q(v)}
  }{
    \CC \prob v.P(v) \gproves \WP {\m{t}} {\CC \prob v.Q(v)}
  }
  \label{rule:c-wp-elim}
\end{proofrule}

By applying the rule twice
(once on the conditioning on \p{len},
and the on the conditioning on $\p{coin}_1 \ne \p{coin}_2$),
we can process the if statement case by case,
and then combine, under the same conditioning, the postconditions
we obtained in each case.
The ``else'' branch is a \code{skip} (omitted) so it preserves the precondition
in the $b=0$ branch.
In the ``then'' branch we can assume $b=1$; we apply \ref{rule:wp-rl-assign}
to the assignments and combine the result with the ``else'' branch by making
the overall postcondition of the if statement to be parametric
on the value of~$b$ (and~$\ell$).

The last non-obvious step is~\eqref{von-neumann:Pi+1} in \cref{fig:von-neumann-outline},
where we show that the conditional postcondition of the if statement
implies the loop invariant $P(i+1)$.
Let
\[
  K(\ell,b) =
  \begin{cases}
    \sure{\Ip{len}{1}=\ell+1\leq i+1} *
    \distAs{\p{out[len]}\at{\I1}}{\Ber{\onehalf}}
    \CASE b=1
    \\
    \sure{\Ip{len}{1}=\ell\leq i+1}
    \CASE b=0
  \end{cases}
\]
then the step is proven as follows:
\begin{eqexplain}
  &
  \CC \prob \ell.\CC \beta b. \bigl(
    \sure{(\p{coin}_1 \ne \p{coin}_2)\at{\I1}=b}
    * \var{Out}_\ell * K(\ell,b)
  \bigr)
  \whichproves
\CC {\prob \pprod \beta} (\ell,b).\bigl(
    \sure{(\p{coin}_1 \ne \p{coin}_2)\at{\I1}=b}
    * \var{Out}_\ell * K(\ell,b)
  \bigr)
  \byrule{c-assoc}
  \whichproves
\CC {\prob \pprod \beta} (\ell,b).\bigl(
    \var{Out}_\ell * K(\ell,b)
  \bigr)
  \byrules{c-cons}
  \whichproves
\CC {\prob''} (\ell',\ell).
    \begin{cases}
      \sure{\Ip{len}{1}=\ell'\leq i+1} *
      \var{Out}_{\ell'-1} *
      \distAs{\p{out[len]}\at{\I1}}{\Ber{\onehalf}}
      \CASE \ell'=\ell+1
      \\
      \sure{\Ip{len}{1}=\ell'\leq i+1} *
      \var{Out}_{\ell'}
      \CASE \ell'=\ell
    \end{cases}
  \byrules{c-transf}
  \whichproves
\CC {\prob''\circ\inv{\proj_1}} \ell'.\bigl(
    \sure{\Ip{len}{1}=\ell'\leq i+1} *
    \var{Out}_{\ell'}
  \bigr)
  \byrules{c-dist-proj}
  \whichproves
\E \prob'.
  \CC {\prob'} \ell'.\bigl(
    \sure{\Ip{len}{1}=\ell'\leq i+1} *
    \var{Out}_{\ell'}
  \bigr)
\end{eqexplain}

The application of \ref{rule:c-transf}
uses the function $f(\ell,b) = (\ell+b, \ell)$ to introduce the new $\ell'$
and then we project away the unused~$\ell$ using the derived \ref{rule:c-dist-proj} (note that the rule applies to $\sure{\hole}$ assertions and multiple ownership assertions in a separating conjunction thanks to \ref{rule:prod-split} and \ref{rule:prod-unsplit}).

\subsection{Monte Carlo: $\p{BETW\_SEQ} \leq \p{BETW}$}
\label{sec:appendix:ex:monte}

  Recall the example sketched in \cref{sec:intro}
where one wants to compare the accuracy of variants of a  Monte Carlo algorithm
(in \cref{fig:between-code})
to estimate whether a number~$x$ is within the extrema of some set~$S$.
\Cref{fig:between-code-repeat} reproduces the code here for convenience,
with the self-assignments to~\p{l} and~\p{r} expanded to their form with
a temporary (primed) variable storing the old value of the assigned variable.

The verification task we accomplish in this section is to compare
the accuracy of the two Monte Carlo algorithms
\p{BETW\_SEQ} and \p{BETW} (the optimized one).

This goal can be encoded as the judgment:
\[
  \cpl{\Ip{l}{1}=\Ip{r}{1}=\Ip{l}{2}=\Ip{r}{2}=0}\withp{\m{\permap}}
  \proves
  \WP {\m[
    \I1: \code{BETW_SEQ($x, S$)},
    \I2: \code{BETW($x, S$)}
  ]} {
    \cpl{\Ip{d}{1} \leq \Ip{d}{2}}
  }
\]
where
$\m{\permap}$ contains full permissions for all the variables.
The judgment states, through the relational lifting, that it is more likely
to get a positive answer from \p{BETW} than from \p{BETW\_SEQ}.
The challenge is implementing the intuitive relational argument
sketched in \cref{sec:intro},
in the presence of very different looping structures.

By \ref{rule:wp-rl-assign}, it is easy to prove that
\[
  \cpl{\Ip{l}{1}\leq\Ip{l}{2} \land
       \Ip{r}{1}\leq\Ip{r}{2}}\withp{\m{\permap}}
  \proves
  \WP {\m[
    \I1: \code{d := r&&l},
    \I2: \code{d := r&&l}
  ]} {
    \cpl{\Ip{d}{1} \leq \Ip{d}{2}}
  }
\]
Therefore we will focus on proving that the loops produce distributions
satisfying $
  Q = \cpl{\Ip{l}{1}\leq\Ip{l}{2} \land
       \Ip{r}{1}\leq\Ip{r}{2}}.
$

Now the main obstacle is that we have a single loop at component~$\I2$
looping $2N$ times, and two sequentially composed loops in $\I1$,
each running $N$ iterations.
In a standard coupling-based logic like pRHL,
such structural differences are usually bridged by invoking a
syntactic transformation (\eg loop splitting) that is provided
by a library of transformations that were proven separately, using meta-reasoning directly on the semantic model,
by the designer of the logic.
In \thelogic\ we aim at:
\begin{itemize}
  \item Avoiding resorting to syntactic transformations;
  \item Avoiding relying on an ad-hoc (incomplete) library of transformations;
  \item Avoiding having to argue for correctness of transformations semantically.
\end{itemize}
To achieve this, we formulate the loop-splitting pattern as a \emph{rule}
which allows to consider $N$ iterations of component $\I2$ against the first loop of $\I1$, and the rest against the second loop of $\I1$.
\begin{proofrule}
  \infer*[lab=wp-loop-split]{
  P_1(N_1) \proves P_2(0)
  \\\\
  \forall i < N_1 \st
    P_1(i) \proves \WP{\m[\I1: t_1, \I2: t]}{P_1(i+1)}
  \\\\
  \forall j < N_2 \st
    P_2(j) \proves \WP{\m[\I1: t_2, \I2: t]}{P_2(j+1)}
}{
  P_1(0) \proves
  \WP{\m[
    \I1: (\Loop{N_1}{t_1}\p;\Loop{N_2}{t_2}),
    \I2: \Loop{(N_1+N_2)}{t}
  ]}{P_2(N_2)}
}   \relabel{rule:wp-loop-split}
\end{proofrule}

Most importantly, such rule is \emph{derivable} from the primitive rules of \thelogic, avoiding semantic reasoning all together.
Once this rule is proven, it can be used any time need for such pattern arises.
Before showing how this rule is derivable,
which we do in \cref{proof:wp-loop-split},
let us show how to use it to close our example.

We want to apply \ref{rule:wp-loop-split} with $N_1=N_2=N$,
$t_1$ as the body of the loop of \p{BelowMax},
$t_2$ as the body of the loop of \p{AboveMin},
and
$t$ as the body of the loop of \p{BETW}.
We define the two loop invariants as follows:
\begin{align*}
  P_1(i) &\is
    \cpl{
      \Ip{r}{1}\leq\Ip{r}{2}
      \land
      \Ip{l}{1}=0\leq\Ip{l}{2}
    }
  &
  P_2(j) &\is
    \cpl{
      \Ip{r}{1}\leq\Ip{r}{2}
      \land
      \Ip{l}{1}\leq\Ip{l}{2}
    }
\end{align*}
Note that they both ignore the iteration number.
Clearly we have:
\begin{align*}
  P_0 &\proves P_1(0)
  &
  P_1(N) &\proves P_2(0)
  &
  P_2(N) &\proves Q
\end{align*}

By applying \ref{rule:wp-loop-split} we reduce the goal to the triples:
\begin{align*}
  \cpl{
    \Ip{r}{1}\leq\Ip{r}{2}
    \land
    \Ip{l}{1}=0\leq\Ip{l}{2}
  }
  &\proves
  \WP{\m[
    \I1: t_1,
    \I2: t
  ]}{
    \cpl{
      \Ip{r}{1}\leq\Ip{r}{2}
      \land
      \Ip{l}{1}=0\leq\Ip{l}{2}
    }
  }
  \\
  \cpl{
    \Ip{r}{1}\leq\Ip{r}{2}
    \land
    \Ip{l}{1}\leq\Ip{l}{2}
  }
  &\proves
  \WP{\m[
    \I1: t_2,
    \I2: t
  ]}{
    \cpl{
      \Ip{r}{1}\leq\Ip{r}{2}
      \land
      \Ip{l}{1}\leq\Ip{l}{2}
    }
  }
\end{align*}
which are easy to obtain by replicating the standard coupling-based reasoning
steps, using \ref{rule:coupling} and \ref{rule:wp-rl-assign}.

\medskip
As promised, we now prove \ref{rule:wp-loop-split} is derivable.
\begin{lemma}
\label{proof:wp-loop-split}
  \Cref{rule:wp-loop-split} is sound.
\end{lemma}

\begin{proof}
  Assume:
  \begin{gather}
  P_1(N_1) \proves P_2(0)
  \label{wp-loop-split:P1-P2}
  \\
  \forall i < N_1 \st
    P_1(i) \proves \WP{\m[\I1: t_1, \I2: t]}{P_1(i+1)}
  \label{wp-loop-split:loop1}
  \\
  \forall j < N_2 \st
    P_2(j) \proves \WP{\m[\I1: t_2, \I2: t]}{P_2(j+1)}
  \label{wp-loop-split:loop2}
  \end{gather}
  We want to show:
  \[
    P_1(0) \proves
    \WP{\m[
      \I1: (\Loop{N_1}{t_1}\p;\Loop{N_2}{t_2}),
      \I2: \Loop{(N_1+N_2)}{t}
    ]}{P_2(N_2)}
  \]
  First, by using \ref{rule:wp-nest} and \ref{rule:wp-seq},
  we can reduce the goal to:
  \[
    P_1(0) \proves
    \WP{\m[\I2: \Loop{(N_1+N_2)}{t}]}[\big]{
      \WP{\m[\I1: \Loop{N_1}{t_1}]}{
        \WP{\m[\I1:\Loop{N_2}{t_2}]}{P_2(N_2)}
       }
    }
  \]
  Now define:
  \[
    P(k) =
    \begin{cases}
    \WP {\m[\I1: \Loop{k}{t_1}]} { P_1(k) }
    \CASE k \leq N_1
    \\
    \WP {\m[\I1: \Loop{N_1}{t_1}]}[\big]{
      \WP {\m[\I1: \Loop{(k-N_1)}{t_2}]} { P_2(k-N_1) }
    }
    \CASE k > N_1
    \end{cases}
  \]
  We have:
  \begin{gather}
    P_1(0) \proves P(0)
    \label{wp-loop-split:P1-P}
    \\
    P(N_1+N_2) \proves
    \WP{\m[\I1: \Loop{N_1}{t_1}]}{
      \WP{\m[\I1:\Loop{N_2}{t_2}]}{P_2(N_2)}
     }
    \label{wp-loop-split:P-P2}
  \end{gather}
  Entailment~\eqref{wp-loop-split:P1-P} holds by \ref{rule:wp-loop-0},
  and \eqref{wp-loop-split:P-P2} holds by definition.
  Therefore, using \ref{rule:wp-cons} we reduced the goal to
  \[
    P(0) \proves \WP{\m[\I2: \Loop{(N_1+N_2)}{t}]}{P(N_1+N_2)}
  \]
  which we can make progress on using \ref{rule:wp-loop}.
  We are left with proving:
  \[
    \forall k < N_1+N_2\st
    P(k) \proves
    \WP {\m[\I2: t]} {P(k+1)}
  \]
  We distinguish three cases:
  \begin{casesplit}
  \case[$k<N_1$] By unfolding the definition of~$P$ we obtain:
    \[
      \WP {\m[\I1: \Loop{k}{t_1}]} { P_1(k) }
      \proves
      \WP {\m[\I2: t]}[\big]{
        \WP {\m[\I1: \Loop{(k+1)}{t_1}]} { P_1(k+1) }
      }
    \]
    Using \ref{rule:wp-loop-unf} on the inner WP we obtain:
    \[
      \WP {\m[\I1: \Loop{k}{t_1}]} { P_1(k) }
      \proves
      \WP {\m[\I2: t]}[\big]{
        \WP {\m[\I1: \Loop{(k)}{t_1}]} { \WP{\m[\I1: t_1]}{P_1(k+1)} }
      }
    \]
    By \ref{rule:wp-nest} we can swap the two topmost WPs:
    \[
      \WP {\m[\I1: \Loop{k}{t_1}]} { P_1(k) }
      \proves
      \WP {\m[\I1: \Loop{k}{t_1}]}[\big]{
        \WP {\m[\I2: t]} { \WP{\m[\I1: t_1]}{P_1(k+1)} }
      }
    \]
    Finally, by \ref{rule:wp-cons} we can eliminate the topmost WP from both sides:
    \[
      P_1(k)
      \proves
      \WP {\m[\I2: t]} { \WP{\m[\I1: t_1]}{P_1(k+1)} }
    \]
    which by \ref{rule:wp-nest} is our assumption~\eqref{wp-loop-split:loop1}
    with $i=k$.

  \case[$k=N_1$] By unfolding the definition of~$P$ we obtain:
    \[
      \WP {\m[\I1: \Loop{N_1}{t_1}]} { P_1(N_1) }
      \proves
      \WP {\m[\I2: t]}*{
        \WP {\m[\I1: \Loop{N_1}{t_1}]}[\big]{
          \WP {\m[\I1: \Loop{1}{t_2}]} { P_2(0) }
        }
      }
    \]
    By a trivial application of \ref{rule:wp-loop}
    we have $
      \WP {\m[\I1: t]} { Q }
      \proves
      \WP {\m[\I1: \Loop{1}{t}]} { Q }
    $, so we can simplify the innermost WP to:
    \[
      \WP {\m[\I1: \Loop{N_1}{t_1}]} { P_1(N_1) }
      \proves
      \WP {\m[\I2: t]}*{
        \WP {\m[\I1: \Loop{N_1}{t_1}]}[\big]{
          \WP {\m[\I1: t_2]} { P_2(1) }
        }
      }
    \]
    Then by \ref{rule:wp-nest} we can swap the topmost WPs:
    \[
      \WP {\m[\I1: \Loop{N_1}{t_1}]} { P_1(N_1) }
      \proves
      \WP {\m[\I1: \Loop{N_1}{t_1}]}*{
        \WP {\m[\I2: t]}[\big]{
          \WP {\m[\I1: t_2]} { P_2(1) }
        }
      }
    \]
    By \ref{rule:wp-cons} we can eliminate the topmost WP from both sides:
    \[
      P_1(N_1)
      \proves
      \WP {\m[\I2: t]}[\big]{
        \WP {\m[\I1: t_2]} { P_2(1) }
      }
    \]
    Using assumption~\eqref{wp-loop-split:P-P2} we can reduce this to:
    \[
      P_2(0)
      \proves
      \WP {\m[\I2: t]}[\big]{
        \WP {\m[\I1: t_2]} { P_2(1) }
      }
    \]
    which by \ref{rule:wp-nest} is our assumption~\eqref{wp-loop-split:loop2}
    with $j=0$.

  \case[$k>N_1$] By unfolding the definition of~$P$ we obtain:
    \begin{align*}
      &\WP {\m[\I1: \Loop{N_1}{t_1}]}{
        \WP {\m[\I1: \Loop{(k-N_1)}{t_2}]} { P_2(k-N_1) }
      }
      \\ {}\proves{}&
      \WP {\m[\I2: t]}*{
        \WP {\m[\I1: \Loop{N_1}{t_1}]}[\big]{
          \WP {\m[\I1: \Loop{(k-N_1+1)}{t_2}]} { P_2(k-N_1+1) }
        }
      }
    \end{align*}
    Using \ref{rule:wp-loop-unf} on the inner WP we obtain:
    \begin{align*}
      &\WP {\m[\I1: \Loop{N_1}{t_1}]}[\big]{
        \WP {\m[\I1: \Loop{(k-N_1)}{t_2}]} { P_2(k-N_1) }
      }
      \\ {}\proves{}&
      \WP {\m[\I2: t]}*{
        \WP {\m[\I1: \Loop{N_1}{t_1}]}[\big]{
          \WP {\m[\I1: \Loop{(k-N_1)}{t_2}]} {
            \WP{\m[\I1: t_2]}{P_2(k-N_1+1)}
          }
        }
      }
    \end{align*}
    By \ref{rule:wp-nest} we can push the topmost WP inside:
    \begin{align*}
      &\WP {\m[\I1: \Loop{N_1}{t_1}]}[\big]{
        \WP {\m[\I1: \Loop{(k-N_1)}{t_2}]} { P_2(k-N_1) }
      }
      \\ {}\proves{}&
      \WP {\m[\I1: \Loop{N_1}{t_1}]}*{
        \WP {\m[\I1: \Loop{(k-N_1)}{t_2}]}[\big]{
          \WP {\m[\I2: t]} {
            \WP{\m[\I1: t_2]}{P_2(k-N_1+1)}
          }
        }
      }
    \end{align*}
    Finally, by \ref{rule:wp-cons} we can eliminate the topmost WPs from both sides:
    \begin{equation*}
    P_2(k-N_1)
    \proves
    \WP {\m[\I2: t]} {
      \WP{\m[\I1: t_2]}{P_2(k-N_1+1)}
    }
    \end{equation*}
    which by \ref{rule:wp-nest} is our assumption~\eqref{wp-loop-split:loop2}
    with $j=k-N_1$.
  \qedhere
  \end{casesplit}
\end{proof}  

\subsection{Monte Carlo: Equivalence Between \p{BETW\_MIX} and \p{BETW\_SEQ}}

  \Cref{fig:between-code-repeat} shows another way in which one can
approximately compute the ``between'' function: \p{BETW\_MIX}.
In this example we want to prove the equivalence between
\p{BETW\_MIX} and \p{BETW\_SEQ}.
Again, the main obstacle to overcome in the proof is that
the structure of the two programs is very different.
\p{BETW\_SEQ} has two loops of~$N$ iterations, with one sample per iteration.
\p{BETW\_MIX} has a single loop of~$N$ iterations, but it samples twice
per iteration.
Note that the equivalence cannot be understood as a generic program transformation: the order in which the samples are taken in the two programs
is drastically different; they are only equivalent because the calculations
done on each of these independent samples are independent from one another.

Intuitively, we want to produce a proof that aligns each iteration
of the first loop of \p{BETW\_SEQ} with half of each iteration of \p{BETW\_MIX},
and each iteration of the second loop of \p{BETW\_SEQ} with the second half of each iteration of \p{BETW\_MIX}.
In the same vein as the previous example, we want to formalize the
proof pattern as a rule that aligns the loops as desired,
prove the rule is derivable, and apply it to the example.
A rule encoding the above pattern is the following:
\begin{proofrule}
  \infer*[lab=wp-loop-mix]{
  \forall i < N \st
    P_1(i) \proves \WP{\m[\I1: t_1, \I2: t_1']}{P_1(i+1)}
  \\
  \forall i < N \st
    P_2(i) \proves \WP{\m[\I1: t_2, \I2: t_2']}{P_2(i+1)}
}{
  P_1(0) * P_2(0)
  \proves
  \WP{\m[
    \I1: (\Loop{N}{t_1}\p;\Loop{N}{t_2}),
    \I2: \Loop{N}{(t_1';t_2')}
  ]}{P_1(N) * P_2(N)}
}   \relabel{rule:wp-loop-mix}
\end{proofrule}
Before showing how this rule is derivable,
which we do in \cref{proof:wp-loop-mix},
let us show how to use it to close our example.

We want to prove the goal:
\[
  \begin{conj}
    \sure{\Ip{r}{1} = \Ip{l}{1} = 0} \land
    \sure{\Ip{r}{2} = \Ip{l}{2} = 0}
  \end{conj}
  \withp{\m{p}}
  \proves
  \WP{\m<
    \I1: \Loop{N}{t_{\p{r}}}\p; \Loop{N}{t_{\p{l}}},
    \I2: \Loop{N}{(t_{\p{r}}\p; t_{\p{l}})}
  >}*{
    \cpl*{
    \begin{conj*}
      \Ip{r}{1} = \Ip{r}{2} \land
      \Ip{l}{1} = \Ip{l}{2}
    \end{conj*}
    }
  }
\]
where
  $\m{p}$ has full permissions for all the relevant variables,
  $t_{\p{r}}$ is the body of the loop of \p{BelowMax}, and
  $t_{\p{l}}$ is the body of the loop of \p{AboveMin}.

As a first manipulation, we use \ref{rule:rl-merge} in the postcondition,
and \cref{rule:coupling} (via \ref{rule:sure-dirac}) to the precondition,
to obtain:
\[
  \begin{conj}
    \cpl{\Ip{r}{1} = \Ip{r}{2}}\withp{\m{p}_{\p{r}}} * {}\\
    \cpl{\Ip{l}{1} = \Ip{l}{2}}\withp{\m{p}_{\p{l}}}
  \end{conj}
  \proves
  \WP{\m<
    \I1: \Loop{N}{t_{\p{r}}}\p; \Loop{N}{t_{\p{l}}},
    \I2: \Loop{N}{(t_{\p{r}}\p; t_{\p{l}})}
  >}*{
    \begin{conj}
      \cpl{\Ip{r}{1} = \Ip{r}{2}}\withp{\m{p}_{\p{r}}} * {}\\
      \cpl{\Ip{l}{1} = \Ip{l}{2}}\withp{\m{p}_{\p{l}}}
    \end{conj}
  }
\]
where
$\m{p}_{\p{r}} = \m[\Ip{r}{1}:1, \Ip{r}{2}:1, \Ip{q}{1}:1, \Ip{q}{2}:1] $, and
$\m{p}_{\p{l}} = \m[\Ip{l}{1}:1, \Ip{l}{2}:1, \Ip{p}{1}:1, \Ip{p}{2}:1] $.
Then \ref{rule:wp-loop-mix} applies and we are left with the two triples
\begin{gather*}
  \cpl{\Ip{r}{1} = \Ip{r}{2}}
  \withp{\m{p}_{\p{r}}}
  \proves
  \WP{\m[
    \I1: t_{\p{r}},
    \I2: t_{\p{r}}
  ]}*{
    \cpl{\Ip{r}{1} = \Ip{r}{2}}
    \withp{\m{p}_{\p{r}}}
  }
  \\
  \cpl{\Ip{l}{1} = \Ip{l}{2}}
  \withp{\m{p}_{\p{l}}}
  \proves
  \WP{\m[
    \I1: t_{\p{l}},
    \I2: t_{\p{l}}
  ]}*{
    \cpl{\Ip{l}{1} = \Ip{l}{2}}
    \withp{\m{p}_{\p{l}}}
  }
\end{gather*}
which are trivially proved using a standard coupling argument.

\medskip
As promised, we now prove \ref{rule:wp-loop-mix} is derivable,
concluding the example.
\begin{lemma}
\label{proof:wp-loop-mix}
  \Cref{rule:wp-loop-mix} is sound.
\end{lemma}

\begin{proof}
  Assume:
  \begin{gather}
    \forall i < N \st
      P_1(i) \proves \WP{\m[\I1: t_1, \I2: t_1']}{P_1(i+1)}
    \label{wp-loop-mix:P1}
    \\
    \forall i < N \st
      P_2(i) \proves \WP{\m[\I1: t_2, \I2: t_2']}{P_2(i+1)}
    \label{wp-loop-mix:P2}
  \end{gather}
  Our goal is to prove:
  \[
    P_1(0) * P_2(0)
    \proves
    \WP{\m[
      \I1: (\Loop{N}{t_1}\p;\Loop{N}{t_2}),
      \I2: \Loop{N}{(t_1';t_2')}
    ]}{P_1(N) * P_2(N)}
  \]
  We first massage the goal to split the sequential composition at \I1.
  By \ref{rule:wp-seq} and \ref{rule:wp-nest} we obtain
  \[
    P_1(0) * P_2(0)
    \proves
    \WPv{\m[\I2: \Loop{N}{(t_1';t_2')}]}{
      \WP{\m[\I1: \Loop{N}{t_1}]}{
        \WP{\m[\I1: \Loop{N}{t_2}]}{
          P_1(N) * P_2(N)
        }
      }
    }
  \]
  Now by applying \ref{rule:wp-frame} in the postcondition (twice) we obtain
  \begin{equation}
    P_1(0) * P_2(0)
    \proves
    \WP{\m[\I2: \Loop{N}{(t_1';t_2')}]}*{
      \begin{conj}
      \WP{\m[\I1: \Loop{N}{t_1}]}{P_1(N)} * {}\\
      \WP{\m[\I1: \Loop{N}{t_2}]}{P_2(N)}
      \end{conj}
    }
  \label{wp-loop-mix:goal-loop}
  \end{equation}
  Define
  \begin{align*}
    P(i) &\is Q_1(i) * Q_2(i)
    &
    Q_1(i) &\is \WP{\m[\I1: \Loop{i}{t_1}]}{P_1(i)}
    &
    Q_2(i) &\is \WP{\m[\I1: \Loop{i}{t_2}]}{P_2(i)}
  \end{align*}
  Clearly we have
  $ P_1(0) * P_2(0) \proves P(0) $ (by \ref{rule:wp-loop-0})
  and $ P(N) $ coincides with the postcondition
  of our goal~\eqref{wp-loop-mix:goal-loop}, which is now rewritten to:
  \[
    P(0)\proves \WP{\m[\I2: \Loop{N}{(t_1'\p;t_2')}]}{P(N)}
  \]
  Now we can apply \ref{rule:wp-loop} with invariant~$P$ and reduce the goal to
  the triples:
  \[
    \forall i < N \st
      Q_1(i)*Q_2(i) \proves \WP{\m[\I2: (t_1'\p;t_2')]}{Q_1(i+1)*Q_2(i+1)}
  \]
  By \ref{rule:wp-seq} and \ref{rule:wp-frame} we can reduce the goal to
  \[
    Q_1(i)*Q_2(i) \proves
    \WP{\m[\I2: t_1']}{Q_1(i+1)} *
    \WP{\m[\I2: t_2']}{Q_2(i+1)}
  \]
  which we can prove by showing the two triples:
  \begin{align*}
    Q_1(i) &\proves
    \WP{\m[\I2: t_1']}{Q_1(i+1)}
    &
    Q_2(i) &\proves
    \WP{\m[\I2: t_2']}{Q_2(i+1)}
  \end{align*}
  We focus on the former as the latter can be dealt with symmetrically.
  By unfolding $Q_1$ we obtain:
  \[
    \WP{\m[\I1: \Loop{i}{t_1}]}{P_1(i)}
    \proves
    \WP{\m[\I2: t_1']}[\big]{\WP{\m[\I1: \Loop{(i+1)}{t_1}]}{P_1(i+1)}}.
  \]
  We then apply \ref{rule:wp-loop-unf} to the innermost WP and \ref{rule:wp-nest}to swap the two WPs in the conclusion:
  \[
    \WP{\m[\I1: \Loop{i}{t_1}]}{P_1(i)}
    \proves
    \WP{\m[\I1: \Loop{i}{t_1}]}[\big]{\WP{\m[\I2: t_1', \I1: t_1]}{P_1(i+1)}}.
  \]
  Finally, by \ref{rule:wp-cons} we can eliminate the topmost WPs
  on both sides and reduce the goal to assumption~\eqref{wp-loop-mix:P1}.
\end{proof}

\subsection{Randomized Cache Management}

Cache is a part of computer memory fast to access but limited in its size.
Cache management algorithms determine which elements are to keep or
evict upon a sequence of requests.
If the computer knows the entire sequence of requests, then its best strategy is
to always evict the element that is next requested furtherest in the future.
We call this strategy the \emph{offline optimal algorithm}.
However, a realistic cache algorithm needs to make decisions
online, which make it impossible to consider requests in the future.
While for some applications deterministic
cache algorithms, such as First-In-First-Out (FIFO), Last-In-First-Out (LIFO),
and Least Recently Used (LRU), perform relatively well, they suffer when the
sequence of requests are produced adversarially to maximize their misses. It is
proved that, for a cache that holds $k$ items, no deterministic cache algorithm
is $< k$-competitive; that means for any deterministic cache algorithm, there
exists some sequence on which the algorithm produces at least $k$ times misses
of the offline optimal algorithm's number of misses.

Randomization could help a lot in this scenario of adversarial requests.
The \emph{Marker algorithm}, a randomized algorithm also known as 1-bit LRU algorithm,
is shown to be $\log(k)$-competitive to the offline optimal.
While the algorithm is relatively intuitive, its analysis involves many parts.
Though we do not have an end-to-end formal proof of that analysis,
we can verify some important components in our logic.

The Marker algorithm attaches each slot of its cache with an 1-bit marker.
Initially, for any $1 \leq i \leq k$, the mark
$\code{cache[$i$][mark]}$ is set to 0.
Then it operates in phases.
At the beginning of each phase, the markers are all set to 0.
Every time it encounters a new request that
it has not seen in the current phase,
either the requested element is already in the cache and then
it simply changes that slot's mark to 1,
or the requested element is not in the cache,
then it replaces the item on a randomly chosen 0-marked slot with the request
and marks that slot to~1.
Thus, after processing $k$ distinct requests in a phase,
all the slots in the cache are marked with~1.
 The algorithm then reset all
these marks to 0 and start a new phase.

We show the details of the implementation of each phase in~\cref{fig:marker:algorithm}.
(The part resetting all marks to 0 and starts a new phase is not modelled.)
A crucial step in the implementation is to sample the eviction index $\code{ev}$
uniformly randomly from all index marked to 0.
There are many ways to implement that. We use a memory-efficient reservoir sampling algorithm~\cite{vitter1985random}, which can also be used in many other applications.
Throughout this example, in both the code and the proof, we use $[a, b]$ to denote the set of integers
that are at least $a$ and at most $b$.

\begin{figure}
\setlength\tabcolsep{0pt}\begin{tabular*}{\textwidth}{
    @{\extracolsep{\fill}}
    *{4}{p{\textwidth/4}}@{}
  }
  \begin{sourcecode*}
  def Marker(seq, N, cache):
    cost := 0
    i := 1
    repeat N:
       hit := False
       j := 1
       repeat k:
	       if cache[j]  == seq[i]:
		        cache[j] := (seq[i], 1)
		        hit := True
	       j := j+1
	    if hit:
		    skip
		  else:
        ev:~Reservoir_Samp(k, cache)
		  	cache[ev] := (seq[i], 1)
		    cost := cost + 1
      i := i+1
  \end{sourcecode*}
  &
  \begin{sourcecode*}
    def Reservoir-Samp(k, cache):
        c = 0
        j = 0
        ev := None
        repeat k:
          c := c + 1
          if cache[j][mark] = 0 then
            j := j + 1
            x :~ unifInt [1,j]
            if x <= 1:
              ev := c
  \end{sourcecode*}
\end{tabular*}
  \caption{The Marker algorithm.}
  \label{fig:marker:algorithm}
\end{figure}

\subsubsection{Reservoir Sampling}

The first component that we verify is the reservoir sampling.
We want to show that for any input $\code{cache}$,
the output $\code{ev}$ is distributed uniformly
among indices $1 \leq l \leq k$  with $\code{cache[$l$][mark]} = 0$.
We state that goal in our logic as
\begin{align*}
  & \forall S.
  \sure{l \in S \iff \code{cache[$l$][mark]} = 1} \\
  \gproves
  &\WP{1: \code{Reservoir-Samp($k$, cache)}}{
    \sure{l \in S \iff \code{cache[$l$][mark]} = 1}
    \ast \distAs{\code{ev}}{\Unif{[1, k] \setminus S}}
  }
\end{align*}
The condition
$\sure{l \in S \iff \code{cache[$l$][mark]} = 1}$
 makes sure that $S$ contains exactly the cache indices
 marked as 0.

The program uses $\p{c}$ to track the number of
slots that we iterate over and \code{j} to count the number of
cache indices with mark 0. It initializes both
$\p{c}$ and $\code{j}$ to 0, and $\code{ev}$ to a non-integer
value \code{None}.
The program then iterates over the entries \code{cache[$c$]} for
$1 \leq c \leq k$ and in each iteration update $\code{ev}$ and \code{j}
to maintain that the loop invariants that
1) $\code{ev}$ is distributed uniformly among
indices we have iterated over that are marked 0, i.e.,
$\distAs{\code{ev}}{\Unif{[1,c] \setminus S}}$
(the uniform distribution over empty set does not make sense,
so we arbitrarily interpret $\Unif{\emptyset}$ as the Dirac distribution
over the value \code{None});
2) the number of element in $\Unif{[1,c] \setminus S}$ is $\code{j}$.
Specifically, in the $c^{th}$ iteration,
if $\code{cache[$c$]} = 1$, we just increase $c$ by 1 and enter the next iteration.
If $\code{cache[$c$]} = 0$,
then we increase both $\code{c, j}$ by 1
and sample an integer \code{x} uniformly from the interval $[1, j]$.
If $\code{x} = 1$, which happens with probability $\frac{1}{j}$,
then we update $\code{ev}$ to the current index $c$;
if $\code{x} \neq 1$, which happens with $\frac{j - 1}{j}$ of the case,
we keep the previous $\code{ev}$, which has equal probability
to be any one of the $j - 1$ elements from $[1,c - 1] \setminus S$.
Thus, the distribution of $\code{ev}$ after this iteration
is a convex combination of these two branches, which we can
show to be $ \Unif{[1,c] \setminus S}$, thus reestablishing the loop invariants.

We outline the proof in~\ref{fig:discrete-reservoir}.
We omit $\ownall$ and permission 1 on all variables get assigned to in all assertions.
Also, because we need to include $\sure{l \in S \iff \code{cache[$l$][mark]} = 1}$
in all the assertions,
we abbreviate it into $\varphi$ to minimize the visual overhead.

Define $P(\p{c})$ to be
\[
  \varphi \ast \distAs{\code{ev}}{\Unif{[1, \p{c}] \setminus S}} \ast \sure{\card{[1, \p{c}] \setminus S} = j}
\]
At the first few steps, by~\ref{rule:wp-assign},
we get
\[
  \varphi \ast \sure{\p{c} = 0} \ast \sure{\code{j} = 0} \ast \sure{\code{ev} = \code{None}}
\]
before entering the first iteration.
This propositionally implies that $\sure{[1, \p{c}] \setminus S = \emptyset}$,
and by~\ref{rule:c-unit-l} and~\ref{rule:c-unit-r},
\[
\sure{\code{ev} = \code{None}} \vdash \distAs{\code{ev}}{\delta_{\code{None}}}
\]
and thus we have $P(0)$.

\begin{mathfig}[\small]
  \begin{proofoutline}
    \PREC{ \sure{l \in S \iff \code{cache[$l$][mark]} = 1}}\\
    \CODE{c = 0}\\
    \CODE{j = 0}\\
    \CODE{ev := None} \\
    \ASSR{\varphi \ast \sure{\p{c} = 0} \ast \sure{\code{j} = 0} \ast \sure{\code{ev} = \code{None}}}\\
    \CODE{repeat k:}\\
    \begin{proofindent}
      \ASSR{P(c)}\\
      \ASSR{
        \varphi \ast \distAs{\code{ev}}{\Unif{[1, \p{c}] \setminus S}} \ast \sure{\card{[1, \p{c}] \setminus S} = j}
      }\\
      \CODE{c := c + 1}\\
      \ASSR{
        \varphi
        \ast \distAs{\code{ev}}{\Unif{[1, \p{c} - 1] \setminus S}} \ast \sure{\card{[1, \p{c} - 1] \setminus S} = j}
      }\\
      \CODE{if cache[c][mark] = 0 then} \\
      \begin{proofindent}
        \CODE{j := j+1}\\
        \ASSR{
          \sure{\code{cache[c][mark]} = 0} \ast \varphi \ast
          \distAs{\code{ev}}{\Unif{[1, \p{c} - 1] \setminus S}} \ast \sure{\card{[1, \p{c} - 1] \setminus S} = j - 1}}\\
        \CODE{x :~ unifInt [1,j]} \\
        \ASSR{
                \sure{\code{cache[c][mark]} = 0} \ast
                \varphi \ast \distAs{\code{ev}}{\Unif{[1, \p{c} - 1] \setminus S}} \ast \sure{\card{[1, \p{c} - 1] \setminus S} = j - 1} \ast \distAs{x}{\Unif{[1, j]}}}\\
\CODE{if x = 1:}\\
            \begin{proofindent}
              \CODE{ev := c}\\
            \end{proofindent}\\
            \ASSR{
              \CMod{\Unif{[1, j]}} w.
              \begin{pmatrix}
                \sure{w \neq 1} \implies \CMod{\Unif{[1, \p{c}-1] \setminus S}} v.\sure{\code{ev} = v}\\
                \sure{w = 1} \implies \sure{\code{ev} = c}
              \end{pmatrix}
              \ast
              \begin{pmatrix}
              \sure{\code{cache[c][mark]} = 0} \\\ast
              \varphi \ast \sure{\card{[1, \p{c} - 1] \setminus S} = j - 1}
              \end{pmatrix}
            }\\
            \ASSR{
              \sure{\code{cache[c][mark]} = 0} \ast
              \varphi \ast \distAs{\code{ev}}{\Unif{[1, \p{c}] \setminus S}} \ast \sure{\card{[1, \p{c}] \setminus S} = j}
            }\TAG[reservoir:end]\\
      \end{proofindent}\\
            \ASSR{
              \begin{pmatrix}
                \sure{\code{cache[c][mark]} = 0} \implies
                \varphi \ast \distAs{\code{ev}}{\Unif{[1, \p{c}] \setminus S}} \ast \sure{\card{[1, \p{c}] \setminus S} = j} \\
                \sure{\code{cache[c][mark]} = 1} \implies
                \varphi \ast \distAs{\code{ev}}{\Unif{[1, \p{c}] \setminus S}} \ast \sure{\card{[1, \p{c}] \setminus S} = \card{[1, \p{c} - 1] \setminus S} = j} \\
              \end{pmatrix}
            } \TAG[reservoir:wp-if-prim]\\
            \ASSR{
                \varphi \ast
                \distAs{\code{ev}}{\Unif{[1, \p{c}] \setminus S}} \ast \sure{\card{[1, \p{c}] \setminus S} = j}
            }\\
    \end{proofindent}\\
    \ASSR{
                \varphi \ast
                \distAs{\code{ev}}{\Unif{[1, k] \setminus S}} \ast \sure{\card{[1, k] \setminus S} = j}
            }
  \end{proofoutline}
  \caption{Reservoir Sampling}
  \label{fig:discrete-reservoir}
\end{mathfig}

We then want to use~\ref{rule:wp-if-prim} to prove~\eqref{reservoir:wp-if-prim}.
The $\sure{\code{cache[c][mark]} = 1}$ branch is trivial.
For the $\sure{\code{cache[c][mark]} = 0}$ branch,
we first apply~\ref{rule:wp-assign} and~\ref{rule:wp-samp} for the assignment and
the sampling. We then apply~\ref{rule:wp-if-unary} and~\ref{rule:wp-assign}
for the $\code{if x = 1}$ branching.
The most crucial step is ~\eqref{reservoir:end}, which we prove in the follow.
For visual presentation,
we abbreviate $\sure{\code{cache[c][mark]} = 0} \ast \varphi$ as $\psi$.
\begin{eqexplain}
&\CMod{\Unif{[1, j]}} w.
              \begin{grp}
                \sure{w \neq 1} \implies \CMod{\Unif{[1, \p{c}-1] \setminus S}} v.\sure{\code{ev} = v}\\
                \sure{w = 1} \implies \sure{\code{ev} = c}
              \end{grp}
              \ast \psi
              \ast \sure{\card{[1, \p{c} - 1] \setminus S} = j - 1}
      \whichproves
\CMod{\Unif{[1, j]}} w.
              \begin{grp}
                \sure{w \neq 1} \implies \CMod{\Unif{[1, \p{c}-1] \setminus S}} v.\sure{\code{ev} = v}\\
                \sure{w = 1} \implies \CMod{\Unif{[1, \p{c}-1] \setminus S}} v. \sure{\code{ev} = \p{c}}
              \end{grp}
              \ast \psi
              \ast \sure{\card{[1, \p{c} - 1] \setminus S} = j - 1}
      \byrules{c-true, c-frame}
      \whichproves
\CMod{\Unif{[1, j]}} w.
              \CMod{\Unif{[1, \p{c}-1] \setminus S}} v. \sure{\code{ev} =
                \ITE{w = 1}{\p{c}}{v} }
              \ast \psi
              \ast \sure{\card{[1, \p{c} - 1] \setminus S} = j - 1}
      \byrule{c-cons}
      \whichproves
\CMod{\Unif{[1, j]} \pprod \Unif{[1, \p{c}-1] \setminus S}} (w, v).  \sure{\code{ev} = \ITE{w = 1}{\p{c}}{v} }
              \ast \psi
              \ast \sure{\card{[1, \p{c} - 1] \setminus S} = j - 1}
      \byrule{c-fuse}
\end{eqexplain}
Because  $\sure{\card{[1, \p{c} - 1] \setminus S} = j - 1}$,
there exists some bijection $h$ between $[1, \p{c} - 1] \setminus S$
and $[1, j-1]$.
For $(x, y) \in [1, j - 1] \times (([1, \p{c} - 1]) \setminus S) \cup \{\p{c}\}$,
we define
\begin{align*}
f(x, y) &= \ITE{y = \p{c}}{1}{x + 1} \\
g(x, y) &= \ITE{y = \p{c}}{h^{-1}(x)}{y}
\end{align*}
Pure reasoning yields that the map
$\langle f, g \rangle : (w, v) \mapsto (f(w,v),g(w,v))$
is a bijection between
$[1, j - 1] \times (([1, \p{c} - 1]) \setminus S) \cup \{\p{c}\}$ and
$[1,j] \times ([1, \p{c} - 1] \setminus S)$.
Furthermore, for any $(x, y)$,
\[
\Unif{[1, j]} \pprod \Unif{[1, \p{c}-1] \setminus S}(f(x), g(y))
=  \frac{1}{j \cdot (j - 1)}
= \Unif{[1, j - 1]} \pprod \Unif{[1, (\p{c}-1] \setminus S) \cup \{c\} }(x, y) ,
\]
Thus, we can apply~\ref{rule:c-transf} with the bijection $\langle f, g \rangle$
to derive that
\begin{align*}
  &\CMod{\Unif{[1, j]} \pprod \Unif{[1, \p{c}-1] \setminus S}} (w, v).
  \sure{\code{ev} = \ITE{w = 1}{\p{c}}{v} }
              \ast \psi
              \ast \sure{\card{[1, \p{c} - 1] \setminus S} = j - 1} \\
  \proves
  &\CMod{\Unif{[1, j - 1]} \pprod \Unif{([1, \p{c}-1] \setminus S) \cup \{c\} }} (x, y).
    \begin{conj}
    \sure{
      \code{ev} =
      \ITE{f(x, y) = 1}{\p{c}}{g(x,y)}
    }\\
    \ast \psi
    \ast \sure{\card{[1, \p{c} - 1] \setminus S} = j - 1}.
    \end{conj}
\end{align*}
Propositional reasoning gives the result that
\begin{align*}
&\ITE{f(x, y) = 1}{\p{c}}{g(x,y)} \\
{}={}
&\ITE{y = \p{c}}{\p{c}}{\left(
    \ITE{y = \p{c}}{h^{-1}(x)}{y}
  \right)}\\
{}={}&  y.
\end{align*}
Thus, we can simplify the assertion above into
\begin{align*}
  &\CMod{\Unif{[1, j - 1]} \pprod \Unif{([1, \p{c}-1] \setminus S) \cup \{c\} }} (x, y).
    \sure{
      \code{ev} = y
    }
    \ast \psi
    \ast \sure{\card{[1, \p{c} - 1] \setminus S} = j - 1}.
\end{align*}
Now recall that $\psi$ is $\sure{\code{cache[c][mark]} = 0} \ast
    \sure{l \in S \iff \code{cache[$l$][mark]} = 1}$,
which implies that the set $([1, \p{c}-1] \setminus S) \cup \{c\}$ is the equivalent to
$([1, \p{c}] \setminus S)$.
Thus, we can further simplify the assertion above into
\begin{align*}
  &\CMod{\Unif{[1, j - 1]} \pprod \Unif{[1, \p{c}] \setminus S}} (x, y).  \sure{\code{ev} = y} \ast \psi
  \ast \sure{\card{[1, \p{c}] \setminus S} = j},
\end{align*}
Then, we apply~\ref{rule:c-sure-proj} to project out the unused part
of the distribution under conditioning:
\begin{align*}
  &\CMod{\Unif{[1, \p{c}] \setminus S}} (x, y).
  \begin{conj}
  \sure{\code{ev} = y}
  \ast \sure{\code{cache[c][mark]} = 0}
  \\
  \ast
    \sure{l \in S \iff \code{cache[$l$][mark]} = 1}  \ast \sure{\card{[1, \p{c}] \setminus S} = j}.
  \end{conj}
\end{align*}
Applying~\ref{rule:sure-str-convex}, we can pull out almost sure assertions
under the conditioning modality and get
\begin{align*}
  &\left( \CMod{\Unif{[1, \p{c}] \setminus S}} (x, y).
  \sure{\code{ev} = y} \right)
  \ast
  \begin{conj}
  \sure{\code{cache[c][mark]} = 0}
  \\
  \ast
    \sure{l \in S \iff \code{cache[$l$][mark]} = 1}  \ast \sure{\card{[1, \p{c}] \setminus S} = j}.
  \end{conj}
\end{align*}
Last, we apply~\ref{rule:c-unit-l} to obtain~\eqref{reservoir:end},
and ~\eqref{reservoir:wp-if-prim} follows from~\ref{rule:wp-if-unary}.

By simplifying~\eqref{reservoir:wp-if-prim},
we close the loop invariant $P(c)$.
Finally, applying the loop rule~\ref{rule:wp-loop}, we get
\begin{align*}
    \sure{l \in S \iff \code{cache[$l$][mark]} = 1} \ast
    \distAs{\code{ev}}{\Unif{[1, \p{c}] \setminus S}} \ast \sure{\card{[1, \p{c}] \setminus S} = j}.
\end{align*}

\subsubsection{Marker algorithm within a phase}
When reasoning about Marker's algorithm, the pen-and-paper proof uses a
reduction-flavored technique. Given the actual sequence of requests
$\code{seq}$, we compute a modified sequence $\code{worseseq}$ by moving
requests that are not in the initial cache, a.k.a., the clean request, ahead of
those that are in the cache, a.k.a., the dirty
requests.
Each clean request incurs cost 1 wherever it appears,
and dirty request incurs cost 1 if the cache slot holding its value
has been evicted in the current phase and incur cost 0 otherwise.
Thus, the clean requests in \code{seq} and \code{worseseq} incur
the same cost, and the dirty requests in \code{worseseq}
incur at least as much cost as in \code{seq}.
This is an inherently relational argument. The rest part of the pen-and-paper
analysis uses a lot of unary reasoning to upper bound the cost of Marker on
$\code{worseseq}$, and concludes that it also upper bounds the cost of processing \code{seq}.
Though we are not able to prove every step in this analysis formally,
it suggests the need to combine unary and n-nary reasoning.
In the following, we verify an important component of the
unary analysis:
when Marker operates on $M$ clean requests
in the beginning of a phase, the $M$ distinct clean requests incurs exactly cost $m$
and the evicted slots are chosen uniformly.

At the beginning of the phase,
it holds that
\[
  \forall 1 \leq l \leq k. \sure{\code{cache[$l$][mark] = 0}}.
\]
We also assume that $\code{cost}$ is $c$ and $\code{cache}$ is $h$ in the beginning
of the phase.
To express that the Marker is handling requests $\code{seq}[i]$ that are not currently in the cache,
we assert that
\[
\forall 1 \leq n \leq M.\  \forall 1 \leq l \leq k. \sure{\code{seq[n]} \neq h\code{[$l$][value]}}.
\]
The cost increases by exactly $M$ only if these $M$ clean requests are distinct,
so we also require
\[
  \forall 1 \leq n' < n \leq M. \sure{\code{seq[n]} \neq \code{seq[n']}}
\]
We combine these into the precondition and outline the proof at~\Cref{fig:cache-clean-phase}.
The assertions throughout should be conjuncted with $\ownall$ with $\land$ and
we omit it for visual simplicity.
For the loop, we use the loop invariant
\[
  P(i) :=
  \CMod{\kappa(i-1)} S.
        \begin{pmatrix}
             \sure{\code{cost} = c + i - 1}
           \ast \left( \forall i \leq n \leq M.\  \forall 1 \leq l \leq k. \sure{\code{seq[n]} \neq h[l][\code{value}]} \right) \\
             {} \ast \left( \forall l. \sure{l \in S \iff \code{cache[$l$][mark]} = 1 \iff \code{cache}[l][\code{value}] \neq h[l][\code{value}]} \right) \\
             {} \ast \left( \forall 1 \leq n' < n \leq M. \sure{\code{seq[n]} \neq \code{seq[n']}} \right)\\
        \end{pmatrix},
\]
where the kernel $\kappa$ is defined such that for any $1 \leq i \leq m$,
$\kappa(i) = \Unif{\{S \subseteq [1, k] \mid \textsf{size}(S) = i\}}$ --
so $\kappa(i- 1) = \Unif{\big\{S \subseteq [1, k] \mid \textsf{size}(S) = i - 1 \big\}}$.
And we write $\sure{A \iff B \iff C}$
as shorthand for $\sure{A \iff B} \ast \sure{B \iff C}$.

\begin{mathfig}[\small]
  \begin{proofoutline}
  \PREC{
        \sure{\code{cost} = c} \ast
        \begin{pmatrix}
        \forall 1 \leq n\leq M.\  \forall 1 \leq l \leq k. \sure{\code{seq[n]} \neq h[l][\code{value}]}\\
        \ast {} \forall 1 \leq l \leq k. \sure{\code{cache[$l$][mark] = 0}}\\
        \forall 1 \leq n' < n \leq M. \sure{\code{seq[n]} \neq \code{seq[n']}}\\
        \end{pmatrix}
      }\\
  \CODE{i := 1}\\
\CODE{repeat M:}\\
  \begin{proofindent}
    \ASSR{P(i)} \TAG[cache:clean:P(i)]\\
\CODE{hit := 0}\\
    \CODE{j := 1}\\
    \ASSR{\CMod{\kappa(i-1)} S.
        \begin{pmatrix}
             \sure{\code{cost} = c + i - 1} \ast \forall i \leq n\leq M.\  \forall 1 \leq l \leq k. \sure{\code{seq[n]} \neq h[l][\code{value}]}\\
             {} \ast \forall l. \sure{l \in S \iff \code{cache[$l$][mark]} = 1 \iff \code{cache}[l][\code{value}] \neq h[l][\code{value}]} \\
             {} \ast \sure{\code{hit} = 0} \ast \sure{\code{j} = 1}
             \ast \forall 1 \leq n' < n \leq M. \sure{\code{seq[n]} \neq \code{seq[n']}}\\
        \end{pmatrix}
      }\\
    \CODE{repeat k:}\\
    \begin{proofindent}
      \ASSR{Q(j)} \TAG[cache:clean:Q(j)]\\
    \CODE{if cache[j][value]  = seq[i]:}\\
      \begin{proofindent}
        \CODE{cache[j][mark] := 1; hit := 1} \\
      \end{proofindent}\\
      \CODE{j := j+1}\\
    \end{proofindent}\\
    \ASSR{
      \CMod{\kappa(i-1)} S.
        \begin{pmatrix}
          \sure{\code{cost} = c + i - 1} \ast \forall i \leq n\leq M.\  \forall 1 \leq l \leq k. \sure{\code{seq[n]} \neq h[l][\code{value}]}\\
             {} \ast \forall l. \sure{l \in S \iff \code{cache[$l$][mark]} = 1 \iff \code{cache}[l][\code{value}] \neq h[l][\code{value}]} \\
             {} \ast \sure{\code{hit} = 0}
             \ast \forall 1 \leq n' < n \leq M. \sure{\code{seq[n]} \neq \code{seq[n']}}\\
        \end{pmatrix}
      } \TAG[cache:clean:6] \\
    \CODE{if $\ \neg$ hit then:} \\
      \begin{proofindent}
        \CODE{ev:~Reservoir-Samp(k, cache)}\\
        \ASSR{
          \CMod{\kappa(i-1)} S.
          \begin{pmatrix}
          \sure{\code{cost} = c + i - 1} \ast \forall i \leq n\leq M.\  \forall 1 \leq l \leq k. \sure{\code{seq[n]} \neq h[l][\code{value}]}\\
             {} \ast \forall l. \sure{l \in S \iff \code{cache[$l$][mark]} = 1 \iff \code{cache}[l][\code{value}] \neq h[l][\code{value}]} \\
             {} \ast \forall 1 \leq n' < n \leq M. \sure{\code{seq[n]} \neq \code{seq[n']}}\\
             {} \ast \sure{\code{hit} = 0} \ast  \distAs{\code{ev}}{\Unif{[1,k] \setminus S}}
          \end{pmatrix}
         }  \TAG[cache:clean:5] \\
        \ASSR{
          \CMod{\kappa(i-1)} S.
          \CMod{\Unif{[1,k] \setminus S}} u.
          \begin{pmatrix}
          \sure{\code{cost} = c + i - 1} \ast \forall i \leq n\leq M.\  \forall 1 \leq l \leq k. \sure{\code{seq[n]} \neq h[l][\code{value}]}\\
             {} \ast \forall l. \sure{l \in S \iff \code{cache[$l$][mark]} = 1 \iff \code{cache}[l][\code{value}] \neq h[l][\code{value}]} \\
             {} \ast \forall 1 \leq n' < n \leq M. \sure{\code{seq[n]} \neq \code{seq[n']}}\\
             {} \ast \sure{\code{hit} = 0} \ast \sure{\code{ev} = u}
          \end{pmatrix}
          }  \TAG[cache:clean:3]\\
\CODE{cache[ev][value] := seq[i]} \\
        \ASSR{
           \CMod{\kappa(i-1)} S.
          \CMod{\Unif{[1,k] \setminus S}} u.
          \begin{pmatrix}
             \sure{\code{cost} = c + i - 1}
             \ast \left(\forall i < n \leq M.\  \forall 1 \leq l \leq k.
           \sure{\code{seq[n]} \neq h[l][\code{value}]}\right)\\
           {} \ast
           \begin{grp}
             \forall l. \sure{l \in S \iff \code{cache[$l$][mark]} = 1}\\
             \ast \sure{l \in S \cup \{u\} \iff \code{cache}[l][\code{value}] \neq h[l][\code{value}]}\\
           \end{grp} \\
             {} \ast \forall 1 \leq n' < n \leq M. \sure{\code{seq[n]} \neq \code{seq[n']}}\\
             {} \ast \sure{\code{hit} = 0} \ast \sure{\code{ev} = u}
          \end{pmatrix}
        } \TAG[cache:clean:7] \\
          \CODE{cache[ev][mark] := 1} \\
        \ASSR{
          \CMod{\kappa(i-1)} S.
          \CMod{\Unif{[1,k] \setminus S}} u.
          \begin{pmatrix}
          \sure{\code{cost} = c + i - 1} \ast \forall i \leq n\leq M.\  \forall 1 \leq l \leq k. \sure{\code{seq[n]} \neq h[l][\code{value}]}\\
             {} \ast \forall l. \sure{l \in (S \cup \{u\}) \iff \code{cache[$l$][mark]} = 1 \iff \code{cache}[l][\code{value}] \neq h[l][\code{value}]} \\
             {} \ast \forall 1 \leq n' < n \leq M. \sure{\code{seq[n]} \neq \code{seq[n']}}\\
             {} \ast \sure{\code{hit} = 0} \ast \sure{\code{ev} = u}
          \end{pmatrix}
        } \TAG[cache:clean:8] \\
          \ASSR{
          \CMod{\kappa(i)} S.
          \begin{pmatrix}
                \sure{\code{cost} = c + i - 1}  \ast
                \forall i + 1 \leq n \leq M.\  \forall 1 \leq l \leq k. \sure{\code{seq[n]} \neq h[l][\code{value}]} \\
                 {} \ast (\forall l. \sure{l \in (S \cup \{u\}) \iff \code{cache[$l$][mark]} = 1 \iff \code{cache}[l][\code{value}] \neq h[l][\code{value}]}) \\
                 {} \ast \forall 1 \leq n' < n \leq M. \sure{\code{seq[n]} \neq \code{seq[n']}}\\
          \end{pmatrix}
          }
          \TAG[cache:clean:4]\\
        \CODE{cost := cost + 1}\\
        \ASSR{
          \CMod{\kappa(i)} S.
          \begin{pmatrix}
            \sure{\code{cost} = c + i }  \ast
            \forall i + 1 \leq n \leq M.\  \forall 1 \leq l \leq k. \sure{\code{seq[n]} \neq h[l][\code{value}]} \\
                 {} \ast (\forall l. \sure{l \in (S \cup \{u\}) \iff \code{cache[$l$][mark]} = 1 \iff \code{cache}[l][\code{value}] \neq h[l][\code{value}]}) \\
                 {} \ast \forall 1 \leq n' < n \leq M. \sure{\code{seq[n]} \neq \code{seq[n']}}\\
          \end{pmatrix}
          }\\
      \end{proofindent} \\
      \CODE{i := i+1}\\
      \end{proofindent}\\
      \ASSR{\CMod{\kappa(M)} S.
        \begin{pmatrix}
             \sure{\code{cost} = c + M}
              \ast \forall 1 \leq n' < n \leq M. \sure{\code{seq[n]} \neq \code{seq[n']}}
             \\
             {} \ast \forall l. \sure*{
             \begin{conj*}
               l \in S \iff \code{h[$l$][mark]} = 1
               \\\iff \code{cache}[l][\code{value}] \neq h[l][\code{value}]
             \end{conj*}
             }
        \end{pmatrix}
      }
  \end{proofoutline}
  \caption{Cache: clean phase}
  \label{fig:cache-clean-phase}
\end{mathfig}

Starting from $P(i)$, we first apply~\ref{rule:wp-assign} twice
and then enter the inner loop with
\begin{align*}
  Q(j) =
      \CMod{\kappa(i-1)} S.
        \begin{pmatrix}
          \sure{\code{cost} = c + i - 1} \ast \left( \forall 1 \leq n\leq M.\  \forall 1 \leq l \leq k. \sure{\code{seq[n]} \neq \code{h[$l$][value]}} \right)\\
             {} \ast \left( \forall l. \sure{l \in S \iff \code{cache[$l$][mark]} = 1 \iff \code{cache}[l][\code{value}] \neq h[l][\code{value}]}\right) \\
             {} \ast \sure{\code{hit} = 0}
             \ast \left(\forall 1 \leq n' < n \leq M. \sure{\code{seq[n]} \neq \code{seq[n']}} \right)\\
        \end{pmatrix}
\end{align*}
To reestablish $Q(j)$ in~\eqref{cache:clean:Q(j)},
we apply~\ref{rule:c-wp-swap} to reason from fixed $S$. Then we
use~\ref{rule:wp-if-prim}:
the condition $ \forall 1 \leq n\leq M.\  \forall 1 \leq l \leq k. \sure{\code{seq[n]} \neq \code{cache[$l$][value]}}$ implies that we do not enter \code{if} block and
the assertion $Q(j)$ still holds afterwards. Last, we
apply~\ref{rule:wp-assign} to derive that $Q(j)$ holds
after the assignment $\code{j := j+1}$. Since $Q(j) = Q(j+1)$,
$Q(j + 1)$ holds as well.
Thus by~\ref{rule:wp-loop} we have~\eqref{cache:clean:6}.

The step~\eqref{cache:clean:5} is derived using~\ref{rule:c-wp-swap}
and the specs of $\code{Reservoir-Samp}$.
For the step~\eqref{cache:clean:3}, we apply~\ref{rule:c-unit-r} to rewrite
$\distAs{\code{ev}}{\Unif{[1,k] \setminus S}}$ as $\CMod{\Unif{[1,k] \setminus
S}} u. \sure{\code{ev} = u}$ and we use~\ref{rule:c-frame} to pull
$\CMod{\Unif{[1,k] \setminus S}} u.$ ahead. In the next step,
using~\ref{rule:wp-assign}, we get
\begin{align*}
  \CMod{\kappa(i-1)} S.
          \CMod{\Unif{[1,k] \setminus S}} u.
          \begin{pmatrix}
             \sure{\code{cost} = c + i - 1}
             \ast \left(\forall i < n \leq M.\  \forall 1 \leq l \leq k.
              \sure{\code{seq[n]} \neq \code{h[$l$][value]}}\right)\\
           {} \ast
           \begin{grp}
             \forall l. \sure{l \in S \iff \code{cache[$l$][mark]} = 1}\\
             \ast \sure{l \in S \land l \neq \code{ev} \implies \code{cache}[l][\code{value}] \neq h[l][\code{value}]}\\
             \ast \sure{ \code{cache}[l][\code{value}] \neq h[l][\code{value}]  \land l \neq \code{ev} \implies l \in S } \\
           \end{grp} \\
             {} \ast \forall 1 \leq n' < n \leq M. \sure{\code{seq[n]} \neq \code{seq[n']}}\\
             {} \ast \sure{\code{hit} = 0} \ast \sure{\code{ev} = u}
          \end{pmatrix}
\end{align*}
With propositional reasoning, we can replace
 $
 \sure{ \code{cache}[l][\code{value}] \neq h[l][\code{value}]  \land l \neq \code{ev} \implies l \in S }
 $
with
$
  \sure{ \code{cache}[l][\code{value}] \neq h[l][\code{value}]  \implies l \in \{ S \cup u\} }
$
under the condition
$\sure{\code{ev} = u}$;
also, because $\forall i < n \leq M.\  \forall 1 \leq l \leq k.
              \sure{\code{seq[n]} \neq \code{h[$l$][value]}}$,
we can strengthen the assertion
$\sure{l \in S \land l \neq \code{ev} \implies \code{cache}[l][\code{value}] \neq h[l][\code{value}]}$
into
$\sure{l \in S \cup \{u\} \implies \code{cache}[l][\code{value}] \neq h[l][\code{value}]}$.
Thus, we have~\eqref{cache:clean:7}.

Similarly, applying~\ref{rule:wp-assign} and using propositional reasoning,
we can derive~\eqref{cache:clean:8}.
For the step~\eqref{cache:clean:4}, we apply~\ref{rule:c-fuse} to combine
$\CMod{\kappa(i - 1)} S. \CMod{\Unif{[1,k] \setminus S}} u.$ into
$\CMod{\kappa(i - 1) \fuse \Unif{[1,k] \setminus S}} (S, u)$.
We then apply~\ref{rule:c-transf} based on the bijection between
$(S, u)$ and $(S \cup \{u\}, u)$ to get that
such that
\begin{align*}
  \CMod{d} (S', u).
          \begin{pmatrix}
          \sure{\code{cost} = c + i - 1}  \ast \forall i + 1 \leq n \leq M.\  \forall 1 \leq l \leq k. \sure{\code{seq[n]} \neq h[l][\code{value]}} \\
                 {} \ast  (\forall l. \sure{l \in S' \iff \code{cache[$l$][mark]} = 1 \iff \code{cache}[l][\code{value}] \neq h[l][\code{value}]}) \\
                 {} \ast \forall 1 \leq n' < n \leq M. \sure{\code{seq[n]} \neq \code{seq[n']}}\\
          \end{pmatrix},
\end{align*}
where
\begin{align*}
  d &= \bind(\kappa(i - 1), \fun S. \bind(\Unif{[1,k] \setminus S}, \fun u. \dirac{(S \cup \{u\},  u)})) \\
    &= \bind(\kappa(i ), \fun S. \bind(\Unif{[1,k] \setminus S}, \fun u. \dirac(S ,  u)))
\end{align*}
Since $u$ is not used,
we could use~\ref{rule:c-sure-proj} to project out $u$ and
derive that
\begin{align*}
  \CMod{\kappa(i)} S'.
          \begin{pmatrix}
                 \sure{\code{cost} = c + i - 1}  \ast  \forall i + 1 \leq n \leq M.\  \forall 1 \leq l \leq k. \sure{\code{seq[n]} \neq h\code{[$l$][value]}} \\
                 {} \ast (\forall l. \sure{l \in S' \iff \code{cache[$l$][mark]} = 1 \iff \code{cache}[l][\code{value}] \neq h[l][\code{value}]}) \\
                 {} \ast \forall 1 \leq n' < n \leq M. \sure{\code{seq[n]} \neq \code{seq[n']}}\\
          \end{pmatrix}
\end{align*}

Then, by~\ref{rule:wp-assign} and~\ref{rule:wp-loop}, we can derive the post-condition
\begin{align*}
  \CMod{\kappa(M)} S.
        \begin{pmatrix}
             \sure{\code{cost} = c + M} \ast \forall 1 \leq n' < n \leq M. \sure{\code{seq[n]} \neq \code{seq[n']}} \\
             {} \ast \forall l.
             \sure*{
               l \in S
               \iff
               \code{cache[$l$][mark]} = 1
               \iff \code{cache}[l][\code{value}] \neq h[l][\code{value}]
             }
        \end{pmatrix}
\end{align*}
Last, with~\ref{rule:sure-str-convex}, we can establish
\begin{align*}
 \sure{\code{cost} = c + M} \ast
  \CMod{\kappa(M)} S.
             \forall l.
             \sure*{
               l \in S
               \iff
               \code{cache[$l$][mark]} = 1
               \iff \code{cache}[l][\code{value}] \neq h[l][\code{value}]
             }
\end{align*}
which states our goal that $\code{cost}$ increases by exactly $M$ and
the set of evicted indices $S$ is picked uniformly.

\end{document}